\documentclass[11pt,english]{article}
\usepackage[T1]{fontenc}
\usepackage[utf8]{inputenc}
\usepackage{xcolor}
\usepackage{babel}
\usepackage{float}
\usepackage{booktabs}
\usepackage{mathrsfs}
\usepackage{mathtools}
\usepackage{amsmath}
\usepackage{amsthm}
\usepackage{amssymb}
\usepackage{graphicx}
\usepackage{geometry}
\usepackage{setspace}
\usepackage{wasysym}
\usepackage[authoryear]{natbib}
\usepackage[pdfusetitle,hypertexnames=false,
 bookmarks=true,bookmarksnumbered=false,bookmarksopen=false,
 breaklinks=true,pdfborder={0 0 1},backref=false,colorlinks=true]
 {hyperref}
\hypersetup{
 pdfborderstyle=,linkcolor=blue,citecolor=blue}

\makeatletter

\providecommand{\tabularnewline}{\\}

\@ifundefined{date}{}{\date{}}
\usepackage{babel}
\allowdisplaybreaks

\usepackage{bbm}

\newtheoremstyle{remboldstyle}
  {}{}{}{}{\bfseries}{.}{.5em}{{\thmname{#1 }}{\thmnumber{#2}}{\thmnote{ (#3)}}}
\theoremstyle{remboldstyle}
\newtheorem{rembold}{Remark}

\providecommand{\algorithmname}{Algorithm}

\usepackage{appendix}
\appendixtitleon

\usepackage{scalerel}
\usepackage{stackengine}\stackMath
\newcommand{\reallywidecheck}[1]{%
\savestack{\tmpbox}{\stretchto{%
  \scaleto{%
    \scalerel*[\widthof{\ensuremath{#1}}]{\kern-.6pt\bigwedge\kern-.6pt}%
    {\rule[-\textheight/2]{1ex}{\textheight}}%WIDTH-LIMITED BIG WEDGE
  }{\textheight}% 
}{0.5ex}}%
\stackon[1pt]{#1}{\scalebox{-1}{\tmpbox}}%
}
\allowdisplaybreaks

\renewenvironment{proof}[1][\proofname]{%
  \par\pushQED{\qed}\normalfont%
  \topsep6\p@\@plus6\p@\relax
  \trivlist\item[\hskip\labelsep\bfseries#1\@addpunct{.}]%
  \ignorespaces
}{%
  \popQED\endtrivlist\@endpefalse
}

\newtheoremstyle{ntnboldstyle}
  {}{}{}{}{\bfseries}{.}{.5em}{{\thmname{#1}}{\thmnote{(#2)}}}
\theoremstyle{ntnboldstyle}

\def\@seccntformat#1{\@ifundefined{#1@cntformat}%
   {\csname the#1\endcsname\quad}  % default
   {\csname #1@cntformat\endcsname}% enable individual control
}
\let\oldappendix\appendix %% save current definition of \appendix
\renewcommand{\appendix}{%
    \oldappendix
    \newcommand{\section@cntformat}{\appendixname~\thesection\quad}
}

\theoremstyle{plain}
\providecommand{\assumptionname}{Assumption}
\providecommand{\lemmaname}{Lemma}
\providecommand{\theoremname}{Theorem}

\theoremstyle{plain}
\providecommand{\algorithmname}{Algorithm}
\newtheorem{assumption}{\protect\assumptionname}\newtheorem{lyxalgorithm}{\protect\algorithmname}

\providecommand{\corollaryname}{Corollary}

\makeatother

\providecommand\theoremname{Theorem}
\theoremstyle{plain}
\newtheorem{thm}{\protect\theoremname}
\providecommand\corollaryname{Corollary}
\newtheorem{cor}{\protect\corollaryname}
\providecommand\lemmaname{Lemma}
\newtheorem{lem}{\protect\lemmaname}

\makeatletter
\let\GRDmaketitle\maketitle
\let\GRDAtmaketitle\@maketitle
\let\GRDthanks\thanks
\let\GRDtitle\title
\let\GRDauthor\author
\let\GRDdate\date
\let\GRDand\and
\makeatother
\begin{document}
\title{\textbf{Generic Covariate Adjustment for }\\
\textbf{Regression Discontinuity Designs\thanks{This version: \today. We thank Matias Cattaneo, Vadim Marmer, Taisuke Otsu and Toru Kitagawa
for their helpful comments. All errors are ours. Jun Ma acknowledges
the financial support from the National Natural Science Foundation
of China (grant numbers 72394392 and 72673194). Zhengfei Yu acknowledges
the financial support from JSPS KAKENHI (grant numbers 21K01419 and
25K05032).}}}
\author{ Jun Ma\thanks{Jun Ma. School of Economics, Renmin University of China. Email: jun.ma@ruc.edu.cn (corresponding author)}
\and Yuya Sasaki\thanks{Yuya Sasaki. Department of Economics, Vanderbilt University.}
\and Zhengfei Yu\thanks{Zhengfei Yu. Graduate School of Economics, The University of Osaka.}}
\date{}
\maketitle
\begin{abstract}
It is standard practice to include covariates in regression discontinuity
designs (RDDs) and regression kink designs (RKDs), but the theoretical
justification for doing so does not generally extend beyond linear
estimands. This paper proposes a novel entropy balancing reweighting
approach for covariate adjustment within a general framework of RDDs
and RKDs. While conventional regression-based covariate adjustment
methods generally fail to deliver consistent estimation for nonlinear
estimands such as quantile treatment effects, our reweighting approach
achieves consistency while improving efficiency. Moreover, even in
settings where the regression-based covariate adjustment method already
improves efficiency, our approach can deliver additional efficiency
gains. Simulation studies corroborate these theoretical findings.
We present an empirical application in which our covariate adjustment
yields statistically significant results that would not be obtained
without covariate adjustment.

\medskip{}
\end{abstract}
%%%%%%%%%%%%%%%%%%%%%%%%%%%%%%%%%%%%%%%%%%%%%%%%%%%%%%%%%%%%%

\section{Introduction}

%%%%%%%%%%%%%%%%%%%%%%%%%%%%%%%%%%%%%%%%%%%%%%%%%%%%%%%%%%%%%
It has long been standard practice to include covariates when estimating
model or causal parameters in empirical research in economics. This
convention is largely motivated by the desire to mitigate omitted
variable bias in regression models. Beyond linear regression, the
practice of including covariates has naturally extended to other econometric
frameworks, including event studies, instrumental variables designs,
randomized controlled trials, and regression discontinuity designs
(RDDs). However, incorporating covariates is not always theoretically
justified in econometric settings beyond correctly specified linear
models. In some cases, it may even lead to undesirable consequences
despite the use of otherwise credible empirical designs. RDDs are
no exception.

This paper develops a general framework for covariate adjustment in
RDDs \citep*{hahn2001identification} and regression kink designs
\citep*[RKDs;][]{Card:2015gc}; {\color{black}see \citet{Cattaneo2022} for a recent overview of the RDD literature.} Rather than simply including covariates
additively in local-to-cutoff regressions, we advocate reweighting
the sample objective to achieve covariate balance. As we show, this
balancing/reweighting approach delivers more robust and more efficient
performance across a broad class of RDDs and RKDs than the regression-based or additive
covariate adjustment method.

Specifically, we demonstrate the theoretical properties of our reweighting-based
covariate adjustment approach in three settings. First, for linear estimands in levels, such as the sharp RDD estimand, our reweighting
method achieves the same first-order efficiency as the regression-based
covariate adjustment method. Second, for nonlinear estimands, such
as the quantile RDD estimand, the conventional regression-based adjustment
approach generally fails to deliver consistent estimation. In contrast,
our reweighting approach yields consistent estimation and improves
efficiency by exploiting covariate information. Third, for derivative
estimation, such as the RKD estimand, the regression-based adjustment
approach delivers consistent estimation and enjoys efficiency gains
by exploiting derivative restrictions on the covariates. Our reweighting
approach further enhances efficiency by additionally exploiting the
level restrictions on the covariates.

%Balancing weights are typically constructed so that the weighted sample moments of the covariates within the treatment and control subsamples match the corresponding unweighted sample moments of the full sample.
In RDDs and RKDs, balancing weights are constructed so that the weighted
sample moments of the covariates coincide on both sides of the cutoff.
Among the many possible balancing weights, we focus on the generalized
entropy balancing (EB) weights,\footnote{For conciseness, we omit the term ``generalized'' hereafter and
simply refer to it as EB.} defined as the weights closest to the uniform weights in the sense
of minimizing the Cressie--Read divergence subject to the balancing
constraints. The Cressie--Read divergence includes as a special case
the relative entropy used by \citet{Hainmueller2012}. We then propose
replacing the uniform weights in the objective function with the EB
weights to achieve the desired properties of covariate adjustment
in a broad class of RDD and RKD frameworks.

While we employ the EB weights, our objective differs fundamentally
from that of \citet{Hainmueller2012}. On the one hand, \citet{Hainmueller2012}
uses EB weights to mitigate selection bias under the unconfoundedness
assumption by balancing covariate moments between the treatment and
control groups. On the other hand, we use EB weights to improve estimation
efficiency. In the context of RDDs and RKDs, a correctly specified
design already accounts for selection bias through the research design
itself. Consequently, correcting for selection bias is not our primary
concern. Instead, we exploit EB weights to improve efficiency while
preserving the identifying assumptions of the design. {\color{black}Indeed, as \citet{Cattaneo2022} note, covariates cannot restore nonparametric identification in RDDs, but they can improve precision.} Our paper thus
provides a complementary contribution to the EB literature.

We are not the first to study the theoretical properties of covariate
adjustment in the RDD framework. In particular, \citet*[CCFT,][]{calonico2019regression}
establish the consistency of covariate-adjusted RDD estimators and
develop valid inference procedures. 
%{\color{red}As \citet{Cattaneo2022} note, covariates cannot restore nonparametric identification in RDDs, but they can improve precision.} 
CCFT's regression-based adjustment
generally improves estimation efficiency. For linear estimands in
levels, such as the sharp RDD estimand, our proposed reweighting approach
is first-order equivalent to the regression-based covariate adjustment
of CCFT. However, the regression-based approach does not naturally
extend to nonlinear estimands such as the quantile RDD estimand, where
it may fail to deliver consistent estimation. In contrast, our proposed
reweighting approach continues to apply in nonlinear settings, delivering
consistent estimation while improving efficiency through the use of
covariate information. In this sense, our paper extends the framework
pioneered by CCFT by providing a generic covariate adjustment method
that applies not only to linear estimands but also to a broad class
of nonlinear RDD/RKD estimands including those studied by \citet*{Babii2023},
\citet*{Chiang2019}, \citet*{Huang2022}, \citet*{Qu2019}, \citet*{Qu2024a},
and \citet*{xu2017regression,Xu2018}.

In addition, as discussed above, our reweighting approach enjoys a
further efficiency advantage in the linear setting when the parameter
of interest involves derivatives, as in \citet{Card:2015gc} and \citet{Dong2015}.
Taken together, these results show that our proposed method provides
a generic framework for covariate adjustment across a broad class
of RDD and RKD models, delivering gains whenever they are available
while never sacrificing first-order efficiency in settings where the
conventional approach is already optimal. The proofs of all theoretical results are collected in the supplement.\footnote{The supplement is available at \url{https://ruc-econ.github.io/generic_RD_supp.pdf}.} \color{black}

\bigskip{}
\noindent \textbf{Notation.} $\sum_{i}$ is understood as $\sum^{n}_{i=1}$.
``$a\coloneqq b$'' means that $a$ is defined by $b$. For any
$k$-times differentiable univariate function $f$, let $f^{\left(k\right)}$
denote the $k$-th order derivative. Let $\mathbbm{1}\left(\cdot\right)$
denote the indicator function. For a $d$-dimensional vector $x$,
let $x^{\top}$ denote its transpose and $\left\Vert x\right\Vert $
denote its Euclidean norm. $x^{\otimes2}$ denotes a vector of the
distinct entries of $x\otimes x$ ($x^{\otimes2}\coloneqq\mathrm{vech}\left(xx^{\top}\right)$,
where $\mathrm{vech}\left(\cdot\right)$ denotes half vectorization).
Let $0_{J}$ denote the $J$-dimensional vector all of whose elements
are zero. For a real-valued function $f:\mathcal{X}\rightarrow\mathbb{R}$,
let $\left\Vert f\right\Vert _{\infty}\coloneqq\mathrm{sup}_{x\in\mathcal{X}}\left|f\left(x\right)\right|$
denote the sup-norm. $\ell^{\infty}[a,b]$ denotes
the family of bounded real-valued functions on $[a,b]$. Let
$\mathrm{e}_{k,j}$ denote the $j$-th unit vector in $\mathbb{R}^{k}$.
Let $\left[a\pm b\right]$ denote the interval $\left[a-b,a+b\right]$.
For $N\in\mathbb{N}$, $\left[N\right]$ denotes $\left\{ 1,\ldots,N\right\} $. Let $\lceil x\rceil$ denote the smallest integer greater than or equal to $x\in \mathbb{R}$.
The standard normal probability density function (PDF) and cumulative
distribution function (CDF) are denoted by $\phi$ and $\varPhi$,
respectively. Let $\rightsquigarrow$ denote convergence in distribution
in the general sense \citep[Section 18.2]{VanDerVaart1998}.

Let $X\in\mathbb{R}$ be a continuous score supported on $\left[\underline{x},\overline{x}\right]$.
Let $f_{X}$ denote its density function. Without loss of generality,
we normalize the cutoff to zero, so that $0\in\left(\underline{x},\overline{x}\right)$.
Denote $\varphi\coloneqq f_{X}\left(0\right)$ for simplicity. For
a random vector (or matrix) $A$, denote $\mu_{A}\left(x\right)\coloneqq\mathrm{E}\left[A\mid X=x\right]$
and let $\mu^{\left(k\right)}_{A}\left(\cdot\right)$ be its $k$-th
derivative. Denote $\mu^{\left(k\right)}_{A,-}\coloneqq\lim_{x\uparrow0}\mu^{\left(k\right)}_{A}\left(x\right)$
and let $\mu^{\left(k\right)}_{A,+}$ be defined similarly with $\lim_{x\uparrow0}$
replaced by $\lim_{x\downarrow0}$. For simplicity, also denote $\mu_{A,+},\mu_{A,-}\coloneqq\mu^{\left(0\right)}_{A,+},\mu^{\left(0\right)}_{A,-}$.
Let $\mu_{A,\pm}\coloneqq\mu_{A,+}+\mu_{A,-}$. For random vectors
$A$ and $A'$, $\mathrm{Var}_{+}\left[A\right]$ is understood as
$\lim_{x\downarrow0}\mathrm{Var}\left[A\mid X=x\right]$ and $\mathrm{Cov}_{+}\left[A,A'\right]$
is understood as $\lim_{x\downarrow0}\mathrm{Cov}\left[A,A'\mid X=x\right]$.
$\mathrm{Var}_{-}\left[A\right]$ and $\mathrm{Cov}_{-}\left[A,A'\right]$
are defined analogously. Also, for notational simplicity, let $\mathrm{Var}_{\pm}\left[A\right]\coloneqq\mathrm{Var}_{+}\left[A\right]+\mathrm{Var}_{-}\left[A\right]$
and $\mathrm{Cov}_{\pm}\left[A,A'\right]\coloneqq\mathrm{Cov}_{+}\left[A,A'\right]+\mathrm{Cov}_{-}\left[A,A'\right]$.
$\mathrm{Var}_{0}\left[A\right]$ and $\mathrm{Cov}_{0}\left[A,A'\right]$
are understood as $\mathrm{Var}\left[A\mid X=0\right]$ and $\mathrm{Cov}\left[A,A'\mid X=0\right]$.
Let $\mathrm{supp}\left(A\right)$ and $\mathrm{supp}\left(A\mid X=x\right)$
denote the unconditional and conditional supports of $A$. Quantile
RDD objects carry the argument $\tau$ and the argument-free symbol
always denotes the mean RDD object.

%%%%%%%%%%%%%%%%%%%%%%%%%%%%%%%%%%%%%%%%%%%%%%%%%%%%%%%%%%%%%

\section{Generic covariate adjustment framework}

%%%%%%%%%%%%%%%%%%%%%%%%%%%%%%%%%%%%%%%%%%%%%%%%%%%%%%%%%%%%%

%In observational studies, the balancing weights satisfy the requirement that the weighted sample moments of the covariates within the control (treatment) group match the unweighted sample moments of the covariates of all units.  Among all balancing weights, \citet{Hainmueller2012} defines the entropy balancing (EB) weights as those being as close as possible to the uniform weights in the sense of minimal relative entropy, i.e., the Kullback-Leibler (KL) divergence. \citeauthor{Hainmueller2012} replaces the uniform weights in simple sample means with the EB weights. \yuya{I have moved the above paragraph to intro.}

Let $Y$ denote the outcome variable and $D\in\left\{ 0,1\right\} $
the binary treatment indicator. Let $d_{z}$ denote the number of
observed predetermined covariates collected in the vector $Z$. The
variables in $Z$ can be continuous, discrete or mixed. We observe
$\left(Y,D,Z\right)$ and the score $X$. Let the data $\left\{ \left(Y_{i},D_{i},X_{i},Z_{i}\right)\right\} ^{n}_{i=1}$
be i.i.d. copies of $\left(Y,D,X,Z\right)$. In a sharp RDD model,
the treatment is assigned if $X\geq0$. Let $\left(Y\left(0\right),Y\left(1\right)\right)$
be the potential outcomes with or without treatment. Likewise, let
$\left(Z\left(0\right),Z\left(1\right)\right)$ be potential covariates.
The assumption $Z\left(1\right)=Z\left(0\right)=Z$ formalizes the
condition that the covariates in $Z$ are predetermined. If both $\mu_{Y\left(1\right)}\left(\cdot\right)$
and $\mu_{Y\left(0\right)}\left(\cdot\right)$ are continuous at 0,
then the conditional average treatment effect for individuals with
zero score $\mathrm{E}\left[Y\left(1\right)-Y\left(0\right)\mid X=0\right]$
is identified as $\vartheta\coloneqq\mu_{Y,+}-\mu_{Y,-}$, the (sharp)
RDD estimand. Consider the following assumption.

\begin{assumption} \label{assu:smoothness Z}(i) $\mu_{Z}\left(\cdot\right)$
is continuous at 0. (ii) $f_{X}$ is continuous at 0 and $\varphi\coloneqq f_{X}\left(0\right)>0$.
\end{assumption}

Assumption \ref{assu:smoothness Z}(i) imposes the local version of
the standard covariate balance condition, which can be represented
as $\mu_{Z,+}=\mu_{Z,-}$. %The two continuity requirements are testable implications of the identifying assumptions of the regression discontinuity design. 
Identification relies on the continuity of the conditional expectations
of the potential outcomes at the cutoff \citep{hahn2001identification},
a hypothesis that is itself untestable. As \citet{lee2008randomized}
argues, however, this continuity has an observable counterpart: when
agents cannot precisely manipulate the score, units just above and
just below the cutoff are comparable, so the distribution of the predetermined
covariates, and in particular its mean $\mu_{Z}\left(\cdot\right)$,
varies continuously through the cutoff, as does the density $f_{X}$.
Continuity of $\mu_{Z}\left(\cdot\right)$ and of $f_{X}$ is thus
often considered an implication of a valid design, rather than an
additional restriction placed upon it. This is precisely why predetermined
covariates and the score density serve as the routine falsification
checks of RDD validity: because a valid design implies that they exhibit
no discontinuity at the cutoff, $\mu_{Z,+}=\mu_{Z,-}$ is the null
hypothesis of the standard covariate-balance tests \citep{lee_lemieux_2010_RD,Cattaneo2019},
while the continuity of $f_{X}$ is the null of the density (manipulation)
test \citep{mccrary2008manipulation}. Assumption \ref{assu:smoothness Z},
therefore, incorporates the empirically verifiable content of the
RDD validity, and its plausibility can be assessed directly from the
data. Because verifying covariate balance and the smoothness of the
score density is already standard practice in credible RDD analysis,
our covariate adjustment method demands nothing that the empirical
researcher is not expected to check anyway. It imposes no assumptions
beyond those a conventional RDD study must already maintain.

%%%%%%%%%%%%%%%%%%%%%%%%%%%%%%%%%%%%%%%%%%%%%%%%%%%%%%%%%%%%%

\subsection{Balancing weights}

%%%%%%%%%%%%%%%%%%%%%%%%%%%%%%%%%%%%%%%%%%%%%%%%%%%%%%%%%%%%%
Now we elaborate on the EB and reweighting approach to covariate adjustment
in general RDD frameworks. Let $K\left(\cdot\right)$ denote a kernel
function satisfying the following assumption.

\begin{assumption} \label{assu:kernel}$K\left(\cdot\right)$ is
a symmetric continuous PDF supported on $\left[-1,1\right]$. \end{assumption}

Let $h$ denote the bandwidth. We assume that $h=h_{n}$ decreases
with the sample size $n$. For notational simplicity, we suppress
the dependence of $h$ on $n$. We drop the subscript $i$ when we
refer to population-level counterparts. Let $p\geq1$ be the integer-valued
local polynomial (LP) order. Let $r_{p}\left(t\right)\coloneqq\left(1,t,\ldots,t^{p}\right)^{\top}$. Denote 
\begin{equation}
\Pi_{+,p}\left(h\right)\coloneqq\frac{1}{nh}\sum_{i}r_{p}\left(\frac{X_{i}}{h}\right)r^{\top}_{p}\left(\frac{X_{i}}{h}\right)K\left(\frac{X_{i}}{h}\right)I_{i},\label{eq:PI -}
\end{equation}
where $I_{i}\coloneqq\mathbbm{1}\left(X_{i}>0\right)$. Let $\Pi_{-,p}\left(h\right)$
be defined similarly by the right-hand side of (\ref{eq:PI -}) with
$I_{i}$ replaced by $\mathbbm{1}\left(X_{i}<0\right)$. Let 
\begin{equation}
W_{+,p,i}\left(h\right)\coloneqq\mathrm{e}^{\top}_{p+1,1}\Pi^{-1}_{+,p}\left(h\right)r_{p}\left(\frac{X_{i}}{h}\right)K\left(\frac{X_{i}}{h}\right)I_{i}.\label{eq:regression weight}
\end{equation}
Let $W_{-,p,i}\left(h\right)$ be defined similarly by the right-hand
side of (\ref{eq:regression weight}) with $I_{i}$ and $\Pi_{+,p}\left(h\right)$
replaced by $\mathbbm{1}\left(X_{i}<0\right)$ and $\Pi_{-,p}\left(h\right)$,
respectively. Let $W_{p,i}\left(h\right)\coloneqq W_{+,p,i}\left(h\right)-W_{-,p,i}\left(h\right)$.

Denote $\bar{Z}_{i}\coloneqq\left(1,Z^{\top}_{i}\right)^{\top}$ and
$V_{p,i}\left(h\right)\coloneqq W_{p,i}\left(h\right)\bar{Z}_{i}$.
The EB approach looks for balancing weights closest to the uniform
weights, where ``closeness'' is measured by the Cressie--Read (CR)
divergence 
\[
\mathit{CR}_{\varrho,n}\left(w_{1},\ldots,w_{n}\right)\coloneqq\frac{1}{n\varrho\left(1+\varrho\right)}\sum_{i}\left\{ \left(nw_{i}\right)^{-\varrho}-1\right\} 
\]
from $\left(w_{1},\ldots,w_{n}\right)$ to the uniform weights $\left(1/n,\ldots,1/n\right)$,
indexed by a divergence parameter $\text{\ensuremath{\varrho}}\in\mathbb{R}$.
We define EB weights $\left(\widehat{w}_{p,1}\left(h\right),\ldots,\widehat{w}_{p,n}\left(h\right)\right)$
as the solution to the following minimum relative entropy
problem: 
\begin{align}
 & \underset{w_{1},\ldots,w_{n}}{\min}\,\mathit{CR}_{\varrho,n}\left(w_{1},\ldots,w_{n}\right)\nonumber \\
 & \textrm{subject to }\sum_{i}w_{i}V_{p,i}\left(h\right)=0_{d_{z}+1},\ensuremath{\sum_{i}w_{i}=1}\textrm{ and }w_{i}\geq0,\textrm{ }\forall i.\label{eq:generalized balancing}
\end{align}
%The objective function in \eqref{eq:generalized balancing} is the CR divergence from $\left(w_{1},...,w_{n}\right)$ to the uniform weights $\left(1/n,...,1/n\right)$. 
When $\varrho=0$, the objective function, defined as the limit as
$\varrho\rightarrow0$, is $-\sum_{i}\mathrm{log}\left(n\cdot w_{i}\right)/n$,
the Kullback--Leibler (KL) divergence from the uniform weights to $\left(w_{1},\ldots,w_{n}\right)$. Similarly, when $\varrho=-1$,
the objective function is the KL divergence $\sum_{i}w_{i}\mathrm{log}\left(n\cdot w_{i}\right)$
from $\left(w_{1},\ldots,w_{n}\right)$ to the uniform weights,
and corresponds to the relative entropy in \citet{Hainmueller2012}.

By solving the minimization problem \eqref{eq:generalized balancing},
we find the set of weights with the least CR divergence from the uniform
weights among all balancing weights. The uniform weights satisfy the
important finite-sample property $\sum_{i}W_{p,i}\left(h\right)=0$.
Imposing $\sum_{i}w_{i}\begin{array}{c}
W_{p,i}\left(h\right)\end{array}=0$ in the constraint of \eqref{eq:generalized balancing} ensures that
the balancing weights satisfy the same property. The balancing weights
should also satisfy $\sum_{i}w_{i}\begin{array}{c}
W_{+,p,i}\left(h\right)Z_{i}\end{array}=\sum_{i}w_{i}\begin{array}{c}
W_{-,p,i}\left(h\right)Z_{i}\end{array}$. This requires that the (kernel-weighted) local averages of the covariates
on both sides of the cutoff coincide in finite samples under the new
weights.

%In the setup of \citet{Hainmueller2012}, the population moments of covariates in the control or treatment group may not be the same as the unconditional population moments, since the treatment status is not independent from the covariates. \citet{Hainmueller2012} proposed reweighting using EB weights under which some intentionally misspecified (biased) balancing constraints to correct for the selection bias. In our case, the balancing constraints are correctly specified and our reweighting aims at enhancing efficiency. The smoothness assumption in Assumption \ref{assu:smoothness Z} implies a restriction $\mu_{Z,+}=\mu_{Z,-}$ on the population distribution of the observed covariates. The EB weights explicitly exploit such information from the covariates. \yuya{The above paragraph should be a remark.}

We use strong duality and concentrate out $\left(w_{1},\ldots,w_{n}\right)$
and the Lagrange multiplier for $\ensuremath{\sum_{i}w_{i}=1}$. We
obtain the concentrated Lagrangian dual problem 
\begin{eqnarray}
 &  & \max_{\lambda}\frac{1}{\varrho(1+\varrho)}\left\{ \left(\frac{1}{n}\sum_{i}(1+V^{\top}_{p,i}\left(h\right)\lambda)^{\varrho/(1+\varrho)}\right)^{1+\varrho}-1\right\} \label{eq:dual}\nonumber \\
 &  & \textrm{subject to }1+V^{\top}_{p,i}\left(h\right)\lambda>0\textrm{ for all }i=1,\ldots,n.
\end{eqnarray}
The dual characterization of the optimal weights is given by 
\begin{equation}
\widehat{w}_{p,i}\left(h\right)=\frac{(1+V^{\top}_{p,i}\left(h\right)\widehat{\lambda}_{p}\left(h\right))^{-1/(1+\varrho)}}{\sum_{j}(1+V^{\top}_{p,j}\left(h\right)\widehat{\lambda}_{p}\left(h\right))^{-1/(1+\varrho)}},\,i=1,2,\ldots,n,\label{eq:balancing weights}
\end{equation}
where $\widehat{\lambda}_{p}\left(h\right)$ denotes the optimizer
of the convex optimization problem (\ref{eq:dual}). Computing the
EB weights requires solving the well-understood convex problem \eqref{eq:dual},
e.g., by Newton's method. The domain of its objective function is
the convex set $\left\{ \lambda:1+V^{\top}_{p,i}\left(h\right)\lambda>0\textrm{ for all }i\right\} $.
The algorithm should either take these positivity constraints into
account or use a modified objective function defined for all $\lambda$.

When $\varrho=-1$, the concentrated dual problem can be simplified
to 
\[
\min_{\lambda}\,\sum_{i}\exp\bigl(-V^{\top}_{p,i}\left(h\right)\lambda\bigr)
\]
with no positivity constraint, and in this case, the optimal weights
are given by 
\[
\frac{\exp(-V^{\top}_{p,i}\left(h\right)\widehat{\lambda}_{p}\left(h\right))}{\sum_{j}\exp(-V^{\top}_{p,j}\left(h\right)\widehat{\lambda}_{p}\left(h\right))},\text{ for }i=1,2,\ldots,n.
\]

When $\varrho=-2$, the CR divergence is proportional to the square
of the Euclidean distance between $\left(w_{1},\ldots,w_{n}\right)$
and the uniform weights. Problem \eqref{eq:dual} is then a quadratic
program. Let 
\begin{equation}
\widetilde{\lambda}_{p}\left(h\right)\coloneqq-\left(\sum_{i}V_{p,i}\left(h\right)V^{\top}_{p,i}\left(h\right)\right)^{-1}\left(\sum_{i}V_{p,i}\left(h\right)\right)\label{eq:unconstrained multiplier}
\end{equation}
be the closed-form optimizer of the unconstrained version. By plugging
$\widetilde{\lambda}_{p}\left(h\right)$ into the right-hand side
of \eqref{eq:balancing weights}, we obtain balancing weights 
\begin{equation}
\widetilde{w}_{p,i}\left(h\right)\coloneqq\frac{1+V^{\top}_{p,i}\left(h\right)\widetilde{\lambda}_{p}\left(h\right)}{\sum_{j}(1+V^{\top}_{p,j}\left(h\right)\widetilde{\lambda}_{p}\left(h\right))},\,i=1,2,\ldots,n,\label{eq:w_til balancing weights}
\end{equation}
that do not require numerical optimization. The weights in \eqref{eq:w_til balancing weights}
satisfy the positivity constraints with probability approaching one.

\begin{rembold}For implementation, we recommend
that practitioners select the bandwidth $h$ first, following the
detailed guidelines in the \hyperref[app:bdw]{Appendix}. The formula for the
asymptotic mean squared error (AMSE)-optimal bandwidth incorporates
the covariate information, but the pilot estimates of the unknown
quantities entering it do not require the EB weights. Bandwidth and
weights can therefore be determined sequentially. Given the selected
$h$, the EB weights are obtained either by solving the dual problem
\eqref{eq:dual} and substituting its optimizer $\widehat{\lambda}_{p}\left(h\right)$
into \eqref{eq:balancing weights}, or, when $\varrho=-2$, by evaluating
the closed form \eqref{eq:w_til balancing weights}.\end{rembold}

%%%%%%%%%%%%%%%%%%%%%%%%%%%%%%%%%%%%%%%%%%%%%%%%%%%%%%%%%%%%%

\subsection{Our proposed approach}

\label{subsec:proposed_approach} %%%%%%%%%%%%%%%%%%%%%%%%%%%%%%%%%%%%%%%%%%%%%%%%%%%%%%%%%%%%%

RDD estimators generally correspond to the regression coefficient of $I_{i}$
in optimization problems of the form 
\begin{equation}
\underset{b_{0},b_{1}}{\min}\sum_{i}K\left(\frac{X_{i}}{h}\right)l\left(Y_{i}-r^{\top}_{p}\left(X_{i}\right)b_{0}-I_{i}\cdot r^{\top}_{p}\left(X_{i}\right)b_{1}\right),\label{eq:general_objective_uniform_weights}
\end{equation}
where $l\left(\cdot\right)$ is a loss function. For example, $l\left(\cdot\right)$ is the squared loss
for the mean RDD (as in Section \ref{subsec:cov_rd}), while it is
the check loss for the quantile RDD (as in Section \ref{subsec:Sharp-quantile-RD}).
To incorporate covariates within this general framework, we propose
replacing the uniform weights in \eqref{eq:general_objective_uniform_weights}
with the EB weights $\widehat{w}_{p,i}\left(h\right)$ given by \eqref{eq:balancing weights}.
Specifically, we consider 
\begin{equation}
\underset{b_{0},b_{1}}{\min}\sum_{i}\widehat{w}_{p,i}\left(h\right)K\left(\frac{X_{i}}{h}\right)l\left(Y_{i}-r^{\top}_{p}\left(X_{i}\right)b_{0}-I_{i}\cdot r^{\top}_{p}\left(X_{i}\right)b_{1}\right).\label{eq:general_objective_balancing_weights}
\end{equation}
Alternatively, the weights $\widehat{w}_{p,i}\left(h\right)$ in \eqref{eq:general_objective_balancing_weights}
may be replaced by $\widetilde{w}_{p,i}\left(h\right)$ given in \eqref{eq:w_til balancing weights}.

As discussed below, this reweighting can improve estimation efficiency.
The magnitude of the efficiency gain relative to existing regression-based
covariate adjustment methods depends on the specific setting. We distinguish
three main cases. First, for linear estimands in levels, such as the
sharp RDD estimand, the regression-based covariate adjustment achieves
the same first-order asymptotic efficiency as our proposed reweighting
approach; see Section \ref{subsec:cov_rd}.

Second, for nonlinear estimands, such as the quantile RDD estimand,
a naïve extension of the regression-based covariate adjustment approach
generally fails to deliver consistent estimation. In contrast, our
proposed reweighting approach not only yields consistent estimation,
but also improves efficiency by exploiting covariate information;
see Section \ref{subsec:Sharp-quantile-RD}.

Third, for linear estimands in derivatives, such as the RKD estimand,
the regression-based covariate adjustment approach yields consistent
estimation and can improve efficiency by exploiting derivative restrictions
on the covariates. Specifically, the regression-based
covariate adjustment solves the optimization problem \eqref{eq:CCFT estimator}
and takes the coefficient of $I_{i}X_{i}$. However, it does not exploit
the information contained in the level restrictions. In contrast,
by reweighting \eqref{eq:CCFT estimator} in Section \ref{subsec:cov_rd}
below with the EB weights $\widehat{w}_{p,i}\left(h\right)$ in \eqref{eq:balancing weights}
or $\widetilde{w}_{p,i}\left(h\right)$ in \eqref{eq:w_til balancing weights},
our proposed approach further enhances efficiency by utilizing both
the level and derivative restrictions; see Section \ref{subsec:Sharp-regression-kink}.

In summary, our proposed procedure delivers efficiency gains for nonlinear
estimands and in settings where the estimand involves derivatives.
At the same time, it incurs no efficiency loss relative to the regression-based
adjustment for linear estimands in levels. Consequently, it provides
a generically applicable framework for covariate adjustment across
RDD and RKD settings.

\begin{rembold}\citet{Hainmueller2012} introduced the EB weights
to mitigate selection bias by balancing covariate moments between
the treatment and control groups. While we employ the same weighting
device, our objective is fundamentally different. Because RDDs and
RKDs address selection through the research design itself, we do not
use EB weights for selection bias correction. Instead, as summarized
above, we use them to improve estimation efficiency.\end{rembold}

%\yuya{I have inserted this small subsection above that summarizes our general proposal. Without this subsection, readers would be lost about what our main estimator is.}

%%%%%%%%%%%%%%%%%%%%%%%%%%%%%%%%%%%%%%%%%%%%%%%%%%%%%%%%%%%%%

\section{Linear estimands in levels: sharp RDD}

\label{subsec:cov_rd} %%%%%%%%%%%%%%%%%%%%%%%%%%%%%%%%%%%%%%%%%%%%%%%%%%%%%%%%%%%%%
In this section, we compare our proposed covariate adjustment approach
with the existing regression-based approach in the sharp RDD. We show
that, in this setting, the regression-based approach attains the same
first-order efficiency as our proposed method. %To fix ideas, we focus throughout this section on the standard sharp RDD with covariates.
\color{black}

Let $B\coloneqq\left(Y,Z^{\top}\right)^{\top}$. The following assumption
is imposed on the population distribution of the observed variables.
Let $\mathbb{B}\subseteq\left[\underline{x},\overline{x}\right]$
denote a neighborhood of 0 and $\mathbb{B}_{0}\coloneqq\mathbb{B}\setminus\left\{ 0\right\} $.
Let $f_{Z\mid X}\left(\cdot\mid x\right)$ denote the conditional
density of $Z$ given $X=x$ with respect to some dominating measure
$m$.

\begin{assumption} \label{assu:data generating process}(i) $\mu_{B}\left(\cdot\right)$
is $\left(p+1\right)$-times continuously differentiable on either side of $\mathbb{B}_{0}$
and $\mu^{\left(p+1\right)}_{B}\left(\cdot\right)$ is uniformly continuous
on either side of $\mathbb{B}_{0}$; (ii) $\mu_{B^{\otimes2}}$ is uniformly continuous
on either side of $\mathbb{B}_{0}$; (iii) $\mathrm{Var}_{+}\left[B\right]$ and
$\mathrm{Var}_{-}\left[B\right]$ are positive definite; (iv) there
exists some envelope $g\left(\cdot\right)$ such that $\sup_{x\in\mathbb{B}_{0}}f_{Z\mid X}\left(\cdot\mid x\right)\leq g$
and $\int\left\Vert z\right\Vert ^{2+\delta}g\left(z\right)m\left(\mathrm{d}z\right)<\infty$
for some $\delta\in\left(0,1\right)$. The same condition holds for
$f_{Y\mid X}$. \end{assumption}

Parts (i)--(iii) of Assumption \ref{assu:data generating process}
are standard regularity conditions and will be directly invoked in
the proofs. These assumptions are satisfied under suitable conditions
imposed on the population distribution of the latent variables. The
smoothness level in (i) is similar to that commonly assumed in the
literature, i.e., the minimal smoothness level ($p+1$) such that
the leading smoothing bias term of the estimator (using a $p$-th
order LP) can be explicitly characterized. Let $\bar{B}\coloneqq\left(Y\left(1\right),Y\left(0\right),Z\right)$.
Part (i) is satisfied if $\mu_{\bar{B}}\left(\cdot\right)$ is $\left(p+1\right)$-times
continuously differentiable on $\mathbb{B}$ with uniformly continuous
derivatives. Part (ii) is satisfied if $\mu_{\bar{B}{}^{\otimes2}}\left(\cdot\right)$
is uniformly continuous on $\mathbb{B}$. Existence of $\mathrm{Var}_{+}\left[B\right]$
and $\mathrm{Var}_{-}\left[B\right]$ is guaranteed under these assumptions.
$\mathrm{Var}_{+}\left[B\right]$ (or $\mathrm{Var}_{-}\left[B\right]$)
is guaranteed to be positive definite if $\mathrm{Var}\left[\left(Y\left(1\right),Z^{\top}\right)^{\top}\mid X=0\right]$
(or $\mathrm{Var}\left[\left(Y\left(0\right),Z^{\top}\right)^{\top}\mid X=0\right]$)
is positive definite. Assumption \ref{assu:data generating process}(iv)
is required for the linearization of the Lagrange multiplier $\widehat{\lambda}_{p}\left(h\right)$.
We note that this condition is not required for the unconstrained
multiplier in \eqref{eq:unconstrained multiplier}, which has an exact
linear representation.

%%%%%%%%%%%%%%%%%%%%%%%%%%%%%%%%%%%%%%%%%%%%%%%%%%%%%%%%%%%%%

\subsection{Augmented regression approach}

%%%%%%%%%%%%%%%%%%%%%%%%%%%%%%%%%%%%%%%%%%%%%%%%%%%%%%%%%%%%%

The simple LP estimator without covariates is given by the regression
coefficient of $I_{i}$ in 
\begin{equation}
\underset{b_{0},b_{1}}{\min}\sum_{i}K\left(\frac{X_{i}}{h}\right)\left\{ Y_{i}-r^{\top}_{p}\left(X_{i}\right)b_{0}-I_{i}\cdot r^{\top}_{p}\left(X_{i}\right)b_{1}\right\} ^{2},\label{eq:LP without covariates}
\end{equation}
which has the closed form $\left(nh\right)^{-1}\sum_{i}W_{p,i}\left(h\right)Y_{i}$.
CCFT take an augmented regression approach to incorporate the predetermined
covariates. Their covariate-adjusted estimator $\widehat{\vartheta}^{\mathit{CCFT}}_{p}\left(h\right)$
for $\vartheta$ is given by the regression coefficient of $I_{i}$
in 
\begin{equation}
\underset{b_{0},b_{1},b_{2}}{\min}\sum_{i}K\left(\frac{X_{i}}{h}\right)\left\{ Y_{i}-r^{\top}_{p}\left(X_{i}\right)b_{0}-I_{i}\cdot r^{\top}_{p}\left(X_{i}\right)b_{1}-Z^{\top}_{i}b_{2}\right\} ^{2}.\label{eq:CCFT estimator}
\end{equation}

Denote $\Lambda_{+,p}\coloneqq\int^{1}_{0}r_{p}\left(t\right)r^{\top}_{p}\left(t\right)K\left(t\right)\mathrm{d}t$.
For $\nu=0,1,\dots,p$, define the $\nu$-th order ``equivalent kernels''
by $\mathcal{K}_{+,p,\nu}\left(t\right)\coloneqq\mathrm{e}^{\top}_{p+1,\nu+1}\Lambda^{-1}_{+,p}r_{p}\left(t\right)K\left(t\right)$,
and let $\left(\Lambda_{-,p},\mathcal{K}_{-,p,\nu}\right)$ be defined
by the same equations with the integral range $\left[0,1\right]$
replaced by $\left[-1,0\right]$. Assumption \ref{assu:kernel} implies
that $\mathcal{K}_{+,p,\nu}\left(t\right)=\left(-1\right)^{\nu}\mathcal{K}_{-,p,\nu}\left(-t\right)$.
Denote $\omega^{j,k}_{+,p,\nu}\coloneqq\int^{1}_{0}t^{j}\mathcal{K}^{k}_{+,p,\nu}\left(t\right)\mathrm{d}t$
and $\omega^{j,k}_{-,p,\nu}\coloneqq\int^{0}_{-1}t^{j}\mathcal{K}^{k}_{-,p,\nu}\left(t\right)\mathrm{d}t$.
By straightforward algebra, we have $\omega^{j,k}_{-,p,\nu}=\left(-1\right)^{\nu k+j}\omega^{j,k}_{+,p,\nu}$.
Let $\omega^{j,k}_{p,\nu}$ denote the common value if $\omega^{j,k}_{+,p,\nu}=\omega^{j,k}_{-,p,\nu}$
($\nu k+j$ is even), and let it denote $\omega^{j,k}_{+,p,\nu}$
if $\omega^{j,k}_{+,p,\nu}=-\omega^{j,k}_{-,p,\nu}$ ($\nu k+j$ is
odd).

Let 
\begin{equation}
\gamma\coloneqq\left(\mathrm{Var}_{\pm}\left[Z\right]\right)^{-1}\mathrm{Cov}_{\pm}\left[Z,Y\right]\mbox{ and }\sigma^{2}\coloneqq\mathrm{Var}_{\pm}\left[\epsilon\right],\label{eq:gamma sigma definition}
\end{equation}
where $\epsilon\coloneqq Y-Z^{\top}\gamma$. Note that 
\[
\sigma^{2}=\sum_{k\in\left\{ 0,1\right\} }\mathrm{Var}_{0}\left[Y\left(k\right)-Z\left(k\right)^{\top}\gamma\right]\textrm{ and }\gamma=\left(\sum_{k\in\left\{ 0,1\right\} }\mathrm{Var}_{0}\left[Z\left(k\right)\right]\right)^{-1}\left(\sum_{k\in\left\{ 0,1\right\} }\mathrm{Cov}_{0}\left[Z\left(k\right),Y\left(k\right)\right]\right),
\]
under the smoothness conditions on $\mu_{\bar{B}}\left(\cdot\right)$
and $\mu_{\bar{B}{}^{\otimes2}}\left(\cdot\right)$. Moreover, since
$\mathrm{Var}_{\pm}\left[Z^{\top}\gamma\right]=\mathrm{Cov}_{\pm}\left[Y,Z^{\top}\gamma\right]$,
we have 
\begin{equation}
\sigma^{2}=\mathrm{Var}_{\pm}\left[Y\right]-\gamma^{\top}\mathrm{Var}_{\pm}\left[Z\right]\gamma\leq\mathrm{Var}_{\pm}\left[Y\right]=\sum_{k\in\left\{ 0,1\right\} }\mathrm{Var}_{0}\left[Y\left(k\right)\right].\label{eq:CCFT efficiency gain}
\end{equation}
By the partitioned regression argument, we have $\widehat{\vartheta}^{\mathit{CCFT}}_{p}\left(h\right)=\left(nh\right)^{-1}\sum_{i}W_{p,i}\left(h\right)\left(Y_{i}-Z^{\top}_{i}\widehat{\gamma}_{p}\left(h\right)\right)$,
where $\widehat{\gamma}_{p}\left(h\right)$ is a consistent estimator
of $\gamma$. CCFT show that 
\begin{equation}
\sqrt{nh}\left(\widehat{\vartheta}^{\mathit{CCFT}}_{p}\left(h\right)-\vartheta-\mathscr{B}^{\mathit{CCFT}}_{p}h^{p+1}\right)\rightarrow_{d}\mathrm{N}\left(0,\mathscr{V}^{\mathit{CCFT}}_{p}\right),\label{eq:CCFT first order}
\end{equation}
where 
\[
\mathscr{B}^{\mathit{CCFT}}_{p}\coloneqq\frac{\mu^{\left(p+1\right)}_{\epsilon,+}\omega^{p+1,1}_{+,p,0}-\mu^{\left(p+1\right)}_{\epsilon,-}\omega^{p+1,1}_{-,p,0}}{\left(p+1\right)!}\textrm{ and }\mathscr{V}^{\mathit{CCFT}}_{p}\coloneqq\frac{\omega^{0,2}_{p,0}\sigma^{2}}{\varphi}
\]
under Assumptions \ref{assu:smoothness Z}, \ref{assu:kernel}, and
\ref{assu:data generating process}, and the condition that the bandwidth
satisfies $nh^{2p+3}=O\left(1\right)$ and $nh\rightarrow\infty$.
Note that the simple LP estimator without covariates has an asymptotic
variance given by $\omega^{0,2}_{p,0}\mathrm{Var}_{\pm}\left[Y\right]/\varphi\geq\mathscr{V}^{\mathit{CCFT}}_{p}$
(see, e.g., \citealp{imbens2011optimal}).

\begin{rembold}CCFT note that there is no definite ranking between their estimator and the standard LP estimator without covariates, and remark that their efficiency results ``are in perfect agreement with those in
the literature on analysis of experiments obtained using Neyman's repeated sampling (Freedman, 2008; Lin, 2013), where it is also found that incorporating covariates in randomized controlled trials using
linear regression leads to efficiency gains only under particular
assumptions'' (CCFT, Section IV.B). As the RDD is often viewed as local randomization,
the asymptotic efficiency gain and CCFT's remark from the perspective
of randomized experiments can be reconciled. Assumption \ref{assu:smoothness Z}(ii) implies that the RDD is analogous to a randomized experiment with equal assignment probabilities. \citet[Theorem 5.2(iv)]{Negi2014} show that in this case, the pooled regression adjustment, whose algorithm is analogous to that of the CCFT estimator, always leads to an asymptotic variance no larger than that
of the unadjusted estimator.\end{rembold}

We also note that including a covariate will not change the asymptotic
variance if and only if the corresponding element in $\gamma$ is
zero. Consider the partition $Z=\left(Z^{\top}_{1},Z^{\top}_{2}\right)^{\top}$
of $Z$ and let $\gamma=\left(\gamma^{\top}_{1},\gamma^{\top}_{2}\right)^{\top}$
be the conformable partition of $\gamma$ such that the dimension
of $\gamma_{j}$ coincides with that of $Z_{j}$, $j=1,2$.
A direct calculation shows that $\mathscr{V}^{\mathit{CCFT}}_{p}$
is equal to the asymptotic variance of the covariate-adjusted estimator
using only $Z_{1}$ if and only if $\gamma_{2}=0$. In this case,
$Z_{2}$ is irrelevant in the sense that dropping $Z_{2}$ has no
first-order impact: it neither leads to efficiency loss nor changes
the asymptotic smoothing bias. In conclusion, if we say that an estimator
achieves an efficiency gain when its asymptotic variance is smaller
than that of the standard estimator without covariates, then the CCFT
estimator achieves an efficiency gain as long as at least one element
of $\gamma$ is nonzero.

%%%%%%%%%%%%%%%%%%%%%%%%%%%%%%%%%%%%%%%%%%%%%%%%%%%%%%%%%%%%%

\subsection{Reweighting approach}

%%%%%%%%%%%%%%%%%%%%%%%%%%%%%%%%%%%%%%%%%%%%%%%%%%%%%%%%%%%%%
Thus far, we have reviewed the CCFT estimator. We now turn to our
reweighting approach introduced in Section \ref{subsec:proposed_approach}.
\color{black} Specifically, we replace the uniform weights in \eqref{eq:LP without covariates}
by the EB weights $\widehat{w}_{p,i}\left(h\right)$ given in \eqref{eq:balancing weights}.
Let our EB reweighting estimator $\widehat{\vartheta}^{\mathit{eb}}_{p}\left(h\right)$
be the regression coefficient of $I_{i}$ in 
\begin{equation}
\underset{b_{0},b_{1}}{\min}\sum_{i}\widehat{w}_{p,i}\left(h\right)K\left(\frac{X_{i}}{h}\right)\left\{ Y_{i}-r^{\top}_{p}\left(X_{i}\right)b_{0}-I_{i}\cdot r^{\top}_{p}\left(X_{i}\right)b_{1}\right\} ^{2}.\label{eq:entropy balancing estimator definition}
\end{equation}
Alternatively, one may replace $\widehat{w}_{p,i}\left(h\right)$
in \eqref{eq:entropy balancing estimator definition} by $\widetilde{w}_{p,i}\left(h\right)$
defined by \eqref{eq:w_til balancing weights}. In this case, the
EB reweighting estimator can be written exactly as $\left(nh\right)^{-1}\sum_{i}W_{p,i}\left(h\right)\left(Y_{i}-Z^{\top}_{i}\widetilde{\gamma}_{p}\left(h\right)\right)$,
where $\widetilde{\gamma}_{p}\left(h\right)$ is some consistent estimator
of $\gamma$. 

\begin{rembold}For expositional simplicity, all
theorems are stated for the EB reweighting estimator constructed from
$\widehat{w}_{p,i}\left(h\right)$. They apply verbatim to the estimator
constructed from $\widetilde{w}_{p,i}\left(h\right)$, which is the
$\varrho=-2$ member of the family with the nonnegativity restrictions
dropped. For each fixed $\varrho\neq-1$, Lemma \ref{lem:reg weights lambda} in Section \ref{sec:Proof-of-Theorem 1} of the supplement
gives $\widehat{\lambda}_{p}\left(h\right)=-(1+\varrho)\widetilde{\lambda}_{p}\left(h\right)+o_{p}\left(\left(nh\right)^{-1/2}\right)$.
Substitution into \eqref{eq:balancing weights} yields a common first-order expansion for the weighting schemes.\end{rembold}

We can state the following relationship between the CCFT estimator
and ours.
\begin{thm}
\label{thm:normality}Suppose that Assumptions \ref{assu:smoothness Z}--\ref{assu:data generating process}
hold. Assume that the bandwidth satisfies $nh^{2p+3}=O\left(1\right)$
and $nh\rightarrow\infty$. Then, $\widehat{\vartheta}^{\mathit{eb}}_{p}\left(h\right)$
is first-order equivalent to $\widehat{\vartheta}^{\mathit{CCFT}}_{p}\left(h\right)$: $\widehat{\vartheta}^{\mathit{eb}}_p(h) - \widehat{\vartheta}^{\mathit{CCFT}}_p(h) = o_p((nh)^{-1/2})$.
In particular, (\ref{eq:CCFT first order}) also holds for $\widehat{\vartheta}^{\mathit{eb}}_{p}\left(h\right)$. 
\end{thm}
The efficiency comparison in \eqref{eq:CCFT efficiency gain} parallels
the insight of \citet{Hirano:2003cz}, whose analysis accounts for
the puzzling phenomenon that the (efficient) inverse probability weighting
estimator using the nonparametrically estimated propensity score has
a smaller asymptotic variance than the one using the true propensity
score. In the missing data example of \citet[Section 3]{Hirano:2003cz}
with a binary covariate, the efficient estimator can be written as
a reweighted version of the estimator based on the true selection
probability (propensity score). Although the latter uses the correct
probability, it does not exploit the information contained in such
knowledge. Reweighting the observations to enforce balance therefore
reduces the asymptotic variance. Our EB estimator operates through
an analogous mechanism: starting from the weights of the unadjusted
estimator, it reweights observations so that the sample counterpart
of the local covariate balance restriction is satisfied. Theorem \ref{thm:normality}
shows that the resulting estimator is first-order equivalent to CCFT's
estimator. Thus, in both settings, the efficiency gain arises because
reweighting uses valid balance information that the original estimator
leaves unexploited.

\begin{rembold}The estimator $\widehat{\vartheta}^{\mathit{eb}}_{p}\left(h\right)$
can be written as a generalized empirical likelihood estimator. The
moment restrictions are the first-order conditions from \eqref{eq:LP without covariates}
augmented by the balance restrictions as in \eqref{eq:generalized balancing}.
The balance restrictions serve as overidentifying restrictions, which
naturally give rise to efficiency gains. Theorem \ref{thm:normality}
shows that EB estimators with different $\varrho$ are all first-order
equivalent. The empirical likelihood estimator, which is also the
EB estimator with $\varrho=0$, has more favorable higher-order bias
properties (see, e.g., \citealp{newey_smith_2004_higher}). \end{rembold}

%\section{Applications}

%%%%%%%%%%%%%%%%%%%%%%%%%%%%%%%%%%%%%%%%%%%%%%%%%%%%%%%%%%%%

\section{Nonlinear estimands: quantile RDD}

\label{subsec:Sharp-quantile-RD} %%%%%%%%%%%%%%%%%%%%%%%%%%%%%%%%%%%%%%%%%%%%%%%%%%%%%%%%%%%%%
In this section, we compare our proposed covariate adjustment approach
with the regression-based approach in the context of the sharp quantile
RDD. We show that, in this setting, a naïve extension of the regression-based
approach may fail to deliver even consistent estimation. In contrast,
our reweighting approach remains consistent and can improve efficiency
by exploiting covariate information. %To fix ideas, we focus throughout this section on the sharp quantile RDD with covariates.

For a scalar random variable $A$, let $F_{A\mid X}\left(\cdot\mid x\right)$,
$f_{A\mid X}\left(\cdot\mid x\right)$, and $Q_{A\mid X}\left(\cdot\mid x\right)$
denote the conditional CDF, PDF, and quantile function, respectively,
of $A$ given $X=x$, and let $Q^{\left(k\right)}_{A\mid X}\left(\tau\mid x\right)\coloneqq\left(\partial/\partial x\right)^{k}Q_{A\mid X}\left(\tau\mid x\right)$.
Let $\kappa_{A,+}\left(\tau\right)\coloneqq\lim_{x\downarrow0}Q_{A\mid X}\left(\tau\mid x\right)$
and $\kappa^{\left(k\right)}_{A,+}\left(\tau\right)\coloneqq\lim_{x\downarrow0}Q^{\left(k\right)}_{A\mid X}\left(\tau\mid x\right)$,
and let $\kappa_{A,-}\left(\tau\right)$ and $\kappa^{\left(k\right)}_{A,-}\left(\tau\right)$
be defined similarly with $x\uparrow0$. For the outcome $Y$, let
$\varphi_{+}\left(\tau\right)\coloneqq\lim_{x\downarrow0}f_{Y\mid X}\left(Q_{Y\mid X}\left(\tau\mid x\right)\mid x\right)$
and let $\varphi_{-}\left(\tau\right)$ be defined similarly. We are
interested in estimating the $\tau$-th quantile RDD (QRD) estimand
\[
\vartheta\left(\tau\right)\coloneqq\kappa_{Y,+}\left(\tau\right)-\kappa_{Y,-}\left(\tau\right).
\]

To identify a treatment effect at the cutoff in quantile terms, let
$F_{Y(d)\mid X}(\cdot\mid x)$ and $Q_{Y(d)\mid X}(\cdot\mid x)$
denote the conditional CDF and the conditional quantile function, respectively,
of the potential outcome $Y(d)$ given $X=x$. Define the $\tau$-th
quantile treatment effect (QTE) at $X=0$ as $\vartheta_{\mathit{qte}}(\tau)\coloneqq Q_{Y(1)\mid X}(\tau\mid0)-Q_{Y(0)\mid X}(\tau\mid0)$,
the horizontal distance between the $\tau$-quantiles of the two conditional
potential-outcome distributions at the cutoff. The identifying condition
is the quantile analogue of the continuity requirement of \citet{hahn2001identification}
for the mean RDD: for every $y$ and each $d\in\{0,1\}$, $x\mapsto F_{Y(d)\mid X}(y\mid x)$
is continuous at $x=0$, or equivalently, $x\mapsto Q_{Y(d)\mid X}(\tau\mid x)$
is continuous at $0$ for every $\tau$. This is a smoothness restriction
across the cutoff that, like its mean counterpart, is untestable.
Under the sharp assignment, for $x>0$
one has $Q_{Y\mid X}(\tau\mid x)=Q_{Y(1)\mid X}(\tau\mid x)$, and
continuity gives $\kappa_{Y,+}(\tau)=Q_{Y(1)\mid X}(\tau\mid0)$. Similarly,
$\kappa_{Y,-}(\tau)=Q_{Y(0)\mid X}(\tau\mid0)$. Therefore, $\vartheta(\tau)=\vartheta_{\mathit{qte}}(\tau)$:
the QRD estimand $\vartheta\left(\tau\right)$ identifies the $\tau$-th
QTE at the cutoff \citep{Qu2019,chiang2019robust}.

Let $\rho_{\tau}\left(u\right)\coloneqq u\left(\tau-\mathbbm{1}\left(u<0\right)\right)$
be the $\tau$-th check function. The $p$-th order local polynomial
estimator $\widehat{\vartheta}^{\mathit{QY}}_{p}\left(\tau\mid h\right)$
of \citet{Qu2019} is given by the regression coefficient of $I_{i}$
in the local quantile regression problem 
\begin{equation}
\underset{b_{0},b_{1}}{\min}\sum_{i}K\left(\frac{X_{i}}{h}\right)\rho_{\tau}\left(Y_{i}-r^{\top}_{p}\left(X_{i}\right)b_{0}-I_{i}\cdot r^{\top}_{p}\left(X_{i}\right)b_{1}\right).\label{eq:quantile_RDD_QY}
\end{equation}

\citet{Qu2019} show that 
\[
\sqrt{nh}\left(\widehat{\vartheta}^{\mathit{QY}}_{p}\left(\tau\mid h\right)-\vartheta\left(\tau\right)-h^{p+1}\mathscr{B}^{\mathit{QY}}_{p}\left(\tau\right)\right)\rightarrow_{d}\mathrm{N}\left(0,\mathscr{V}^{\mathit{QY}}_{p}\left(\tau\right)\right),
\]
where 
\begin{eqnarray*}
\mathscr{B}^{\mathit{QY}}_{p}\left(\tau\right) & \coloneqq & \frac{\omega^{p+1,1}_{+,p,0}\kappa^{\left(p+1\right)}_{Y,+}\left(\tau\right)}{\left(p+1\right)!}-\frac{\omega^{p+1,1}_{-,p,0}\kappa^{\left(p+1\right)}_{Y,-}\left(\tau\right)}{\left(p+1\right)!}\\
\mathscr{V}^{\mathit{QY}}_{p}\left(\tau\right) & \coloneqq & \frac{\omega^{0,2}_{p,0}}{\varphi}\left(\frac{\tau\left(1-\tau\right)}{\varphi^{2}_{+}\left(\tau\right)}+\frac{\tau\left(1-\tau\right)}{\varphi^{2}_{-}\left(\tau\right)}\right).
\end{eqnarray*}

%%%%%%%%%%%%%%%%%%%%%%%%%%%%%%%%%%%%%%%%%%%%%%%%%%%%%%%%%%%%%

\subsection{Augmented regression approach}

%%%%%%%%%%%%%%%%%%%%%%%%%%%%%%%%%%%%%%%%%%%%%%%%%%%%%%%%%%%%%

An algorithmic extension of CCFT's covariate adjustment method is
to solve the augmented local quantile regression 
\begin{equation}
\underset{b_{0},b_{1},b_{2}}{\min}\sum_{i}K\left(\frac{X_{i}}{h}\right)\rho_{\tau}\left(Y_{i}-r^{\top}_{p}\left(X_{i}\right)b_{0}-I_{i}\cdot r^{\top}_{p}\left(X_{i}\right)b_{1}-Z^{\top}_{i}b_{2}\right).\label{eq:CCFT algorithmic extension}
\end{equation}
We now argue that this algorithmic extension of CCFT may fail to produce
even a consistent estimator of $\vartheta\left(\tau\right)$.

Consider the partial minimization problem of \eqref{eq:CCFT algorithmic extension}
for fixed $b_{2}$: 
\begin{equation}
H_{\tau}\left(b_{2}\right)\coloneqq\underset{b_{0},b_{1}}{\min}\sum_{i}K\left(\frac{X_{i}}{h}\right)\rho_{\tau}\left(Y_{i}-Z^{\top}_{i}b_{2}-r^{\top}_{p}\left(X_{i}\right)b_{0}-I_{i}\cdot r^{\top}_{p}\left(X_{i}\right)b_{1}\right).\label{eq:partial minimization}
\end{equation}
The regression coefficients of $Z_{i}$ in (\ref{eq:CCFT algorithmic extension})
can be written as $\arg\min_{b_{2}}H_{\tau}\left(b_{2}\right)$. For
fixed $b_{2}$, it can be shown that $(nh)^{-1}H_{\tau}\left(b_{2}\right)$
converges in probability to $\bar{H}_{\tau}\left(b_{2}\right)$ (up
to a constant), where 
\begin{eqnarray*}
\bar{H}_{\tau}\left(b_{2}\right) & \coloneqq & \underset{x\downarrow0}{\lim}\,\mathrm{E}\left[\rho_{\tau}\left(\left(Y-Z^{\top}b_{2}\right)-\kappa_{Y-Z^{\top}b_{2},+}\left(\tau\right)\right)\mid X=x\right]\\
 &  & +\underset{x\uparrow0}{\lim}\,\mathrm{E}\left[\rho_{\tau}\left(\left(Y-Z^{\top}b_{2}\right)-\kappa_{Y-Z^{\top}b_{2},-}\left(\tau\right)\right)\mid X=x\right].
\end{eqnarray*}
Therefore, the regression coefficients of $Z_{i}$ converge in probability
to $\bar{\gamma}\left(\tau\right)\coloneqq\arg\min_{b_{2}}\bar{H}_{\tau}\left(b_{2}\right)$.
It follows that the probability limit of the regression coefficient
of $I_{i}$ in \eqref{eq:CCFT algorithmic extension} is $\kappa_{Y-Z^{\top}\bar{\gamma}\left(\tau\right),+}\left(\tau\right)-\kappa_{Y-Z^{\top}\bar{\gamma}\left(\tau\right),-}\left(\tau\right)$,
which may not equal the QRD estimand.

We use a numerical example to illustrate the inconsistency of the
algorithmic extension defined by \eqref{eq:CCFT algorithmic extension}.
Under data-generating process (DGP) I, described in Sections~\ref{sec:simu_design} and \ref{sec:simu_qrd},
the QRD estimand specializes to $\vartheta(\tau)=1.49+(\sqrt{11.8}-1)\varPhi^{-1}(\tau)$,
and %Since $(u_{y},u_{z})$ is bivariate normal, for each fixed
%$b_{2}\in\mathbb{R}$ the one-sided limits of the conditional distribution of
%$Y-b_{2}Z$ at the cutoff are Gaussian. 
%Hence,
%\[
%\kappa_{Y-b_{2}Z,s}(\tau)=\alpha_{s}+(\beta_{s}-b_{2})m_{z,s}\left(0\right)
%+\sigma_{s}(b_{2})\varPhi^{-1}(\tau),\,s\in\left\{ -,+\right\} ,
%\]
%where
%$\sigma^{2}_{s}(b_{2})\coloneqq(\beta_{s}-b_{2})^{2}+2\rho(\beta_{s}-b_{2})+1$.
%It follows that
%\[
%\bar{H}_{\tau}(b_{2})=\phi\left(\varPhi^{-1}(\tau)\right)
%\left\{\sigma_{+}(b_{2})+\sigma_{-}(b_{2})\right\},
%\]
%where $\phi$ denotes the standard normal density. 
%Then the minimizer $\bar{\gamma}(\tau)$ does not depend on $\tau$. 
%Writing
%$\sigma_{s}(b_{2})=\{(\beta_{s}-b_{2}+\rho)^{2}+1-\rho^{2}\}^{1/2}$, the
%first-order condition leads to
%$\beta_{+}-b_{2}+\rho=-\left(\beta_{-}-b_{2}+\rho\right)$, so that
%\[
%\bar{\gamma}(\tau)=\frac{\beta_{-}+\beta_{+}}{2}+\rho=1.8,\textrm{ and }\sigma_{+}(\bar{\gamma}\left(\tau\right))=\sigma_{-}(\bar{\gamma}\left(\tau\right))=\sqrt{3.16}.
%\]
the probability limit of the coefficient of $I_{i}$ in \eqref{eq:CCFT algorithmic extension}
equals $1.49$ at every $\tau$. Therefore, the asymptotic bias, defined
as the probability limit minus $\vartheta(\tau)$, is $-2.435\times\varPhi^{-1}(\tau)$
and is reported in the first row of Table~\ref{tab:asymptotic bias}.
The asymptotic bias vanishes at the median but grows in the tails.
The remaining rows of Table~\ref{tab:asymptotic bias} report simulated
finite-sample biases with $nh$ fixed at $2{,}500$. As the sample
size increases (bandwidth decreases), the simulated bias approaches
its asymptotic counterpart. When the bandwidth is relatively large
($h=0.5$), the bias away from the median is smaller in magnitude than
its asymptotic value but remains substantial.

\begin{table}[!t]
\centering \caption{Asymptotic and finite-sample biases of the algorithmic extension \eqref{eq:CCFT algorithmic extension},
DGP I in Sections~\ref{sec:simu_design} and \ref{sec:simu_qrd}.}
\label{tab:asymptotic bias} %
\begin{tabular}{lccccc}
\toprule 
 & $\tau=0.10$  & $\tau=0.25$  & $\tau=0.50$  & $\tau=0.75$  & $\tau=0.90$ \tabularnewline
\midrule 
Asymptotic bias  & 3.121  & 1.642  & 0.000  & $-1.642$  & $-3.121$ \tabularnewline
\midrule 
\multicolumn{6}{l}{Finite-sample bias}\tabularnewline
$n=50{,}000$, $h=0.05$  & 2.735  & 1.442  & $-0.001$  & $-1.433$  & $-2.701$ \tabularnewline
$n=25{,}000$, $h=0.10$  & 2.372  & 1.250  & $-0.005$  & $-1.250$  & $-2.379$ \tabularnewline
$n=5{,}000$, $h=0.50$  & 1.106  & 0.610  & 0.055  & $-0.491$  & $-0.996$ \tabularnewline
\bottomrule
\end{tabular}
\end{table}

%%%%%%%%%%%%%%%%%%%%%%%%%%%%%%%%%%%%%%%%%%%%%%%%%%%%%%%%%%%%%

\subsection{Reweighting approach}

\label{subsec:quantile_RDD_reweighting} %%%%%%%%%%%%%%%%%%%%%%%%%%%%%%%%%%%%%%%%%%%%%%%%%%%%%%%%%%%%%

We now turn to our proposed reweighting approach presented in Section
\ref{subsec:proposed_approach}. Specifically, we replace the uniform
weights in \eqref{eq:quantile_RDD_QY} by the EB weights $\widehat{w}_{p,i}\left(h\right)$ in \eqref{eq:balancing weights}.
The QRD estimator with balancing weights $\widehat{\vartheta}^{\mathit{eb}}_{p}\left(\tau\mid h\right)$
is defined as the regression coefficient of $I_{i}$ in 
\begin{equation}
\underset{b_{0},b_{1}}{\min}\sum_{i}\widehat{w}_{p,i}\left(h\right)K\left(\frac{X_{i}}{h}\right)\rho_{\tau}\left(Y_{i}-r^{\top}_{p}\left(X_{i}\right)b_{0}-I_{i}\cdot r^{\top}_{p}\left(X_{i}\right)b_{1}\right).\label{QRD estimator bw}
\end{equation}
Alternatively, $\widehat{w}_{p,i}\left(h\right)$ can be replaced by
the closed-form weights $\widetilde{w}_{p,i}\left(h\right)$ in \eqref{eq:w_til balancing weights}. Consider the following
assumption.

\begin{assumption} \label{assu:QRD smoothness}(i) $Q_{Y\mid X}\left(\cdot\mid\cdot\right)$
is uniformly continuous on either side of $\left[\underline{\tau},\overline{\tau}\right]\times\mathbb{B}_{0}$,
for fixed $0<\underline{\tau}<\overline{\tau}<1$. (ii) $Q_{Y\mid X}\left(\tau\mid\cdot\right)$
is $\left(p+1\right)$-times continuously differentiable on either side of $\mathbb{B}_{0}$ and $Q^{\left(p+1\right)}_{Y\mid X}\left(\cdot\mid\cdot\right)$
is uniformly continuous on either side of $\left[\underline{\tau},\overline{\tau}\right]\times\mathbb{B}_{0}$.
(iii) $f_{Y\mid X}\left(\cdot\mid\cdot\right)$ is uniformly continuous
on either side of $\left\{ \left(y,x\right):y\in\mathrm{supp}\left(Y\mid X=x\right),x\in\mathbb{B}_{0}\right\} $.
(iv) There exists $\varepsilon>0$ such that $f_{Y\mid X}\left(y\mid x\right)$
is bounded away from zero, uniformly over $y\in\left[Q_{Y\mid X}\left(\tau\mid x\right)-\varepsilon,Q_{Y\mid X}\left(\tau\mid x\right)+\varepsilon\right]$,
$\tau\in\left[\underline{\tau},\overline{\tau}\right]$, and $x\in\mathbb{B}_{0}$;
in particular, $\inf_{\tau\in\left[\underline{\tau},\overline{\tau}\right]}\varphi_{s}\left(\tau\right)>0$
for $s\in\left\{ -,+\right\} $. (v) For each $z\in\mathrm{supp}\left(Z\right)$, $F_{Y\mid XZ}\left(\cdot\mid\cdot,z\right)$ (the conditional CDF of $Y$ given $X$ and $Z$) is uniformly continuous
on either side of $\mathbb{R}\times\mathbb{B}_{0}$. (vi) For each $z\in\mathrm{supp}\left(Z\right)$,
$f_{Z\mid X}\left(z\mid\cdot\right)$ is uniformly continuous on either side of $\mathbb{B}_{0}$.
\end{assumption} Assumption \ref{assu:QRD smoothness}(i)--(iv)
are standard regularity conditions. Similar assumptions are also found
in \citet{Qu2019}. Under Assumption \ref{assu:QRD smoothness}(i)
and (iii), $\tau\mapsto\varphi_{s}\left(\tau\right)$ is continuous
on $\left[\underline{\tau},\overline{\tau}\right]$. Denote $U_{i}\left(\tau\right)\coloneqq\tau-\mathbbm{1}\left(Y_{i}\leq Q_{Y\mid X}\left(\tau\mid X_{i}\right)\right)$.
Let 
\begin{eqnarray*}
\gamma\left(\tau\right) & \coloneqq & \left(\mathrm{Var}_{\pm}\left[Z\right]\right)^{-1}\left(\frac{\mu_{ZU\left(\tau\right),+}}{\varphi_{+}\left(\tau\right)}+\frac{\mu_{ZU\left(\tau\right),-}}{\varphi_{-}\left(\tau\right)}\right)\\
\mathscr{B}_{\star,p}\left(\tau\right) & \coloneqq & \mathscr{B}^{\mathit{QY}}_{p}\left(\tau\right)-\left(\frac{\mu^{\left(p+1\right)}_{Z,+}\omega^{p+1,1}_{+,p,0}}{\left(p+1\right)!}-\frac{\mu^{\left(p+1\right)}_{Z,-}\omega^{p+1,1}_{-,p,0}}{\left(p+1\right)!}\right)^{\top}\gamma\left(\tau\right).
\end{eqnarray*}
Note that the smoothness assumptions in Assumption \ref{assu:QRD smoothness}(i),
(v), and (vi) guarantee that $\tau\mapsto\mu_{ZU\left(\tau\right),s}$
is continuous on $\left[\underline{\tau},\overline{\tau}\right]$
for $s\in\left\{ -,+\right\} $. Define the covariance kernel 
\begin{multline}
\mathscr{C}_{\star,p}\left(\tau,\tau'\right)\coloneqq\\
\frac{\omega^{0,2}_{p,0}}{\varphi}\sum_{s\in\left\{ +,-\right\} }\left\{ \frac{\min\left\{ \tau,\tau'\right\} -\tau\tau'}{\varphi_{s}\left(\tau\right)\varphi_{s}\left(\tau'\right)}-\frac{\mu^{\top}_{ZU\left(\tau\right),s}\gamma\left(\tau'\right)}{\varphi_{s}\left(\tau\right)}-\frac{\gamma^{\top}\left(\tau\right)\mu_{ZU\left(\tau'\right),s}}{\varphi_{s}\left(\tau'\right)}+\gamma^{\top}\left(\tau\right)\mathrm{Var}_{s}\left[Z\right]\gamma\left(\tau'\right)\right\} ,\label{eq:QRD covariance kernel}
\end{multline}
and the asymptotic variance 
\begin{equation}
\mathscr{V}_{\star,p}\left(\tau\right)\coloneqq\mathscr{C}_{\star,p}\left(\tau,\tau\right)=\mathscr{V}^{\mathit{QY}}_{p}\left(\tau\right)-\frac{\omega^{0,2}_{p,0}\gamma^{\top}\left(\tau\right)\mathrm{Var}_{\pm}\left[Z\right]\gamma\left(\tau\right)}{\varphi}.\label{eq:QRD asymptotic variance}
\end{equation}
Denote 
\[
\mathbb{Z}^{\mathit{eb}}_{p}\left(\tau\mid h\right)\coloneqq\sqrt{nh}\left(\widehat{\vartheta}^{\mathit{eb}}_{p}\left(\tau\mid h\right)-\vartheta\left(\tau\right)-h^{p+1}\mathscr{B}_{\star,p}\left(\tau\right)\right).
\]
The following result establishes the asymptotic distribution of the reweighted estimator jointly over the quantile range. It shows asymptotic Gaussianity and provides the foundation for uniform inference.

\begin{thm}
\label{thm:sharp QRD}Suppose that Assumptions \ref{assu:smoothness Z}--\ref{assu:QRD smoothness}
hold. Assume that the bandwidth satisfies $nh^{2p+3}=O\left(1\right)$,
$nh\rightarrow\infty$, and $\log\left(n\right)/\sqrt{nh}\rightarrow0$.
Then, $\mathbb{Z}^{\mathit{eb}}_{p}\left(\cdot\mid h\right)\rightsquigarrow\mathbb{Z}_{p}$
in $\ell^{\infty}\left[\underline{\tau},\overline{\tau}\right]$,
where $\left\{ \mathbb{Z}_{p}\left(\tau\right):\tau\in\left[\underline{\tau},\overline{\tau}\right]\right\} $
is a centered tight Gaussian process with covariance kernel $\mathscr{C}_{\star,p}\left(\cdot,\cdot\right)$
given by (\ref{eq:QRD covariance kernel}). In particular, for each
fixed $\tau\in\left[\underline{\tau},\overline{\tau}\right]$, $\mathbb{Z}^{\mathit{eb}}_{p}\left(\tau\mid h\right)\rightarrow_{d}\mathrm{N}\left(0,\mathscr{V}_{\star,p}\left(\tau\right)\right)$,
with $\mathscr{V}_{\star,p}\left(\tau\right)$ given by (\ref{eq:QRD asymptotic variance}). 
\end{thm}
% {\color{red}[Yuya: shouldn't the covariance kernel be $\mathscr{C}_\ast\left(\cdot,\cdot\right)$ rather than $\mathscr{C}\left(\cdot,\cdot\right)$> Similarly, $\mathscr{V}\left(\tau\right)$ should change to $\mathscr{V}_\ast\left(\tau\right)$ in the last part of the theorem.]}

The second term in \eqref{eq:QRD asymptotic variance} represents
the reduction in asymptotic variance due to the covariate information.
Figure \ref{figure:qrd_variance} illustrates this efficiency gain
for $\widehat{\vartheta}^{\mathit{eb}}_{p}\left(\tau\mid h\right)$
under DGP I in Sections~\ref{sec:simu_design} and \ref{sec:simu_qrd}.

\begin{figure}[!t]
\caption{Asymptotic variance for estimating $\vartheta\left(\tau\right)$ with
and without covariate information}
\centering{}\label{figure:qrd_variance} \includegraphics{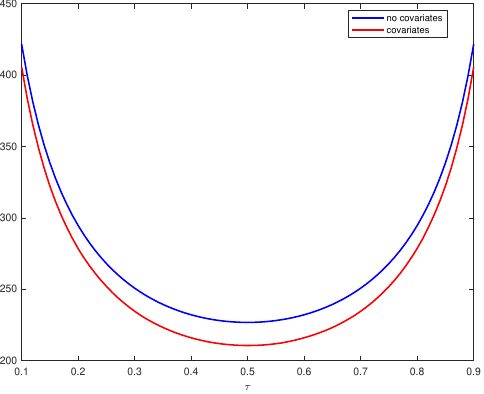} 
\end{figure}

%%%%%%%%%%%%%%%%%%%%%%%%%%%%%%%%%%%%%%%%%%%%%%%%%%%%%%%%%%%%%

\subsection{Quantile RDD inference with covariate adjustment}

\label{subsec:quantile_RDD_inference} %%%%%%%%%%%%%%%%%%%%%%%%%%%%%%%%%%%%%%%%%%%%%%%%%%%%%%%%%%%%%

The standard error constructed from the linearization of $\widehat{\vartheta}^{\mathit{eb}}_{p}\left(\tau\mid h\right)$
and the plug-in method requires choosing additional tuning parameters.
We propose using the nonparametric bootstrap to construct confidence
intervals (CIs) and uniform confidence bands (UCBs).

A nonparametric bootstrap sample $\left\{ Y^{*}_{i},X^{*}_{i},Z^{*}_{i}\right\} ^{n}_{i=1}$
consists of $n$ independent draws from the original sample with replacement.
Let $\left(W^{*}_{+,p,i}\left(h\right),W^{*}_{-,p,i}\left(h\right),W^{*}_{p,i}\left(h\right)\right)$
be the bootstrap analogues of $\left(W_{+,p,i}\left(h\right),W_{-,p,i}\left(h\right),W_{p,i}\left(h\right)\right)$:
$W^{*}_{+,p,i}\left(h\right)$ is defined by the right-hand side of
\eqref{eq:regression weight} with $X_{i}$ replaced by $X^{*}_{i}$,
$W^{*}_{-,p,i}\left(h\right)$ is defined similarly, and $W^{*}_{p,i}\left(h\right)\coloneqq W^{*}_{+,p,i}\left(h\right)-W^{*}_{-,p,i}\left(h\right)$.
Similarly, let $\widehat{\lambda}^{*}_{p}\left(h\right)$ be the bootstrap
analogue of $\widehat{\lambda}_{p}\left(h\right)$: $\widehat{\lambda}^{*}_{p}\left(h\right)$
is the optimizer of the convex optimization problem defined by \eqref{eq:dual}
with $V_{p,i}\left(h\right)$ replaced by $V^{*}_{p,i}\left(h\right)\coloneqq W^{*}_{p,i}\left(h\right)\bar{Z}^{*}_{i}$,
where $\bar{Z}^{*}_{i}\coloneqq\left(1,Z^{*\top}_{i}\right)^{\top}$.
Finally, let $\widehat{w}^{*}_{p,i}\left(h\right)$ denote the bootstrap
analogue of the balancing weights $\widehat{w}_{p,i}\left(h\right)$,
which is computed from the right-hand side of \eqref{eq:balancing weights}
with $\left(V_{p,i}\left(h\right),\widehat{\lambda}_{p}\left(h\right)\right)$
replaced by $\left(V^{*}_{p,i}\left(h\right),\widehat{\lambda}^{*}_{p}\left(h\right)\right)$.
We then define the bootstrap analogue $\widehat{\vartheta}^{*}_{p}\left(\tau\mid h\right)$
as the regression coefficient of $I^{*}_{i}$ in the bootstrap analogue of \eqref{QRD estimator bw}: 
\begin{equation}
\underset{b_{0},b_{1}}{\min}\sum_{i}\widehat{w}^{*}_{p,i}\left(h\right)K\left(\frac{X^{*}_{i}}{h}\right)\rho_{\tau}\left(Y^{*}_{i}-r^{\top}_{p}\left(X^{*}_{i}\right)b_{0}-I^{*}_{i}\cdot r^{\top}_{p}\left(X^{*}_{i}\right)b_{1}\right),\label{QRD estimator bw boot}
\end{equation}
where $I^{*}_{i}\coloneqq\mathbbm{1}\left(X^{*}_{i}>0\right)$. Alternatively,
one may replace $\widehat{w}^{*}_{p,i}\left(h\right)$ by $\widetilde{w}^{*}_{p,i}\left(h\right)$,
the bootstrap analogue of $\widetilde{w}_{p,i}\left(h\right)$.

%{\color{red}[Yuya: Shouldn't we resample $X_i$ inside $K(\cdot)$ and first $r_p(\cdot)$? ALso, don't we resample $I_i$?]}

Let $\mathbb{W}_{n}$ be a random element that depends on the bootstrap
sample and takes values in some Banach space $\mathbb{D}$. Let $\mathbb{W}$
be a tight random element in $\mathbb{D}$. ``$\mathbb{W}_{n}\rightsquigarrow_{*}\mathbb{W}$
in $\mathbb{D}$'' denotes conditional convergence in distribution
given the original sample:
\[
\underset{f\in\mathit{BL}_{1}\left(\mathbb{D}\right)}{\sup}\left|\mathrm{E}_{*}\left[f\left(\mathbb{W}_{n}\right)\right]-\mathrm{E}\left[f\left(\mathbb{W}\right)\right]\right|\rightarrow_{p}0,
\]
as $n\uparrow\infty$ (see \citealp[Section 23.2.1]{VanDerVaart1998}),
where $\mathrm{E}_{*}\left[\cdot\right]$ is the expectation with
respect to the resampling distribution (i.e., conditional distribution
given the original sample) and $\mathit{BL}_{1}\left(\mathbb{D}\right)$
denotes the collection of real-valued functions defined
on $\mathbb{D}$ that are bounded in absolute value by 1 and have Lipschitz constant at most 1. 

Let $\mathbb{Z}^{*}_{p}\left(\tau\mid h\right)\coloneqq\sqrt{nh}\left(\widehat{\vartheta}^{*}_{p}\left(\tau\mid h\right)-\widehat{\vartheta}^{\mathit{eb}}_{p}\left(\tau\mid h\right)\right)$.
The following result is a bootstrap analogue of Theorem \ref{thm:sharp QRD}. For expositional simplicity, the proof of Theorem \ref{thm:QRD bootstrap}
is presented for the bootstrap estimator constructed using the closed-form
weights $\widetilde{w}^{*}_{p,i}\left(h\right)$. The same conclusion
holds when $\widetilde{w}^{*}_{p,i}\left(h\right)$ is replaced by
$\widehat{w}^{*}_{p,i}\left(h\right)$, by a straightforward adaptation
of the argument using the corresponding bootstrap first-order expansions (a bootstrap analogue of Lemma \ref{lem:reg weights lambda} in the supplement).
\begin{thm}
\label{thm:QRD bootstrap}Suppose that Assumptions \ref{assu:smoothness Z}--\ref{assu:QRD smoothness}
hold. Assume that the bandwidth $h$ satisfies the conditions in the
statement of Theorem \ref{thm:sharp QRD}. Then, $\mathbb{Z}^{*}_{p}\left(\cdot\mid h\right)\rightsquigarrow_{*}\mathbb{Z}_{p}$
in $\ell^{\infty}\left[\underline{\tau},\overline{\tau}\right]$.
\end{thm}
Theorems \ref{thm:sharp QRD} and \ref{thm:QRD bootstrap} are stated
for a bandwidth that is common to all quantiles. Quantile-dependent
bandwidths are standard in the literature on conditional quantile
processes: the bandwidth in \citet{Qu2015} varies across quantiles
to adapt to data sparsity, and \citet{Qu2019} allow for such bandwidths
in uniform inference on QTEs in the sharp RDD without covariates.
We model the bandwidth used at quantile $\tau$ as $h\left(\tau\right)\coloneqq c\left(\tau\right)h_{0}$,
$\tau\in\left[\underline{\tau},\overline{\tau}\right]$, where $h_{0}=h_{0,n}$
denotes a baseline bandwidth sequence and $c\left(\cdot\right)$ is
a scale function. We impose the following mild condition. \begin{assumption}
\label{assu:varying bandwidth}$c:\left[\underline{\tau},\overline{\tau}\right]\rightarrow\left[\underline{c},\overline{c}\right]$
is Lipschitz continuous, with $0<\underline{c}\leq\overline{c}<\infty$.
\end{assumption} Assumption \ref{assu:varying bandwidth} is the
same as \citet[Assumption 5]{Qu2019}. \begin{assumption} \label{assu:bounded variation}$K\left(\cdot\right)$
is Lipschitz continuous on $\mathbb{R}$. \end{assumption} Assumption
\ref{assu:bounded variation} strengthens Assumption \ref{assu:kernel}
only mildly: it is satisfied by many commonly used kernels (e.g.,
the triangular, Epanechnikov, and biweight kernels) and is implied
by the differentiability required in \citet[Assumption 4]{Qu2019}.\footnote{Assumption \ref{assu:bounded variation} implies that $K(\cdot)$ is of bounded
variation, which controls the uniform entropy of the function classes
indexed by the bandwidth that arise when the bandwidth varies across
quantiles.}

When the bandwidth varies across quantiles, the entire two-step procedure
is run at bandwidth $h\left(\tau\right)$ for each $\tau\in\left[\underline{\tau},\overline{\tau}\right]$.
The balancing weights $\widehat{w}_{p,i}\left(h\right)$ solve (\ref{eq:generalized balancing})
with $V_{p,i}\left(h\right)$ computed at bandwidth $h\left(\tau\right)$,
and $\widehat{\vartheta}^{\mathit{eb}}_{p}\left(\tau\mid h\left(\tau\right)\right)$
is the coefficient of $I_{i}$ in (\ref{QRD estimator bw}) with $h$
replaced by $h\left(\tau\right)$. A quantile-dependent bandwidth
alters the limiting covariance structure because the estimators at
quantiles $\tau$ and $\tau'$ average the data over local windows
of different widths $h\left(\tau\right)$ and $h\left(\tau'\right)$.
For $a,b\in\left[\underline{c},\overline{c}\right]$, define the constant
\begin{equation}
\Omega_{p}\left(a,b\right)\coloneqq\frac{1}{\sqrt{ab}}\int^{\infty}_{0}\mathcal{K}_{+,p,0}\left(\frac{u}{a}\right)\mathcal{K}_{+,p,0}\left(\frac{u}{b}\right)\mathrm{d}u.\label{eq:kernel overlap}
\end{equation}
Since $\mathcal{K}_{-,p,0}\left(t\right)=\mathcal{K}_{+,p,0}\left(-t\right)$
under Assumption \ref{assu:kernel}, the same value is obtained when
$\mathcal{K}_{+,p,0}$ is replaced by $\mathcal{K}_{-,p,0}$ and the
integration range by $\left(-\infty,0\right)$. Define the covariance
kernel 
\begin{equation}
\mathscr{C}^{c}_{\star,p}\left(\tau,\tau'\right)\coloneqq\frac{\Omega_{p}\left(c\left(\tau\right),c\left(\tau'\right)\right)}{\omega^{0,2}_{p,0}}\mathscr{C}_{\star,p}\left(\tau,\tau'\right).\label{eq:QRD covariance kernel vb}
\end{equation}
Now we extend the convergence results of Theorems \ref{thm:sharp QRD} and \ref{thm:QRD bootstrap} to bandwidths that vary across quantiles.

\begin{thm}
\label{thm:QRD varying bandwidth}Suppose that Assumptions \ref{assu:smoothness Z}--\ref{assu:bounded variation}
hold. Assume that the baseline bandwidth $h_{0}$ satisfies the conditions
imposed on $h$ in the statement of Theorem \ref{thm:sharp QRD}.
Then, $\mathbb{Z}^{\mathit{eb}}_{p}\left(\cdot\mid h\left(\cdot\right)\right)\rightsquigarrow\mathbb{Z}^{c}_{p}$
in $\ell^{\infty}\left[\underline{\tau},\overline{\tau}\right]$,
where $\left\{ \mathbb{Z}^{c}_{p}\left(\tau\right):\tau\in\left[\underline{\tau},\overline{\tau}\right]\right\} $
is a centered tight Gaussian process with covariance kernel $\mathscr{C}^{c}_{\star,p}\left(\cdot,\cdot\right)$
given by (\ref{eq:QRD covariance kernel vb}). Moreover, $\mathbb{Z}^{*}_{p}\left(\cdot\mid h\left(\cdot\right)\right)\rightsquigarrow_{*}\mathbb{Z}^{c}_{p}$
in $\ell^{\infty}\left[\underline{\tau},\overline{\tau}\right]$. 
\end{thm}
By a change of variables, we have $\Omega_{p}\left(a,a\right)=\omega^{0,2}_{p,0}$
and $\mathscr{C}^{c}_{\star,p}\left(\tau,\tau\right)=\mathscr{V}_{\star,p}\left(\tau\right)$:
the pointwise asymptotic bias and variance are exactly those of Theorem
\ref{thm:sharp QRD} evaluated at the bandwidth $h\left(\tau\right)$.
When $c\left(\cdot\right)$ is constant, $\mathscr{C}^{c}_{\star,p}=\mathscr{C}_{\star,p}$
and Theorem \ref{thm:QRD varying bandwidth} reduces to Theorem \ref{thm:sharp QRD}.
The overlap structure in (\ref{eq:kernel overlap})--(\ref{eq:QRD covariance kernel vb})
parallels the covariance function of the limiting process in \citet[Proposition 1]{Qu2019}.

Let $\mathbb{Z}^{\mathit{us}}_{p}\left(\tau\mid h\right)\coloneqq\sqrt{nh}\left(\widehat{\vartheta}^{\mathit{eb}}_{p}\left(\tau\mid h\right)-\vartheta\left(\tau\right)\right)$
for $\tau\in\left[\underline{\tau},\overline{\tau}\right]$. Theorem
\ref{thm:QRD varying bandwidth} implies that under the undersmoothing
assumption $nh^{2p+3}_{0}=o\left(1\right)$, the bootstrap analogue
$\mathbb{Z}^{*}_{p}\left(\cdot\mid h\left(\cdot\right)\right)$ of the process
$\mathbb{Z}^{\mathit{us}}_{p}\left(\cdot\mid h\left(\cdot\right)\right)$ conditionally converges
in distribution to the same limiting Gaussian process. By the continuous
mapping theorem (CMT) and the bootstrap CMT, $\left\Vert \mathbb{Z}^{\mathit{us}}_{p}\left(\cdot\mid h\left(\cdot\right)\right)\right\Vert _{\infty}\rightsquigarrow\left\Vert \mathbb{Z}^{c}_{p}\right\Vert _{\infty}$
and $\left\Vert \mathbb{Z}^{*}_{p}\left(\cdot\mid h\left(\cdot\right)\right)\right\Vert _{\infty}\rightsquigarrow_{*}\left\Vert \mathbb{Z}^{c}_{p}\right\Vert _{\infty}$
under $nh^{2p+3}_{0}=o\left(1\right)$. These results justify the
asymptotic validity of the bootstrap CIs and UCBs we propose below.

Let $\mathrm{Pr}_{*}\left[\cdot\right]$ be the conditional probability
given the original sample and let 
\begin{equation}
s_{p,q}\left(\tau\mid h\right)\coloneqq\inf\left\{ u\in\mathbb{R}:\mathrm{Pr}_{*}\left[\widehat{\vartheta}^{*}_{p}\left(\tau\mid h\right)\leq u\right]\geq q\right\} ,\,q\in\left(0,1\right),\label{eq:theta_rho_hat_star quantile}
\end{equation}
be the $q$-th quantile of the resampling distribution of $\widehat{\vartheta}^{*}_{p}\left(\tau\mid h\right)$.
Let $1-\alpha$ be the nominal coverage probability. The bootstrap
percentile CI $\mathit{CI}_{p,\alpha}\left(\tau\mid h\right)\coloneqq\left[s_{p,\alpha/2}\left(\tau\mid h\right),s_{p,1-\alpha/2}\left(\tau\mid h\right)\right]$
is asymptotically valid, i.e., $\Pr\left[\vartheta\left(\tau\right)\in\mathit{CI}_{p,\alpha}\left(\tau\mid h\right)\right]\rightarrow1-\alpha$
as $n\uparrow\infty$, by Theorem \ref{thm:QRD bootstrap} and standard
arguments. The following algorithm summarizes the procedure to compute
the bootstrap CI using the quantiles in \eqref{eq:theta_rho_hat_star quantile}.
Let $N_{B}$ denote the number of bootstrap replications. \begin{lyxalgorithm}[Bootstrap
percentile confidence interval] \label{alg:qrd_ci}\textbf{ }Step
1: In each of the replications $b\in\left[N_{B}\right]$, independently
draw $\left\{ Y^{*(b)}_{i},X^{*(b)}_{i},Z^{*(b)}_{i}\right\} ^{n}_{i=1}$
with replacement from the original sample. Step 2: For all $b\in\left[N_{B}\right]$,
compute the balancing weights $\widehat{w}^{*(b)}_{p,i}\left(h\right)$
using $\left\{ X^{*(b)}_{i},Z^{*(b)}_{i}\right\} ^{n}_{i=1}$. Step
3: Compute $\widehat{\vartheta}^{*(b)}_{p}\left(\tau\mid h\right)$
as the coefficient of $I^{*(b)}_{i}$ in (\ref{QRD estimator bw boot}), evaluated with $\widehat{w}^{*(b)}_{p,i}\left(h\right)$, and $\left\{ Y^{*(b)}_{i},X^{*(b)}_{i},Z^{*(b)}_{i}\right\} ^{n}_{i=1}$.
Step 4: Sort $\left\{ \widehat{\vartheta}^{*(b)}_{p}\left(\tau\mid h\right)\right\} ^{N_{B}}_{b=1}$
to obtain the order statistics $\vartheta_{\left\langle 1\right\rangle }\leq\cdots\leq\vartheta_{\left\langle N_{B}\right\rangle }$.
Step 5: Return the CI $\left[\vartheta_{\left\langle \left\lceil N_{B}\times\left(\alpha/2\right)\right\rceil \right\rangle },\vartheta_{\left\langle \left\lceil N_{B}\times\left(1-\alpha/2\right)\right\rceil \right\rangle }\right]$
for $\vartheta\left(\tau\right)$. \end{lyxalgorithm} The convergence
results $\left\Vert \mathbb{Z}^{\mathit{us}}_{p}\left(\cdot\mid h\left(\cdot\right)\right)\right\Vert _{\infty}\rightsquigarrow\left\Vert \mathbb{Z}^{c}_{p}\right\Vert _{\infty}$
and $\left\Vert \mathbb{Z}^{*}_{p}\left(\cdot\mid h\left(\cdot\right)\right)\right\Vert _{\infty}\rightsquigarrow_{*}\left\Vert \mathbb{Z}^{c}_{p}\right\Vert _{\infty}$
justify the asymptotic validity of the nonstudentized UCBs. Studentized UCBs additionally account for variation in the asymptotic variance
$\mathscr{V}_{\star,p}\left(\tau\right)$ across quantiles. We follow the
construction of \citet{Chernozhukov2018} and studentize $\mathbb{Z}^{*}_{p}\left(\cdot\mid h\left(\cdot\right)\right)$
and $\mathbb{Z}^{\mathit{us}}_{p}\left(\cdot\mid h\left(\cdot\right)\right)$ by the interquartile
range (IQR). Let $z_{q}$ be the $q$-th quantile of the standard normal distribution.
Since $\sqrt{nh}\left(s_{p,q}\left(\tau\mid h\right)-\widehat{\vartheta}^{\mathit{eb}}_{p}\left(\tau\mid h\right)\right)$
is the $q$-th quantile of the resampling distribution of $\mathbb{Z}^{*}_{p}\left(\tau\mid h\right)$,
it follows from $\mathbb{Z}^{*}_{p}\left(\cdot\mid h\right)\rightsquigarrow_{*}\mathbb{Z}_{p}$
and standard arguments that 
\[
\sqrt{nh}\left(s_{p,q}\left(\tau\mid h\right)-\widehat{\vartheta}^{\mathit{eb}}_{p}\left(\tau\mid h\right)\right)\rightarrow_{p}\sqrt{\mathscr{V}_{\star,p}\left(\tau\right)}\cdot z_{q},
\]
for all $\tau\in\left[\underline{\tau},\overline{\tau}\right]$. Therefore,
$\left(nh\right)\left(\mathit{IQR}_{p}\left(\tau\mid h\right)/\left(z_{0.75}-z_{0.25}\right)\right)^{2}$
is a consistent estimator of $\mathscr{V}_{\star,p}\left(\tau\right)$,
where $\mathit{IQR}_{p}\left(\tau\mid h\right)\coloneqq s_{p,0.75}\left(\tau\mid h\right)-s_{p,0.25}\left(\tau\mid h\right)$
is the IQR of the resampling distribution of $\widehat{\vartheta}^{*}_{p}\left(\tau\mid h\right)$.
Let 
\[
s^{\mathit{ucb}}_{p,q}\left(h\left(\cdot\right)\right)\coloneqq\inf\left\{ u\in\mathbb{R}:\mathrm{Pr}_{*}\left[\underset{\tau\in\left[\underline{\tau},\overline{\tau}\right]}{\sup}\frac{\left|\widehat{\vartheta}^{*}_{p}\left(\tau\mid h\left(\tau\right)\right)-\widehat{\vartheta}^{\mathit{eb}}_{p}\left(\tau\mid h\left(\tau\right)\right)\right|}{\mathit{IQR}_{p}\left(\tau\mid h\left(\tau\right)\right)/\left(z_{0.75}-z_{0.25}\right)}\leq u\right]\geq q\right\} 
\]
be the $q$-th quantile of the resampling distribution of the supremum
of the studentized version of $\mathbb{Z}^{*}_{p}\left(\cdot\mid h\left(\cdot\right)\right)$.
A UCB is given by the collection of intervals $\left\{ \mathit{CB}_{p,\alpha}\left(\tau\mid h\left(\cdot\right)\right):\tau\in\left[\underline{\tau},\overline{\tau}\right]\right\} $,
where 
\begin{equation}
\mathit{CB}_{p,\alpha}\left(\tau\mid h\left(\cdot\right)\right)\coloneqq\left[\widehat{\vartheta}^{\mathit{eb}}_{p}\left(\tau\mid h\left(\tau\right)\right)\pm s^{\mathit{ucb}}_{p,1-\alpha}\left(h\left(\cdot\right)\right)\cdot\frac{\mathit{IQR}_{p}\left(\tau\mid h\left(\tau\right)\right)}{z_{0.75}-z_{0.25}}\right],\,\tau\in\left[\underline{\tau},\overline{\tau}\right].\label{eq:CB_Q variable width}
\end{equation}

We summarize the procedure in the following algorithm. Let $T$ be
a large positive integer and let $\mathcal{T}\coloneqq\left\{ \tau^{\left(1\right)},\ldots,\tau^{\left(T\right)}\right\} $
be equally spaced grid points in $\left[\underline{\tau},\overline{\tau}\right]$.
\begin{lyxalgorithm}[Bootstrap uniform confidence band] \label{alg:qrd_ucb}
Steps 1--3: Same as those in Algorithm \ref{alg:qrd_ci}. Step 4:
Compute $\widehat{\vartheta}^{*(b)}_{p}\left(\tau\mid h\left(\tau\right)\right)$
for $\left(b,\tau\right)\in\left[N_{B}\right]\times\mathcal{T}$ and
compute $\widehat{\vartheta}^{\mathit{eb}}_{p}\left(\tau\mid h\left(\tau\right)\right)$
for $\tau\in\mathcal{T}$. Step 5: Compute the order statistics $\vartheta_{\left\langle 1\right\rangle }\left(\tau\right)\leq\cdots\leq\vartheta_{\left\langle N_{B}\right\rangle }\left(\tau\right)$
corresponding to $\left\{ \widehat{\vartheta}^{*(b)}_{p}\left(\tau\mid h\left(\tau\right)\right)\right\} ^{N_{B}}_{b=1}$
for all $\tau\in\mathcal{T}$. Step 6: Compute 
\[
\left\{ \underset{\tau\in\mathcal{T}}{\max}\frac{\left|\widehat{\vartheta}^{*(b)}_{p}\left(\tau\mid h\left(\tau\right)\right)-\widehat{\vartheta}^{\mathit{eb}}_{p}\left(\tau\mid h\left(\tau\right)\right)\right|}{\left(\vartheta_{\left\langle \left\lceil N_{B}\times0.75\right\rceil \right\rangle }\left(\tau\right)-\vartheta_{\left\langle \left\lceil N_{B}\times0.25\right\rceil \right\rangle }\left(\tau\right)\right)/\left(z_{0.75}-z_{0.25}\right)}\right\} ^{N_{B}}_{b=1}
\]
and obtain the corresponding order statistics $s^{\mathit{ucb}}_{\left\langle 1\right\rangle }\leq\cdots\leq s^{\mathit{ucb}}_{\left\langle N_{B}\right\rangle }$
and the critical value $s^{\mathit{ucb}}_{\left\langle \left\lceil N_{B}\left(1-\alpha\right)\right\rceil \right\rangle }$.
Step 7: Return the UCB 
\[
\left\{ \left[\widehat{\vartheta}^{\mathit{eb}}_{p}\left(\tau\mid h\left(\tau\right)\right)\pm s^{\mathit{ucb}}_{\left\langle \left\lceil N_{B}\left(1-\alpha\right)\right\rceil \right\rangle }\left(\frac{\vartheta_{\left\langle \left\lceil N_{B}\times0.75\right\rceil \right\rangle }\left(\tau\right)-\vartheta_{\left\langle \left\lceil N_{B}\times0.25\right\rceil \right\rangle }\left(\tau\right)}{z_{0.75}-z_{0.25}}\right)\right]\right\} _{\tau\in\mathcal{T}}.
\]
\end{lyxalgorithm} To compute the QRD estimator with balancing weights,
we run the weighted LP quantile regression (\ref{QRD estimator bw})
with $p=2$ and with the bandwidth determined by minimizing the AMSE
for $p=1$: 
\begin{equation}
\widehat{h}(\tau)\coloneqq\left(\frac{\widehat{\mathscr{V}}_{\star}\left(\tau\right)}{4\cdot\widehat{\mathscr{B}}^{2}_{\star}\left(\tau\right)}\right)^{1/5}n^{-1/5}.\label{eq:MSEbw_qte}
\end{equation}
The combination of quadratic local quantile regression and AMSE-optimal
bandwidth for the linear fit internalizes bias correction in the spirit
of \citet[Remark 7]{calonico2014robust}. \footnote{\color{black}\citet{calonico2018optimal} show that the AMSE-optimal bandwidth remains valid for robust bias-corrected inference, although it is not coverage-error optimal. Because our analysis focuses on first-order properties, exploring coverage-error-optimal bandwidth selection is beyond the scope of this paper. We leave this as an important direction for future research.} 
Appendix \ref{app:bdw_qrd}
provides detailed steps to obtain the estimates $\widehat{\mathscr{B}}_{\star}\left(\tau\right)$
and $\widehat{\mathscr{V}}_{\star}\left(\tau\right)$ for $\mathscr{B}_{\star,1}\left(\tau\right)$
and $\mathscr{V}_{\star,1}\left(\tau\right)$. The AMSE-optimal bandwidth
$\widehat{h}\left(\tau\right)$ varies across quantiles. We propose
the CI $\mathit{CI}_{2,\alpha}\left(\tau\mid\widehat{h}\left(\tau\right)\right)$
for fixed $\tau$ and the UCB $\left\{ \mathit{CB}_{2,\alpha}\left(\tau\mid\widehat{h}\left(\cdot\right)\right):\tau\in\left[\underline{\tau},\overline{\tau}\right]\right\} $
for practical implementation.

%%%%%%%%%%%%%%%%%%%%%%%%%%%%%%%%%%%%%%%%%%%%%%%%%%%%%%%%%%%%%

\section{Derivative estimation: sharp regression kink design}

\label{subsec:Sharp-regression-kink} %%%%%%%%%%%%%%%%%%%%%%%%%%%%%%%%%%%%%%%%%%%%%%%%%%%%%%%%%%%%%
In this section, we compare our proposed reweighting-based covariate
adjustment approach with the regression-based approach in the context
of derivative estimation for the conditional mean. In this setting, CCFT show that the regression-based approach yields consistent estimation and improves
efficiency by exploiting derivative restrictions on the covariates.
Our proposed approach can further enhance efficiency by also exploiting
the level restrictions, in addition to the derivative restrictions
on the covariates. %To fix ideas, we focus throughout this section on the sharp regression kink design (RKD) with covariates.

In the sharp RKD, $\vartheta_{\dagger}\coloneqq\mu^{\left(1\right)}_{Y,+}-\mu^{\left(1\right)}_{Y,-}$
is proportional to the treatment effect generated by a kink in the
treatment rule. The following stronger continuity assumption is implied
by a valid sharp RKD. \begin{assumption} \label{assu:RKD smoothness}$\mu_{Z}\left(\cdot\right)$
is continuously differentiable at 0. \end{assumption} Assumption
\ref{assu:RKD smoothness} is the RKD counterpart of the covariate
continuity condition imposed in the sharp RDD and plays the same role
there. Because the sharp RKD measures the treatment effect from the
change in the slope of the regression function at the cutoff rather
than from a jump in its level, identification requires smoothness
rather than mere continuity: \citet{Card:2015gc} show that identifying
the sharp RKD effect requires the conditional density of the score
given the unobservable to be continuously differentiable at the cutoff
\citep[Assumption 4]{Card:2015gc}, a strengthening of the continuity
that suffices for the sharp RDD \citep[Condition 2b]{lee2008randomized},
so that a kink in the outcome can be attributed to the policy. Continuous
differentiability of $\mu_{Z}\left(\cdot\right)$ at the cutoff is
implied by the model assumptions \citep[Proposition 1(a) and Corollary 2]{Card:2015gc}.
It imposes both $\mu_{Z,+}=\mu_{Z,-}$ and $\mu^{\left(1\right)}_{Z,+}=\mu^{\left(1\right)}_{Z,-}$,
strengthening the level balance of the sharp RDD case, so that predetermined
covariates exhibit neither a jump nor a kink at the cutoff. As in
the RDD case, this is a testable implication of a valid design rather
than an additional restriction. Like its RDD analogue, Assumption
\ref{assu:RKD smoothness} therefore requires nothing that a careful
RKD study does not already check.

Recall from Section \ref{subsec:cov_rd} the equivalent kernels $\mathcal{K}_{+,p,\nu}$
and $\mathcal{K}_{-,p,\nu}$. Let $\varsigma_{p}\coloneqq\int^{1}_{0}\mathcal{K}_{+,p,0}\left(t\right)\mathcal{K}_{+,p,1}\left(t\right)\mathrm{d}t$.
One can verify that $\int^{0}_{-1}\mathcal{K}_{-,p,0}\left(t\right)\mathcal{K}_{-,p,1}\left(t\right)\mathrm{d}t=-\varsigma_{p}$.

%%%%%%%%%%%%%%%%%%%%%%%%%%%%%%%%%%%%%%%%%%%%%%%%%%%%%%%%%%%%%

\subsection{Augmented regression approach}

\label{subsec:RKD_regression} %%%%%%%%%%%%%%%%%%%%%%%%%%%%%%%%%%%%%%%%%%%%%%%%%%%%%%%%%%%%%

CCFT's estimator $\widehat{\vartheta}^{\mathit{CCFT}}_{\dagger,p}\left(h\right)$
is the regression coefficient of $I_{i}X_{i}$ in (\ref{eq:CCFT estimator}).
CCFT show that under the condition $\mu^{\left(1\right)}_{Z,+}=\mu^{\left(1\right)}_{Z,-}$
implied by Assumption \ref{assu:RKD smoothness}, 
\[
\sqrt{nh^{3}}\left(\widehat{\vartheta}^{\mathit{CCFT}}_{\dagger,p}\left(h\right)-\vartheta_{\dagger}-h^{p}\mathscr{B}_{\dagger,p}\right)\rightarrow_{d}\mathrm{N}\left(0,\mathscr{V}_{\dagger,p}\right),
\]
where 
\[
\mathscr{B}_{\dagger,p}\coloneqq\frac{\mu^{\left(p+1\right)}_{\epsilon,+}\omega^{p+1,1}_{+,p,1}-\mu^{\left(p+1\right)}_{\epsilon,-}\omega^{p+1,1}_{-,p,1}}{\left(p+1\right)!}\mbox{ and }\mathscr{V}_{\dagger,p}\coloneqq\frac{\omega^{0,2}_{p,1}\sigma^{2}}{\varphi}.
\]
Recall that $\epsilon\coloneqq Y-Z^{\top}\gamma$ and the constants
$\left(\gamma,\sigma^{2}\right)$ are defined by (\ref{eq:gamma sigma definition}).
The simple LP derivative estimator without covariates, i.e., the coefficient
of $I_{i}X_{i}$ in \eqref{eq:LP without covariates}, has an asymptotic
variance given by $\omega^{0,2}_{p,1}\mathrm{Var}_{\pm}\left[Y\right]/\varphi\geq\mathscr{V}_{\dagger,p}$.
We can extend the result in Theorem \ref{thm:normality} to show that
$\widehat{\vartheta}^{\mathit{CCFT}}_{\dagger,p}\left(h\right)$ is
first-order equivalent to an estimator using balancing weights that
incorporate the ``derivative'' balance condition $\mu^{\left(1\right)}_{Z,+}=\mu^{\left(1\right)}_{Z,-}$
for the first-order derivatives of $\mu_{Z}\left(\cdot\right)$.

%%%%%%%%%%%%%%%%%%%%%%%%%%%%%%%%%%%%%%%%%%%%%%%%%%%%%%%%%%%%%

\subsection{Augmentation plus reweighting}

\label{subsec:RKD_reweighting}

As argued in the previous subsection, the CCFT estimator achieves
efficiency gains by exploiting the covariate information embodied
in the derivative balance condition. We show that further efficiency
gains are possible by additionally incorporating the level balance
condition, $\mu_{Z,+}=\mu_{Z,-}$, through the reweighting approach
proposed in Section \ref{subsec:proposed_approach}. Let $\widehat{\vartheta}^{\mathit{eb}}_{\dagger,p}\left(h\right)$
be the regression coefficient of $I_{i}X_{i}$ in 
\begin{equation}
\underset{b_{0},b_{1},b_{2}}{\min}\sum_{i}\widehat{w}_{p,i}\left(h\right)K\left(\frac{X_{i}}{h}\right)\left\{ Y_{i}-r^{\top}_{p}\left(X_{i}\right)b_{0}-I_{i}\cdot r^{\top}_{p}\left(X_{i}\right)b_{1}-Z^{\top}_{i}b_{2}\right\} ^{2}.\label{eq:adjusted CCFT kink}
\end{equation}

Let 
\begin{equation}
\gamma_{\dagger,p}\coloneqq\left(\frac{\varsigma_{p}}{\omega^{0,2}_{p,0}}\right)\left(\mathrm{Var}_{\pm}\left[Z\right]\right)^{-1}\left(\mathrm{Cov}_{+}\left[Z,\epsilon\right]-\mathrm{Cov}_{-}\left[Z,\epsilon\right]\right).\label{eq:gamma_RKD}
\end{equation}
The following result establishes the asymptotic normality of the proposed RKD estimator $\widehat{\vartheta}^{\mathit{eb}}_{\dagger,p}\left(h\right)$ and quantifies the additional variance reduction from incorporating the level balance condition.

\begin{thm}
\label{thm:kink}Suppose that Assumptions \ref{assu:smoothness Z}--\ref{assu:data generating process}
and \ref{assu:RKD smoothness} hold, and the bandwidth satisfies $nh^{2p+3}=O\left(1\right)$
and $nh^{3}\rightarrow\infty$. Then, 
\[
\sqrt{nh^{3}}\left(\widehat{\vartheta}^{\mathit{eb}}_{\dagger,p}\left(h\right)-\vartheta_{\dagger}-h^{p}\mathscr{B}_{\star\dagger,p}\right)\rightarrow_{d}\mathrm{N}\left(0,\mathscr{V}_{\star\dagger,p}\right),
\]
where 
\[
\mathscr{B}_{\star\dagger,p}\coloneqq\mathscr{B}_{\dagger,p}-\frac{\left(\mu^{\left(p+1\right)}_{Z,+}\omega^{p+1,1}_{+,p,0}-\mu^{\left(p+1\right)}_{Z,-}\omega^{p+1,1}_{-,p,0}\right)^{\top}\gamma_{\dagger,p}}{\left(p+1\right)!}\mbox{ and }\mathscr{V}_{\star\dagger,p}\coloneqq\mathscr{V}_{\dagger,p}-\frac{\omega^{0,2}_{p,0}\gamma^{\top}_{\dagger,p}\mathrm{Var}_{\pm}\left[Z\right]\gamma_{\dagger,p}}{\varphi}.
\]
\end{thm}

The second term in the expression for $\mathscr{V}_{\star\dagger,p}$ represents the further reduction in asymptotic variance arising from the level balance condition, on top of the reduction from $\omega^{0,2}_{p,1}\mathrm{Var}_{\pm}\left[Y\right]/\varphi$
to $\mathscr{V}_{\dagger,p}$ delivered by the derivative balance condition
alone. Table \ref{tab:rkd_variance} illustrates the size of these
two reductions for $\widehat{\vartheta}^{\mathit{eb}}_{\dagger,p}\left(h\right)$
under the simulation designs described in Sections \ref{sec:simu_design} and \ref{sec:simu_kink}. In DGP
I, the ``Reg'' and ``EBW'' estimators attain the same efficiency
gain because $\gamma_{\dagger,p}$ in \eqref{eq:gamma_RKD} is close
to zero in that design, so that the level balance condition does
not further reduce the asymptotic variance.

\begin{table}[!t]
\centering \caption{Asymptotic standard deviation (Asy.\ SD) for estimating $\vartheta_{\dagger}$
with and without covariate information, $p=3$: ``Standard'' is the
LP derivative estimator without covariates, ``Reg'' is the CCFT estimator
and ``EBW'' is the EB reweighting estimator; covariate balance conditions
(Cov.\ Cond): $\mu_{Z,+}=\mu_{Z,-}$ (Level) and $\mu^{\left(1\right)}_{Z,+}=\mu^{\left(1\right)}_{Z,-}$
(Deri). Eff.\ gain is the proportional reduction in Asy.\ SD relative
to the standard estimator without using any covariate information.}
\label{tab:rkd_variance} %
\begin{tabular}{clcccc}
\toprule 
\multicolumn{1}{c}{DGP} & \multicolumn{1}{l}{Method} & \multicolumn{2}{l}{Cov. Cond} & \multicolumn{1}{c}{Asy.\ SD} & \multicolumn{1}{c}{Eff.\ gain}\tabularnewline
\multicolumn{1}{c}{} & \multicolumn{1}{c}{} & \multicolumn{1}{l}{Level} & \multicolumn{1}{l}{Deri} &  & \tabularnewline
\midrule 
I  & Standard  & No  & No  & 74.9  & \tabularnewline
 & Reg  & No  & Yes  & 66.0  & 12.0\%\tabularnewline
 & EBW  & Yes  & Yes  & 66.0  & 12.0\%\tabularnewline
\midrule 
II  & Standard  & No  & No  & 102.7  & \tabularnewline
 & Reg  & No  & Yes  & 85.8  & 16.5\%\tabularnewline
 & EBW  & Yes  & Yes  & 71.4  & 30.5\%\tabularnewline
\midrule 
III  & Standard  & No  & No  & 118.9  & \tabularnewline
 & Reg  & No  & Yes  & 100.8  & 15.2\%\tabularnewline
 & EBW  & Yes  & Yes  & 83.5  & 29.7\%\tabularnewline
\bottomrule
\end{tabular}
\end{table}

%%%%%%%%%%%%%%%%%%%%%%%%%%%%%%%%%%%%%%%%%%%%%%%%%%%%%%%%%%%%%

\subsection{RKD inference with covariate adjustment}

\label{subsec:RKD_inference} %%%%%%%%%%%%%%%%%%%%%%%%%%%%%%%%%%%%%%%%%%%%%%%%%%%%%%%%%%%%%
Let $\eta_{i}\coloneqq Y_{i}-\mu_{Y}\left(X_{i}\right)$ and $\xi_{i}\coloneqq Z_{i}-\mu_{Z}\left(X_{i}\right)$
be regression errors, and denote $\zeta_{i}\coloneqq\eta_{i}-\xi^{\top}_{i}\gamma$.
Let $\left(W^{\dagger}_{p,i}\left(h\right),W^{\dagger}_{+,p,i}\left(h\right),W^{\dagger}_{-,p,i}\left(h\right)\right)$
be defined by the formulas of $\left(W_{p,i}\left(h\right),W_{+,p,i}\left(h\right),W_{-,p,i}\left(h\right)\right)$
with $\mathrm{e}_{p+1,1}$ replaced by $\mathrm{e}_{p+1,2}$. We have
\begin{equation}
\sqrt{nh^{3}}\left(\widehat{\vartheta}^{\mathit{eb}}_{\dagger,p}\left(h\right)-\vartheta_{\dagger}-\mathscr{B}_{\star\dagger,p}h^{p}\right)=\frac{1}{\sqrt{nh}}\sum_{i}\left(W^{\dagger}_{p,i}\left(h\right)\zeta_{i}-W_{p,i}\left(h\right)\xi^{\top}_{i}\gamma_{\dagger,p}\right)+o_{p}\left(1\right).\label{eq:theta_hat_dag - theta_dag linearization}
\end{equation}
By taking a plug-in residual approach, we propose an estimator of
the pre-asymptotic variance of the leading term on the right-hand
side of \eqref{eq:theta_hat_dag - theta_dag linearization} for studentization.

Let $\widetilde{\eta}_{+,p,i}\left(h\right)$ denote the $i$-th fitted
regression residual from the LP regression 
\[
\underset{b}{\min}\sum_{i}K\left(\frac{X_{i}}{h}\right)I_{i}\left\{ Y_{i}-r^{\top}_{p}\left(X_{i}\right)b\right\} ^{2}.
\]
Let $\widetilde{\xi}_{+,p,i}\left(h\right)$ be the $d_{z}$-dimensional
vector of fitted residuals from the analogous LP regressions with
$Z_{i}$ as the dependent variable. Let 
\begin{equation}
\widetilde{\zeta}_{+,p,i}\left(h\right)\coloneqq\widetilde{\eta}_{+,p,i}\left(h\right)-\widetilde{\xi}^{\top}_{+,p,i}\left(h\right)\widetilde{\gamma},\label{eq:zeta_til}
\end{equation}
where $\widetilde{\gamma}$ is a consistent estimator of $\gamma$
whose formula can be found in Appendix \ref{app:imp_rkd}. The asymptotic
variance estimator can be taken as $\widetilde{\mathscr{V}}_{+,p}\left(h\right)+\widetilde{\mathscr{V}}_{-,p}\left(h\right)$:
\begin{eqnarray}
\widetilde{\mathscr{V}}_{+,p}\left(h\right) & \coloneqq & \frac{1}{nh}\sum_{i}\left\{ \left(W^{\dagger}_{+,p,i}\left(h\right)\right)^{2}\widetilde{\zeta}^{2}_{+,p,i}\left(h\right)+W^{2}_{+,p,i}\left(h\right)\left(\widetilde{\xi}^{\top}_{+,p,i}\left(h\right)\widetilde{\gamma}_{\dagger,p}\right)^{2}\right.\label{eq:scr_V_til}\nonumber \\
 &  & \left.-2\cdot W^{\dagger}_{+,p,i}\left(h\right)W_{+,p,i}\left(h\right)\widetilde{\zeta}_{+,p,i}\left(h\right)\left(\widetilde{\xi}^{\top}_{+,p,i}\left(h\right)\widetilde{\gamma}_{\dagger,p}\right)\right\} ,
\end{eqnarray}
and $\widetilde{\mathscr{V}}_{-,p}\left(h\right)$ is defined similarly,
where $\widetilde{\gamma}_{\dagger,p}$ is some consistent estimator
of $\gamma_{\dagger,p}$. The standard error for studentization is
now given by $\mathit{se}_{p}\left(h\right)\coloneqq\sqrt{\left(\widetilde{\mathscr{V}}_{+,p}\left(h\right)+\widetilde{\mathscr{V}}_{-,p}\left(h\right)\right)/\left(nh^{3}\right)}$. Under undersmoothing, the resulting studentized statistic is asymptotically standard normal.
  
\begin{cor}
Under the assumptions in the statement of Theorem \ref{thm:kink}
and $nh^{2p+3}=o\left(1\right)$, we have 
\[
\frac{\widehat{\vartheta}^{\mathit{eb}}_{\dagger,p}\left(h\right)-\vartheta_{\dagger}}{\mathit{se}_{p}\left(h\right)}\rightarrow_{d}\mathrm{N}\left(0,1\right).
\]
\end{cor}
An asymptotically valid $1-\alpha$ CI for $\vartheta_{\dagger}$
is $\mathit{CI}^{\dagger}_{p,\alpha}\left(h\right)\coloneqq\left[\widehat{\vartheta}^{\mathit{eb}}_{\dagger,p}\left(h\right)\pm z_{1-\alpha/2}\mathit{se}_{p}\left(h\right)\right]$. To compute $\widehat{\vartheta}^{\mathit{eb}}_{\dagger,p}$ that incorporates
both level and derivative balance conditions, we run the LP regression
(\ref{eq:adjusted CCFT kink}) with $p=3$ in $r_{p}(x)$ and with
the following bandwidth, determined by minimizing the AMSE for $p=2$:
\begin{equation}
\widehat{h}\coloneqq\left(\frac{3}{4}\cdot\frac{\widehat{\mathscr{V}}_{\star\dagger}}{\widehat{\mathscr{B}}^{2}_{\star\dagger}}\right)^{1/7}n^{-1/7},\label{eq:rkd-h}
\end{equation}
where $\left(\widehat{\mathscr{V}}_{\star\dagger},\widehat{\mathscr{B}}_{\star\dagger}\right)$
are consistent estimators of $\left(\mathscr{V}_{\star\dagger,2},\mathscr{B}_{\star\dagger,2}\right)$.
Appendix \ref{app:imp_rkd} provides detailed steps to calculate $\left(\widehat{\mathscr{V}}_{\star\dagger},\widehat{\mathscr{B}}_{\star\dagger}\right)$.
Our proposed CI is given by $\mathit{CI}^{\dagger}_{3,\alpha}\left(\widehat{h}\right)$.

\subsection{Treatment effect derivative}

The object $\vartheta_{\dagger}$ has a distinct interpretation in
the sharp RDD. \citet{Dong2015} define the treatment effect derivative
(TED) at the cutoff by 
\begin{equation}
\vartheta_{ted}\coloneqq\left.\frac{\mathrm{d}}{\mathrm{d}x}\mathrm{E}\left[Y\left(1\right)-Y\left(0\right)\mid X=x\right]\right|_{x=0}.\label{eq:TED definition}
\end{equation}
In an RDD, it measures how rapidly the conditional treatment effect
changes with the score at the cutoff. A large absolute value indicates that
the RDD treatment effect may change substantially when the cutoff
is moved slightly. Testing $\vartheta_{ted}=0$ therefore provides
a local diagnostic for the stability and external validity of the
RDD treatment effect. Suppose that $x\mapsto\mathrm{E}\left[Y\left(d\right)\mid X=x\right]$
is continuously differentiable in a neighborhood of zero for $d\in\left\{ 0,1\right\} $.
\citet{Dong2015} show that the right-hand derivative of the observed
conditional mean identifies the derivative of $\mathrm{E}\left[Y\left(1\right)\mid X=x\right]$
at $x=0$, whereas the left-hand derivative identifies the derivative
of $\mathrm{E}\left[Y\left(0\right)\mid X=x\right]$. Consequently,
we have $\vartheta_{ted}=\vartheta_{\dagger}$. Thus, the regression-
and reweighting-based covariate adjustment methods and their relative
efficiency ordering discussed before also apply to the TED, provided
that Assumption \ref{assu:RKD smoothness} is maintained.

However, as the TED is often used to evaluate the external validity
of the RDD, a researcher may prefer the covariate balance condition
in Assumption \ref{assu:smoothness Z} to the stronger version in
Assumption \ref{assu:RKD smoothness}. In that case, the regression-based
covariate adjustment discussed in Section \ref{subsec:RKD_regression}
may not lead to consistent estimation because the derivative balance
condition $\mu^{\left(1\right)}_{Z,+}=\mu^{\left(1\right)}_{Z,-}$
is not maintained in the RDD. In contrast, our reweighting approach
still applies. By taking the regression coefficient of $I_{i}X_{i}$ in \eqref{eq:entropy balancing estimator definition},
the reweighting approach incorporates into the derivative estimation
the level balance condition $\mu_{Z,+}=\mu_{Z,-}$ maintained in
Assumption \ref{assu:smoothness Z} and enhances efficiency relative
to the no-covariate derivative estimator (the regression coefficient of $I_{i}X_{i}$
in \eqref{eq:LP without covariates}). %%%%%%%%%%%%%%%%%%%%%%%%%%%%%%%%%%%%%%%%%%%%%%%%%%%%%%%%%%%%%

\section{Monte Carlo simulations}

%%%%%%%%%%%%%%%%%%%%%%%%%%%%%%%%%%%%%%%%%%%%%%%%%%%%%%%%%%%%%

\subsection{Simulation design}

\label{sec:simu_design}

We conduct simulations to evaluate the finite-sample performance of
the EB reweighting covariate adjustment approach for the sharp
RDD and RKD. The generation of the outcome variable
$Y_{i}$, the score $X_{i}$ and the first covariate $Z_{1,i}$ is
based on the simulation design of CCFT. The incorporation of additional
covariates $\widetilde{Z}_{i}=\left(Z_{2,i},\ldots,Z_{d_{z},i}\right)^{\top}$
follows that of \citet{arai2021regression}. We consider three DGPs:
DGP I has one covariate ($d_{z}=1$), DGP II has three covariates
($d_{z}=3$) and DGP III has five covariates ($d_{z}=5$). Let 
\begin{eqnarray*}
m_{y,-}\left(x\right) & \coloneqq & \alpha_{-}+0.96x+5.47x^{2}+15.28x^{3}+15.87x^{4}+5.14x^{5},\\
m_{y,+}\left(x\right) & \coloneqq & \alpha_{+}+0.62x-2.84x^{2}+8.42x^{3}-10.24x^{4}+4.31x^{5},\\
m_{z,-}\left(x\right) & \coloneqq & 0.49+\theta_{-}x+5.74x^{2}+17.14x^{3}+19.75x^{4}+7.47x^{5},\\
m_{z,+}\left(x\right) & \coloneqq & 0.49+\theta_{+}x-0.23x^{2}-3.46x^{3}+6.43x^{4}-3.48x^{5}.
\end{eqnarray*}
The observed data are a sample of size $n$ of $\left(Y_{i},X_{i},Z_{1,i},\dots,Z_{d_{z},i}\right)$,
where the score $X_{i}\sim2\times\mathrm{Beta}\left(2,4\right)-1$,
the first covariate 
\begin{eqnarray*}
Z_{1,i}= & \begin{cases}
m_{z,-}\left(X_{i}\right)+u_{z,i} & \text{if }X_{i}<0,\\
m_{z,+}\left(X_{i}\right)+u_{z,i} & \text{if }X_{i}>0,
\end{cases}
\end{eqnarray*}
additional covariates $\widetilde{Z}_{i}\sim\mathrm{N}\left(0,\Sigma\right)$
where each element of $\Sigma$ is $\mathrm{Cov}\left[Z_{j,i},Z_{k,i}\right]=0.5^{\left|j-k\right|}$
for any $j,k=2,\dots,d_{z}$, and the outcome 
\begin{eqnarray*}
Y_{i} & = & \begin{cases}
m_{y,-}\left(X_{i}\right)+\beta_{-}Z_{1,i}+\widetilde{Z}^{\top}_{i}\pi_{-}+u_{y,i} & \text{if }X_{i}<0,\\
m_{y,+}\left(X_{i}\right)+\beta_{+}Z_{1,i}+\widetilde{Z}^{\top}_{i}\pi_{+}+u_{y,i} & \text{if }X_{i}>0.
\end{cases}
\end{eqnarray*}
Error terms $\left(u_{y,i},u_{z,i}\right)$ are bivariate normal with
mean $0$, standard deviation $1$ and correlation coefficient $\rho=0.3$.
The additional covariates $\widetilde{Z}_{i}$ are drawn independently
of $X_{i}$ and $\left(u_{y,i},u_{z,i}\right)$. Throughout this section,
the sample sizes are $n=2{,}000$ and $5{,}000$. The covariate adjustment
uses the unconstrained balancing weights $\widetilde{w}_{p,i}\left(h\right)$
in (\ref{eq:w_til balancing weights}) and, in the bootstrap, their
analogues $\widetilde{w}^{*}_{p,i}\left(h\right)$. The number of
Monte Carlo replications is $1{,}000$.

\subsection{Quantile RDD}

\label{sec:simu_qrd}
For the DGPs in Section~\ref{sec:simu_design},
let $m_{-}\left(x\right)\coloneqq m_{y,-}\left(x\right)+\beta_{-}m_{z,-}\left(x\right)$,
$m_{+}\left(x\right)\coloneqq m_{y,+}\left(x\right)+\beta_{+}m_{z,+}\left(x\right)$,
$v_{-}\coloneqq1+2\rho\beta_{-}+\beta^{2}_{-}$ and $v_{+}\coloneqq1+2\rho\beta_{+}+\beta^{2}_{+}$.
The quantile function of $Y$ conditional on
$X$ is given by 
\[
Q_{Y\mid X}\left(\tau\mid x\right)=\begin{cases}
m_{+}\left(x\right)+\left(v_{+}+\pi^{\top}_{+}\Sigma\pi_{+}\right)^{1/2}\varPhi^{-1}\left(\tau\right) & \text{if }x>0\\
m_{-}\left(x\right)+\left(v_{-}+\pi^{\top}_{-}\Sigma\pi_{-}\right)^{1/2}\varPhi^{-1}\left(\tau\right) & \text{if }x<0
\end{cases}
\]
and thus the $\tau$-th QRD estimand is 
\begin{eqnarray*}
\vartheta\left(\tau\right) & = & \left\{ m_{+}\left(0\right)-m_{-}\left(0\right)\right\} +\left\{ \left(v_{+}+\pi^{\top}_{+}\Sigma\pi_{+}\right)^{1/2}-\left(v_{-}+\pi^{\top}_{-}\Sigma\pi_{-}\right)^{1/2}\right\} \varPhi^{-1}\left(\tau\right).
\end{eqnarray*}

We set the parameters $\alpha_{-}=0.36$, $\alpha_{+}=0.38$, $\beta_{-}=0$,
$\beta_{+}=3$, $\theta_{-}=1.06$, $\theta_{+}=0.61$, and $\pi_{-}=\pi_{+}=[\text{0.8}^{1},\dots,\text{0.8}^{d_{z}-1}]^{\top}$.
Tables~\ref{tab: qte_pw_dz_1}--\ref{tab: qte_pw_dz_5} report the
finite-sample performance of the quantile RDD estimators for selected
quantile levels $\tau$. We consider two approaches, labeled ``Standard''
and ``EBW'' in the tables. The standard estimator is the quantile
RDD estimator calculated as the regression coefficient of $I_{i}$ in \eqref{eq:quantile_RDD_QY}
without using covariate information. EBW refers to our EB reweighting
approach proposed in Section~\ref{subsec:quantile_RDD_reweighting}
that incorporates the covariate balance condition. The EBW CI corresponds to
$\mathit{CI}_{2,\alpha}\left(\tau\mid \widehat{h}(\tau)\right)$ in Section~\ref{subsec:quantile_RDD_inference} and is formed
following Algorithm~\ref{alg:qrd_ci} 
with $N_{B}=1{,}000$ bootstrap replications. The standard CI
is constructed analogously without using the covariate information. For the bandwidth choice, we first compute the AMSE-optimal
bandwidth at the median $\widehat{h}\left(0.5\right)$ using (\ref{eq:MSEbw_qte})
and then rescale it across quantiles by the Gaussian reference rule of
\citet{YuJones1998} to obtain bandwidths for
other values of $\tau$: 
\[
\widehat{h}^{\mathit{YJ}}\left(\tau\right)=\left(\frac{2\tau\left(1-\tau\right)}{\pi\phi^{2}\left(\varPhi^{-1}\left(\tau\right)\right)}\right)^{1/5}\cdot\widehat{h}\left(0.5\right).
\]
The reported bias and root mean squared error (RMSE) are 
those of the local linear ($p=1$) estimator evaluated at $\widehat{h}\left(\tau\right)$,
while the coverage probabilities and interval lengths are those of
the bootstrap intervals based on the local quadratic ($p=2$) fit
at the same bandwidth. 
The last column reports the resulting efficiency gain,
measured as the proportional reduction in the length of the $95\%$
CI relative to the standard estimator.

Both the standard and EBW quantile RDD estimators exhibit small bias,
and their bootstrap CIs attain approximately the nominal coverage
in every configuration. Relative to the standard estimator, the covariate-adjusted EBW
estimator has smaller RMSE and shorter CIs for every DGP, quantile
level, and sample size. The asymptotic efficiency improvement
of Theorem~\ref{thm:sharp QRD} therefore materializes in finite samples.
Averaged over the thirty configurations, covariate adjustment reduces
the length of the $90\%$ CI by $9.9\%$ (range: $5.3\%$--$14.2\%$)
and that of the $95\%$ CI by $9.8\%$ (range: $5.0\%$--$14.0\%$).
The reduction is largest at the
median, where interval length falls by $13.0\%$ on average, and is
smallest in the tails at $\tau=0.10$ and $\tau=0.90$, where it is
close to $7\%$.

Table~\ref{tab:qte_uni_dz} reports the performance of the UCB for
the QTE $\vartheta(\tau)$ with and without covariate adjustment.
The EBW UCB corresponds to 
$\left\{ \mathit{CB}_{2,\alpha}\left(\tau\mid\widehat{h}\left(\cdot\right)\right):\tau\in[0.1,0.9]\right\}$ in
Section~\ref{subsec:quantile_RDD_inference} and is formed following Algorithm~\ref{alg:qrd_ucb}. The standard UCB is constructed
analogously without using the covariate information. The last column reports the proportional reduction in the width of the $95\%$ EBW UCB relative
to the standard UCB. In every configuration, the average width of the EBW UCB 
taken across the quantile levels $\tau\in[0.1,0.9]$ is smaller than
that of the standard UCB. Specifically, covariate
adjustment narrows the $90\%$ UCB by $9.7\%$ on average and the
$95\%$ UCB by $10.1\%$. The reductions range from $8.2\%$ to $11.4\%$
across configurations.

\begin{table}[!t]
\centering \caption{Finite-sample performance of quantile RDD estimators under DGP I:
$\tau$ = quantile level, CP = coverage probability for the confidence
interval, CIL = average length of the confidence interval. Eff.\ gain
is the proportional reduction in CIL95 achieved by the EBW estimator
relative to the standard estimator without covariates.}
\label{tab: qte_pw_dz_1} %
\begin{tabular}{ccllcccccccccc}
\toprule 
 \multicolumn{1}{c}{$\tau$} & \multicolumn{1}{c}{$n$} & \multicolumn{1}{l}{Method} & \multicolumn{1}{l}{Covariate} & \multicolumn{1}{c}{Bias} & \multicolumn{1}{c}{RMSE} & \multicolumn{1}{c}{CP90} & \multicolumn{1}{c}{CP95} & \multicolumn{1}{c}{CIL90} & \multicolumn{1}{c}{CIL95} & \multicolumn{1}{c}{Eff.\ gain} \tabularnewline
\midrule 
 0.10  & 2,000  & Standard  & No  & 0.080  & 0.595  & 0.926  & 0.959  & 2.408  & 2.910  &  \tabularnewline
  &  & EBW  & Yes  & 0.030  & 0.518  & 0.947  & 0.976  & 2.235  & 2.718  & 6.6\%  \tabularnewline
  & 5,000  & Standard  & No  & 0.057  & 0.381  & 0.937  & 0.974  & 1.603  & 1.928  &  \tabularnewline
  &  & EBW  & Yes  & 0.017  & 0.328  & 0.955  & 0.982  & 1.495  & 1.792  & 7.1\%  \tabularnewline
\midrule 
 0.25  & 2,000  & Standard  & No  & 0.029  & 0.497  & 0.932  & 0.971  & 1.969  & 2.364   & \tabularnewline
  &  & EBW  & Yes  & -0.023  & 0.393  & 0.954  & 0.982  & 1.737  & 2.085  & 11.8\%  \tabularnewline
  & 5,000  & Standard  & No  & 0.032  & 0.316  & 0.917  & 0.964  & 1.320  & 1.581    & \tabularnewline
  &  & EBW  & Yes  & -0.011  & 0.251  & 0.940  & 0.971  & 1.168  & 1.398  & 11.6\% \tabularnewline
\midrule 
 0.50  & 2,000  & Standard  & No  & 0.059  & 0.493  & 0.928  & 0.972  & 1.843  & 2.211  & \tabularnewline
  &  & EBW  & Yes  & 0.001  & 0.389  & 0.951  & 0.979  & 1.582  & 1.902  & 14.0\% \tabularnewline
  & 5,000  & Standard  & No  & 0.042  & 0.309  & 0.921  & 0.966  & 1.229  & 1.473   & \tabularnewline
  &  & EBW  & Yes  & -0.002  & 0.233  & 0.945  & 0.976  & 1.064  & 1.276  & 13.4\% \tabularnewline
\midrule 
 0.75  & 2,000  & Standard  & No  & 0.071  & 0.479  & 0.934  & 0.971  & 1.953  & 2.350   & \tabularnewline
  &  & EBW  & Yes  & 0.008  & 0.395  & 0.942  & 0.977  & 1.703  & 2.053  & 12.6\%  \tabularnewline
  & 5,000  & Standard  & No  & 0.073  & 0.307  & 0.933  & 0.972  & 1.291  & 1.549  &  \tabularnewline
  &  & EBW  & Yes & 0.024  & 0.240  & 0.960  & 0.981  & 1.134  & 1.358  & 12.3\% \tabularnewline
\midrule 
 0.90  & 2,000  & Standard  & No  & 0.062  & 0.581  & 0.966  & 0.988  & 2.404  & 2.899   & \tabularnewline
  &  & EBW  & Yes  & -0.006  & 0.507  & 0.972  & 0.990  & 2.210  & 2.668  & 8.0\%  \tabularnewline
  & 5,000  & Standard  & No  & 0.043  & 0.376  & 0.938  & 0.981  & 1.556  & 1.871  &  \tabularnewline
  &  & EBW  & Yes  & -0.007  & 0.327  & 0.949  & 0.979  & 1.438  & 1.730  & 7.5\%  \tabularnewline
\bottomrule
\end{tabular}
\end{table}

\begin{table}[!t]
\centering \caption{Finite-sample performance of quantile RDD estimators under DGP II:
$\tau$ = quantile level, CP = coverage probability for the confidence
interval, CIL = average length of the confidence interval. Eff.\ gain
is the proportional reduction in CIL95 achieved by the EBW estimator
relative to the standard estimator without covariates.}
\label{tab: qte_pw_dz_3} %
\begin{tabular}{ccllcccccccccc}
\toprule 
\multicolumn{1}{c}{$\tau$} & \multicolumn{1}{c}{$n$} & \multicolumn{1}{l}{Method} & \multicolumn{1}{l}{Covariate} & \multicolumn{1}{c}{Bias} & \multicolumn{1}{c}{RMSE} & \multicolumn{1}{c}{CP90} & \multicolumn{1}{c}{CP95} & \multicolumn{1}{c}{CIL90} & \multicolumn{1}{c}{CIL95} & \multicolumn{1}{c}{Eff.\ gain} \tabularnewline
\midrule 
 0.10  & 2,000  & Standard  & No  & 0.067  & 0.659  & 0.925  & 0.968  & 2.581  & 3.114    & \tabularnewline
  &  & EBW  & Yes  & 0.034  & 0.578  & 0.940  & 0.970  & 2.421  & 2.929  & 5.9\%  \tabularnewline
  & 5,000  & Standard  & No  & 0.064  & 0.397  & 0.929  & 0.973  & 1.715  & 2.061    & \tabularnewline
  &  & EBW  & Yes  & 0.005  & 0.359  & 0.950  & 0.978  & 1.605  & 1.933  & 6.2\%  \tabularnewline
\midrule 
 0.25  & 2,000  & Standard  & No  & 0.025  & 0.543  & 0.928  & 0.963  & 2.100  & 2.515    & \tabularnewline
  &  & EBW  & Yes  & -0.010  & 0.427  & 0.947  & 0.974  & 1.867  & 2.245  & 10.7\%  \tabularnewline
  & 5,000  & Standard  & No  & 0.031  & 0.323  & 0.928  & 0.970  & 1.400  & 1.679    & \tabularnewline
  &  & EBW  & Yes  & -0.027  & 0.272  & 0.934  & 0.978  & 1.244  & 1.492  & 11.1\% \tabularnewline
\midrule 
 0.50  & 2,000  & Standard  & No  & 0.043  & 0.546  & 0.926  & 0.959  & 1.976  & 2.368    & \tabularnewline
  &  & EBW  & Yes  & -0.004  & 0.412  & 0.947  & 0.982  & 1.714  & 2.061  & 13.0\% \tabularnewline
  & 5,000  & Standard  & No  & 0.041  & 0.320  & 0.943  & 0.983  & 1.305  & 1.566  &   \tabularnewline
  &  & EBW  & Yes  & -0.019  & 0.257  & 0.945  & 0.982  & 1.135  & 1.361  & 13.1\%  \tabularnewline
\midrule 
 0.75  & 2,000  & Standard  & No  & 0.057  & 0.539  & 0.918  & 0.962  & 2.078  & 2.497  & \tabularnewline
  &  & EBW  & Yes  & 0.013  & 0.439  & 0.952  & 0.979  & 1.833  & 2.211  & 11.5\% \tabularnewline
  & 5,000  & Standard  & No  & 0.086  & 0.336  & 0.903  & 0.965  & 1.375  & 1.651  &  \tabularnewline
  &  & EBW  & Yes  & 0.018  & 0.263  & 0.926  & 0.975  & 1.214  & 1.457  & 11.8\%  \tabularnewline
\midrule 
 0.90  & 2,000  & Standard  & No  & 0.050  & 0.609  & 0.956  & 0.986  & 2.571  & 3.098    & \tabularnewline
  &  & EBW  & Yes  & -0.003  & 0.551  & 0.971  & 0.993  & 2.385  & 2.883  & 6.9\%  \tabularnewline
  & 5,000  & Standard  & No  & 0.044  & 0.396  & 0.944  & 0.978  & 1.680  & 2.013  &   \tabularnewline
  &  & EBW  & Yes  & -0.015  & 0.357  & 0.947  & 0.972  & 1.556  & 1.864  & 7.4\% \tabularnewline
\bottomrule
\end{tabular}
\end{table}

\begin{table}[!t]
\centering \caption{Finite-sample performance of quantile RDD estimators under DGP III:
$\tau$ = quantile level, CP = coverage probability for the confidence
interval, CIL = average length of the confidence interval. Eff.\ gain
is the proportional reduction in CIL95 achieved by the EBW estimator
relative to the standard estimator without covariates.}
\label{tab: qte_pw_dz_5} %
\begin{tabular}{ccllcccccccccc}
\toprule 
 \multicolumn{1}{c}{$\tau$} & \multicolumn{1}{c}{$n$} & \multicolumn{1}{l}{Method} & \multicolumn{1}{l}{Covariate} & \multicolumn{1}{c}{Bias} & \multicolumn{1}{c}{RMSE} & \multicolumn{1}{c}{CP90} & \multicolumn{1}{c}{CP95} & \multicolumn{1}{c}{CIL90} & \multicolumn{1}{c}{CIL95} & \multicolumn{1}{c}{Eff.\ gain} \tabularnewline
\midrule 
  0.10  & 2,000  & Standard  & No  & 0.082  & 0.667  & 0.943  & 0.976  & 2.742  & 3.301   & \tabularnewline
   &  & EBW  & Yes  & 0.035  & 0.596  & 0.956  & 0.981  & 2.597  & 3.135  & 5.0\%  \tabularnewline
   & 5,000  & Standard  & No  & 0.053  & 0.410  & 0.945  & 0.970  & 1.811  & 2.174    & \tabularnewline
   &  & EBW  & Yes  & 0.010  & 0.358  & 0.956  & 0.982  & 1.681  & 2.025  & 6.9\%  \tabularnewline
\midrule 
  0.25  & 2,000  & Standard  & No  & 0.026  & 0.545  & 0.931  & 0.963  & 2.196  & 2.641    & \tabularnewline
   &  & EBW  & Yes  & -0.020  & 0.461  & 0.938  & 0.986  & 1.995  & 2.399  & 9.2\%  \tabularnewline
   & 5,000  & Standard  & No  & 0.030  & 0.361  & 0.929  & 0.963  & 1.479  & 1.770    & \tabularnewline
   &  & EBW  & Yes  & -0.021  & 0.282  & 0.941  & 0.972  & 1.321  & 1.578  & 10.8\%\tabularnewline
\midrule 
  0.50  & 2,000  & Standard  & No  & 0.042  & 0.539  & 0.927  & 0.970  & 2.091  & 2.506  & \tabularnewline
   &  & EBW  & Yes  & -0.013  & 0.438  & 0.943  & 0.981  & 1.841  & 2.219  & 11.5\% \tabularnewline
   & 5,000  & Standard  & No  & 0.044  & 0.363  & 0.906  & 0.955  & 1.380  & 1.655  &   \tabularnewline
   &  & EBW  & Yes  & -0.005  & 0.288  & 0.921  & 0.970  & 1.201  & 1.437  & 13.2\%  \tabularnewline
\midrule 
  0.75  & 2,000  & Standard  & No  & 0.084  & 0.523  & 0.936  & 0.970  & 2.185  & 2.624    & \tabularnewline
   &  & EBW  & Yes  & 0.029  & 0.452  & 0.937  & 0.978  & 1.957  & 2.359  & 10.1\%  \tabularnewline
   & 5,000  & Standard  & No  & 0.085  & 0.379  & 0.914  & 0.948  & 1.451  & 1.741  &   \tabularnewline
   &  & EBW  & Yes  & 0.034  & 0.300  & 0.934  & 0.967  & 1.288  & 1.546  & 11.2\%  \tabularnewline
\midrule 
  0.90  & 2,000  & Standard  & No  & 0.044  & 0.641  & 0.966  & 0.984  & 2.707  & 3.263    & \tabularnewline
   &  & EBW  & Yes  & -0.036  & 0.584  & 0.967  & 0.987  & 2.530  & 3.059  & 6.3\%  \tabularnewline
   & 5,000  & Standard  & No  & 0.045  & 0.439  & 0.924  & 0.963  & 1.772  & 2.126  &   \tabularnewline
   &  & EBW  & Yes  & -0.002  & 0.384  & 0.941  & 0.976  & 1.653  & 1.986  & 6.6\% \tabularnewline
\bottomrule
\end{tabular}
\end{table}

\begin{table}[!t]
\centering \caption{Uniform confidence bands for quantile RDD: $\tau\in[0.1,0.9]$ with
step size $0.01$, no covariate adjustment versus balancing covariate
adjustment. SCP = simultaneous coverage probability for the uniform
confidence band, CBW = average width of the uniform confidence band.
Bandwidth: AMSE-optimal. Eff.\ gain is the proportional reduction
in CBW95 achieved by the EBW estimator relative to the standard estimator
without covariates.}
\label{tab:qte_uni_dz} %
\begin{tabular}{ccclcccccc}
\toprule 
\multicolumn{1}{c}{DGP} & \multicolumn{1}{c}{} & \multicolumn{1}{c}{$n$} & \multicolumn{1}{l}{Method} & \multicolumn{1}{l}{Covariate} & \multicolumn{1}{c}{SCP90} & \multicolumn{1}{c}{SCP95} & \multicolumn{1}{c}{CBW90} & \multicolumn{1}{c}{CBW95} & \multicolumn{1}{c}{Eff.\ gain}\tabularnewline
\midrule 
I  &  & 2,000  & Standard  & No  & 0.948  & 0.982  & 3.899  & 4.492  & \tabularnewline
 &  &  & EBW  & Yes  & 0.968  & 0.990  & 3.478  & 3.982  & 11.4\%\tabularnewline
\midrule 
 &  & 5,000  & Standard  & No  & 0.957  & 0.984  & 2.445  & 2.780  & \tabularnewline
 &  &  & EBW  & Yes  & 0.967  & 0.989  & 2.194  & 2.482  & 10.7\%\tabularnewline
\midrule 
II  &  & 2,000  & Standard  & No  & 0.935  & 0.976  & 4.116  & 4.729  & \tabularnewline
 &  &  & EBW  & Yes  & 0.966  & 0.994  & 3.711  & 4.247  & 10.2\%\tabularnewline
\midrule 
 &  & 5,000  & Standard  & No  & 0.950  & 0.984  & 2.569  & 2.917  & \tabularnewline
 &  &  & EBW  & Yes  & 0.965  & 0.987  & 2.323  & 2.624  & 10.0\%\tabularnewline
\midrule 
III  &  & 2,000  & Standard  & No  & 0.952  & 0.983  & 4.335  & 4.984  & \tabularnewline
 &  &  & EBW  & Yes  & 0.958  & 0.993  & 3.981  & 4.568  & 8.3\%\tabularnewline
\midrule 
 &  & 5,000  & Standard  & No  & 0.944  & 0.977  & 2.713  & 3.084  & \tabularnewline
 &  &  & EBW  & Yes  & 0.948  & 0.981  & 2.453  & 2.775  & 10.0\%\tabularnewline
\bottomrule
\end{tabular}
\end{table}

\subsection{Regression kink}

\label{sec:simu_kink}

To mimic a sharp RKD, we choose the parameters of the DGPs in Section~\ref{sec:simu_design} to satisfy
two restrictions: $\alpha_{-}+0.49\beta_{-}=\alpha_{+}+0.49\beta_{+}$
and $\theta_{-}=\theta_{+}$, which guarantee, respectively, that
the conditional mean of the outcome is continuous at the
cutoff: $\lim_{x\downarrow0}\mathrm{E}\left[Y_{i}\mid X_{i}=x\right]=\lim_{x\uparrow0}\mathrm{E}\left[Y_{i}\mid X_{i}=x\right]$,
and that the first covariate satisfies the derivative balance condition:
$\lim_{x\downarrow0}\partial\mathrm{E}\left[Z_{1,i}\mid X_{i}=x\right]/\partial x=\lim_{x\uparrow0}\partial\mathrm{E}\left[Z_{1,i}\mid X_{i}=x\right]/\partial x$.
We set $\beta_{-}=0.22$, $\beta_{+}=0.28$, $\alpha_{-}=0.36$, $\alpha_{+}=\alpha_{-}+0.49(\beta_{-}-\beta_{+})=0.3306$,
$\theta_{-}=\theta_{+}=1.06$, $\pi_{-}=\left[(-0.8)^{1},\dots,(-0.8)^{d_{z}-1}\right]^{\top}$
and $\pi_{+}=\left[\text{0.8}^{1},\dots,\text{0.8}^{d_{z}-1}\right]^{\top}$.\footnote{If $\pi_{-}=\pi_{+}$, imposing the level balance condition $\lim_{x\downarrow0}\mathrm{E}\left[Z_{j,i}\mid X_{i}=x\right]=\lim_{x\uparrow0}\mathrm{E}\left[Z_{j,i}\mid X_{i}=x\right]$
for $j=2,\dots,d_{z}$ does not reduce the asymptotic variance of
the sharp RKD estimator.} Recall that the estimand is $\vartheta_{\dagger}=\mu^{\left(1\right)}_{Y,+}-\mu^{\left(1\right)}_{Y,-}$,
which is proportional to the sharp RKD parameter.

We examine the finite-sample performance of three approaches, labeled
``Standard'', ``Reg'' and ``EBW'' in Table~\ref{tab: kink}:
the first does not use the covariates $Z$ and its point estimate
corresponds to the regression coefficient of $I_{i}X_{i}$ in \eqref{eq:LP without covariates};
the second incorporates the covariate balance in the first derivative,
i.e., $\mu^{\left(1\right)}_{Z,+}=\mu^{\left(1\right)}_{Z,-}$, using
regression adjustment as discussed in Section~\ref{subsec:RKD_regression};
the third is our EB reweighting proposal, which incorporates the covariate
balance conditions both in levels and in first derivatives, as discussed
in Section~\ref{subsec:RKD_reweighting}. We take the point estimates
from the local quadratic regression ($p=2$) and select the bandwidth $\widehat{h}$ by (\ref{eq:rkd-h}). The CI of EBW corresponds to $\mathit{CI}^{\dagger}_{3,\alpha}(\widehat{h})$ in Section~\ref{subsec:RKD_inference}. The CIs of the standard and Reg estimators are formed analogously,
by setting $\widetilde{\gamma}=0$ in \eqref{eq:zeta_til}
and $\widetilde{\gamma}_{\dagger,p}=0$ in \eqref{eq:scr_V_til} for the former,
and by setting $\widetilde{\gamma}_{\dagger,p}=0$ for the latter.

According to Table~\ref{tab: kink}, both the Reg and the EBW
estimators lead to efficiency improvements relative to the standard
estimator without covariate adjustment. The last column reports the
proportional reduction in CIL95 achieved by each of the two covariate-adjusted
estimators relative to the standard estimator. The efficiency gain
ranges from $11.8\%$ to $24.0\%$ when only the derivative balance
condition is imposed and from $11.6\%$ to $38.4\%$ when the level
balance condition is added. In DGPs II and III, the EBW estimator,
which imposes both balance conditions, improves further on the Reg
estimator. This ordering is consistent with Theorem~\ref{thm:kink}.
%which delivers $\mathscr{V}_{\star,p}\leq\mathscr{V}_{\dagger,p}\leq\omega^{0,2}_{p,1}\mathrm{Var}_{\pm}\left[Y\right]/\varphi$. The second inequality measures the asymptotic variance reduction due to the incorporation of the derivative balancing condition $\mu^{\left(1\right)}_{Z,+}=\mu^{\left(1\right)}_{Z,-}$:the reduction equals $\omega^{0,2}_{p,1}\gamma^{\top}\mathrm{Var}_{\pm}\left[Z\right]\gamma/\varphi$ and is thus driven by the part of the variation of the outcome atthe cutoff that $Z$ explains. The first inequality measures the efficiency improvement brought by the level balance condition $\mu_{Z,+}=\mu_{Z,-}$. The corresponding variance reduction equals $\omega^{0,2}_{p,0}\gamma^{\top}_{\dagger,p}\mathrm{Var}_{\pm}\left[Z\right]\gamma_{\dagger,p}/\varphi$,where $\gamma_{\dagger,p}$ is proportional to $\mathrm{Cov}_{+}\left[Z,\epsilon\right]-\mathrm{Cov}_{-}\left[Z,\epsilon\right]$,the difference across the two sides of the cutoff in the covariancebetween the covariates and the covariate-adjusted outcome $\epsilon=Y-Z^{\top}\gamma$.
In DGP I, the covariate $Z$ is related to $Y$ in much the same way
on either side of the cutoff, so that $\mathrm{Cov}_{+}\left[Z,\epsilon\right]-\mathrm{Cov}_{-}\left[Z,\epsilon\right]$
and thus $\gamma_{\dagger,p}$ in \eqref{eq:gamma_RKD} are close
to zero. As a result, the additional inclusion of the level balance
condition hardly reduces the asymptotic variance. %Whenever $Z$ is related to $Y$ in much the same way on either side of the cutoff, the level balance condition therefore adds little to the derivative balance condition. This is precisely what happens in the DGP I, in which the balance condition in level hardly reduces the asymptotic variance because its efficiency improvement is proportional to $(\beta_{+}-\beta_{-})^{2}$, which takes a tiny value in this DGP.

\begin{table}[!t]
\centering \caption{Finite-sample performance of estimators for the sharp RKD: covariate
balance conditions (Cov. Cond): $\mu_{Z,+}=\mu_{Z,-}$ (Level) and
$\mu^{\left(1\right)}_{Z,+}=\mu^{\left(1\right)}_{Z,-}$ (Deri); CP
= coverage probability for the confidence interval; CIL = average
length of the confidence interval. Eff.\ gain is the proportional
reduction in CIL95 relative to the standard estimator without using
any covariate information.}
\label{tab: kink} %
\begin{tabular}{cllccccccccc}
\toprule 
\multicolumn{1}{c}{DGP} & \multicolumn{1}{c}{$n$} & \multicolumn{1}{l}{Method} & \multicolumn{2}{l}{Cov. Cond} & \multicolumn{1}{c}{Bias} & \multicolumn{1}{c}{RMSE} & \multicolumn{1}{c}{CP90} & \multicolumn{1}{c}{CP95} & \multicolumn{1}{c}{CIL90} & \multicolumn{1}{c}{CIL95} & \multicolumn{1}{c}{Eff.\ gain}\tabularnewline
\multicolumn{1}{c}{} & \multicolumn{1}{c}{} & \multicolumn{1}{c}{} & \multicolumn{1}{l}{Level} & \multicolumn{1}{l}{Deri} &  &  &  &  &  &  & \tabularnewline
\midrule 
I  & 2,000  & Standard  & No  & No  & 0.460  & 6.12  & 0.925  & 0.963  & 17.0  & 20.2  & \tabularnewline
 &  & Reg  & No  & Yes  & 0.255  & 5.33  & 0.927  & 0.967  & 15.0  & 17.8  & 11.8\%\tabularnewline
 &  & EBW  & Yes  & Yes  & 0.255  & 5.35  & 0.928  & 0.967  & 15.0  & 17.9  & 11.6\%\tabularnewline
\cmidrule{3-12}
 & 5,000  & Standard  & No  & No  & 0.686  & 3.27  & 0.867  & 0.940  & 7.8  & 9.3  & \tabularnewline
 &  & Reg  & No  & Yes  & 0.250  & 2.85  & 0.894  & 0.954  & 6.9  & 8.2  & 11.9\%\tabularnewline
 &  & EBW  & Yes  & Yes  & 0.253  & 2.85  & 0.895  & 0.951  & 6.9  & 8.2  & 11.8\%\tabularnewline
\midrule 
II  & 2,000  & Standard  & No  & No  & 0.551  & 7.26  & 0.911  & 0.960  & 18.6  & 22.2  & \tabularnewline
 &  & Reg  & No  & Yes  & 0.195  & 6.23  & 0.876  & 0.941  & 15.2  & 18.1  & 18.4\%\tabularnewline
 &  & EBW  & Yes  & Yes  & 0.247  & 5.19  & 0.882  & 0.944  & 13.2  & 15.8  & 28.9\%\tabularnewline
\cmidrule{3-12}
 & 5,000  & Standard  & No  & No  & 0.601  & 4.50  & 0.903  & 0.954  & 11.3  & 13.5  & \tabularnewline
 &  & Reg  & No  & Yes  & 0.254  & 3.80  & 0.871  & 0.940  & 9.0  & 10.7  & 20.7\%\tabularnewline
 &  & EBW  & Yes  & Yes  & 0.303  & 3.10  & 0.872  & 0.939  & 7.7  & 9.2  & 32.2\%\tabularnewline
\midrule 
III  & 2,000  & Standard  & No  & No  & 0.694  & 8.65  & 0.911  & 0.958  & 22.5  & 26.8  & \tabularnewline
 &  & Reg  & No  & Yes  & 0.266  & 7.28  & 0.873  & 0.946  & 17.9  & 21.4  & 20.2\%\tabularnewline
 &  & EBW  & Yes  & Yes  & 0.296  & 5.69  & 0.896  & 0.950  & 15.3  & 18.3  & 31.8\%\tabularnewline
\cmidrule{3-12}
 & 5,000  & Standard  & No  & No  & 0.709  & 5.26  & 0.903  & 0.958  & 13.6  & 16.2  & \tabularnewline
 &  & Reg  & No  & Yes  & 0.316  & 4.35  & 0.860  & 0.932  & 10.4  & 12.4  & 24.0\%\tabularnewline
 &  & EBW  & Yes  & Yes  & 0.382  & 3.38  & 0.860  & 0.928  & 8.4  & 10.0  & 38.4\%\tabularnewline
\bottomrule
\end{tabular}
\end{table}

%%%%%%%%%%%%%%%%%%%%%%%%%%%%%%%%%%%%%%%%%%%%%%%%%%%%%%%%%%%%%

\section{Empirical illustration}

%%%%%%%%%%%%%%%%%%%%%%%%%%%%%%%%%%%%%%%%%%%%%%%%%%%%%%%%%%%%%
%We present two empirical applications, one based on the quantile RDD (Section \ref{sec:islamic}) and the other on the RKD (Section \ref{subsec:landais}).
%%%%%%%%%%%%%%%%%%%%%%%%%%%%%%%%%%%%%%%%%%%%%%%%%%%%%%%%%%%%%
%\subsection{Islamic political representation and women's educational attainment}\label{sec:islamic}
%%%%%%%%%%%%%%%%%%%%%%%%%%%%%%%%%%%%%%%%%%%%%%%%%%%%%%%%%%%%%
We apply our EB reweighting covariate adjustment method to the Turkish
election data originally analyzed by \citet{Meyersson2014} and subsequently
revisited by \citet{Cattaneo2019} as a leading empirical application
of RDDs. This application examines the causal effects of Islamic political
representation on women's educational attainment. The dataset covers
$2{,}629$ Turkish municipalities, and the design is a sharp RDD. The
score $X$ is the vote-share margin of the Islamic party in the 1994
mayoral elections, while the outcome $Y$ is the percentage
of women aged 15--20 who had completed high school by the year 2000.
We include four covariates: the Islamic party's vote share in 1994,
the number of parties receiving votes in 1994, the logarithm of the
1994 population, and an indicator for district-center status. Table~\ref{empi: meyersson}
reports the estimated average treatment effect and QTEs, both with
and without covariate adjustment. The average treatment effect is
estimated using the \texttt{R} package \texttt{rdrobust} \citep{calonico2015rdrobust,Cattaneo2019}.
The QTEs and their $90\%$ CIs are computed using our method developed
in Section~\ref{subsec:Sharp-quantile-RD}, with the EB weights given
in \eqref{eq:w_til balancing weights}. In all specifications, the
bandwidth is the covariate-adjusted AMSE-optimal bandwidth in \eqref{eq:MSEbw_qte}.

Previous analyses of these data have focused on the average treatment
effect at the cutoff, documenting a positive effect of Islamic political
representation on the educational attainment of young women; see the
last two rows of Table~\ref{empi: meyersson}.\footnote{Table II of \citet{Meyersson2014} reports a mean estimate (without
covariates) of $3.2$ with a standard error of $1.0$ at a bandwidth
of $24$, similar to the mean estimates reported in Table~\ref{empi: meyersson}.} For the average treatment effect, the estimated impact is already
statistically significant even without incorporating covariate information.
Consequently, although covariate adjustment yields a shorter CI, it
does not affect the substantive conclusion regarding statistical significance.

\begin{table}[!t]
\centering \caption{Estimates of the $\tau$-th QTEs and the average treatment effect
of Islamic rule on women's education. \textquotedblleft Covariates\textquotedblright{}
indicates whether covariate adjustment is employed. \textquotedblleft Bandwidth\textquotedblright{}
reports the AMSE-optimal bandwidth with covariate adjustment. The
sample size is $n=2{,}629$. Efficiency gains are computed as the
percentage reduction in the length of the confidence interval when
covariate adjustment is used relative to when it is not.}
\label{empi: meyersson} %
\begin{tabular}{lccccc}
\toprule 
Estimand  & Covariates  & Bandwidth  & Estimate  & 90\% CI  & Efficiency gain\tabularnewline
\midrule 
$\tau=0.25$  & No  & 29.8  & 3.070  & $[-0.051,\,5.751]$  & \tabularnewline
 & Yes  & 29.8  & 3.070  & $[\phantom{-}0.333,\,5.635]$  & 8.6\%\tabularnewline
\cmidrule{1-1}\cmidrule(lr){2-6}
$\tau=0.50$  & No  & 28.8  & 3.728  & $[-0.223,\,6.900]$  & \tabularnewline
 & Yes  & 28.8  & 3.734  & $[\phantom{-}0.006,\,6.482]$  & 9.1\%\tabularnewline
\cmidrule{1-1}\cmidrule(lr){2-6}
$\tau=0.75$  & No  & 29.8  & 2.683  & $[-1.357,\,6.823]$  & \tabularnewline
 & Yes  & 29.8  & 2.558  & $[-1.274,\,6.380]$  & 6.4\%\tabularnewline
\midrule 
\multicolumn{5}{r}{Average efficiency gain over $\tau\in[0.2,0.8]$} & 8.3\%\tabularnewline
\midrule 
Mean  & No  & 27.4  & 2.949  & $[\phantom{-}0.069,\,5.621]$  & \tabularnewline
 & Yes  & 27.4  & 3.003  & $[\phantom{-}0.918,\,5.617]$  & 15.4\%\tabularnewline
\midrule 
\addlinespace
%\multicolumn{6}{l}{\footnotesize Average efficiency gain over $\tau\in[0.2,0.8]$: 8.32\%.}
 &  &  &  &  & \tabularnewline
\end{tabular}
\end{table}

We now turn to the corresponding QTEs. The point estimates reported
in the first three row groups of Table~\ref{empi: meyersson} reveal
heterogeneity across the outcome distribution. Using covariate adjustment
together with the AMSE-optimal bandwidth, the estimated treatment
effect is approximately $3.07$ percentage points at the $0.25$ quantile,
increases to $3.73$ percentage points at the median, and then declines
to $2.56$ percentage points at the $0.75$ quantile.

Table~\ref{empi: meyersson} further shows that incorporating covariates
has little impact on the point estimates but improves their precision:
covariate adjustment shortens the CIs by $6.4\%$ to $9.1\%$ at the
reported quantiles and by $8.3\%$ on average over $\tau\in[0.2,0.8]$.
In particular, at $\tau=0.25$ and $\tau=0.50$, the covariate-adjusted
estimator yields statistically significant treatment effects, whereas
the corresponding estimators without covariate adjustment do not.
Thus, unlike in the mean RDD analysis, covariate adjustment changes
the empirical conclusion for the QTEs. This highlights the practical
importance of our proposed method, which provides valid and more efficient
covariate adjustment in nonlinear settings such as the quantile RDD.

\section{Conclusion}

This paper proposes a generic approach to covariate adjustment in
RDDs and RKDs. Rather than entering the covariates additively in the
local polynomial regression, we reweight the local objective function
by the EB weights that equate the weighted covariate moments on the
two sides of the cutoff. Because the adjustment operates through the
weights rather than through the specification, our EB reweighting-based
covariate adjustment applies to a broad class of RDD-related estimands
defined by a local optimization problem of the form \eqref{eq:general_objective_uniform_weights},
whether the estimand is linear or nonlinear, and whether it is a level
or a derivative at the cutoff. 

For linear estimands in levels, as in the sharp RDD, our estimator
is first-order equivalent to the regression-based covariate adjustment
of CCFT, so nothing is lost where the existing method is already efficient.
For nonlinear estimands, as in the quantile RDD, the regression-based
adjustment is generally inconsistent, whereas our estimator remains
consistent and attains a smaller asymptotic variance than its counterpart
without covariates. We also provide bootstrap confidence intervals
and uniform confidence bands. For estimands in derivatives, as in
the RKD and the TED of \citet{Dong2015}, the regression-based adjustment
exploits only the derivative balance conditions, whereas our reweighting
approach exploits the level balance conditions as well and lowers
the asymptotic variance further. The Monte Carlo experiments and the
empirical illustration confirm that these gains materialize in finite
samples without distorting coverage.

The application of EB reweighting-based covariate adjustment is not
limited to the designs studied here. Extensions to other nonlinear
estimands such as those studied by \citet*{Chiang2019}, \citet*{Huang2022},
\citet{Qu2024a}, and \citet*{xu2017regression,Xu2018} are left for
future research.

%\subsubsection*{Declaration of generative AI and AI-assisted technologies in the manuscript preparation process}
%During the preparation of this work, the authors used AI-assisted technologies (Claude and Codex) for language refinement and readability improvements. After using these tools, the authors reviewed and edited the content as needed and take full responsibility for the content of the published article.

\bibliographystyle{chicago}
\bibliography{generic}

\appendix
%% Single appendix: print the heading as "Appendix  Bandwidth selection" without the
%% letter; subsections keep the numbers A.1 and A.2.
\makeatletter
\renewcommand{\section@cntformat}{\appendixname\quad}
\makeatother

\section{Bandwidth selection}

\label{app:bdw}

In this appendix, for a vector $x$, $x^{\left[j\right]}$ denotes the $j$-th
coordinate of $x$.

\subsection{Quantile RDD}

\label{app:bdw_qrd}

\subsubsection*{Step 1: Determining the preliminary bandwidth and estimating $\left(\varphi,\varphi_{+}\left(\tau\right),\varphi_{-}\left(\tau\right)\right)$}

For a generic pilot bandwidth $\ell$, let 
\[
\widetilde{\varphi}\left(\ell\right)\coloneqq\frac{1}{n\ell}\sum_{i}K\left(\frac{X_{i}}{\ell}\right)
\]
denote the ordinary kernel density estimator of $\varphi$. Let $s_{X}$
denote the sample standard deviation and $\mathit{IQR}_{X}$ denote
the sample IQR of $X$. Let $\varpi_{1}\coloneqq\int u^{2}K\left(u\right)\mathrm{d}u$,
$\varpi_{2}\coloneqq\int K^{2}\left(u\right)\mathrm{d}u$, and 
\[
h^{X}_{rot}\coloneqq\min\left\{ s_{X},\frac{\mathit{IQR}_{X}}{1.349}\right\} \left(\frac{8\sqrt{\pi}}{3}\cdot\frac{\varpi_{2}}{\varpi^{2}_{1}}\right)^{1/5}n^{-1/5}
\]
be the modified Silverman rule-of-thumb (ROT) bandwidth. Let $h^{Y}_{rot}$
be defined by the same formula using the sample standard deviation
and IQR of $Y$. In the implementation, we take the estimator of $\varphi$
to be $\widetilde{\varphi}\coloneqq\widetilde{\varphi}\left(h^{X}_{rot}\right)$.

Let 
\[
\widetilde{\kappa}_{Y,+}\left(\tau\mid\ell\right)\coloneqq\mathrm{e}^{\top}_{2,1}\underset{\beta\in\mathbb{R}^{2}}{\arg\min}\sum_{i}K\left(\frac{X_{i}}{\ell}\right)I_{i}\rho_{\tau}\left(Y_{i}-r^{\top}_{1}\left(X_{i}\right)\beta\right).
\]
We consider the following simulation-based implementation of the conditional
density estimator of \citet[Section S.1]{Qu2019}. Let $U_{1},\dots,U_{M}$
be i.i.d. draws from $\mathrm{Uniform}\left(0,1\right)$ for some
large $M$. Let 
\[
\widetilde{\varphi}_{+}\left(\tau\mid\ell\right)\coloneqq\frac{1}{M\ell}\sum^{M}_{i=1}K\left(\frac{\widetilde{\kappa}_{Y,+}\left(U_{i}\mid h^{X}_{rot}\right)-\widetilde{\kappa}_{Y,+}\left(\tau\mid h^{X}_{rot}\right)}{\ell}\right)
\]
be a consistent estimator of $\varphi_{+}\left(\tau\right)$. Let
$\widetilde{\varphi}_{-}\left(\tau\mid\ell\right)$ be defined analogously.
In our implementation, we take the estimators of $\left(\varphi_{+}\left(\tau\right),\varphi_{-}\left(\tau\right)\right)$
to be $\left(\widetilde{\varphi}_{+}\left(\tau\right),\widetilde{\varphi}_{-}\left(\tau\right)\right)\coloneqq\left(\widetilde{\varphi}_{+}\left(\tau\mid2h^{Y}_{rot}\right),\widetilde{\varphi}_{-}\left(\tau\mid2h^{Y}_{rot}\right)\right)$.

\subsubsection*{Step 2: Estimating $\gamma\left(\tau\right)$ and constructing $\widehat{\mathscr{V}}_{\star}\left(\tau\right)$}

Denote 
\[
\widetilde{\beta}_{j,+}\left(\ell\right)\coloneqq\underset{\beta\in\mathbb{R}^{2}}{\arg\min}\sum_{i}K\left(\frac{X_{i}}{\ell}\right)I_{i}\left(Z^{\left[j\right]}_{i}-r^{\top}_{1}\left(X_{i}\right)\beta\right)^{2}.
\]
Let $\widetilde{\xi}_{+,i}\left(\ell\right)$ be the vector whose
$j$-th coordinate $\widetilde{\xi}^{\left[j\right]}_{+,i}\left(\ell\right)$
is defined by $\widetilde{\xi}^{\left[j\right]}_{+,i}\left(\ell\right)\coloneqq Z^{\left[j\right]}_{i}-r^{\top}_{1}\left(X_{i}\right)\widetilde{\beta}_{j,+}\left(\ell\right)$.
Let $\widetilde{\xi}_{-,i}\left(\ell\right)$ be defined analogously.
Let 
\[
\widetilde{\mu}_{\xi\xi^{\top},+}\left(\ell\right)\coloneqq\frac{1}{n\ell}\sum_{i}W_{+,1,i}\left(\ell\right)\widetilde{\xi}_{+,i}\left(\ell\right)\widetilde{\xi}_{+,i}^{\top}\left(\ell\right)
\]
and let $\widetilde{\mu}_{\xi\xi^{\top},-}\left(\ell\right)$ be defined
similarly. Let 
\[
\widetilde{\beta}_{\tau,+}\left(\ell\right)\coloneqq\underset{\beta\in\mathbb{R}^{2}}{\arg\min}\sum_{i}K\left(\frac{X_{i}}{\ell}\right)I_{i}\rho_{\tau}\left(Y_{i}-r^{\top}_{1}\left(X_{i}\right)\beta\right)
\]
and let 
\[
\widetilde{U}_{+,i}\left(\tau\mid\ell\right)\coloneqq\tau-\mathbbm{1}\left(Y_{i}\leq r^{\top}_{1}\left(X_{i}\right)\widetilde{\beta}_{\tau,+}\left(\ell\right)\right).
\]
Let $\widetilde{U}_{-,i}\left(\tau\mid\ell\right)$ be defined similarly.
Let 
\[
\widetilde{\mu}_{\xi U\left(\tau\right),+}\left(\ell\right)\coloneqq\frac{1}{n\ell}\sum_{i}W_{+,1,i}\left(\ell\right)\widetilde{\xi}_{+,i}\left(\ell\right)\widetilde{U}_{+,i}\left(\tau\mid\ell\right).
\]
Let $\widetilde{\mu}_{\xi\xi^{\top},+}\coloneqq\widetilde{\mu}_{\xi\xi^{\top},+}\left(h^{X}_{rot}\right)$
and $\widetilde{\mu}_{\xi U\left(\tau\right),+}\coloneqq\widetilde{\mu}_{\xi U\left(\tau\right),+}\left(h^{X}_{rot}\right)$.
Let $\left(\widetilde{\mu}_{\xi\xi^{\top},-},\widetilde{\mu}_{\xi U\left(\tau\right),-}\right)$
be defined similarly and let 
\[
\widetilde{\gamma}\left(\tau\right)\coloneqq\left(\widetilde{\mu}_{\xi\xi^{\top},+}+\widetilde{\mu}_{\xi\xi^{\top},-}\right)^{-1}\left(\frac{\widetilde{\mu}_{\xi U\left(\tau\right),+}}{\widetilde{\varphi}_{+}\left(\tau\right)}+\frac{\widetilde{\mu}_{\xi U\left(\tau\right),-}}{\widetilde{\varphi}_{-}\left(\tau\right)}\right)
\]
be the estimate for $\gamma\left(\tau\right)$ in the implementation,
where $\widetilde{\varphi}_{+}\left(\tau\right)$ and $\widetilde{\varphi}_{-}\left(\tau\right)$
are from Step 1. Let 
\[
\widehat{\mathscr{V}}_{\star}\left(\tau\right)\coloneqq\frac{\omega^{0,2}_{1,0}}{\widetilde{\varphi}}\left\{ \left(\frac{\tau\left(1-\tau\right)}{\widetilde{\varphi}^{2}_{+}\left(\tau\right)}+\frac{\tau\left(1-\tau\right)}{\widetilde{\varphi}^{2}_{-}\left(\tau\right)}\right)-\widetilde{\gamma}^{\top}\left(\tau\right)\left(\widetilde{\mu}_{\xi\xi^{\top},+}+\widetilde{\mu}_{\xi\xi^{\top},-}\right)\widetilde{\gamma}\left(\tau\right)\right\} .
\]

\subsubsection*{Step 3: Constructing $\widehat{\mathscr{B}}_{\star}\left(\tau\right)$}

Let $b$ be the pilot bandwidth. Run a local quadratic quantile regression
to estimate $\left(\kappa^{\left(2\right)}_{Y,+}\left(\tau\right),\kappa^{\left(2\right)}_{Y,-}\left(\tau\right)\right)$.
Let 
\[
\widetilde{\kappa}^{\left(2\right)}_{Y,+}\left(\tau\mid b\right)\coloneqq2!\cdot\mathrm{e}^{\top}_{3,3}\underset{\beta\in\mathbb{R}^{3}}{\arg\min}\sum_{i}K\left(\frac{X_{i}}{b}\right)I_{i}\rho_{\tau}\left(Y_{i}-r^{\top}_{2}\left(X_{i}\right)\beta\right)
\]
and let $\widetilde{\kappa}^{\left(2\right)}_{Y,-}\left(\tau\mid b\right)$
be defined similarly. Run a local quadratic regression to estimate
$\left(\mu^{\left(2\right)}_{Z^{\left[j\right]},+},\mu^{\left(2\right)}_{Z^{\left[j\right]},-}\right)$.
Let 
\[
\widetilde{\mu}^{\left(2\right)}_{Z^{\left[j\right]},+}\left(b\right)\coloneqq2!\cdot\mathrm{e}^{\top}_{3,3}\underset{\beta\in\mathbb{R}^{3}}{\arg\min}\sum_{i}K\left(\frac{X_{i}}{b}\right)I_{i}\left(Z^{\left[j\right]}_{i}-r^{\top}_{2}\left(X_{i}\right)\beta\right)^{2}
\]
and let $\widetilde{\mu}^{\left(2\right)}_{Z^{\left[j\right]},-}\left(b\right)$
be defined similarly. Let $\widetilde{\mu}^{\left(2\right)}_{Z,+}\left(b\right)\coloneqq\bigl(\widetilde{\mu}^{\left(2\right)}_{Z^{\left[1\right]},+}\left(b\right),\ldots,\widetilde{\mu}^{\left(2\right)}_{Z^{\left[d_{z}\right]},+}\left(b\right)\bigr)^{\top}$,
and let $\widetilde{\mu}^{\left(2\right)}_{Z,-}\left(b\right)$ be
defined similarly. Let 
\[
\widehat{\mathscr{B}}_{\star}\left(\tau\mid b\right)\coloneqq\omega^{2,1}_{1,0}\left\{ \left(\frac{\widetilde{\kappa}^{\left(2\right)}_{Y,+}\left(\tau\mid b\right)}{2}-\frac{\widetilde{\kappa}^{\left(2\right)}_{Y,-}\left(\tau\mid b\right)}{2}\right)-\left(\frac{\widetilde{\mu}^{\left(2\right)}_{Z,+}\left(b\right)}{2}-\frac{\widetilde{\mu}^{\left(2\right)}_{Z,-}\left(b\right)}{2}\right)^{\top}\widetilde{\gamma}\left(\tau\right)\right\} 
\]
be the estimate for $\mathscr{B}_{\star,1}\left(\tau\right)$, where
$\widetilde{\gamma}\left(\tau\right)$ is obtained from Step 2.

Note that the estimation error of $\widetilde{\gamma}\left(\tau\right)$
is negligible under smoothness conditions. Let 
\begin{eqnarray*}
\mathscr{B}_{\star\star}\left(\tau\right) & \coloneqq & \frac{\kappa^{\left(3\right)}_{Y,+}\left(\tau\right)-\gamma^{\top}\left(\tau\right)\mu^{\left(3\right)}_{Z,+}}{3!}\cdot\omega^{3,1}_{+,2,2}-\frac{\kappa^{\left(3\right)}_{Y,-}\left(\tau\right)-\gamma^{\top}\left(\tau\right)\mu^{\left(3\right)}_{Z,-}}{3!}\cdot\omega^{3,1}_{-,2,2}\\
\mathscr{V}_{\star\star}\left(\tau\right) & \coloneqq & \frac{\omega^{0,2}_{2,2}}{\varphi}\left\{ \frac{\tau\left(1-\tau\right)}{\varphi^{2}_{+}\left(\tau\right)}+\frac{\tau\left(1-\tau\right)}{\varphi^{2}_{-}\left(\tau\right)}\right.\\
 &  & \left.-\left(\frac{\mu_{ZU\left(\tau\right),+}}{\varphi_{+}\left(\tau\right)}+\frac{\mu_{ZU\left(\tau\right),-}}{\varphi_{-}\left(\tau\right)}\right)^{\top}\left(\mathrm{Var}_{\pm}\left[Z\right]\right)^{-1}\left(\frac{\mu_{ZU\left(\tau\right),+}}{\varphi_{+}\left(\tau\right)}+\frac{\mu_{ZU\left(\tau\right),-}}{\varphi_{-}\left(\tau\right)}\right)\right\} .
\end{eqnarray*}
Then, the AMSE of $\widehat{\mathscr{B}}_{\star}\left(\tau\mid b\right)$
is proportional to 
\[
\mathscr{B}^{2}_{\star\star}\left(\tau\right)b^{2}+\frac{1}{nb^{5}}\mathscr{V}_{\star\star}\left(\tau\right).
\]
The AMSE-optimal pilot bandwidth is given by 
\[
b_{\mathit{opt}}\left(\tau\right)\coloneqq\left(\frac{5}{2}\cdot\frac{\mathscr{V}_{\star\star}\left(\tau\right)}{\mathscr{B}^{2}_{\star\star}\left(\tau\right)}\right)^{1/7}n^{-1/7}.
\]

Estimation of $\mathscr{B}_{\star\star}\left(\tau\right)$ requires
estimating $\kappa^{\left(3\right)}_{Y,+}\left(\tau\right)$, $\mu^{\left(3\right)}_{Z,+}$,
$\kappa^{\left(3\right)}_{Y,-}\left(\tau\right)$ and $\mu^{\left(3\right)}_{Z,-}$.
We use a parametric reference model to estimate $\kappa^{\left(3\right)}_{Y,+}\left(\tau\right)$
and $\kappa^{\left(3\right)}_{Y,-}\left(\tau\right)$. Consider the
following piecewise cubic global quantile regression model: 
\begin{eqnarray*}
Y & = & r^{\top}_{3}\left(X\right)\alpha_{0}+\left(I\cdot r_{3}\left(X\right)\right)^{\top}\alpha_{1}+u\\
Q_{u\mid X}\left(\tau\mid X\right) & = & 0.
\end{eqnarray*}
Let $\left(\widetilde{\alpha}^{\top}_{0},\widetilde{\alpha}^{\top}_{1}\right)^{\top}$
be the regression coefficients from a global quantile regression at
quantile level $\tau$. The estimate of $\kappa^{\left(3\right)}_{Y,+}\left(\tau\right)$
is given by $\widetilde{\kappa}^{\left(3\right)}_{Y,+}\left(\tau\right)\coloneqq3!\left\{ \left(\mathrm{e}^{\top}_{4,4}\widetilde{\alpha}_{0}\right)+\left(\mathrm{e}^{\top}_{4,4}\widetilde{\alpha}_{1}\right)\right\} $.
The estimate of $\kappa^{\left(3\right)}_{Y,-}\left(\tau\right)$
is given by $\widetilde{\kappa}^{\left(3\right)}_{Y,-}\left(\tau\right)\coloneqq3!\left(\mathrm{e}^{\top}_{4,4}\widetilde{\alpha}_{0}\right)$.
The estimates of $\mu^{\left(3\right)}_{Z^{\left[j\right]},+}$ and
$\mu^{\left(3\right)}_{Z^{\left[j\right]},-}$ are obtained by estimating
the following parametric reference (piecewise cubic) model: 
\begin{eqnarray*}
Z^{\left[j\right]} & = & r^{\top}_{3}\left(X\right)\alpha_{0}+\left(I\cdot r_{3}\left(X\right)\right)^{\top}\alpha_{1}+u\\
\mathrm{E}\left[u\mid X\right] & = & 0.
\end{eqnarray*}
Let $\widetilde{\alpha}_{0}$ be the regression coefficients of $r_{3}\left(X\right)$
and let $\widetilde{\alpha}_{1}$ be the regression coefficients of
$I\cdot r_{3}\left(X\right)$. The estimate of $\mu^{\left(3\right)}_{Z^{\left[j\right]},+}$
is given by $\widetilde{\mu}^{\left(3\right)}_{Z^{\left[j\right]},+}\coloneqq3!\left\{ \left(\mathrm{e}^{\top}_{4,4}\widetilde{\alpha}_{0}\right)+\left(\mathrm{e}^{\top}_{4,4}\widetilde{\alpha}_{1}\right)\right\} $.
The estimate of $\mu^{\left(3\right)}_{Z^{\left[j\right]},-}$ is
given by $\widetilde{\mu}^{\left(3\right)}_{Z^{\left[j\right]},-}\coloneqq3!\left(\mathrm{e}^{\top}_{4,4}\widetilde{\alpha}_{0}\right)$.
Let $\widetilde{\mu}^{\left(3\right)}_{Z,+}\coloneqq\bigl(\widetilde{\mu}^{\left(3\right)}_{Z^{\left[1\right]},+},\ldots,\widetilde{\mu}^{\left(3\right)}_{Z^{\left[d_{z}\right]},+}\bigr)^{\top}$,
and let $\widetilde{\mu}^{\left(3\right)}_{Z,-}$ be defined similarly.
The estimate $\widetilde{\gamma}\left(\tau\right)$ of $\gamma\left(\tau\right)$
is from Step 2. Now we have the estimate 
\[
\widetilde{\mathscr{B}}_{\star\star}\left(\tau\right)\coloneqq\frac{\widetilde{\kappa}^{\left(3\right)}_{Y,+}\left(\tau\right)-\widetilde{\gamma}^{\top}\left(\tau\right)\widetilde{\mu}^{\left(3\right)}_{Z,+}}{3!}\cdot\omega^{3,1}_{+,2,2}-\frac{\widetilde{\kappa}^{\left(3\right)}_{Y,-}\left(\tau\right)-\widetilde{\gamma}^{\top}\left(\tau\right)\widetilde{\mu}^{\left(3\right)}_{Z,-}}{3!}\cdot\omega^{3,1}_{-,2,2}
\]
for $\mathscr{B}_{\star\star}\left(\tau\right)$. We also have the
estimate 
\[
\widetilde{\mathscr{V}}_{\star\star}\left(\tau\right)\coloneqq\frac{\omega^{0,2}_{2,2}}{\widetilde{\varphi}}\left\{ \frac{\tau\left(1-\tau\right)}{\widetilde{\varphi}^{2}_{+}\left(\tau\right)}+\frac{\tau\left(1-\tau\right)}{\widetilde{\varphi}^{2}_{-}\left(\tau\right)}-\widetilde{\gamma}^{\top}\left(\tau\right)\left(\widetilde{\mu}_{\xi\xi^{\top},+}+\widetilde{\mu}_{\xi\xi^{\top},-}\right)\widetilde{\gamma}\left(\tau\right)\right\} ,
\]
where $\left(\widetilde{\varphi},\widetilde{\varphi}_{+}\left(\tau\right),\widetilde{\varphi}_{-}\left(\tau\right)\right)$
are from Step 1 and $\left(\widetilde{\mu}_{\xi U\left(\tau\right),+},\widetilde{\mu}_{\xi U\left(\tau\right),-},\widetilde{\mu}_{\xi\xi^{\top},+},\widetilde{\mu}_{\xi\xi^{\top},-}\right)$
are from Step 2. Now we determine the pilot bandwidth $\widehat{b}\left(\tau\right)$
to be 
\[
\widehat{b}\left(\tau\right)\coloneqq\left(\frac{5}{2}\cdot\frac{\widetilde{\mathscr{V}}_{\star\star}\left(\tau\right)}{\widetilde{\mathscr{B}}^{2}_{\star\star}\left(\tau\right)}\right)^{1/7}n^{-1/7}.
\]
Finally, the estimator of $\mathscr{B}_{\star,1}\left(\tau\right)$
is given by $\widehat{\mathscr{B}}_{\star}\left(\tau\right)\coloneqq\widehat{\mathscr{B}}_{\star}\left(\tau\mid\widehat{b}\left(\tau\right)\right)$.

\subsection{Regression kink}

\label{app:imp_rkd}

\subsubsection*{Step 1: Preliminary bandwidth and $\widetilde{\varphi}$}

This step is the same as Step 1 of Appendix \ref{app:bdw_qrd}: using
the modified Silverman ROT bandwidth $h^{X}_{rot}$, we compute $\widetilde{\varphi}\coloneqq\widetilde{\varphi}\left(h^{X}_{rot}\right)$.

\subsubsection*{Step 2: Estimating $\gamma$, $\gamma_{\dagger,p}$ and constructing
$\widehat{\mathscr{V}}_{\star\dagger}$}

Residualize $Y$ and $Z$ by one-sided local linear regressions with
ROT bandwidth $h^{X}_{rot}$. Specifically, let 
\[
\widetilde{\beta}_{+}\left(\ell\right)\coloneqq\underset{\beta\in\mathbb{R}^{2}}{\arg\min}\sum_{i}K\left(\frac{X_{i}}{\ell}\right)I_{i}\left(Y_{i}-r^{\top}_{1}\left(X_{i}\right)\beta\right)^{2}
\]
and let $\widetilde{\eta}_{+,i}\left(\ell\right)\coloneqq Y_{i}-r^{\top}_{1}\left(X_{i}\right)\widetilde{\beta}_{+}\left(\ell\right)$.
Obtain $\widetilde{\eta}_{-,i}\left(\ell\right)$ similarly. Construct
\begin{eqnarray*}
\widetilde{\mu}_{\eta^{2},+}\left(\ell\right) & \coloneqq & \frac{1}{n\ell}\sum_{i}W_{+,1,i}\left(\ell\right)\widetilde{\eta}^{2}_{+,i}\left(\ell\right)\\
\widetilde{\mu}_{\xi\eta,+}\left(\ell\right) & \coloneqq & \frac{1}{n\ell}\sum_{i}W_{+,1,i}\left(\ell\right)\widetilde{\xi}_{+,i}\left(\ell\right)\widetilde{\eta}_{+,i}\left(\ell\right).
\end{eqnarray*}
Let $\widetilde{\mu}_{\eta^{2},+}\coloneqq\widetilde{\mu}_{\eta^{2},+}\left(h^{X}_{rot}\right)$
and $\widetilde{\mu}_{\xi\eta,+}\coloneqq\widetilde{\mu}_{\xi\eta,+}\left(h^{X}_{rot}\right)$.
Let $\left(\widetilde{\mu}_{\eta^{2},-},\widetilde{\mu}_{\xi\eta,-}\right)$
be defined similarly. We then construct the estimators for the coefficients
$\gamma$ and $\gamma_{\dagger,p}$ as follows. Let 
\[
\widetilde{\gamma}\coloneqq\left(\widetilde{\mu}_{\xi\xi^{\top},+}+\widetilde{\mu}_{\xi\xi^{\top},-}\right)^{-1}\left(\widetilde{\mu}_{\xi\eta,+}+\widetilde{\mu}_{\xi\eta,-}\right).
\]
Let 
\begin{eqnarray*}
\widetilde{\zeta}_{+,i}\left(\ell\right) & \coloneqq & \widetilde{\eta}_{+,i}\left(\ell\right)-\widetilde{\xi}^{\top}_{+,i}\left(\ell\right)\widetilde{\gamma}\\
\widetilde{\mu}_{\xi\zeta,+}\left(\ell\right) & \coloneqq & \frac{1}{n\ell}\sum_{i}W_{+,1,i}\left(\ell\right)\widetilde{\xi}_{+,i}\left(\ell\right)\widetilde{\zeta}_{+,i}\left(\ell\right)
\end{eqnarray*}
and let $\widetilde{\mu}_{\xi\zeta,+}\coloneqq\widetilde{\mu}_{\xi\zeta,+}\left(h^{X}_{rot}\right)$.
Let $\widetilde{\mu}_{\xi\zeta,-}$ be defined similarly. Now define
\begin{eqnarray*}
\widetilde{\gamma}_{\dagger,p} & \coloneqq & \left(\frac{\varsigma_{p}}{\omega^{0,2}_{p,0}}\right)\left(\widetilde{\mu}_{\xi\xi^{\top},+}+\widetilde{\mu}_{\xi\xi^{\top},-}\right)^{-1}\left(\widetilde{\mu}_{\xi\zeta,+}-\widetilde{\mu}_{\xi\zeta,-}\right).
\end{eqnarray*}
We obtain the following estimator of the variance $\mathscr{V}_{\star\dagger,2}$
\begin{equation}
\widehat{\mathscr{V}}_{\star\dagger}\coloneqq\frac{\omega^{0,2}_{2,1}}{\widetilde{\varphi}}\left(\widetilde{\mu}_{\eta^{2},+}+\widetilde{\mu}_{\eta^{2},-}\right)-\frac{\omega^{0,2}_{2,1}}{\widetilde{\varphi}}\widetilde{\gamma}^{\top}\left(\widetilde{\mu}_{\xi\xi^{\top},+}+\widetilde{\mu}_{\xi\xi^{\top},-}\right)\widetilde{\gamma}-\frac{\omega^{0,2}_{2,0}}{\widetilde{\varphi}}\widetilde{\gamma}^{\top}_{\dagger,2}\left(\widetilde{\mu}_{\xi\xi^{\top},+}+\widetilde{\mu}_{\xi\xi^{\top},-}\right)\widetilde{\gamma}_{\dagger,2}.\label{eq:rkd-Vstar}
\end{equation}

\subsubsection*{Step 3: Constructing $\widehat{\mathscr{B}}_{\star\dagger}$}

We construct an estimator $\widehat{\mathscr{B}}_{\star\dagger}\left(b\right)$
for $\mathscr{B}_{\star\dagger,2}$ and determine the bandwidth $b$.
We run a local cubic regression to estimate $\bigl(\mu^{\left(3\right)}_{Y,+},\mu^{\left(3\right)}_{Y,-}\bigr)$.
Let 
\begin{eqnarray*}
\widetilde{\mu}^{\left(3\right)}_{Y,+}\left(b\right) & \coloneqq & 3!\cdot\mathrm{e}^{\top}_{4,4}\underset{\beta\in\mathbb{R}^{4}}{\arg\min}\sum_{i}K\left(\frac{X_{i}}{b}\right)I_{i}\left(Y_{i}-r^{\top}_{3}\left(X_{i}\right)\beta\right)^{2}.
\end{eqnarray*}
$\widetilde{\mu}^{\left(3\right)}_{Y,-}\left(b\right)$ and $\bigl(\widetilde{\mu}^{\left(3\right)}_{Z^{\left[j\right]},+}\left(b\right),\widetilde{\mu}^{\left(3\right)}_{Z^{\left[j\right]},-}\left(b\right)\bigr)$
are defined similarly. Let $\widetilde{\mu}^{\left(3\right)}_{Z,+}\left(b\right)\coloneqq\bigl(\widetilde{\mu}^{\left(3\right)}_{Z^{\left[1\right]},+}\left(b\right),\ldots,\widetilde{\mu}^{\left(3\right)}_{Z^{\left[d_{z}\right]},+}\left(b\right)\bigr)^{\top}$
and let $\widetilde{\mu}^{\left(3\right)}_{Z,-}\left(b\right)$ be
defined analogously. Let 
\begin{eqnarray*}
\widehat{\mathscr{B}}_{\star\dagger}\left(b\right) & \coloneqq & \frac{1}{3!}\left\{ \left(\widetilde{\mu}^{\left(3\right)}_{Y,+}\left(b\right)\omega^{3,1}_{+,2,1}-\widetilde{\mu}^{\left(3\right)}_{Y,-}\left(b\right)\omega^{3,1}_{-,2,1}\right)-\left(\widetilde{\mu}^{\left(3\right)}_{Z,+}\left(b\right)\omega^{3,1}_{+,2,1}-\widetilde{\mu}^{\left(3\right)}_{Z,-}\left(b\right)\omega^{3,1}_{-,2,1}\right)^{\top}\widetilde{\gamma}\right.\\
 &  & \left.-\left(\widetilde{\mu}^{\left(3\right)}_{Z,+}\left(b\right)\omega^{3,1}_{+,2,0}-\widetilde{\mu}^{\left(3\right)}_{Z,-}\left(b\right)\omega^{3,1}_{-,2,0}\right)^{\top}\widetilde{\gamma}_{\dagger,2}\right\} .
\end{eqnarray*}
Using the symmetries $\omega^{3,1}_{+,2,1}=\omega^{3,1}_{-,2,1}=\omega^{3,1}_{2,1}$
and $\omega^{3,1}_{+,2,0}=-\omega^{3,1}_{-,2,0}=\omega^{3,1}_{2,0}$,
$\widehat{\mathscr{B}}_{\star\dagger}\left(b\right)$ can be written
as 
\[
\widehat{\mathscr{B}}_{\star\dagger}\left(b\right)=\frac{1}{3!}\omega^{3,1}_{2,1}\left\{ \left(\widetilde{\mu}^{\left(3\right)}_{Y,+}\left(b\right)-\widetilde{\mu}^{\left(3\right)}_{Y,-}\left(b\right)\right)-\left(\widetilde{\mu}^{\left(3\right)}_{Z,+}\left(b\right)-\widetilde{\mu}^{\left(3\right)}_{Z,-}\left(b\right)\right)^{\top}\widetilde{\gamma}-\left(\widetilde{\mu}^{\left(3\right)}_{Z,+}\left(b\right)+\widetilde{\mu}^{\left(3\right)}_{Z,-}\left(b\right)\right)^{\top}\widetilde{\gamma}_{\ddagger}\right\} ,
\]
where $\widetilde{\gamma}_{\ddagger}\coloneqq\left(\omega^{3,1}_{2,0}/\omega^{3,1}_{2,1}\right)\widetilde{\gamma}_{\dagger,2}$
and $\gamma_{\ddagger}\coloneqq\left(\omega^{3,1}_{2,0}/\omega^{3,1}_{2,1}\right)\gamma_{\dagger,2}$.

Note that the estimation errors of $\widetilde{\gamma}$ and $\widetilde{\gamma}_{\ddagger}$
are negligible under smoothness conditions. The AMSE of $\widehat{\mathscr{B}}_{\star\dagger}\left(b\right)$
is proportional to $\mathscr{B}^{2}_{\star\star}b^{2}+\mathscr{V}_{\star\star}/\left(nb^{7}\right)$,
where 
\begin{eqnarray*}
\mathscr{B}_{\star\star} & \coloneqq & \frac{\mu^{\left(4\right)}_{Y,+}-\left(\mu^{\left(4\right)}_{Z,+}\right)^{\top}\left(\gamma+\gamma_{\ddagger}\right)}{4!}\cdot\omega^{4,1}_{+,3,3}-\frac{\mu^{\left(4\right)}_{Y,-}-\left(\mu^{\left(4\right)}_{Z,-}\right)^{\top}\left(\gamma-\gamma_{\ddagger}\right)}{4!}\cdot\omega^{4,1}_{-,3,3}\\
\mathscr{V}_{\star\star} & \coloneqq & \frac{\omega^{0,2}_{3,3}}{\varphi}\left\{ \mu_{\eta^{2},+}-2\cdot\mu^{\top}_{\xi\eta,+}\left(\gamma+\gamma_{\ddagger}\right)+\left(\gamma+\gamma_{\ddagger}\right)^{\top}\mu_{\xi\xi^{\top},+}\left(\gamma+\gamma_{\ddagger}\right)\right.\\
 &  & \left.+\mu_{\eta^{2},-}-2\cdot\mu^{\top}_{\xi\eta,-}\left(\gamma-\gamma_{\ddagger}\right)+\left(\gamma-\gamma_{\ddagger}\right)^{\top}\mu_{\xi\xi^{\top},-}\left(\gamma-\gamma_{\ddagger}\right)\right\} .
\end{eqnarray*}
The AMSE-optimal bandwidth is given by 
\begin{equation}
b_{\mathit{opt}}\coloneqq\left(\frac{7}{2}\cdot\frac{\mathscr{V}_{\star\star}}{\mathscr{B}^{2}_{\star\star}}\right)^{1/9}n^{-1/9}.\label{eq:pilot btw RKD}
\end{equation}
To estimate $b_{\mathit{opt}}$, we first obtain the estimates of
$\mu^{\left(4\right)}_{Y,+}$ and $\mu^{\left(4\right)}_{Y,-}$ from
a parametric reference model. We consider the following global piecewise
quartic regression model 
\begin{eqnarray*}
Y & = & r^{\top}_{4}\left(X\right)\alpha_{0}+\left(I\cdot r_{4}\left(X\right)\right)^{\top}\alpha_{1}+u\\
\mathrm{E}\left[u\mid X\right] & = & 0.
\end{eqnarray*}
Let $\widetilde{\alpha}_{0}$ be the regression coefficients of $r_{4}\left(X\right)$
and let $\widetilde{\alpha}_{1}$ be the regression coefficients of
$I\cdot r_{4}\left(X\right)$. The estimate of $\mu^{\left(4\right)}_{Y,+}$
is $\widetilde{\mu}^{\left(4\right)}_{Y,+}\coloneqq4!\left\{ \left(\mathrm{e}^{\top}_{5,5}\widetilde{\alpha}_{0}\right)+\left(\mathrm{e}^{\top}_{5,5}\widetilde{\alpha}_{1}\right)\right\} $,
and the estimate of $\mu^{\left(4\right)}_{Y,-}$ is $\widetilde{\mu}^{\left(4\right)}_{Y,-}\coloneqq4!\left(\mathrm{e}^{\top}_{5,5}\widetilde{\alpha}_{0}\right)$.
Let $\widetilde{\mu}^{\left(4\right)}_{Z,+}$ and $\widetilde{\mu}^{\left(4\right)}_{Z,-}$
be defined similarly. Now we have the estimate 
\[
\widetilde{\mathscr{B}}_{\star\star}\coloneqq\frac{\widetilde{\mu}^{\left(4\right)}_{Y,+}-\left(\widetilde{\mu}^{\left(4\right)}_{Z,+}\right)^{\top}\left(\widetilde{\gamma}+\widetilde{\gamma}_{\ddagger}\right)}{4!}\cdot\omega^{4,1}_{+,3,3}-\frac{\widetilde{\mu}^{\left(4\right)}_{Y,-}-\left(\widetilde{\mu}^{\left(4\right)}_{Z,-}\right)^{\top}\left(\widetilde{\gamma}-\widetilde{\gamma}_{\ddagger}\right)}{4!}\cdot\omega^{4,1}_{-,3,3}
\]
for $\mathscr{B}_{\star\star}$. We also have the estimate 
\begin{eqnarray*}
\widetilde{\mathscr{V}}_{\star\star} & \coloneqq & \frac{\omega^{0,2}_{3,3}}{\widetilde{\varphi}}\left\{ \widetilde{\mu}_{\eta^{2},+}-2\cdot\widetilde{\mu}^{\top}_{\xi\eta,+}\left(\widetilde{\gamma}+\widetilde{\gamma}_{\ddagger}\right)+\left(\widetilde{\gamma}+\widetilde{\gamma}_{\ddagger}\right)^{\top}\widetilde{\mu}_{\xi\xi^{\top},+}\left(\widetilde{\gamma}+\widetilde{\gamma}_{\ddagger}\right)\right.\\
 &  & \left.+\widetilde{\mu}_{\eta^{2},-}-2\cdot\widetilde{\mu}^{\top}_{\xi\eta,-}\left(\widetilde{\gamma}-\widetilde{\gamma}_{\ddagger}\right)+\left(\widetilde{\gamma}-\widetilde{\gamma}_{\ddagger}\right)^{\top}\widetilde{\mu}_{\xi\xi^{\top},-}\left(\widetilde{\gamma}-\widetilde{\gamma}_{\ddagger}\right)\right\} 
\end{eqnarray*}
for $\mathscr{V}_{\star\star}$. Now we determine the pilot bandwidth
$\widehat{b}$ by using the right-hand side of \eqref{eq:pilot btw RKD}
with $\left(\mathscr{V}_{\star\star},\mathscr{B}_{\star\star}\right)$
replaced by $\left(\widetilde{\mathscr{V}}_{\star\star},\widetilde{\mathscr{B}}_{\star\star}\right)$.
Our estimate for $\mathscr{B}_{\star\dagger,2}$ is given by $\widehat{\mathscr{B}}_{\star\dagger}\coloneqq\widehat{\mathscr{B}}_{\star\dagger}\left(\widehat{b}\right)$. 

%%%%%%%%%%%%%%%%%%%%%%%%%%%%%%%%%%%%%%%%%%%%%%%%%%%%%%%%%%%%%%%%%%%%%%%%%%%%%%%
%%                                SUPPLEMENT                                  %%
%%                                                                            %%
%% Everything below reproduces draft_V13_supp.tex verbatim. The block that    %%
%% follows restores the numbering the supplement had as a standalone document:%%
%%   pages     1, 2, 3, ...   (restarted; see \setcounter{page}{1} below)      %%
%%   sections  S.1, S.2, ...                                                  %%
%%   equations (S.1), (S.2), ...                                              %%
%%   lemmas    S.1, S.2, ...                                                  %%
%% To number the supplement's pages continuously with the main text instead,  %%
%% delete the \setcounter{page}{1} line (and, optionally, the \theHpage line).%%
%%%%%%%%%%%%%%%%%%%%%%%%%%%%%%%%%%%%%%%%%%%%%%%%%%%%%%%%%%%%%%%%%%%%%%%%%%%%%%%

\clearpage

%% Restore the title-page machinery consumed by the main paper's \maketitle.
\makeatletter
\let\maketitle\GRDmaketitle
\let\@maketitle\GRDAtmaketitle
\let\thanks\GRDthanks
\let\title\GRDtitle
\let\author\GRDauthor
\let\date\GRDdate
\let\and\GRDand
%% The main paper's \appendix left \section@cntformat defined (it prints
%% "Appendix ..."); make it undefined again so supplement headings print "S.1 ...".
\let\section@cntformat\relax
\makeatother

%% Restart the counters so the supplement numbers exactly as it did standalone.
\setcounter{section}{0}
\setcounter{subsection}{0}
\setcounter{equation}{0}
\setcounter{lem}{0}
\setcounter{footnote}{0}
\setcounter{page}{1}

%% Numbering in the supplement carries the prefix S: Section S.1, S.2, ...;
%% equations (S.1), (S.2), ...; Lemma S.1, S.2, ...
\renewcommand{\theequation}{S.\arabic{equation}}
\renewcommand{\thelem}{S.\arabic{lem}}

%% (PDF anchors are kept distinct from the main paper's by the hypertexnames=false
%% option passed to hyperref in the preamble.)

\title{\textbf{Supplement to ``Generic Covariate Adjustment for }\\
\textbf{Regression Discontinuity Designs''}\thanks{This version: \today.}}
\author{ Jun Ma\thanks{Jun Ma. School of Economics, Renmin University of China.}
\and Yuya Sasaki\thanks{Yuya Sasaki. Department of Economics, Vanderbilt University.}
\and Zhengfei Yu\thanks{Zhengfei Yu. Graduate School of Economics, The University of Osaka.}}
\date{}
\maketitle

\noindent This supplement contains the proofs of Theorems \ref{thm:normality}--\ref{thm:kink}
of the main text. Section, equation, and lemma numbers
in this supplement carry the prefix ``S''. Numbered references without this prefix,
such as Theorem \ref{thm:normality}, Assumption \ref{assu:kernel}, and (\ref{eq:balancing weights}),
refer to the main text.

\textbf{Notation.} $\max_{i}$ ($\min_{i}$) is understood as $\max_{1\leq i\leq n}$
($\min_{1\leq i\leq n}$). For real numbers $a$ and $b$, $a\land b\coloneqq\min\left\{ a,b\right\} $.
For a vector $x$ and a matrix $\mathrm{A}$,
$x^{\left[j\right]}$ denotes the $j$-th coordinate of $x$ and $\mathrm{A}^{\left[jk\right]}$
denotes the $jk$-th element of $\mathrm{A}$. For a square matrix
$\mathrm{A}$, $\left\Vert \mathrm{A}\right\Vert $ is understood
as the spectral norm of $\mathrm{A}$ and $\mathrm{mineig}\left(\mathrm{A}\right)$
denotes its smallest eigenvalue. ``$a\apprle b$'' is understood
as $a\leq C\cdot b$, for some universal constant $C>0$ that does
not depend on the distribution of the variables or the sample size
but may depend on the kernel function. Denote $K_{+,p}\left(u\right)\coloneqq K\left(u\right)r_{p}\left(u\right)\mathbbm{1}\left(u>0\right)$
and let $K_{-,p}$ be defined analogously. Let $\mathrm{H}$ be the $\left(p+1\right)\times\left(p+1\right)$
diagonal matrix with $\left(1,h,\ldots,h^{p}\right)$ on the diagonal. We adopt the following
abbreviations: ``wpa1'' for ``with probability approaching one'',
``LIE'' for ``law of iterated expectations'', and ``DCT'' for
``dominated convergence theorem''. In the proofs, for notational
brevity, we suppress the dependence on the bandwidth $h$: we write
$\Pi_{+,p}$, $W_{+,p,i}$, $V_{p,i}$, and $\widehat{\lambda}_{p}$
for $\Pi_{+,p}\left(h\right)$, $W_{+,p,i}\left(h\right)$, $V_{p,i}\left(h\right)$,
and $\widehat{\lambda}_{p}\left(h\right)$, and we drop $h$ from
the arguments of the estimators, writing $\widehat{\vartheta}^{\mathit{eb}}_{p}$
for $\widehat{\vartheta}^{\mathit{eb}}_{p}\left(h\right)$, and similarly
for related quantities.

\renewcommand{\thesection}{S.\arabic{section}}

\section{Proof of Theorem \ref{thm:normality}}

\label{sec:Proof-of-Theorem 1}

Denote 
\[
\bar{W}_{+,p,i}\coloneqq\mathrm{e}^{\top}_{p+1,1}\left(\varphi\Lambda_{+,p}\right)^{-1}K_{+,p}\left(\frac{X_{i}}{h}\right),
\]
let $\bar{W}_{-,p,i}$ be defined analogously, and let $\bar{W}_{p,i}\coloneqq\bar{W}_{+,p,i}-\bar{W}_{-,p,i}$.
It follows from Chebyshev's inequality, a change of variables, a Taylor
expansion, and the continuity of $f_{X}$ that $\Pi_{+,p}-\mathrm{E}\left[\Pi_{+,p}\right]=O_{p}\left(\left(nh\right)^{-1/2}\right)$.
By a change of variables and the continuity of $f_{X}$, $\mathrm{E}\left[\Pi_{+,p}\right]=\varphi\Lambda_{+,p}+o\left(1\right)$.
It now follows that 
\begin{equation}
\left|\text{\ensuremath{\mathrm{mineig}}}\left(\Pi_{+,p}\right)-\text{\ensuremath{\mathrm{mineig}}}\left(\varphi\Lambda_{+,p}\right)\right|\leq\left\Vert \Pi_{+,p}-\varphi\Lambda_{+,p}\right\Vert =o_{p}\left(1\right).\label{eq:min eigen converge 2}
\end{equation}
Then, it follows from this result, the equality 
\begin{equation}
\mathrm{A}^{-1}-\mathrm{B}^{-1}=-\mathrm{B}^{-1}\left(\mathrm{A}-\mathrm{B}\right)\mathrm{B}^{-1}+\mathrm{B}^{-1}\left(\mathrm{A}-\mathrm{B}\right)\mathrm{A}^{-1}\left(\mathrm{A}-\mathrm{B}\right)\mathrm{B}^{-1}\label{eq:A_inv - B_inv}
\end{equation}
for positive definite matrices $\mathrm{A}$ and $\mathrm{B}$, and
$\text{\ensuremath{\mathrm{mineig}}}\left(\Lambda_{+,p}\right)>0$
that 
\begin{eqnarray}
\left\Vert \Pi^{-1}_{+,p}-\left(\varphi\Lambda_{+,p}\right)^{-1}\right\Vert  & = & o_{p}\left(1\right).\label{eq:PI_hat_inv - PI_inv rate}
\end{eqnarray}
Similar results hold for $\Pi_{-,p}$. 
\begin{lem}
\label{lem: reg weights basic}Let $A$ denote a random variable,
and let $\left\{ \left(A_{i},X_{i}\right)\right\} ^{n}_{i=1}$ be
i.i.d. copies of $\left(A,X\right)$. Suppose that Assumptions \ref{assu:smoothness Z}(ii)
and \ref{assu:kernel} hold. Assume that $nh\uparrow\infty$ and $h\downarrow0$.
Let $\mathbb{B}\subseteq\left[\underline{x},\overline{x}\right]$
denote an open neighborhood of 0. The following results hold for all
$k\in\mathbb{N}$. (i) If $\mu_{A}$ is uniformly continuous on $\mathbb{B}\setminus\left\{ 0\right\} $,
then 
\[
\mathrm{E}\left[\frac{1}{h}\bar{W}^{k}_{+,p}A\right]=\frac{\mu_{A,+}\omega^{0,k}_{+,p,0}}{\varphi^{k-1}}+o\left(1\right).
\]
(ii) If $\mu_{A}$ is $\left(p+1\right)$-times continuously differentiable
with uniformly continuous $\mu^{\left(p+1\right)}_{A}$ on $\mathbb{B}\setminus\left\{ 0\right\} $,
then 
\[
\frac{1}{nh}\sum_{i}W_{+,p,i}\mu_{A}\left(X_{i}\right)=\mu_{A,+}+\frac{\mu^{\left(p+1\right)}_{A,+}}{\left(p+1\right)!}\omega^{p+1,1}_{+,p,0}h^{p+1}+o_{p}\left(h^{p+1}\right).
\]
(iii) If for some $\delta>0$, $\mu_{\left|A\right|^{1+\delta}}$
is bounded on $\mathbb{B}\setminus\left\{ 0\right\} $, then 
\[
\frac{1}{nh}\sum_{i}W^{k}_{+,p,i}A_{i}-\mathrm{E}\left[\frac{1}{h}\bar{W}^{k}_{+,p}A\right]=o_{p}\left(1\right).
\]
(iv) If the conditional density $f_{A\mid X}\left(\cdot\mid x\right)$
of $A$ given $X=x$ with respect to some dominating measure $m$
is dominated by $g\left(\cdot\right)$, uniformly in $x\in\mathbb{B}$,
with $\int\left|y\right|^{r}g\left(y\right)m\left(\mathrm{d}y\right)<\infty$
for some $r>0$, then we have $\max_{i}\left|W_{+,p,i}A_{i}\right|=o_{p}\left(\left(nh\right)^{1/r}\right)$.
Similar results with ``$+$'' replaced by ``$-$'' also hold. 
\end{lem}
\begin{proof}[Proof of Lemma \ref{lem: reg weights basic}]
For Part (i), note that 
\begin{eqnarray*}
\mathrm{E}\left[\frac{1}{h}\bar{W}^{k}_{+,p}A\right] & = & \int^{\overline{x}}_{0}\frac{1}{h}\left(\mathrm{e}^{\top}_{p+1,1}\left(\varphi\Lambda_{+,p}\right)^{-1}r_{p}\left(\frac{x}{h}\right)\right)^{k}K^{k}\left(\frac{x}{h}\right)\mu_{A}\left(x\right)f_{X}\left(x\right)\mathrm{d}x\\
 & = & \int^{\frac{\overline{x}}{h}}_{0}\left(\mathrm{e}^{\top}_{p+1,1}\left(\varphi\Lambda_{+,p}\right)^{-1}r_{p}\left(y\right)\right)^{k}K^{k}\left(y\right)\mu_{A}\left(hy\right)f_{X}\left(hy\right)\mathrm{d}y\\
 & = & \frac{\mu_{A,+}\omega^{0,k}_{+,p,0}}{\varphi^{k-1}}+o\left(1\right)
\end{eqnarray*}
where the first equality follows from the LIE, the second equality
follows from a change of variables, and the third equality follows
from the continuity of $\mu_{A}$ and $f_{X}$, and the DCT.

By Taylor's theorem, for $X_{i}>0$, 
\begin{equation}
\mu_{A}\left(X_{i}\right)=\mu^{\top}_{+}r_{p}\left(X_{i}\right)+\frac{\mu^{\left(p+1\right)}_{A}\left(\dot{t}X_{i}\right)}{\left(p+1\right)!}X^{p+1}_{i},\label{eq:mu Taylor}
\end{equation}
for some $\dot{t}\in\left(0,1\right)$, where $\mu_{+}\coloneqq\left(\mu_{A,+},\mu^{\left(1\right)}_{A,+}/1!,\ldots,\mu^{\left(p\right)}_{A,+}/p!\right)^{\top}$.
Then, we write 
\begin{equation}
\frac{1}{nh}\sum_{i}W_{+,p,i}\mu_{A}\left(X_{i}\right)=\frac{1}{nh}\sum_{i}W_{+,p,i}\left(r^{\top}_{p}\left(X_{i}\right)\mu_{+}\right)+\frac{1}{nh}\sum_{i}W_{+,p,i}\frac{\mu^{\left(p+1\right)}_{A}\left(\dot{t}X_{i}\right)}{\left(p+1\right)!}X^{p+1}_{i}.\label{eq:smoothing bias decompose 1}
\end{equation}
By the definition of $W_{+,p,i}$, the first term on the right-hand
side of (\ref{eq:smoothing bias decompose 1}) equals $\mu_{A,+}$.
Write 
\begin{eqnarray}
\frac{1}{nh}\sum_{i}W_{+,p,i}\frac{\mu^{\left(p+1\right)}_{A}\left(\dot{t}X_{i}\right)}{\left(p+1\right)!}X^{p+1}_{i} & = & \frac{1}{nh}\sum_{i}W_{+,p,i}\frac{\left(\mu^{\left(p+1\right)}_{A}\left(\dot{t}X_{i}\right)-\mu^{\left(p+1\right)}_{A,+}\right)}{\left(p+1\right)!}X^{p+1}_{i}\nonumber \\
 &  & +\frac{\mu^{\left(p+1\right)}_{A,+}}{\left(p+1\right)!}\left(\frac{1}{nh}\sum_{i}W_{+,p,i}X^{p+1}_{i}\right).\label{eq:smoothing bias decompose 2}
\end{eqnarray}
For the first term on the right-hand side, we have 
\begin{eqnarray*}
 &  & \left|\frac{1}{nh}\sum_{i}W_{+,p,i}\frac{\left(\mu^{\left(p+1\right)}_{A}\left(\dot{t}X_{i}\right)-\mu^{\left(p+1\right)}_{A,+}\right)}{\left(p+1\right)!}X^{p+1}_{i}\right|\\
 & \leq & \left(\frac{1}{nh}\sum_{i}\left|W_{+,p,i}\frac{\left(X_{i}/h\right)^{p+1}}{\left(p+1\right)!}\right|\right)\left(\underset{0<x<h}{\mathrm{sup}}\left|\mu^{\left(p+1\right)}_{A}\left(x\right)-\mu^{\left(p+1\right)}_{A,+}\right|\right)h^{p+1}.
\end{eqnarray*}
By Markov's inequality and (\ref{eq:PI_hat_inv - PI_inv rate}), $\left(nh\right)^{-1}\sum_{i}\left|W_{+,p,i}\left(X_{i}/h\right)^{p+1}\right|=O_{p}\left(1\right)$
and therefore, the first term on the right-hand side of (\ref{eq:smoothing bias decompose 2})
is $o_{p}\left(h^{p+1}\right)$. It now follows from this result,
(\ref{eq:smoothing bias decompose 1}), and (\ref{eq:smoothing bias decompose 2})
that 
\begin{equation}
\frac{1}{nh}\sum_{i}W_{+,p,i}\mu_{A}\left(X_{i}\right)=\mu_{A,+}+\frac{\mu^{\left(p+1\right)}_{A,+}}{\left(p+1\right)!}\left(\frac{1}{nh}\sum_{i}W_{+,p,i}X^{p+1}_{i}\right)+o_{p}\left(h^{p+1}\right).\label{eq:smoothing bias regression}
\end{equation}
By (\ref{eq:PI_hat_inv - PI_inv rate}), Chebyshev's inequality and
a change of variables,
\begin{equation}
\frac{1}{nh}\sum_{i}W_{+,p,i}\left(\frac{X_{i}}{h}\right)^{p+1}=\omega^{p+1,1}_{+,p,0}+o_{p}\left(1\right).\label{eq:(X/h)^p+1 term}
\end{equation}
Therefore, the second term on the right-hand side of (\ref{eq:smoothing bias regression})
equals $\left(\mu^{\left(p+1\right)}_{A,+}\omega^{p+1,1}_{+,p,0}h^{p+1}\right)/\left(p+1\right)!+o_{p}\left(h^{p+1}\right)$.
Part (ii) follows from this result and (\ref{eq:smoothing bias regression}).

For Part (iii), write 
\begin{equation}
\frac{1}{nh}\sum_{i}W^{k}_{+,p,i}A_{i}-\mathrm{E}\left[\frac{1}{h}\bar{W}^{k}_{+,p}A\right]=\frac{1}{nh}\sum_{i}\left(W^{k}_{+,p,i}-\bar{W}^{k}_{+,p,i}\right)A_{i}+\left\{ \frac{1}{nh}\sum_{i}\bar{W}^{k}_{+,p,i}A_{i}-\mathrm{E}\left[\frac{1}{h}\bar{W}^{k}_{+,p}A\right]\right\} .\label{eq:decompose 3}
\end{equation}
Then, by the binomial theorem, the boundedness of $\mu_{\left|A\right|}\left(\cdot\right)$,
Markov's inequality, and (\ref{eq:PI_hat_inv - PI_inv rate}), the
first term on the right-hand side of (\ref{eq:decompose 3}) is $o_{p}\left(1\right)$.
By the von Bahr--Esseen inequality (see, e.g., \citealp[inequality 9.3.a]{lin_bai_2010_probability_inequalities})
and the elementary inequality $\mathrm{E}\left[\left|A-\mathrm{E}\left[A\right]\right|^{1+\delta}\right]\leq2^{1+\delta}\mathrm{E}\left[\left|A\right|^{1+\delta}\right]$,
\begin{eqnarray}
\mathrm{E}\left[\left|\frac{1}{nh}\sum_{i}\bar{W}^{k}_{+,p,i}A_{i}-\mathrm{E}\left[\frac{1}{h}\bar{W}^{k}_{+,p}A\right]\right|^{1+\delta}\right] & \apprle & \frac{1}{\left(nh\right)^{1+\delta}}\sum_{i}\mathrm{E}\left[\left|\bar{W}^{k}_{+,p,i}A_{i}\right|^{1+\delta}\right]\nonumber \\
 & = & o\left(1\right),\label{eq:1+delta absolute moment bound}
\end{eqnarray}
where the equality follows from (\ref{eq:PI_hat_inv - PI_inv rate})
and the boundedness of $\mu_{\left|A\right|^{1+\delta}}\left(\cdot\right)$.
Now Markov's inequality yields that the second term on the right-hand
side of (\ref{eq:decompose 3}) is $o_{p}\left(1\right)$. Part (iii)
follows from these results.

For Part (iv), first we note that $\max_{i}\left|W_{+,p,i}A_{i}\right|\leq\left\Vert \Pi^{-1}_{+,p}\right\Vert \max_{i}\left\Vert K_{+,p}\left(X_{i}/h\right)\right\Vert \left|A_{i}\right|$.
It follows from (\ref{eq:PI_hat_inv - PI_inv rate}) that $\left\Vert \Pi^{-1}_{+,p}\right\Vert =O_{p}\left(1\right)$.
Fix an arbitrary $\varepsilon>0$. By the union bound, 
\begin{eqnarray}
 &  & \Pr\left[\left(nh\right)^{-1/r}\underset{i}{\max}\left\Vert K_{+,p}\left(\frac{X_{i}}{h}\right)\right\Vert \left|A_{i}\right|>\varepsilon\right]\nonumber \\
 & \leq & \sum_{i}\Pr\left[\left\Vert K_{+,p}\left(\frac{X_{i}}{h}\right)\right\Vert ^{r}\left|A_{i}\right|^{r}>\left(nh\right)\varepsilon^{r}\right]\nonumber \\
 & \leq & \mathrm{E}\left[\frac{1}{h}\left\Vert K_{+,p}\left(\frac{X_{i}}{h}\right)\right\Vert ^{r}\left|A\right|^{r}\mathbbm{1}\left(\left|A\right|>\frac{\left(nh\right)^{1/r}\varepsilon}{\max_{u\in\mathbb{R}}\left\Vert K_{+,p}\left(u\right)\right\Vert }\right)\right]/\varepsilon^{r}.\label{eq:max upper bound}
\end{eqnarray}
Since $nh\uparrow\infty$, by the DCT, the right-hand side of the
second inequality above converges to zero as $n\uparrow\infty$. It
now follows that $\max_{i}\left\Vert K_{+,p}\left(X_{i}/h\right)\right\Vert \left|A_{i}\right|=o_{p}\left(\left(nh\right)^{1/r}\right)$. 
\end{proof}

By the monotonicity of $x\mapsto\left(x^{1+\varrho}-1\right)/\left(\varrho(1+\varrho)\right)$
on $\left(0,\infty\right)$, $\widehat{\lambda}_{p}$ maximizes 
\[
P_{\varrho}\left(\lambda\right)\coloneqq\frac{1}{\varrho}\sum_{i}\left(1+V^{\top}_{p,i}\lambda\right)^{\varrho/(1+\varrho)}
\]
subject to $1+V^{\top}_{p,i}\lambda>0$ for all $i=1,\ldots,n$. It
follows from simple calculations that $P_{\varrho}\left(\cdot\right)$
is concave on the constraint set. 
\begin{lem}
\label{lem:reg weights lambda}Suppose that the assumptions in the
statement of Theorem \ref{thm:normality} hold. Then, $\widehat{\lambda}_{p}=O_{p}\left(\left(nh\right)^{-1/2}\right)$
and 
\begin{equation}
\widehat{\lambda}_{p}=\left(1+\varrho\right)\left(\sum_{i}V_{p,i}V^{\top}_{p,i}\right)^{-1}\sum_{i}V_{p,i}+o_{p}\left(\left(nh\right)^{-1/2}\right).\label{eq:lambda_hat linearization}
\end{equation}
\end{lem}
\begin{proof}
Denote $\mathbb{L}\coloneqq\left\{ \lambda:\left\Vert \lambda\right\Vert \leq\left(nh\right)^{-1/\left(2+\delta\right)}\right\} $.
By Lemma \ref{lem: reg weights basic}(iv), we have $\max_{i}\sup_{\lambda\in\mathbb{L}}\left|V^{\top}_{p,i}\lambda\right|=o_{p}\left(1\right)$.
The set $\mathbb{L}$ is contained in the constraint set $\left\{ \lambda:1+V^{\top}_{p,i}\lambda>0,\,\forall i\right\} $
wpa1, and $\bar{\lambda}_{p}\coloneqq\arg\max_{\lambda\in\mathbb{L}}P_{\varrho}\left(\lambda\right)$
exists wpa1. By a Taylor expansion, 
\begin{eqnarray}
P_{\varrho}\left(0\right) & \leq & P_{\varrho}\left(\bar{\lambda}_{p}\right)\nonumber \\
 & = & P_{\varrho}\left(0\right)+\frac{1}{1+\varrho}\left(\sum_{i}V_{p,i}\right)^{\top}\bar{\lambda}_{p}-\frac{1}{2\left(1+\varrho\right)^{2}}\sum_{i}\frac{\left(V^{\top}_{p,i}\bar{\lambda}_{p}\right)^{2}}{\left(1+\dot{t}\cdot V^{\top}_{p,i}\bar{\lambda}_{p}\right)^{1/\left(1+\varrho\right)+1}},\label{eq:P expansion}
\end{eqnarray}
for some $\dot{t}\in\left(0,1\right)$. By the continuity of $x\mapsto x^{-1/\left(1+\varrho\right)-1}$
at $x=1$ and $\max_{i}\left|V^{\top}_{p,i}\bar{\lambda}_{p}\right|=o_{p}\left(1\right)$,
we have $\min_{i}\left(1+\dot{t}\cdot V^{\top}_{p,i}\bar{\lambda}_{p}\right)^{-1/\left(1+\varrho\right)-1}>1/2$
wpa1. It now follows that $\mathrm{mineig}\left(\sum_{i}V_{p,i}V^{\top}_{p,i}\right)\left\Vert \bar{\lambda}_{p}\right\Vert \leq4\left|1+\varrho\right|\left\Vert \sum_{i}V_{p,i}\right\Vert $
wpa1. By Lemma \ref{lem: reg weights basic}(i) and (iii), 
\[
\mathrm{mineig}\left(\frac{1}{nh}\sum_{i}V_{p,i}V^{\top}_{p,i}\right)=\left(\frac{\omega^{0,2}_{p,0}}{\varphi}\right)\mathrm{mineig}\left(\mu_{\bar{Z}\bar{Z}^{\top},\pm}\right)+o_{p}\left(1\right)
\]
and under our assumptions, $\mu_{\bar{Z}\bar{Z}^{\top},\pm}$ is positive
definite and $\mathrm{mineig}\left(\mu_{\bar{Z}\bar{Z}^{\top},\pm}\right)>0$.

Let $\xi_{i}\coloneqq Z_{i}-\mu_{Z}\left(X_{i}\right)$. By Lemma
\ref{lem: reg weights basic}(ii), 
\begin{eqnarray}
\frac{1}{nh}\sum_{i}W_{p,i}Z_{i} & = & \frac{1}{nh}\sum_{i}W_{p,i}\mu_{Z}\left(X_{i}\right)+\frac{1}{nh}\sum_{i}W_{p,i}\xi_{i}\nonumber \\
 & = & \frac{1}{nh}\sum_{i}W_{p,i}\xi_{i}+\frac{\left(\mu^{\left(p+1\right)}_{Z,+}\omega^{p+1,1}_{+,p,0}-\mu^{\left(p+1\right)}_{Z,-}\omega^{p+1,1}_{-,p,0}\right)}{\left(p+1\right)!}h^{p+1}+o_{p}\left(h^{p+1}\right).\label{eq:PSI_Z_hat decompose}
\end{eqnarray}
By Chebyshev's inequality and a change of variables, we have $\left(nh\right)^{-1/2}\sum_{i}K_{s,p}\left(X_{i}/h\right)\xi_{i}=O_{p}\left(1\right)$
for $s\in\left\{ -,+\right\} $, and therefore, by this result and
\eqref{eq:PI_hat_inv - PI_inv rate},
\begin{eqnarray}
\frac{1}{\sqrt{nh}}\sum_{i}W_{p,i}\xi^{\top}_{i} & = & \mathrm{e}^{\top}_{p+1,1}\Pi^{-1}_{+,p}\left(\frac{1}{\sqrt{nh}}\sum_{i}K_{+,p}\left(\frac{X_{i}}{h}\right)\xi^{\top}_{i}\right)-\mathrm{e}^{\top}_{p+1,1}\Pi^{-1}_{-,p}\left(\frac{1}{\sqrt{nh}}\sum_{i}K_{-,p}\left(\frac{X_{i}}{h}\right)\xi^{\top}_{i}\right)\nonumber \\
 & = & O_{p}\left(1\right).\label{eq:W*ksi expansion}
\end{eqnarray}
By this result and \eqref{eq:PSI_Z_hat decompose}, we have $\left(nh\right)^{-1/2}\sum_{i}W_{p,i}Z_{i}=O_{p}\left(1\right)$.
Then, it follows from this result and 
\begin{equation}
\frac{1}{\sqrt{nh}}\sum_{i}V_{p,i}=\left(\begin{array}{c}
0\\
\frac{1}{\sqrt{nh}}\sum_{i}W_{p,i}Z_{i}
\end{array}\right),\label{eq:V decompose}
\end{equation}
that 
\begin{equation}
\frac{1}{\sqrt{nh}}\sum_{i}V_{p,i}=O_{p}\left(1\right).\label{eq:V_average rate}
\end{equation}

Combining the inequality $\mathrm{mineig}\left(\sum_{i}V_{p,i}V^{\top}_{p,i}\right)\left\Vert \bar{\lambda}_{p}\right\Vert \leq4\left|1+\varrho\right|\left\Vert \sum_{i}V_{p,i}\right\Vert $
with the established fact that $\mathrm{mineig}\left(\left(nh\right)^{-1}\sum_{i}V_{p,i}V^{\top}_{p,i}\right)$
is bounded away from $0$ wpa1 and (\ref{eq:V_average rate}) yields
$\left\Vert \bar{\lambda}_{p}\right\Vert =O_{p}\left(\left(nh\right)^{-1/2}\right)$.
Since $\left(nh\right)^{-1/2}/\left(nh\right)^{-1/\left(2+\delta\right)}\rightarrow0$,
$\bar{\lambda}_{p}$ lies in the interior of $\mathbb{L}$ wpa1. The
objective $P_{\varrho}$ is concave on the open convex constraint
set $\left\{ \lambda:1+V^{\top}_{p,i}\lambda>0,\,\forall i\right\} $,
which contains $\mathbb{L}$ wpa1. Hence any local maximizer of $P_{\varrho}$
in the interior of $\mathbb{L}$ is also its global maximizer over
the full constraint set, and so $\widehat{\lambda}_{p}=\bar{\lambda}_{p}$
wpa1. In particular, $\widehat{\lambda}_{p}=O_{p}\left(\left(nh\right)^{-1/2}\right)$,
which establishes the first claim of the lemma. Since $\widehat{\lambda}_{p}\in\mathbb{L}$
wpa1, the bound $\max_{i}\sup_{\lambda\in\mathbb{L}}\left|V^{\top}_{p,i}\lambda\right|=o_{p}\left(1\right)$
noted above implies $\max_{i}\left|V^{\top}_{p,i}\widehat{\lambda}_{p}\right|=o_{p}\left(1\right)$.

Denote $\iota_{\varrho}\left(x\right)\coloneqq x^{-1/\left(1+\varrho\right)}$.
The first-order condition is given by $\sum_{i}V_{p,i}\iota_{\varrho}\left(1+V^{\top}_{p,i}\widehat{\lambda}_{p}\right)=0.$
By a Taylor expansion, 
\begin{equation}
\sum_{i}V_{p,i}\left\{ \iota_{\varrho}\left(1\right)+\iota_{\varrho}'\left(1\right)\left(V^{\top}_{p,i}\widehat{\lambda}_{p}\right)+\frac{1}{2}\iota_{\varrho}''\left(1+\dot{t}\cdot V^{\top}_{p,i}\widehat{\lambda}_{p}\right)\left(V^{\top}_{p,i}\widehat{\lambda}_{p}\right)^{2}\right\} =0,\label{eq:FOC Taylor}
\end{equation}
for some $\dot{t}\in\left(0,1\right)$. Using $\iota_{\varrho}\left(1\right)=1$
and $\iota'_{\varrho}\left(1\right)=-1/\left(1+\varrho\right)$, (\ref{eq:FOC Taylor})
rearranges to 
\begin{equation}
\widehat{\lambda}_{p}=\left(1+\varrho\right)\left(\sum_{i}V_{p,i}V^{\top}_{p,i}\right)^{-1}\sum_{i}V_{p,i}+\frac{1+\varrho}{2}\left(\sum_{i}V_{p,i}V^{\top}_{p,i}\right)^{-1}\sum_{i}\iota_{\varrho}''\left(1+\dot{t}\cdot V^{\top}_{p,i}\widehat{\lambda}_{p}\right)\left(V^{\top}_{p,i}\widehat{\lambda}_{p}\right)^{2}V_{p,i}.\label{eq:lambda hat expansion}
\end{equation}
Since $\widehat{\lambda}_{p}\in\mathbb{L}$ wpa1 gives $\max_{i}\left|V^{\top}_{p,i}\widehat{\lambda}_{p}\right|=o_{p}\left(1\right)$,
by the continuity of $\iota_{\varrho}''$ at $1$, 
\begin{equation}
\max_{i}\left|\iota_{\varrho}''\left(1+\dot{t}\cdot V^{\top}_{p,i}\widehat{\lambda}_{p}\right)\right|\leq\sup_{\left|t\right|\le\max_{i}\left|V^{\top}_{p,i}\widehat{\lambda}_{p}\right|}\left|\iota_{\varrho}''\left(1+t\right)\right|=O_{p}\left(1\right).\label{eq:max iota'' bound}
\end{equation}
Hence, by this result, the triangle inequality, and $\left\Vert \left(nh\right)^{-1}\sum_{i}V_{p,i}V^{\top}_{p,i}\right\Vert =O_{p}\left(1\right)$,
\begin{eqnarray*}
\left\Vert \frac{1}{nh}\sum_{i}\iota_{\varrho}''\left(1+\dot{t}\cdot V^{\top}_{p,i}\widehat{\lambda}_{p}\right)\left(V^{\top}_{p,i}\widehat{\lambda}_{p}\right)^{2}V_{p,i}\right\Vert  & \leq & O_{p}\left(1\right)\cdot\max_{i}\left\Vert V_{p,i}\right\Vert \cdot\left\Vert \frac{1}{nh}\sum_{i}V_{p,i}V^{\top}_{p,i}\right\Vert \left\Vert \widehat{\lambda}_{p}\right\Vert ^{2}\\
 & = & o_{p}\left(\left(nh\right)^{-1/2}\right).
\end{eqnarray*}
It follows that the remainder term on the right-hand side of (\ref{eq:lambda hat expansion})
is $o_{p}\left(\left(nh\right)^{-1/2}\right)$. Substituting back
into (\ref{eq:lambda hat expansion}) yields the second claim of the
lemma. 
\end{proof}

By a Taylor expansion similar to that in (\ref{eq:FOC Taylor}), applied
to the scalar $\iota_{\varrho}\left(1+V^{\top}_{p,i}\widehat{\lambda}_{p}\right)$,
we have 
\begin{equation}
\iota_{\varrho}\left(1+V^{\top}_{p,i}\widehat{\lambda}_{p}\right)=1-\frac{V^{\top}_{p,i}\widehat{\lambda}_{p}}{1+\varrho}+\frac{1}{2}\iota{}_{\varrho}''\left(1+\dot{t}\cdot V^{\top}_{p,i}\widehat{\lambda}_{p}\right)\left(V^{\top}_{p,i}\widehat{\lambda}_{p}\right)^{2},\label{eq:iota Taylor}
\end{equation}
for some $\dot{t}\in\left(0,1\right)$. Averaging (\ref{eq:iota Taylor})
over $i=1,\ldots,n$ and bounding the two remainder terms, we have
$n^{-1}\sum_{i}V^{\top}_{p,i}\widehat{\lambda}_{p}=O_{p}\left(n^{-1}\right)$
by (\ref{eq:V_average rate}) and Lemma \ref{lem:reg weights lambda},
and also 
\[
\left|\frac{1}{n}\sum_{i}\iota{}_{\varrho}''\left(1+\dot{t}\cdot V^{\top}_{p,i}\widehat{\lambda}_{p}\right)\left(V^{\top}_{p,i}\widehat{\lambda}_{p}\right)^{2}\right|\leq O_{p}\left(1\right)\left\Vert \widehat{\lambda}_{p}\right\Vert ^{2}\left\Vert \frac{1}{n}\sum_{i}V_{p,i}V^{\top}_{p,i}\right\Vert =O_{p}\left(n^{-1}\right)
\]
by Lemma \ref{lem: reg weights basic}(i, iii) and Lemma \ref{lem:reg weights lambda}.
Hence $n^{-1}\sum_{i}\iota_{\varrho}(1+V^{\top}_{p,i}\widehat{\lambda}_{p})=1+O_{p}\left(n^{-1}\right)$,
and substituting (\ref{eq:iota Taylor}) into (\ref{eq:balancing weights})
yields

\begin{equation}
\widehat{w}_{p,i}=\frac{1}{n}\left\{ \iota_{\varrho}\left(1\right)+\iota_{\varrho}'\left(1\right)\left(V^{\top}_{p,i}\widehat{\lambda}_{p}\right)+\frac{1}{2}\iota_{\varrho}''\left(1+\dot{t}\cdot V^{\top}_{p,i}\widehat{\lambda}_{p}\right)\left(V^{\top}_{p,i}\widehat{\lambda}_{p}\right)^{2}\right\} \left\{ 1+O_{p}\left(n^{-1}\right)\right\} ,\label{eq:w_hat Taylor expansion}
\end{equation}
where $\dot{t}\in\left(0,1\right)$ denotes the mean value.

Let 
\[
\Pi^{\mathit{eb}}_{+,p}\coloneqq\frac{1}{h}\sum_{i}\widehat{w}_{p,i}K\left(\frac{X_{i}}{h}\right)I_{i}r_{p}\left(\frac{X_{i}}{h}\right)r^{\top}_{p}\left(\frac{X_{i}}{h}\right)\textrm{ and }W^{\mathit{eb}}_{+,p,i}\coloneqq\mathrm{e}^{\top}_{p+1,1}\left(\Pi^{\mathit{eb}}_{+,p}\right)^{-1}K_{+,p}\left(\frac{X_{i}}{h}\right),
\]
and let $\left(\Pi^{\mathit{eb}}_{-,p},W^{\mathit{eb}}_{-,p,i}\right)$
be defined similarly. Let $W^{\mathit{eb}}_{p,i}\coloneqq W^{\mathit{eb}}_{+,p,i}-W^{\mathit{eb}}_{-,p,i}$.

It follows from (\ref{eq:w_hat Taylor expansion}) and Lemma \ref{lem:reg weights lambda}
that $\left\Vert \Pi^{\mathit{eb}}_{+,p}-\Pi_{+,p}\right\Vert =O_{p}\left(\left(nh\right)^{-1/2}\right)$.
By the same arguments used to show (\ref{eq:PI_hat_inv - PI_inv rate}),
\begin{equation}
\left\Vert \left(\Pi^{\mathit{eb}}_{+,p}\right)^{-1}-\Pi^{-1}_{+,p}\right\Vert =O_{p}\left(\left(nh\right)^{-1/2}\right).\label{eq:PI_eb_inv}
\end{equation}
A similar result holds for $\Pi^{\mathit{eb}}_{-,p}$. 
\begin{proof}[Proof of Theorem \ref{thm:normality}]
It is clear that $\widehat{\vartheta}^{\mathit{eb}}_{p}$ has a closed
form given by 
\begin{equation}
\widehat{\vartheta}^{\mathit{eb}}_{p}=\frac{1}{h}\sum_{i}\widehat{w}_{p,i}W^{\mathit{eb}}_{p,i}Y_{i}.\label{eq:theta_hat_eb close form}
\end{equation}
Denote $G_{p,i}\coloneqq W^{\mathit{eb}}_{p,i}Y_{i}$. Since by \eqref{eq:PI_hat_inv - PI_inv rate}
and \eqref{eq:PI_eb_inv}, $\left\Vert \left(\Pi^{\mathit{eb}}_{s,p}\right)^{-1}\right\Vert =O_{p}\left(1\right)$
and $\left\Vert \Pi^{-1}_{s,p}\right\Vert =O_{p}\left(1\right)$,
$\left(nh\right)^{-1}\sum_{i}\left|G_{p,i}\right|\left\Vert V_{p,i}\right\Vert =O_{p}\left(1\right)$
follows easily from these results and Markov's inequality. Then, by
this result, \eqref{eq:w_hat Taylor expansion}, $\widehat{\lambda}_{p}=O_{p}\left(\left(nh\right)^{-1/2}\right)$
from Lemma \ref{lem:reg weights lambda} and $\max_{i}\left|V^{\top}_{p,i}\widehat{\lambda}_{p}\right|=o_{p}\left(1\right)$,
we have 
\[
\widehat{\vartheta}^{\mathit{eb}}_{p}=\frac{1}{nh}\sum_{i}G_{p,i}\left\{ \iota_{\varrho}\left(1\right)+\iota_{\varrho}'\left(1\right)\left(V^{\top}_{p,i}\widehat{\lambda}_{p}\right)\right\} +o_{p}\left(\left(nh\right)^{-1/2}\right).
\]
Then, by \eqref{eq:lambda_hat linearization} and \eqref{eq:V_average rate},
\begin{equation}
\widehat{\vartheta}^{\mathit{eb}}_{p}=\frac{1}{nh}\sum_{i}G_{p,i}-\left(\frac{1}{nh}\sum_{i}G_{p,i}V^{\top}_{p,i}\right)\left(\frac{1}{nh}\sum_{i}V_{p,i}V^{\top}_{p,i}\right)^{-1}\frac{1}{nh}\sum_{i}V_{p,i}+o_{p}\left(\left(nh\right)^{-1/2}\right).\label{eq:thm1 reg eq 1}
\end{equation}

By \eqref{eq:A_inv - B_inv}, $\left\Vert \Pi^{\mathit{eb}}_{+,p}-\Pi_{+,p}\right\Vert =O_{p}\left(\left(nh\right)^{-1/2}\right)$,
and \eqref{eq:PI_eb_inv}, 
\begin{equation}
\frac{1}{nh}\sum_{i}\left(W^{\mathit{eb}}_{+,p,i}-W_{+,p,i}\right)Y_{i}=-\mathrm{e}^{\top}_{p+1,1}\Pi^{-1}_{+,p}\left(\Pi^{\mathit{eb}}_{+,p}-\Pi_{+,p}\right)\widehat{\beta}_{+}+o_{p}\left(\left(nh\right)^{-1/2}\right),\label{eq:W_eb - W expansion}
\end{equation}
where $\widehat{\beta}_{+}\coloneqq\left(nh\right)^{-1}\sum_{i}\Pi^{-1}_{+,p}K_{+,p}\left(X_{i}/h\right)Y_{i}$.
By \eqref{eq:w_hat Taylor expansion}, \eqref{eq:W_eb - W expansion},
and the definition in \eqref{eq:regression weight}, we have 
\[
\frac{1}{nh}\sum_{i}\left(W^{\mathit{eb}}_{+,p,i}-W_{+,p,i}\right)Y_{i}=\frac{1}{1+\varrho}\left(\frac{1}{nh}\sum_{i}W^{2}_{+,p,i}\left(\widehat{\lambda}^{\top}_{p}\bar{Z}_{i}\right)\left(r^{\top}_{p}\left(\frac{X_{i}}{h}\right)\widehat{\beta}_{+}\right)\right)+o_{p}\left(\left(nh\right)^{-1/2}\right).
\]
By standard arguments, $\widehat{\beta}_{+}=\mu_{Y,+}\mathrm{e}_{p+1,1}+o_{p}\left(1\right)$.
By this result and Lemma \ref{lem: reg weights basic}, 
\begin{eqnarray}
\frac{1}{nh}\sum_{i}W^{2}_{+,p,i}\bar{Z}_{i}\left(r^{\top}_{p}\left(\frac{X_{i}}{h}\right)\widehat{\beta}_{+}\right) & = & \left(\frac{\omega^{0,2}_{p,0}}{\varphi}\right)\mu_{\bar{Z}}\mu_{Y,+}+o_{p}\left(1\right)\nonumber \\
\frac{1}{nh}\sum_{i}V_{p,i}V^{\top}_{p,i} & = & \left(\frac{\omega^{0,2}_{p,0}}{\varphi}\right)\mu_{\bar{Z}\bar{Z}^{\top},\pm}+o_{p}\left(1\right),\label{eq:VV limit}
\end{eqnarray}
where $\mu_{\bar{Z}}$ denotes the common value of $\mu_{\bar{Z},+}$
and $\mu_{\bar{Z},-}$, which exists under Assumption \ref{assu:smoothness Z}(i).
It now follows from these results and Lemma \ref{lem:reg weights lambda}
that 
\[
\frac{1}{nh}\sum_{i}\left(W^{\mathit{eb}}_{+,p,i}-W_{+,p,i}\right)Y_{i}=\left(\frac{1}{nh}\sum_{i}V_{p,i}\right)^{\top}\mu^{-1}_{\bar{Z}\bar{Z}^{\top},\pm}\mu_{\bar{Z}}\mu_{Y,+}+o_{p}\left(\left(nh\right)^{-1/2}\right).
\]
Similarly, 
\[
\frac{1}{nh}\sum_{i}\left(W^{\mathit{eb}}_{-,p,i}-W_{-,p,i}\right)Y_{i}=-\left(\frac{1}{nh}\sum_{i}V_{p,i}\right)^{\top}\mu^{-1}_{\bar{Z}\bar{Z}^{\top},\pm}\mu_{\bar{Z}}\mu_{Y,-}+o_{p}\left(\left(nh\right)^{-1/2}\right).
\]
Recall that $\bar{Z}\coloneqq\left(\begin{array}{cc}
1 & Z^{\top}\end{array}\right)^{\top}$. Clearly, $\mu_{\bar{Z}\bar{Z}^{\top},\pm}\mathrm{e}_{d_{z}+1,1}=2\mu_{\bar{Z}}$.
Then, by these results and \eqref{eq:V decompose}, we have 
\begin{eqnarray}
\frac{1}{nh}\sum_{i}G_{p,i}-\frac{1}{nh}\sum_{i}W_{p,i}Y_{i} & = & \left(\frac{1}{nh}\sum_{i}V_{p,i}\right)^{\top}\mathrm{e}_{d_{z}+1,1}\left(\frac{\mu_{Y,\pm}}{2}\right)+o_{p}\left(\left(nh\right)^{-1/2}\right)\nonumber \\
 & = & o_{p}\left(\left(nh\right)^{-1/2}\right).\label{eq:W_eb - W negligible}
\end{eqnarray}

Next, by Lemma \ref{lem: reg weights basic} and \eqref{eq:PI_eb_inv},
\begin{equation}
\frac{1}{nh}\sum_{i}G_{p,i}V^{\top}_{p,i}=\left(\frac{\omega^{0,2}_{p,0}}{\varphi}\right)\mu_{Y\bar{Z}^{\top},\pm}+o_{p}\left(1\right).\label{eq:GV VV limits}
\end{equation}
Therefore, by this result and \eqref{eq:VV limit}, 
\begin{equation}
\left(\frac{1}{nh}\sum_{i}G_{p,i}V^{\top}_{p,i}\right)\left(\frac{1}{nh}\sum_{i}V_{p,i}V^{\top}_{p,i}\right)^{-1}=\mu_{Y\bar{Z}^{\top},\pm}\mu^{-1}_{\bar{Z}\bar{Z}^{\top},\pm}+o_{p}\left(1\right).\label{eq:gamma_hat}
\end{equation}
We write $\mu_{Y\bar{Z}^{\top},\pm}=\left[\begin{array}{cc}
\mu_{Y,\pm} & \mu_{YZ^{\top},\pm}\end{array}\right]$, write $\mu_{\bar{Z}\bar{Z}^{\top},\pm}$ as a block matrix and apply
the block inversion formula to get 
\begin{equation}
\left[\begin{array}{cc}
2 & \mu_{Z^{\top},\pm}\\
\mu_{Z,\pm} & \mu_{ZZ^{\top},\pm}
\end{array}\right]^{-1}=\left[\begin{array}{cc}
\frac{1}{2}+\frac{1}{4}\mu_{Z^{\top},\pm}\mathrm{S}^{-1}\mu_{Z,\pm} & -\frac{1}{2}\mu_{Z^{\top},\pm}\mathrm{S}^{-1}\\
-\frac{1}{2}\mathrm{S}^{-1}\mu_{Z,\pm} & \mathrm{S}^{-1}
\end{array}\right],\label{eq:block inversion}
\end{equation}
where $\mathrm{S}\coloneqq\mu_{ZZ^{\top},\pm}-\mu_{Z,\pm}\mu_{Z^{\top},\pm}/2=\mathrm{Var}_{\pm}\left[Z\right]$.
By this result and \eqref{eq:V decompose}, we get 
\[
\mu_{Y\bar{Z}^{\top},\pm}\mu^{-1}_{\bar{Z}\bar{Z}^{\top},\pm}\left(\frac{1}{nh}\sum_{i}V_{p,i}\right)=\frac{1}{nh}\sum_{i}W_{p,i}Z^{\top}_{i}\gamma.
\]
Then by this result and \eqref{eq:gamma_hat}, we have 
\[
\left(\frac{1}{nh}\sum_{i}G_{p,i}V^{\top}_{p,i}\right)\left(\frac{1}{nh}\sum_{i}V_{p,i}V^{\top}_{p,i}\right)^{-1}\frac{1}{nh}\sum_{i}V_{p,i}=\frac{1}{nh}\sum_{i}W_{p,i}Z^{\top}_{i}\gamma+o_{p}\left(\left(nh\right)^{-1/2}\right).
\]
Note that this result, \eqref{eq:thm1 reg eq 1} and \eqref{eq:W_eb - W negligible}
yield that $\widehat{\vartheta}^{\mathit{eb}}_{p}=\left(nh\right)^{-1}\sum_{i}W_{p,i}\epsilon_{i}+o_{p}\left(\left(nh\right)^{-1/2}\right)$,
where $\epsilon_{i}\coloneqq Y_{i}-Z^{\top}_{i}\gamma$. Note that
$\widehat{\vartheta}^{\mathit{CCFT}}_{p}$ has the same asymptotic
representation under Assumption \ref{assu:smoothness Z}(i). The conclusion
follows from these observations. 
\end{proof}

\section{Proof of Theorem \ref{thm:sharp QRD}}

Let 
\[
\widehat{\kappa}_{+,p}\left(\tau\right)\coloneqq\underset{b}{\arg\min}\sum_{i}\widehat{w}_{p,i}K\left(\frac{X_{i}}{h}\right)I_{i}\rho_{\tau}\left(Y_{i}-r^{\top}_{p}\left(X_{i}\right)b\right)
\]
and let $\widehat{\kappa}_{-,p}\left(\tau\right)$ be defined by the
same formula with $I_{i}$ replaced by $\mathbbm{1}\left(X_{i}<0\right)$.
Clearly, $\widehat{\vartheta}^{\mathit{eb}}_{p}\left(\tau\right)=\mathrm{e}^{\top}_{p+1,1}\left(\widehat{\kappa}_{+,p}\left(\tau\right)-\widehat{\kappa}_{-,p}\left(\tau\right)\right)$.
Let $\psi_{\tau}\left(u\right)\coloneqq\tau-\mathbbm{1}\left(u\leq0\right)$,
$\kappa_{+}\left(\tau\right)\coloneqq\left(\kappa_{Y,+}\left(\tau\right),\kappa^{\left(1\right)}_{Y,+}\left(\tau\right),\ldots,\kappa^{\left(p\right)}_{Y,+}\left(\tau\right)/p!\right)^{\top}$
and 
\begin{eqnarray*}
R_{+}\left(\varDelta,\tau\right) & \coloneqq & \sqrt{\frac{n}{h}}\sum_{i}\widehat{w}_{p,i}K_{+,p}\left(\frac{X_{i}}{h}\right)\psi_{\tau}\left(\tilde{Y}_{+,i}-\frac{r^{\top}_{p}\left(X_{i}/h\right)\varDelta}{\sqrt{nh}}\right),\textrm{ where}\\
\tilde{Y}_{+,i} & \coloneqq & Y_{i}-r^{\top}_{p}\left(X_{i}\right)\kappa_{+}\left(\tau\right).
\end{eqnarray*}

Let $\mathbb{P}g$ denote $\mathrm{E}\left[g\left(Y,X,Z\right)\right]$
and let $\mathbb{P}_{n}g$ denote the sample average. Let $\mathbb{G}_{n}g\coloneqq\sqrt{n}\left(\mathbb{P}_{n}-\mathbb{P}\right)g$
and let $\left\{ \mathbb{G}_{n}g:g\in\mathscr{G}\right\} $ denote
the empirical process indexed by $g\in\mathscr{G}$, where $\mathscr{G}$
is a function class. Let $\left\Vert \mathbb{G}_{n}\right\Vert _{\mathscr{G}}\coloneqq\sup_{g\in\mathscr{G}}\left|\mathbb{G}_{n}g\right|$.
The proof of Theorem \ref{thm:sharp QRD} follows an adaptation of
the approach in \citet{Koenker1994}. 
\begin{lem}
\label{lem:R_delta linearization}Suppose that the assumptions in
the statement of Theorem \ref{thm:sharp QRD} hold. For all $M>0$,
\begin{equation}
\sup_{\|\varDelta\|\leq M,\,\tau\in\left[\underline{\tau},\overline{\tau}\right]}\|R_{+}(\varDelta,\tau)+D_{+}(\tau)\varDelta-R_{+}(0,\tau)\|=o_{p}(1),\label{eq:asymptotic linearization}
\end{equation}
where $D_{+}\left(\tau\right)\coloneqq\varphi_{+}\left(\tau\right)\varphi\Lambda_{+,p}$.
\end{lem}
\begin{proof}[Proof of Lemma \ref{lem:R_delta linearization}]
By (\ref{eq:w_hat Taylor expansion}), we have 
\begin{eqnarray*}
 &  & R_{+}\left(\varDelta,\tau\right)-R_{+}\left(0,\tau\right)\\
 & = & \sqrt{\frac{n}{h}}\sum_{i}\widehat{w}_{p,i}K_{+,p}\left(\frac{X_{i}}{h}\right)\left\{ \mathbbm{1}\left(\tilde{Y}_{+,i}\leq0\right)-\mathbbm{1}\left(\tilde{Y}_{+,i}\leq\frac{r^{\top}_{p}\left(X_{i}/h\right)\varDelta}{\sqrt{nh}}\right)\right\} \\
 & = & T_{1}\left(\varDelta,\tau\right)\left\{ 1+O_{p}\left(n^{-1}\right)\right\} +T_{2}\left(\varDelta,\tau\right)\left\{ 1+O_{p}\left(n^{-1}\right)\right\} ,
\end{eqnarray*}
where 
\begin{eqnarray*}
T_{1}\left(\varDelta,\tau\right) & \coloneqq & \frac{1}{\sqrt{nh}}\sum_{i}K_{+,p}\left(\frac{X_{i}}{h}\right)\left\{ \mathbbm{1}\left(\tilde{Y}_{+,i}\leq0\right)-\mathbbm{1}\left(\tilde{Y}_{+,i}\leq\frac{r^{\top}_{p}\left(X_{i}/h\right)\varDelta}{\sqrt{nh}}\right)\right\} \\
T_{2}\left(\varDelta,\tau\right) & \coloneqq & \frac{1}{\sqrt{nh}}\sum_{i}K_{+,p}\left(\frac{X_{i}}{h}\right)\left(\iota_{\varrho}\left(1+V^{\top}_{p,i}\widehat{\lambda}_{p}\right)-1\right)\left\{ \mathbbm{1}\left(\tilde{Y}_{+,i}\leq0\right)-\mathbbm{1}\left(\tilde{Y}_{+,i}\leq\frac{r^{\top}_{p}\left(X_{i}/h\right)\varDelta}{\sqrt{nh}}\right)\right\} .
\end{eqnarray*}
We decompose 
\begin{eqnarray}
T_{1}\left(\varDelta,\tau\right) & = & \left\{ T_{1}\left(\varDelta,\tau\right)-\mathrm{E}\left[T_{1}\left(\varDelta,\tau\right)\right]\right\} +\mathrm{E}\left[T_{1}\left(\varDelta,\tau\right)\right]\nonumber \\
 & = & \frac{1}{\sqrt{n}}\sum_{i}\left\{ \frac{1}{\sqrt{h}}K_{+,p}\left(\frac{X_{i}}{h}\right)\left(\mathbbm{1}\left(\tilde{Y}_{+,i}\leq0\right)-\mathbbm{1}\left(\tilde{Y}_{+,i}\leq\frac{r^{\top}_{p}\left(X_{i}/h\right)\varDelta}{\sqrt{nh}}\right)\right)\right.\nonumber \\
 &  & \left.-\mathrm{E}\left[\frac{1}{\sqrt{h}}K_{+,p}\left(\frac{X}{h}\right)\left(\mathbbm{1}\left(\tilde{Y}_{+}\leq0\right)-\mathbbm{1}\left(\tilde{Y}_{+}\leq\frac{r^{\top}_{p}\left(X/h\right)\varDelta}{\sqrt{nh}}\right)\right)\right]\right\} \nonumber \\
 &  & +\sqrt{nh}\cdot\mathrm{E}\left[\frac{1}{h}K_{+,p}\left(\frac{X}{h}\right)\left(\mathbbm{1}\left(\tilde{Y}_{+}\leq0\right)-\mathbbm{1}\left(\tilde{Y}_{+}\leq\frac{r^{\top}_{p}\left(X/h\right)\varDelta}{\sqrt{nh}}\right)\right)\right].\label{eq:T_1 decomposition}
\end{eqnarray}
Under Assumption \ref{assu:QRD smoothness}(i)--(iii), for $x=hu$
with $u\in\left(0,1\right]$, we have 
\begin{eqnarray}
 &  & F_{Y\mid X}\left(r^{\top}_{p}\left(x\right)\kappa_{+}\left(\tau\right)\mid x\right)-F_{Y\mid X}\left(r^{\top}_{p}\left(x\right)\kappa_{+}\left(\tau\right)+\frac{r^{\top}_{p}\left(x/h\right)\varDelta}{\sqrt{nh}}\mid x\right)\nonumber \\
 & = & \left\{ -f_{Y\mid X}\left(Q_{Y\mid X}\left(\tau\mid x\right)\mid x\right)+o\left(1\right)\right\} \frac{r^{\top}_{p}\left(x/h\right)\varDelta}{\sqrt{nh}},\label{eq:F difference bound 1}
\end{eqnarray}
uniformly for all $\left\Vert \varDelta\right\Vert \leq M,\tau\in\left[\underline{\tau},\overline{\tau}\right]$.
By this result and the LIE, we have 
\begin{eqnarray*}
 &  & \frac{1}{\sqrt{nh}}\mathrm{E}\left[T_{1}\left(\varDelta,\tau\right)\right]\\
 & = & \mathrm{E}\left[\frac{1}{h}K_{+,p}\left(\frac{X}{h}\right)\left(F_{Y\mid X}\left(r^{\top}_{p}\left(X\right)\kappa_{+}\left(\tau\right)\mid X\right)-F_{Y\mid X}\left(r^{\top}_{p}\left(X\right)\kappa_{+}\left(\tau\right)+\frac{r^{\top}_{p}\left(X/h\right)\varDelta}{\sqrt{nh}}\mid X\right)\right)\right]\\
 & = & \left\{ -\mathrm{E}\left[\frac{1}{h}K_{+,p}\left(\frac{X}{h}\right)r^{\top}_{p}\left(\frac{X}{h}\right)f_{Y\mid X}\left(Q_{Y\mid X}\left(\tau\mid X\right)\mid X\right)\right]+o\left(1\right)\right\} \frac{\varDelta}{\sqrt{nh}}\\
 & = & -D_{+}\left(\tau\right)\frac{\varDelta}{\sqrt{nh}}+o\left(\left(nh\right)^{-1/2}\right),
\end{eqnarray*}
uniformly for all $\left\Vert \varDelta\right\Vert \leq M,\tau\in\left[\underline{\tau},\overline{\tau}\right]$.
For any $j=1,\ldots,p+1$, write $T^{\left[j\right]}_{1}\left(\varDelta,\tau\right)-\mathrm{E}\left[T^{\left[j\right]}_{1}\left(\varDelta,\tau\right)\right]=\mathbb{G}_{n}g_{1}\left(\cdot\mid\varDelta,\tau\right)$,
where 
\begin{equation}
g_{1}\left(y,x\mid\varDelta,\tau\right)\coloneqq\frac{1}{\sqrt{h}}K^{\left[j\right]}_{+,p}\left(\frac{x}{h}\right)\left\{ \mathbbm{1}\left(y\leq r^{\top}_{p}\left(x\right)\kappa_{+}\left(\tau\right)\right)-\mathbbm{1}\left(y\leq r^{\top}_{p}\left(x\right)\left(\kappa_{+}\left(\tau\right)+\frac{\mathrm{H}^{-1}\varDelta}{\sqrt{nh}}\right)\right)\right\} .\label{eq:g_1 definition}
\end{equation}
Let $\mathscr{G}_{1}\coloneqq\left\{ g_{1}\left(\cdot\mid\varDelta,\tau\right):\|\varDelta\|\leq M,\,\tau\in\left[\underline{\tau},\overline{\tau}\right]\right\} $.
By Markov's inequality, it suffices to show that $\mathrm{E}\left[\left\Vert \mathbb{G}_{n}\right\Vert _{\mathscr{G}_{1}}\right]=o\left(1\right)$.
The function class 
\begin{equation}
\left\{ \left(y,x\right)\mapsto\mathbbm{1}\left(y\leq r^{\top}_{p}\left(x\right)b\right):b\in\mathbb{R}^{p+1}\right\} \label{eq:indicator class 1}
\end{equation}
is Vapnik--Chervonenkis (VC)-subgraph with index no more than $p+4$,
because its elements are indicators of nonnegativity sets of functions
in a $\left(p+2\right)$-dimensional vector space (\citealp[Lemma 9.6]{Kosorok2007}).
Since multiplying by a fixed function preserves the VC property (\citealp[Lemma 9.9(vi)]{Kosorok2007}),
by \citet[Theorem 9.3]{Kosorok2007} and \citet[Lemma A.6]{chernozhukov2014gaussian},
$\mathscr{G}_{1}$ is VC-type (see, e.g., \citealp[Definition 2.1]{Chen2020jackknife})
with a constant envelope $G_{1}\apprle h^{-1/2}$ and characteristics
independent of $n$.

Note that for all $a,b$, 
\begin{equation}
\left(\mathbbm{1}\left(x\leq a\right)-\mathbbm{1}\left(x\leq b\right)\right)^{2}=\left|\mathbbm{1}\left(x\leq a\right)-\mathbbm{1}\left(x\leq b\right)\right|=\mathbbm{1}\left(\min\left\{ a,b\right\} <x\leq\max\left\{ a,b\right\} \right).\label{eq:a.b indicator inequalities}
\end{equation}
By this result, the LIE, (\ref{eq:F difference bound 1}), and a change
of variables, 
\begin{eqnarray}
 &  & \mathbb{P}g^{2}_{1}\left(\cdot\mid\varDelta,\tau\right)\nonumber \\
 & \leq & \mathrm{E}\left[\frac{1}{h}\left\Vert K_{+,p}\left(\frac{X}{h}\right)\right\Vert ^{2}\left|F_{Y\mid X}\left(r^{\top}_{p}\left(X\right)\kappa_{+}\left(\tau\right)\mid X\right)-F_{Y\mid X}\left(r^{\top}_{p}\left(X\right)\kappa_{+}\left(\tau\right)+\frac{r^{\top}_{p}\left(X/h\right)\varDelta}{\sqrt{nh}}\mid X\right)\right|\right]\nonumber \\
 & = & O\left(\left(nh\right)^{-1/2}\right),\label{eq:g_1^2 expectation bound}
\end{eqnarray}
uniformly in $\left\Vert \varDelta\right\Vert \leq M,\tau\in\left[\underline{\tau},\overline{\tau}\right]$.
By \citet[Corollary 5.5]{Chen2020jackknife}, 
\begin{eqnarray}
\mathrm{E}\left[\left\Vert \mathbb{G}_{n}\right\Vert _{\mathscr{G}_{1}}\right] & \apprle & \sqrt{\left(\sup_{\|\varDelta\|\leq M,\,\tau\in\left[\underline{\tau},\overline{\tau}\right]}\mathbb{P}g^{2}_{1}\left(\cdot\mid\varDelta,\tau\right)\right)\log\left(n\right)}+\frac{\log\left(n\right)}{\sqrt{nh}}\nonumber \\
 & = & O\left(\frac{\sqrt{\log\left(n\right)}}{\left(nh\right)^{1/4}}+\frac{\log\left(n\right)}{\sqrt{nh}}\right).\label{eq:maximal inequality}
\end{eqnarray}
Under $nh\to\infty$ and $\left(\log n\right)/\sqrt{nh}\to0$, both
terms are $o\left(1\right)$, so $\mathrm{E}\left[\left\Vert \mathbb{G}_{n}\right\Vert _{\mathscr{G}_{1}}\right]=o\left(1\right)$
holds and Markov's inequality yields $T_{1}\left(\varDelta,\tau\right)-\mathrm{E}\left[T_{1}\left(\varDelta,\tau\right)\right]=o_{p}\left(1\right)$
uniformly for all $\left\Vert \varDelta\right\Vert \leq M,\tau\in\left[\underline{\tau},\overline{\tau}\right]$.
It now follows that 
\begin{equation}
T_{1}\left(\varDelta,\tau\right)=-D_{+}\left(\tau\right)\varDelta+o_{p}\left(1\right),\label{eq:T_1 uniform approximation}
\end{equation}
and $T_{1}\left(\varDelta,\tau\right)=O_{p}\left(1\right)$, uniformly
for all $\left\Vert \varDelta\right\Vert \leq M,\tau\in\left[\underline{\tau},\overline{\tau}\right]$.

By the triangle inequality and (\ref{eq:a.b indicator inequalities}),
we have 
\begin{multline*}
\sup_{\|\varDelta\|\leq M}\left\Vert T_{2}\left(\varDelta,\tau\right)\right\Vert \leq\\
\frac{1}{\sqrt{nh}}\sum_{i}\left\Vert K_{+,p}\left(\frac{X_{i}}{h}\right)\right\Vert \left|\iota_{\varrho}\left(1+V^{\top}_{p,i}\widehat{\lambda}_{p}\right)-1\right|\mathbbm{1}\left(-\frac{\left\Vert r_{p}\left(X_{i}/h\right)\right\Vert M}{\sqrt{nh}}<\tilde{Y}_{+,i}\leq\frac{\left\Vert r_{p}\left(X_{i}/h\right)\right\Vert M}{\sqrt{nh}}\right).
\end{multline*}
Since $\max_{i}\left|\iota_{\varrho}\left(1+V^{\top}_{p,i}\widehat{\lambda}_{p}\right)-1\right|$
does not depend on $\tau$, the right-hand side is bounded by 
\[
\max_{i}\left|\iota_{\varrho}\left(1+V^{\top}_{p,i}\widehat{\lambda}_{p}\right)-1\right|\left(\sqrt{n}\cdot\mathbb{P}_{n}g_{2}\left(\cdot\mid\tau\right)\right),
\]
where 
\[
g_{2}\left(y,x\mid\tau\right)\coloneqq\frac{1}{\sqrt{h}}\left\Vert K_{+,p}\left(\frac{x}{h}\right)\right\Vert \mathbbm{1}\left(-\frac{\left\Vert r_{p}\left(x/h\right)\right\Vert M}{\sqrt{nh}}<y-r^{\top}_{p}\left(x\right)\kappa_{+}\left(\tau\right)\leq\frac{\left\Vert r_{p}\left(x/h\right)\right\Vert M}{\sqrt{nh}}\right).
\]

For the first factor, by Lemma \ref{lem:reg weights lambda}, $\widehat{\lambda}_{p}=O_{p}\left(\left(nh\right)^{-1/2}\right)$,
and $\max_{i}\left\Vert V_{p,i}\right\Vert =O_{p}\left(\left(nh\right)^{1/\left(2+\delta\right)}\right)$,
so $\max_{i}\left|V^{\top}_{p,i}\widehat{\lambda}_{p}\right|=o_{p}\left(1\right)$.
Since $\iota_{\varrho}$ is continuously differentiable at $1$ with
$\iota_{\varrho}\left(1\right)=1$, we have 
\[
\max_{i}\left|\iota_{\varrho}\left(1+V^{\top}_{p,i}\widehat{\lambda}_{p}\right)-1\right|\leq\sup_{\left|t\right|\leq\max_{i}\left|V^{\top}_{p,i}\widehat{\lambda}_{p}\right|}\left|\iota_{\varrho}'\left(1+t\right)\right|\cdot\max_{i}\left|V^{\top}_{p,i}\widehat{\lambda}_{p}\right|=o_{p}\left(1\right).
\]
Since $\sqrt{n}\cdot\mathbb{P}_{n}g_{2}\left(\cdot\mid\tau\right)=\sqrt{n}\cdot\mathbb{P}g_{2}\left(\cdot\mid\tau\right)+\mathbb{G}_{n}g_{2}\left(\cdot\mid\tau\right)$,
it suffices to bound $\sup_{\tau\in\left[\underline{\tau},\overline{\tau}\right]}\sqrt{n}\cdot\mathbb{P}g_{2}\left(\cdot\mid\tau\right)$.
By a result similar to (\ref{eq:F difference bound 1}), we have $\sqrt{n}\cdot\mathbb{P}g_{2}\left(\cdot\mid\tau\right)=O\left(1\right)$,
uniformly in $\tau$. For the fluctuation term, note that by similar
arguments, the class $\mathscr{G}_{2}\coloneqq\left\{ g_{2}\left(\cdot\mid\tau\right):\tau\in\left[\underline{\tau},\overline{\tau}\right]\right\} $
is VC-type with characteristics independent of $n$ and a constant
envelope $G_{2}\apprle h^{-1/2}$. By similar calculations, we can
show that $\sup_{\tau\in\left[\underline{\tau},\overline{\tau}\right]}\mathbb{P}g^{2}_{2}\left(\cdot\mid\tau\right)$
is of order $\left(nh\right)^{-1/2}$. The same maximal inequality
used in the proof of $\mathrm{E}\left[\left\Vert \mathbb{G}_{n}\right\Vert _{\mathscr{G}_{1}}\right]=o\left(1\right)$
yields $\mathrm{E}\left[\left\Vert \mathbb{G}_{n}\right\Vert _{\mathscr{G}_{2}}\right]=o\left(1\right)$.
Therefore $T_{2}\left(\varDelta,\tau\right)=o_{p}\left(1\right)$,
uniformly for all $\left\Vert \varDelta\right\Vert \leq M,\tau\in\left[\underline{\tau},\overline{\tau}\right]$. 
\end{proof}

Let $\tilde{K}_{+,p}\left(u\right)\coloneqq K^{2}\left(u\right)\mathbbm{1}\left(u>0\right)r_{p}\left(u\right)\left(\mathrm{e}^{\top}_{p+1,1}\Lambda^{-1}_{+,p}r_{p}\left(u\right)\right)$
and let $\pi_{+}\coloneqq\int^{1}_{0}\tilde{K}_{+,p}\left(u\right)\mathrm{d}u$.
Recall that $\xi_{i}\coloneqq Z_{i}-\mu_{Z}\left(X_{i}\right)$ and
$U_{i}\left(\tau\right)\coloneqq\tau-\mathbbm{1}\left(Y_{i}\leq Q_{Y\mid X}\left(\tau\mid X_{i}\right)\right)$.
Let $\gamma_{+}\left(\tau\right)\coloneqq\left(\mathrm{Var}_{\pm}\left[Z\right]\right)^{-1}\mu_{ZU\left(\tau\right),+}$
and let $\gamma_{-}\left(\tau\right)$ be defined similarly, so that
$\gamma\left(\tau\right)=\gamma_{+}\left(\tau\right)/\varphi_{+}\left(\tau\right)+\gamma_{-}\left(\tau\right)/\varphi_{-}\left(\tau\right)$
by the definition of $\gamma\left(\tau\right)$ in Section~\ref{subsec:quantile_RDD_reweighting}. 
\begin{lem}
\label{lem:R_0 linearization}Suppose that the assumptions in the
statement of Theorem \ref{thm:sharp QRD} hold. We have the following
linearization for $R_{+}(0,\tau)$: 
\begin{eqnarray*}
R_{+}(0,\tau) & = & \frac{1}{\sqrt{nh}}\sum_{i}K_{+,p}\left(\frac{X_{i}}{h}\right)U_{i}\left(\tau\right)-\left(\frac{\pi_{+}\varphi}{\omega^{0,2}_{p,0}}\right)\gamma^{\top}_{+}\left(\tau\right)\left\{ \frac{1}{\sqrt{nh}}\sum_{i}W_{p,i}\xi_{i}\right\} \\
 &  & +\bar{\mathscr{B}}_{+}\left(\tau\right)\sqrt{nh}h^{p+1}+o_{p}\left(1\right),
\end{eqnarray*}
uniformly in $\tau\in\left[\underline{\tau},\overline{\tau}\right]$,
where 
\[
\bar{\mathscr{B}}_{+}\left(\tau\right)\coloneqq\left(\frac{\int^{1}_{0}K_{+,p}\left(u\right)u^{p+1}\mathrm{d}u}{\left(p+1\right)!}\right)\varphi_{+}\left(\tau\right)\varphi\kappa^{\left(p+1\right)}_{Y,+}\left(\tau\right)-\left(\frac{\pi_{+}\varphi}{\omega^{0,2}_{p,0}}\right)\left\{ \frac{\gamma^{\top}_{+}\left(\tau\right)\left(\mu^{\left(p+1\right)}_{Z,+}\omega^{p+1,1}_{+,p,0}-\mu^{\left(p+1\right)}_{Z,-}\omega^{p+1,1}_{-,p,0}\right)}{\left(p+1\right)!}\right\} .
\]
\end{lem}
\begin{proof}[Proof of Lemma \ref{lem:R_0 linearization}]
We have 
\begin{eqnarray*}
R_{+}\left(0,\tau\right) & = & \sqrt{\frac{n}{h}}\sum_{i}\widehat{w}_{p,i}K_{+,p}\left(\frac{X_{i}}{h}\right)\psi_{\tau}\left(\tilde{Y}_{+,i}\right)\\
 & = & T_{1}\left(\tau\right)\left(1+O_{p}\left(n^{-1}\right)\right)+T_{2}\left(\tau\right)\left(1+O_{p}\left(n^{-1}\right)\right),
\end{eqnarray*}
where 
\begin{eqnarray*}
T_{1}\left(\tau\right) & \coloneqq & \frac{1}{\sqrt{nh}}\sum_{i}K_{+,p}\left(\frac{X_{i}}{h}\right)\left\{ 1-\frac{V^{\top}_{p,i}\widehat{\lambda}_{p}}{1+\varrho}\right\} \psi_{\tau}\left(\tilde{Y}_{+,i}\right)\\
T_{2}\left(\tau\right) & \coloneqq & \frac{1}{\sqrt{nh}}\sum_{i}K_{+,p}\left(\frac{X_{i}}{h}\right)\left\{ \iota_{\varrho}\left(1+V^{\top}_{p,i}\widehat{\lambda}_{p}\right)-\left(1-\frac{V^{\top}_{p,i}\widehat{\lambda}_{p}}{1+\varrho}\right)\right\} \psi_{\tau}\left(\tilde{Y}_{+,i}\right).
\end{eqnarray*}
By the Taylor expansion of $\iota_{\varrho}$ in (\ref{eq:iota Taylor}),
\[
\iota_{\varrho}\left(1+V^{\top}_{p,i}\widehat{\lambda}_{p}\right)-\left(1-\frac{V^{\top}_{p,i}\widehat{\lambda}_{p}}{1+\varrho}\right)=\frac{1}{2}\iota_{\varrho}''\left(1+\dot{t}\cdot V^{\top}_{p,i}\widehat{\lambda}_{p}\right)\left(V^{\top}_{p,i}\widehat{\lambda}_{p}\right)^{2}.
\]
Using $\left|\psi_{\tau}\left(\cdot\right)\right|\leq1$, (\ref{eq:max iota'' bound}),
$\max_{i}\left|V^{\top}_{p,i}\widehat{\lambda}_{p}\right|=o_{p}\left(1\right)$,
$\left(V^{\top}_{p,i}\widehat{\lambda}_{p}\right)^{2}\leq\left|V^{\top}_{p,i}\widehat{\lambda}_{p}\right|\cdot\left\Vert V_{p,i}\right\Vert \cdot\left\Vert \widehat{\lambda}_{p}\right\Vert $,
and 
\[
\frac{1}{nh}\sum_{i}\left\Vert K_{+,p}\left(\frac{X_{i}}{h}\right)\right\Vert \left\Vert V_{p,i}\right\Vert =O_{p}\left(1\right)
\]
by Markov's inequality and $\left\Vert \left(nh\right)^{-1}\sum_{i}V_{p,i}V^{\top}_{p,i}\right\Vert =O_{p}\left(1\right)$,
we have 
\begin{eqnarray*}
\left\Vert T_{2}\left(\tau\right)\right\Vert  & \apprle & \max_{i}\left|\iota_{\varrho}''\left(1+\dot{t}\cdot V^{\top}_{p,i}\widehat{\lambda}_{p}\right)\right|\cdot\max_{i}\left|V^{\top}_{p,i}\widehat{\lambda}_{p}\right|\cdot\left\Vert \widehat{\lambda}_{p}\right\Vert \cdot\frac{1}{\sqrt{nh}}\sum_{i}\left\Vert K_{+,p}\left(\frac{X_{i}}{h}\right)\right\Vert \left\Vert V_{p,i}\right\Vert \\
 & = & o_{p}\left(1\right).
\end{eqnarray*}
The bound is uniform in $\tau$ because no term on the right-hand
side depends on $\tau$.

Then, 
\begin{equation}
T_{1}\left(\tau\right)=\frac{1}{\sqrt{nh}}\sum_{i}K_{+,p}\left(\frac{X_{i}}{h}\right)\psi_{\tau}\left(\tilde{Y}_{+,i}\right)-\frac{1}{1+\varrho}\cdot\frac{1}{\sqrt{nh}}\sum_{i}K_{+,p}\left(\frac{X_{i}}{h}\right)\psi_{\tau}\left(\tilde{Y}_{+,i}\right)V^{\top}_{p,i}\widehat{\lambda}_{p}.\label{eq:R_1 decompose}
\end{equation}
We have 
\[
\frac{1}{\sqrt{nh}}\sum_{i}K_{+,p}\left(\frac{X_{i}}{h}\right)\psi_{\tau}\left(\tilde{Y}_{+,i}\right)=\frac{1}{\sqrt{nh}}\sum_{i}K_{+,p}\left(\frac{X_{i}}{h}\right)U_{i}\left(\tau\right)+T_{3}\left(\tau\right)
\]
where 
\begin{equation}
T_{3}\left(\tau\right)\coloneqq\frac{1}{\sqrt{nh}}\sum_{i}K_{+,p}\left(\frac{X_{i}}{h}\right)\left\{ \mathbbm{1}\left(Y_{i}\leq Q_{Y\mid X}\left(\tau\mid X_{i}\right)\right)-\mathbbm{1}\left(Y_{i}\leq r^{\top}_{p}\left(X_{i}\right)\kappa_{+}\left(\tau\right)\right)\right\} .\label{eq:T_3 definition}
\end{equation}
Since $\mathrm{E}\left[U\left(\tau\right)\mid X\right]=0$, the sum
is centered: for any $j=1,\ldots,p+1$, 
\[
\frac{1}{\sqrt{nh}}\sum_{i}K^{\left[j\right]}_{+,p}\left(\frac{X_{i}}{h}\right)U_{i}\left(\tau\right)=\mathbb{G}_{n}g_{3}\left(\cdot\mid\tau\right),
\]
where 
\[
g_{3}\left(y,x\mid\tau\right)\coloneqq\frac{1}{\sqrt{h}}K^{\left[j\right]}_{+,p}\left(\frac{x}{h}\right)\left\{ \tau-\mathbbm{1}\left(y\leq Q_{Y\mid X}\left(\tau\mid x\right)\right)\right\} .
\]
By the monotonicity of $\tau\mapsto Q_{Y\mid X}\left(\tau\mid x\right)$,
the function class 
\begin{equation}
\left\{ \left(y,x\right)\mapsto\mathbbm{1}\left(y\leq Q_{Y\mid X}\left(\tau\mid x\right)\right):\tau\in\left[\underline{\tau},\overline{\tau}\right]\right\} \label{eq:indicator class 2}
\end{equation}
is VC-subgraph with index 2 by \citet[Lemma 9.10]{Kosorok2007}. Then
it follows from \citet[Theorem 9.3 and Lemma 9.9(vi)]{Kosorok2007}
and \citet[Lemma A.6]{chernozhukov2014gaussian} that the class $\mathscr{G}_{3}\coloneqq\left\{ g_{3}\left(\cdot\mid\tau\right):\tau\in\left[\underline{\tau},\overline{\tau}\right]\right\} $
is VC-type with characteristics independent of $n$ and an envelope
$G_{3}\left(x\right)\coloneqq h^{-1/2}\left\Vert K_{+,p}\left(x/h\right)\right\Vert $.
Therefore, the uniform entropy integral is finite and independent
of $n$ (see, e.g., calculations in the proof of \citealp[Corollary 5.1]{chernozhukov2014gaussian}).
By a change of variables, $\mathbb{P}G^{2}_{3}=O\left(1\right)$.

By the maximal inequality in \citet[Theorem 2.14.1]{VanDerVaartWellner1996},
\begin{equation}
\mathrm{E}\left[\left\Vert \mathbb{G}_{n}\right\Vert _{\mathscr{G}_{3}}\right]\apprle\sqrt{\mathbb{P}G^{2}_{3}}=O\left(1\right).\label{eq:g_3 bound}
\end{equation}
Markov's inequality yields that $\mathbb{G}_{n}g_{3}\left(\cdot\mid\tau\right)=O_{p}\left(1\right)$,
uniformly in $\tau\in\left[\underline{\tau},\overline{\tau}\right]$.
For the bias term $T_{3}\left(\tau\right)$, we write $T_{3}\left(\tau\right)=T_{4}\left(\tau\right)+T_{5}\left(\tau\right)$,
where 
\begin{eqnarray*}
T_{4}\left(\tau\right) & \coloneqq & \frac{1}{\sqrt{nh}}\sum_{i}K_{+,p}\left(\frac{X_{i}}{h}\right)\left\{ \left(\mathbbm{1}\left(Y_{i}\leq Q_{Y\mid X}\left(\tau\mid X_{i}\right)\right)-\mathbbm{1}\left(Y_{i}\leq r^{\top}_{p}\left(X_{i}\right)\kappa_{+}\left(\tau\right)\right)\right)\right.\\
 &  & \left.-\left(\tau-F_{Y\mid X}\left(r^{\top}_{p}\left(X_{i}\right)\kappa_{+}\left(\tau\right)\mid X_{i}\right)\right)\right\} \\
T_{5}\left(\tau\right) & \coloneqq & \frac{1}{\sqrt{nh}}\sum_{i}K_{+,p}\left(\frac{X_{i}}{h}\right)\left\{ \tau-F_{Y\mid X}\left(r^{\top}_{p}\left(X_{i}\right)\kappa_{+}\left(\tau\right)\mid X_{i}\right)\right\} .
\end{eqnarray*}

Under Assumption \ref{assu:QRD smoothness}(ii) and (iii), as $x\downarrow0$,
\begin{equation}
F_{Y\mid X}\left(Q_{Y\mid X}\left(\tau\mid x\right)\mid x\right)-F_{Y\mid X}\left(r^{\top}_{p}\left(x\right)\kappa_{+}\left(\tau\right)\mid x\right)=\frac{x^{p+1}}{\left(p+1\right)!}\left(\varphi_{+}\left(\tau\right)\kappa^{\left(p+1\right)}_{Y,+}\left(\tau\right)+o\left(1\right)\right),\label{eq:F(Q) taylor expansion}
\end{equation}
uniformly in $\tau\in\left[\underline{\tau},\overline{\tau}\right]$.

For any $j=1,\ldots,p+1$, write the summand of $T^{\left[j\right]}_{4}\left(\tau\right)$
as $g_{4}\left(Y_{i},X_{i}\mid\tau\right)-g_{5}\left(X_{i}\mid\tau\right)$,
where 
\begin{eqnarray*}
g_{4}\left(y,x\mid\tau\right) & \coloneqq & \frac{1}{\sqrt{h}}K^{\left[j\right]}_{+,p}\left(\frac{x}{h}\right)\left\{ \mathbbm{1}\left(y\leq Q_{Y\mid X}\left(\tau\mid x\right)\right)-\mathbbm{1}\left(y\leq r^{\top}_{p}\left(x\right)\kappa_{+}\left(\tau\right)\right)\right\} ,\\
g_{5}\left(x\mid\tau\right) & \coloneqq & \frac{1}{\sqrt{h}}K^{\left[j\right]}_{+,p}\left(\frac{x}{h}\right)\left\{ \tau-F_{Y\mid X}\left(r^{\top}_{p}\left(x\right)\kappa_{+}\left(\tau\right)\mid x\right)\right\} .
\end{eqnarray*}
By the LIE, $\mathrm{E}\left[g_{4}\left(Y,X\mid\tau\right)\mid X\right]=g_{5}\left(X\mid\tau\right)$,
so $\mathbb{P}g_{4}\left(\cdot\mid\tau\right)=\mathbb{P}g_{5}\left(\cdot\mid\tau\right)$
and $T_{4}\left(\tau\right)=\mathbb{G}_{n}\left(g_{4}-g_{5}\right)\left(\cdot\mid\tau\right)$.
Define $\mathscr{G}_{4}\coloneqq\left\{ \left(g_{4}-g_{5}\right)\left(\cdot\mid\tau\right):\tau\in\left[\underline{\tau},\overline{\tau}\right]\right\} $.
We have shown above that both (\ref{eq:indicator class 1}) and (\ref{eq:indicator class 2})
are VC-subgraph. Since the subgraph inequality $t\leq F_{Y\mid X}\left(r^{\top}_{p}\left(x\right)b\mid x\right)$
for $t\in\left(0,1\right]$ is equivalent to $r^{\top}_{p}\left(x\right)b\geq Q_{Y\mid X}\left(t\mid x\right)$
(\citealp[Lemma 21.1(i)]{VanDerVaart1998}), the subgraph of $x\mapsto F_{Y\mid X}\left(r^{\top}_{p}\left(x\right)b\mid x\right)$
equals 
\[
\left\{ \left(x,t\right)\in\mathbb{R}\times\left(0,1\right]:r^{\top}_{p}\left(x\right)b-Q_{Y\mid X}\left(t\mid x\right)\geq0\right\} \cup\left(\mathbb{R}\times\left(-\infty,0\right]\right).
\]
The former set can be represented as a non-negativity set of the $\left(p+2\right)$-dimensional
vector space spanned by $r_{p}\left(\cdot\right)$ and $Q_{Y\mid X}\left(\cdot\mid\cdot\right)$.
Therefore, by \citet[Lemma 9.6]{Kosorok2007} and preservation results
in \citet[Lemma 9.9]{Kosorok2007}, $\left\{ \left(y,x\right)\mapsto F_{Y\mid X}\left(r^{\top}_{p}\left(x\right)b\mid x\right):b\in\mathbb{R}^{p+1}\right\} $
is VC-subgraph with index at most $p+4$. It now follows from \citet[Theorem 9.3 and Lemma 9.9(vi)]{Kosorok2007}
and \citet[Lemma A.6]{chernozhukov2014gaussian} that $\mathscr{G}_{4}$
is VC-type with characteristics independent of $n$ and a constant
envelope $G_{4}\apprle h^{-1/2}$. Since $\mathrm{E}\left[\left(g_{4}\left(Y,X\mid\tau\right)-g_{5}\left(X\mid\tau\right)\right)^{2}\mid X\right]\leq\mathrm{E}\left[g^{2}_{4}\left(Y,X\mid\tau\right)\mid X\right]$,
by the LIE, (\ref{eq:F(Q) taylor expansion}) and (\ref{eq:a.b indicator inequalities}),
we have 
\begin{eqnarray}
\mathbb{P}\left(g_{4}-g_{5}\right)^{2}\left(\cdot\mid\tau\right) & \leq & \mathbb{P}g^{2}_{4}\left(\cdot\mid\tau\right)\nonumber \\
 & = & \mathrm{E}\left[\frac{1}{h}\left(K^{\left[j\right]}_{+,p}\right)^{2}\left(\frac{X}{h}\right)\left|F_{Y\mid X}\left(Q_{Y\mid X}\left(\tau\mid X\right)\mid X\right)-F_{Y\mid X}\left(r^{\top}_{p}\left(X\right)\kappa_{+}\left(\tau\right)\mid X\right)\right|\right]\nonumber \\
 & = & O\left(h^{p+1}\right),\label{eq:(g_4-g_5)^2 expectation bound}
\end{eqnarray}
uniformly in $\tau\in\left[\underline{\tau},\overline{\tau}\right]$.
By \citet[Corollary 5.5]{Chen2020jackknife}, 
\[
\mathrm{E}\left[\sup_{\tau\in\left[\underline{\tau},\overline{\tau}\right]}\left|\mathbb{G}_{n}\left(g_{4}-g_{5}\right)\left(\cdot\mid\tau\right)\right|\right]=O\left(\sqrt{h^{p+1}\log\left(n\right)}+\frac{\log\left(n\right)}{\sqrt{nh}}\right).
\]
Markov's inequality yields that $T_{4}\left(\tau\right)=o_{p}\left(1\right)$,
uniformly in $\tau\in\left[\underline{\tau},\overline{\tau}\right]$.

We can write $T^{\left[j\right]}_{5}\left(\tau\right)-\mathrm{E}\left[T^{\left[j\right]}_{5}\left(\tau\right)\right]=\mathbb{G}_{n}g_{5}\left(\cdot\mid\tau\right)$.
By similar arguments, the function class $\mathscr{G}_{5}\coloneqq\left\{ g_{5}\left(\cdot\mid\tau\right):\tau\in\left[\underline{\tau},\overline{\tau}\right]\right\} $
is VC-type with characteristics independent of $n$ and admits an
envelope given by 
\[
G_{5}\left(x\right)\coloneqq\frac{1}{\sqrt{h}}\left|K^{\left[j\right]}_{+,p}\left(\frac{x}{h}\right)\right|\sup_{\tau\in\left[\underline{\tau},\overline{\tau}\right]}\left|\tau-F_{Y\mid X}\left(r^{\top}_{p}\left(x\right)\kappa_{+}\left(\tau\right)\mid x\right)\right|.
\]
By a change of variables and (\ref{eq:F(Q) taylor expansion}), $\mathbb{P}G^{2}_{5}=O\left(h^{2\left(p+1\right)}\right)$.
It follows from the same arguments used to prove (\ref{eq:g_3 bound})
that $\mathrm{E}\left[\left\Vert \mathbb{G}_{n}\right\Vert _{\mathscr{G}_{5}}\right]=O\left(h^{p+1}\right)$.
Markov's inequality yields that $\mathbb{G}_{n}g_{5}\left(\cdot\mid\tau\right)=o_{p}\left(1\right)$,
uniformly in $\tau\in\left[\underline{\tau},\overline{\tau}\right]$.
Therefore, $T_{5}\left(\tau\right)-\mathrm{E}\left[T_{5}\left(\tau\right)\right]=o_{p}\left(1\right)$,
uniformly in $\tau\in\left[\underline{\tau},\overline{\tau}\right]$.

By (\ref{eq:F(Q) taylor expansion}), we have 
\begin{eqnarray*}
\frac{1}{\sqrt{nh}}\mathrm{E}\left[T_{5}\left(\tau\right)\right] & = & \mathrm{E}\left[\frac{1}{h}K_{+,p}\left(\frac{X}{h}\right)\left\{ F_{Y\mid X}\left(Q_{Y\mid X}\left(\tau\mid X\right)\mid X\right)-F_{Y\mid X}\left(r^{\top}_{p}\left(X\right)\kappa_{+}\left(\tau\right)\mid X\right)\right\} \right]\\
 & = & h^{p+1}\left\{ \left(\frac{\int^{1}_{0}K_{+,p}\left(u\right)u^{p+1}\mathrm{d}u}{\left(p+1\right)!}\right)\varphi_{+}\left(\tau\right)\kappa^{\left(p+1\right)}_{Y,+}\left(\tau\right)\varphi+o\left(1\right)\right\} ,
\end{eqnarray*}
uniformly in $\tau\in\left[\underline{\tau},\overline{\tau}\right]$.
Combining these results, we have 
\begin{eqnarray}
\frac{1}{\sqrt{nh}}\sum_{i}K_{+,p}\left(\frac{X_{i}}{h}\right)\psi_{\tau}\left(\tilde{Y}_{+,i}\right) & = & \sqrt{nh}h^{p+1}\left\{ \left(\frac{\int^{1}_{0}K_{+,p}\left(u\right)u^{p+1}\mathrm{d}u}{\left(p+1\right)!}\right)\varphi_{+}\left(\tau\right)\kappa^{\left(p+1\right)}_{Y,+}\left(\tau\right)\varphi+o\left(1\right)\right\} \nonumber \\
 &  & +\frac{1}{\sqrt{nh}}\sum_{i}K_{+,p}\left(\frac{X_{i}}{h}\right)U_{i}\left(\tau\right)+o_{p}\left(1\right),\label{eq:T_1 first term linearization}
\end{eqnarray}
uniformly in $\tau\in\left[\underline{\tau},\overline{\tau}\right]$.

Next, we derive the linearization for the second term on the right-hand
side of (\ref{eq:R_1 decompose}). Note that for all $\tau\in\left[\underline{\tau},\overline{\tau}\right]$,
\begin{eqnarray}
 &  & \left\Vert \frac{1}{nh}\sum_{i}K_{+,p}\left(\frac{X_{i}}{h}\right)\psi_{\tau}\left(\tilde{Y}_{+,i}\right)V^{\top}_{p,i}-\left(\frac{1}{nh}\sum_{i}\tilde{K}_{+,p}\left(\frac{X_{i}}{h}\right)\psi_{\tau}\left(\tilde{Y}_{+,i}\right)\bar{Z}^{\top}_{i}\right)/\varphi\right\Vert \nonumber \\
 & \leq & \left(\frac{2}{nh}\sum_{i}K^{2}\left(\frac{X_{i}}{h}\right)\left\Vert r_{p}\left(\frac{X_{i}}{h}\right)\right\Vert ^{2}\left\Vert \bar{Z}_{i}\right\Vert \right)\left(\left\Vert \Pi^{-1}_{+,p}-\varphi^{-1}\Lambda^{-1}_{+,p}\right\Vert \right)\nonumber \\
 & = & o_{p}\left(1\right).\label{eq:K to K_til approximation}
\end{eqnarray}
Write 
\[
\frac{1}{nh}\sum_{i}\tilde{K}_{+,p}\left(\frac{X_{i}}{h}\right)\psi_{\tau}\left(\tilde{Y}_{+,i}\right)\bar{Z}^{\top}_{i}=T_{6}\left(\tau\right)+T_{7}\left(\tau\right),
\]
where 
\begin{eqnarray}
T_{6}\left(\tau\right) & \coloneqq & \frac{1}{nh}\sum_{i}\tilde{K}_{+,p}\left(\frac{X_{i}}{h}\right)U_{i}\left(\tau\right)\bar{Z}^{\top}_{i}\nonumber \\
T_{7}\left(\tau\right) & \coloneqq & \frac{1}{nh}\sum_{i}\tilde{K}_{+,p}\left(\frac{X_{i}}{h}\right)\left\{ \mathbbm{1}\left(Y_{i}\leq Q_{Y\mid X}\left(\tau\mid X_{i}\right)\right)-\mathbbm{1}\left(Y_{i}\leq r^{\top}_{p}\left(X_{i}\right)\kappa_{+}\left(\tau\right)\right)\right\} \bar{Z}^{\top}_{i}.\label{eq:T_6 and T_7 definitions}
\end{eqnarray}

By Assumption \ref{assu:data generating process}(iv), Assumption
\ref{assu:QRD smoothness}(i), (v), and (vi), and the DCT, $\left(\tau,x\right)\mapsto\mu_{U\left(\tau\right)\bar{Z}^{\top}}\left(x\right)$
is jointly uniformly continuous on $\left[\underline{\tau},\overline{\tau}\right]\times\mathbb{B}_{0}$.
By this result and the LIE, we have 
\[
\mathrm{E}\left[T_{6}\left(\tau\right)\right]=\pi_{+}\varphi\mu_{U\left(\tau\right)\bar{Z}^{\top},+}+o\left(1\right),
\]
uniformly in $\tau\in\left[\underline{\tau},\overline{\tau}\right]$.
By Assumption \ref{assu:QRD smoothness}(i) and (v), for all $z\in\mathrm{supp}\left(Z\right)$,
as $x\downarrow0$, 
\[
\sup_{\tau\in\left[\underline{\tau},\overline{\tau}\right]}\left|F_{Y\mid XZ}\left(Q_{Y\mid X}\left(\tau\mid x\right)\mid x,z\right)-F_{Y\mid XZ}\left(r^{\top}_{p}\left(x\right)\kappa_{+}\left(\tau\right)\mid x,z\right)\right|=o\left(1\right).
\]
By this result, Assumption \ref{assu:QRD smoothness}(vi) and the
DCT, as $x\downarrow0$, 
\[
\sup_{\tau\in\left[\underline{\tau},\overline{\tau}\right]}\left|\mathrm{E}\left[\left\{ F_{Y\mid XZ}\left(Q_{Y\mid X}\left(\tau\mid X\right)\mid X,Z\right)-F_{Y\mid XZ}\left(r^{\top}_{p}\left(X\right)\kappa_{+}\left(\tau\right)\mid X,Z\right)\right\} \bar{Z}\mid X=x\right]\right|=o\left(1\right).
\]
By this result and the LIE, we have 
\begin{eqnarray*}
\mathrm{E}\left[T_{7}\left(\tau\right)\right] & = & \mathrm{E}\left[\frac{1}{h}\tilde{K}_{+,p}\left(\frac{X}{h}\right)\left\{ F_{Y\mid XZ}\left(Q_{Y\mid X}\left(\tau\mid X\right)\mid X,Z\right)-F_{Y\mid XZ}\left(r^{\top}_{p}\left(X\right)\kappa_{+}\left(\tau\right)\mid X,Z\right)\right\} \bar{Z}\right]\\
 & = & o\left(1\right),
\end{eqnarray*}
uniformly in $\tau\in\left[\underline{\tau},\overline{\tau}\right]$.

Working coordinate-wise, fix $j\in\left\{ 1,\ldots,p+1\right\} $
and $k\in\left\{ 1,\ldots,d_{z}+1\right\} $. We write 
\[
T^{\left[jk\right]}_{6}\left(\tau\right)-\mathrm{E}\left[T^{\left[jk\right]}_{6}\left(\tau\right)\right]=\frac{1}{\sqrt{nh}}\mathbb{G}_{n}g_{6}\left(\cdot\mid\tau\right),
\]
where $\bar{z}\coloneqq\left(1,z^{\top}\right)^{\top}$ and 
\[
g_{6}\left(y,x,z\mid\tau\right)\coloneqq\frac{1}{\sqrt{h}}\tilde{K}^{\left[j\right]}_{+,p}\left(\frac{x}{h}\right)\left\{ \tau-\mathbbm{1}\left(y\leq Q_{Y\mid X}\left(\tau\mid x\right)\right)\right\} \bar{z}^{\left[k\right]}.
\]
By similar arguments, we can show that $\mathscr{G}_{6}\coloneqq\left\{ g_{6}\left(\cdot\mid\tau\right):\tau\in\left[\underline{\tau},\overline{\tau}\right]\right\} $
is VC-type with characteristics independent of $n$ and admits an
envelope $G_{6}\left(y,x,z\right)\coloneqq h^{-1/2}\left\Vert \tilde{K}_{+,p}\left(x/h\right)\right\Vert \left|\bar{z}^{\left[k\right]}\right|$.
By the LIE, Assumption \ref{assu:data generating process}(iv) (which
gives $\mathrm{E}\left[\left(\bar{Z}^{\left[k\right]}\right)^{2}\mid X=x\right]$
bounded uniformly in $x$ near $0$), and a change of variables, $\mathbb{P}G^{2}_{6}=O\left(1\right)$.
By the maximal inequality in \citet[Theorem 2.14.1]{VanDerVaartWellner1996},
$\mathrm{E}\left[\left\Vert \mathbb{G}_{n}\right\Vert _{\mathscr{G}_{6}}\right]=O\left(1\right)$.
It follows that $T_{6}\left(\tau\right)-\mathrm{E}\left[T_{6}\left(\tau\right)\right]=o_{p}\left(1\right)$,
uniformly in $\tau\in\left[\underline{\tau},\overline{\tau}\right]$.
It follows from the same arguments that $T_{7}\left(\tau\right)-\mathrm{E}\left[T_{7}\left(\tau\right)\right]=o_{p}\left(1\right)$,
uniformly in $\tau\in\left[\underline{\tau},\overline{\tau}\right]$.

Therefore, 
\[
\frac{1}{nh}\sum_{i}K_{+,p}\left(\frac{X_{i}}{h}\right)\psi_{\tau}\left(\tilde{Y}_{+,i}\right)V^{\top}_{p,i}=\pi_{+}\mu_{U\left(\tau\right)\bar{Z}^{\top},+}+o_{p}\left(1\right),
\]
uniformly in $\tau\in\left[\underline{\tau},\overline{\tau}\right]$.
By this result, Lemma \ref{lem:reg weights lambda} and $\left(nh\right)^{-1}\sum_{i}V_{p,i}V^{\top}_{p,i}=\left(\omega^{0,2}_{p,0}/\varphi\right)\mu_{\bar{Z}\bar{Z}^{\top},\pm}+o_{p}\left(1\right)$,
\begin{eqnarray*}
 &  & \frac{1}{1+\varrho}\cdot\frac{1}{\sqrt{nh}}\sum_{i}K_{+,p}\left(\frac{X_{i}}{h}\right)\psi_{\tau}\left(\tilde{Y}_{+,i}\right)V^{\top}_{p,i}\widehat{\lambda}_{p}\\
 & = & \pi_{+}\mu_{U\left(\tau\right)\bar{Z}^{\top},+}\left(\left(\frac{\omega^{0,2}_{p,0}}{\varphi}\right)\mu_{\bar{Z}\bar{Z}^{\top},\pm}\right)^{-1}\frac{1}{\sqrt{nh}}\sum_{i}V_{p,i}+o_{p}\left(1\right)\\
 & = & \left(\frac{\pi_{+}\varphi}{\omega^{0,2}_{p,0}}\right)\gamma^{\top}_{+}\left(\tau\right)\left\{ \frac{1}{\sqrt{nh}}\sum_{i}W_{p,i}\xi_{i}+\frac{\left(\mu^{\left(p+1\right)}_{Z,+}\omega^{p+1,1}_{+,p,0}-\mu^{\left(p+1\right)}_{Z,-}\omega^{p+1,1}_{-,p,0}\right)}{\left(p+1\right)!}\sqrt{nh}h^{p+1}\right\} +o_{p}\left(1\right),
\end{eqnarray*}
uniformly in $\tau\in\left[\underline{\tau},\overline{\tau}\right]$.
The conclusion follows from this result, (\ref{eq:R_1 decompose}),
(\ref{eq:T_1 first term linearization}) and $R_{+}\left(0,\tau\right)=T_{1}\left(\tau\right)+o_{p}\left(1\right)$. 
\end{proof}

Let $\varDelta_{+}\left(\tau\right)\coloneqq\sqrt{nh}\mathrm{H}\left(\widehat{\kappa}_{+,p}\left(\tau\right)-\kappa_{+}\left(\tau\right)\right)$.
Recall that $\mathrm{H}$ is the diagonal matrix with $\left(1,h,\ldots,h^{p}\right)$
on its diagonal. Let $\varDelta_{-}\left(\tau\right)$ be defined
analogously. Then we have 
\begin{equation}
\sqrt{nh}\left(\widehat{\vartheta}^{\mathit{eb}}_{p}\left(\tau\right)-\vartheta\left(\tau\right)\right)=\mathrm{e}^{\top}_{p+1,1}\left(\varDelta_{+}\left(\tau\right)-\varDelta_{-}\left(\tau\right)\right).\label{eq:theta_hat to Delta}
\end{equation}
The following lemma provides an asymptotic linearization for $\varDelta_{+}\left(\tau\right)$.
A similar result holds for $\varDelta_{-}\left(\tau\right)$. 
\begin{lem}
\label{lem:delta linearization}Suppose that the assumptions in the
statement of Theorem \ref{thm:sharp QRD} hold. We have the following
linearization: 
\[
\varDelta_{+}\left(\tau\right)=D^{-1}_{+}(\tau)R_{+}(0,\tau)+o_{p}\left(1\right),
\]
uniformly in $\tau\in\left[\underline{\tau},\overline{\tau}\right]$. 
\end{lem}
\begin{proof}[Proof of Lemma \ref{lem:delta linearization}]
By \citet[Lemma A.2]{Cai2008}, we have 
\begin{equation}
\left\Vert \sum_{i}\widehat{w}_{p,i}K_{+,p}\left(\frac{X_{i}}{h}\right)\psi_{\tau}\left(Y_{i}-r^{\top}_{p}\left(X_{i}\right)\widehat{\kappa}_{+,p}\left(\tau\right)\right)\right\Vert \leq\left(p+1\right)\underset{i}{\max}\left\Vert \widehat{w}_{p,i}K_{+,p}\left(\frac{X_{i}}{h}\right)\right\Vert ,\label{eq:quantile FOC}
\end{equation}
where the left-hand side equals $\sqrt{h/n}\left\Vert R_{+}\left(\varDelta_{+}\left(\tau\right),\tau\right)\right\Vert $
by the definitions of $R_{+}\left(\varDelta,\tau\right)$ and $\varDelta_{+}\left(\tau\right)$.
We can show that the right-hand side is $O_{p}\left(n^{-1}\right)$.
Indeed, by (\ref{eq:w_hat Taylor expansion}), $\max_{i}\left|V^{\top}_{p,i}\widehat{\lambda}_{p}\right|=o_{p}\left(1\right)$,
(\ref{eq:max iota'' bound}) and $\iota_{\varrho}\left(1\right)=1$,
we have $\max_{i}n\widehat{w}_{p,i}=1+o_{p}\left(1\right)$. Hence,
\[
\left(p+1\right)\max_{i}\left\Vert \widehat{w}_{p,i}K_{+,p}\left(\frac{X_{i}}{h}\right)\right\Vert \apprle\max_{i}\widehat{w}_{p,i}=O_{p}\left(n^{-1}\right).
\]
These results imply that $R_{+}(\varDelta_{+}\left(\tau\right),\tau)=O_{p}(\left(nh\right)^{-1/2})$,
uniformly in $\tau\in\left[\underline{\tau},\overline{\tau}\right]$.

Since 
\[
-\varDelta^{\top}R_{+}(\lambda\varDelta,\tau)=\sqrt{\frac{n}{h}}\sum_{i}\widehat{w}_{p,i}K\left(\frac{X_{i}}{h}\right)\mathbbm{1}\left(X_{i}>0\right)\left(-\varDelta^{\top}r_{p}\left(\frac{X_{i}}{h}\right)\right)\psi_{\tau}\left(\tilde{Y}_{+,i}-\frac{\lambda\left(r^{\top}_{p}\left(X_{i}/h\right)\varDelta\right)}{\sqrt{nh}}\right)
\]
and $\psi_{\tau}\left(\cdot\right)$ is increasing, we have 
\begin{equation}
-\varDelta^{\top}R_{+}(\lambda\varDelta,\tau)\geq-\varDelta^{\top}R_{+}(\varDelta,\tau),\label{eq:convexity}
\end{equation}
for all $\lambda\geq1$.

For all $M>0$, $\eta>0$, by the union bound, we have 
\begin{eqnarray}
 &  & \Pr\left[\inf_{\|\varDelta\|=M,\,\tau\in\left[\underline{\tau},\overline{\tau}\right]}-\varDelta^{\top}R_{+}(\varDelta,\tau)<M\eta\right]\nonumber \\
 & \leq & \Pr\left[\inf_{\|\varDelta\|=M,\,\tau\in\left[\underline{\tau},\overline{\tau}\right]}-\varDelta^{\top}R_{+}(\varDelta,\tau)<M\eta,\,\inf_{\|\varDelta\|=M,\,\tau\in\left[\underline{\tau},\overline{\tau}\right]}-\varDelta^{\top}\bigl(-D_{+}(\tau)\varDelta+R_{+}(0,\tau)\bigr)\geq2M\eta\right]\nonumber \\
 &  & +\Pr\left[\inf_{\|\varDelta\|=M,\,\tau\in\left[\underline{\tau},\overline{\tau}\right]}-\varDelta^{\top}\bigl(-D_{+}(\tau)\varDelta+R_{+}(0,\tau)\bigr)<2M\eta\right].\label{eq:step1-decomp}
\end{eqnarray}
By the Cauchy--Schwarz inequality, on the first event on the right-hand
side of (\ref{eq:step1-decomp}), there exist $\varDelta_{0}$ with
$\|\varDelta_{0}\|=M$ and $\tau_{0}\in\left[\underline{\tau},\overline{\tau}\right]$
such that 
\begin{equation}
-M\|R_{+}(\varDelta_{0},\tau_{0})+D_{+}(\tau_{0})\varDelta_{0}-R_{+}(0,\tau_{0})\|-\varDelta^{\top}_{0}\bigl(-D_{+}(\tau_{0})\varDelta_{0}+R_{+}(0,\tau_{0})\bigr)<M\eta,\label{eq:step1-cs}
\end{equation}
and hence 
\begin{equation}
-M\|R_{+}(\varDelta_{0},\tau_{0})+D_{+}(\tau_{0})\varDelta_{0}-R_{+}(0,\tau_{0})\|+\inf_{\|\varDelta\|=M,\,\tau\in\left[\underline{\tau},\overline{\tau}\right]}\bigl[-\varDelta^{\top}\bigl(-D_{+}(\tau)\varDelta+R_{+}(0,\tau)\bigr)\bigr]<M\eta.\label{eq:step1-inf}
\end{equation}
Since the infimum on the left-hand side is no less than $2\eta M$
on this event, there exist $\varDelta_{0}$ with $\|\varDelta_{0}\|=M$
and $\tau_{0}\in\left[\underline{\tau},\overline{\tau}\right]$ such
that 
\begin{equation}
\eta\leq\|R_{+}(\varDelta_{0},\tau_{0})+D_{+}(\tau_{0})\varDelta_{0}-R_{+}(0,\tau_{0})\|.\label{eq:step1-final}
\end{equation}
Therefore, 
\begin{eqnarray}
 &  & \Pr\left[\inf_{\|\varDelta\|=M,\,\tau\in\left[\underline{\tau},\overline{\tau}\right]}-\varDelta^{\top}R_{+}(\varDelta,\tau)<M\eta,\,\inf_{\|\varDelta\|=M,\,\tau\in\left[\underline{\tau},\overline{\tau}\right]}-\varDelta^{\top}\bigl(-D_{+}(\tau)\varDelta+R_{+}(0,\tau)\bigr)\geq2M\eta\right]\nonumber \\
 & \leq & \Pr\left[\sup_{\|\varDelta\|=M,\,\tau\in\left[\underline{\tau},\overline{\tau}\right]}\|R_{+}(\varDelta,\tau)+D_{+}(\tau)\varDelta-R_{+}(0,\tau)\|\geq\eta\right].\label{eq:step1-bound1}
\end{eqnarray}
For all $\|\varDelta\|=M$, $\tau\in\left[\underline{\tau},\overline{\tau}\right]$,
we have 
\begin{eqnarray}
\varDelta^{\top}D_{+}(\tau)\varDelta & \geq & M^{2}\left\{ \inf_{\tau\in\left[\underline{\tau},\overline{\tau}\right]}\mathrm{mineig}\left(D_{+}(\tau)\right)\right\} ,\label{eq:quad-lb}\\
-\varDelta^{\top}R_{+}(0,\tau) & \geq & -M\left\{ \sup_{\tau\in\left[\underline{\tau},\overline{\tau}\right]}\|R_{+}(0,\tau)\|\right\} ,\label{eq:lin-lb}
\end{eqnarray}
where 
\[
\inf_{\tau\in\left[\underline{\tau},\overline{\tau}\right]}\mathrm{mineig}\left(D_{+}(\tau)\right)=\varphi\left\{ \inf_{\tau\in\left[\underline{\tau},\overline{\tau}\right]}\varphi_{+}\left(\tau\right)\right\} \mathrm{mineig}\left(\Lambda_{+,p}\right)>0.
\]
Hence 
\begin{eqnarray}
 &  & \Pr\left[\inf_{\|\varDelta\|=M,\,\tau\in\left[\underline{\tau},\overline{\tau}\right]}-\varDelta^{\top}\bigl(-D_{+}(\tau)\varDelta+R_{+}(0,\tau)\bigr)<2M\eta\right]\nonumber \\
 & \leq & \Pr\left[M^{2}\left\{ \inf_{\tau\in\left[\underline{\tau},\overline{\tau}\right]}\mathrm{mineig}\left(D_{+}(\tau)\right)\right\} -M\left\{ \sup_{\tau\in\left[\underline{\tau},\overline{\tau}\right]}\|R_{+}(0,\tau)\|\right\} <2M\eta\right].\label{eq:step1-bound2}
\end{eqnarray}
We have shown that for all $M>0$, $\eta>0$, 
\begin{eqnarray}
 &  & \Pr\left[\inf_{\|\varDelta\|=M,\,\tau\in\left[\underline{\tau},\overline{\tau}\right]}-\varDelta^{\top}R_{+}(\varDelta,\tau)<M\eta\right]\nonumber \\
 & \leq & \Pr\left[\sup_{\|\varDelta\|=M,\,\tau\in\left[\underline{\tau},\overline{\tau}\right]}\|R_{+}(\varDelta,\tau)+D_{+}(\tau)\varDelta-R_{+}(0,\tau)\|\geq\eta\right]\nonumber \\
 &  & +\Pr\left[\sup_{\tau\in\left[\underline{\tau},\overline{\tau}\right]}\|R_{+}(0,\tau)\|>M\left\{ \inf_{\tau\in\left[\underline{\tau},\overline{\tau}\right]}\mathrm{mineig}\left(D_{+}(\tau)\right)\right\} -2\eta\right].\label{eq:step1-conclusion}
\end{eqnarray}
By Lemma \ref{lem:R_0 linearization}, we have $\sup_{\tau\in\left[\underline{\tau},\overline{\tau}\right]}\|R_{+}(0,\tau)\|=O_{p}(1)$.
It follows from this result, (\ref{eq:step1-conclusion}), and (\ref{eq:asymptotic linearization})
that for all $\eta,\varepsilon>0$, there exists an $M>0$ such that,
for all sufficiently large $n$, 
\begin{equation}
\Pr\left[\inf_{\|\varDelta\|=M,\,\tau\in\left[\underline{\tau},\overline{\tau}\right]}-\varDelta^{\top}R_{+}(\varDelta,\tau)<M\eta\right]<\varepsilon.\label{eq:step1-final-bound}
\end{equation}

By (\ref{eq:convexity}), for all $\|\varDelta\|=M$, $\tau\in\left[\underline{\tau},\overline{\tau}\right]$
and $\lambda\geq1$, 
\begin{equation}
-\left(\frac{\varDelta}{M}\right)^{\top}R_{+}(\varDelta,\tau)\leq-\left(\frac{\varDelta}{M}\right)^{\top}R_{+}(\lambda\varDelta,\tau)\leq\|R_{+}(\lambda\varDelta,\tau)\|,\label{eq:step2-chain}
\end{equation}
and therefore, for all $M>0$ and $\eta>0$, 
\begin{equation}
\Pr\left[\inf_{\|\varDelta\|\geq M,\,\tau\in\left[\underline{\tau},\overline{\tau}\right]}\|R_{+}(\varDelta,\tau)\|<\eta\right]\leq\Pr\left[\inf_{\|\varDelta\|=M,\,\tau\in\left[\underline{\tau},\overline{\tau}\right]}-\varDelta^{\top}R_{+}(\varDelta,\tau)<M\eta\right].\label{eq:R_+ inf bound}
\end{equation}
It follows from the above inequality and (\ref{eq:step1-final-bound})
that for all $\varepsilon,\eta>0$, there exists an $M>0$ such that,
for all sufficiently large $n$, 
\begin{equation}
\Pr\left[\inf_{\|\varDelta\|\geq M,\,\tau\in\left[\underline{\tau},\overline{\tau}\right]}\|R_{+}(\varDelta,\tau)\|<\eta\right]<\varepsilon.\label{eq:step2-final}
\end{equation}

Clearly, we have 
\begin{eqnarray}
 &  & \Pr\left[\sup_{\tau\in\left[\underline{\tau},\overline{\tau}\right]}\|\varDelta_{+}\left(\tau\right)\|>M\right]\nonumber \\
 & \leq & \Pr\left[\sup_{\tau\in\left[\underline{\tau},\overline{\tau}\right]}\|\varDelta_{+}\left(\tau\right)\|>M,\,\sup_{\tau\in\left[\underline{\tau},\overline{\tau}\right]}\|R_{+}(\varDelta_{+}\left(\tau\right),\tau)\|<\eta\right]+\Pr\left[\sup_{\tau\in\left[\underline{\tau},\overline{\tau}\right]}\|R_{+}(\varDelta_{+}\left(\tau\right),\tau)\|\geq\eta\right]\nonumber \\
 & \leq & \Pr\left[\inf_{\|\varDelta\|\geq M,\,\tau\in\left[\underline{\tau},\overline{\tau}\right]}\|R_{+}(\varDelta,\tau)\|<\eta\right]+\Pr\left[\sup_{\tau\in\left[\underline{\tau},\overline{\tau}\right]}\|R_{+}(\varDelta_{+}\left(\tau\right),\tau)\|\geq\eta\right].\label{eq:step3}
\end{eqnarray}
It follows from the above inequality, (\ref{eq:step2-final}) and
the uniform-in-$\tau$ bound $\|R_{+}(\varDelta_{+}\left(\tau\right),\tau)\|=o_{p}(1)$
that $\varDelta_{+}\left(\tau\right)=O_{p}(1)$, uniformly in $\tau\in\left[\underline{\tau},\overline{\tau}\right]$.

For all $\varepsilon>0$ and $M>0$, by the union bound, 
\begin{eqnarray}
 &  & \Pr\left[\sup_{\tau\in\left[\underline{\tau},\overline{\tau}\right]}\left\Vert R_{+}(\varDelta_{+}\left(\tau\right),\tau)+D_{+}(\tau)\varDelta_{+}\left(\tau\right)-R_{+}(0,\tau)\right\Vert >\varepsilon\right]\nonumber \\
 & \leq & \Pr\left[\sup_{\tau\in\left[\underline{\tau},\overline{\tau}\right]}\left\Vert \varDelta_{+}\left(\tau\right)\right\Vert >M\right]+\Pr\left[\sup_{\|\varDelta\|\leq M,\,\tau\in\left[\underline{\tau},\overline{\tau}\right]}\left\Vert R_{+}(\varDelta,\tau)+D_{+}(\tau)\varDelta-R_{+}(0,\tau)\right\Vert >\varepsilon\right].\nonumber \\
 &  & \label{eq:step3-final}
\end{eqnarray}
Now it follows from this inequality, $\sup_{\tau\in\left[\underline{\tau},\overline{\tau}\right]}\left\Vert \varDelta_{+}\left(\tau\right)\right\Vert =O_{p}(1)$
and (\ref{eq:asymptotic linearization}) that 
\[
R_{+}(\varDelta_{+}\left(\tau\right),\tau)+D_{+}(\tau)\varDelta_{+}\left(\tau\right)-R_{+}(0,\tau)=o_{p}\left(1\right),
\]
uniformly in $\tau\in\left[\underline{\tau},\overline{\tau}\right]$.
The conclusion follows from this result and the uniform-in-$\tau$
bound $\left\Vert R_{+}(\varDelta_{+}\left(\tau\right),\tau)\right\Vert =o_{p}\left(1\right)$. 
\end{proof}

\begin{proof}[Proof of Theorem \ref{thm:sharp QRD}]
We have $\mathrm{e}^{\top}_{p+1,1}D^{-1}_{+}\left(\tau\right)K_{+,p}\left(X_{i}/h\right)=\bar{W}_{+,p,i}/\varphi_{+}\left(\tau\right)$
and $\mathrm{e}^{\top}_{p+1,1}\Lambda^{-1}_{+,p}\pi_{+}=\omega^{0,2}_{p,0}$.
The second equality implies that $\mathrm{e}^{\top}_{p+1,1}D^{-1}_{+}\left(\tau\right)\pi_{+}\varphi/\omega^{0,2}_{p,0}=1/\varphi_{+}\left(\tau\right)$.
Moreover, we have $\mathrm{e}^{\top}_{p+1,1}\Lambda^{-1}_{+,p}\int^{1}_{0}K_{+,p}\left(u\right)u^{p+1}\mathrm{d}u=\omega^{p+1,1}_{+,p,0}$,
so that $\mathrm{e}^{\top}_{p+1,1}D^{-1}_{+}\left(\tau\right)\bar{\mathscr{B}}_{+}\left(\tau\right)$
equals the quantity $\mathscr{B}_{+,p}\left(\tau\right)$ defined
by \eqref{eq:Delta_+ linearization} below. Then, by Lemmas \ref{lem:R_delta linearization}--\ref{lem:delta linearization},
we have 
\begin{eqnarray}
\mathrm{e}^{\top}_{p+1,1}\varDelta_{+}\left(\tau\right) & = & \frac{1}{\sqrt{nh}}\sum_{i}\bar{W}_{+,p,i}\left(\frac{U_{i}\left(\tau\right)}{\varphi_{+}\left(\tau\right)}\right)-\left(\frac{\gamma_{+}\left(\tau\right)}{\varphi_{+}\left(\tau\right)}\right)^{\top}\left\{ \frac{1}{\sqrt{nh}}\sum_{i}\bar{W}_{p,i}\xi_{i}\right\} \nonumber \\
 &  & +\sqrt{nh}h^{p+1}\mathscr{B}_{+,p}\left(\tau\right)+o_{p}\left(1\right),\mbox{ where}\nonumber \\
\mathscr{B}_{+,p}\left(\tau\right) & \coloneqq & \frac{\omega^{p+1,1}_{+,p,0}\kappa^{\left(p+1\right)}_{Y,+}\left(\tau\right)}{\left(p+1\right)!}-\left(\frac{\gamma_{+}\left(\tau\right)}{\varphi_{+}\left(\tau\right)}\right)^{\top}\left(\frac{\mu^{\left(p+1\right)}_{Z,+}\omega^{p+1,1}_{+,p,0}-\mu^{\left(p+1\right)}_{Z,-}\omega^{p+1,1}_{-,p,0}}{\left(p+1\right)!}\right),\label{eq:Delta_+ linearization}
\end{eqnarray}
uniformly in $\tau\in\left[\underline{\tau},\overline{\tau}\right]$.
By similar arguments, we obtain a linearization for $\mathrm{e}^{\top}_{p+1,1}\varDelta_{-}\left(\tau\right)$.
Combining these results and using (\ref{eq:theta_hat to Delta}),
we get 
\begin{equation}
\sqrt{nh}\left(\widehat{\vartheta}^{\mathit{eb}}_{p}\left(\tau\right)-\vartheta\left(\tau\right)-h^{p+1}\mathscr{B}_{\star,p}\left(\tau\right)\right)=\mathbb{Z}_{p,n}\left(\tau\right)+o_{p}\left(1\right),\label{eq:uniform linearization QRD}
\end{equation}
uniformly in $\tau\in\left[\underline{\tau},\overline{\tau}\right]$,
where 
\begin{equation}
\phi_{i}\left(\tau\right)\coloneqq\left(\frac{\bar{W}_{+,p,i}}{\varphi_{+}\left(\tau\right)}-\frac{\bar{W}_{-,p,i}}{\varphi_{-}\left(\tau\right)}\right)U_{i}\left(\tau\right)-\bar{W}_{p,i}\xi^{\top}_{i}\gamma\left(\tau\right)\textrm{ and }\mathbb{Z}_{p,n}\left(\tau\right)\coloneqq\frac{1}{\sqrt{nh}}\sum_{i}\phi_{i}\left(\tau\right).\label{eq:phi_i definition}
\end{equation}

We now upgrade the uniform linearization above to convergence in distribution
in $\ell^{\infty}\left[\underline{\tau},\overline{\tau}\right]$ via
the changing-class Donsker theorem of \citet[Theorem 19.28]{VanDerVaart1998}.
Define 
\begin{multline*}
g_{7}\left(y,x,z\mid\tau\right)\coloneqq\\
\frac{1}{\varphi\sqrt{h}}\left\{ \left(\frac{\mathcal{K}_{+,p,0}\left(x/h\right)\mathbbm{1}\left(x>0\right)}{\varphi_{+}\left(\tau\right)}-\frac{\mathcal{K}_{-,p,0}\left(x/h\right)\mathbbm{1}\left(x<0\right)}{\varphi_{-}\left(\tau\right)}\right)\left\{ \tau-\mathbbm{1}\left(y\leq Q_{Y\mid X}\left(\tau\mid x\right)\right)\right\} \right.\\
\left.-\left(\mathcal{K}_{+,p,0}\left(\frac{x}{h}\right)\mathbbm{1}\left(x>0\right)-\mathcal{K}_{-,p,0}\left(\frac{x}{h}\right)\mathbbm{1}\left(x<0\right)\right)\left(z-\mu_{Z}\left(x\right)\right)^{\top}\gamma\left(\tau\right)\right\} ,
\end{multline*}
so that $g_{7}\left(Y_{i},X_{i},Z_{i}\mid\tau\right)=\phi_{i}\left(\tau\right)/\sqrt{h}$.
Let $\mathscr{G}_{7}\coloneqq\left\{ g_{7}\left(\cdot\mid\tau\right):\tau\in\left[\underline{\tau},\overline{\tau}\right]\right\} $.
Since $\mathrm{E}\left[U\left(\tau\right)\mid X\right]=0$ and $\mathrm{E}\left[\xi\mid X\right]=0$,
$\mathrm{E}\left[\phi\left(\tau\right)\right]=0$, and hence $\mathbb{G}_{n}g_{7}\left(\cdot\mid\tau\right)=\mathbb{Z}_{p,n}\left(\tau\right)$.
We take the index-set semimetric to be the Euclidean metric, under
which $\left[\underline{\tau},\overline{\tau}\right]$ is totally
bounded.

Define the envelope 
\begin{eqnarray}
G_{7}\left(y,x,z\right) & \coloneqq & \frac{1}{\varphi\sqrt{h}}\left(\left|\mathcal{K}_{+,p,0}\left(\frac{x}{h}\right)\right|+\left|\mathcal{K}_{-,p,0}\left(\frac{x}{h}\right)\right|\right)\nonumber \\
 &  & \times\left\{ \left(\inf_{\tau\in\left[\underline{\tau},\overline{\tau}\right]}\min\left\{ \varphi_{+}\left(\tau\right),\varphi_{-}\left(\tau\right)\right\} \right)^{-1}+\left\Vert z-\mu_{Z}\left(x\right)\right\Vert \sup_{\tau\in\left[\underline{\tau},\overline{\tau}\right]}\left\Vert \gamma\left(\tau\right)\right\Vert \right\} .\label{eq:G_7 definition}
\end{eqnarray}
Then $\sup_{\tau\in\left[\underline{\tau},\overline{\tau}\right]}\left|g_{7}\left(\cdot\mid\tau\right)\right|\leq G_{7}$.
A change of variables yields $\mathbb{P}G^{2}_{7}=O\left(1\right)$
and $\mathbb{P}G^{2+\delta}_{7}=O\left(h^{-\delta/2}\right)$. By
Markov's inequality, 
\[
\mathbb{P}G^{2}_{7}\mathbbm{1}\left(G_{7}>\eta\sqrt{n}\right)\leq\mathbb{P}G^{2+\delta}_{7}\left(\eta\sqrt{n}\right)^{-\delta}=O\left(\left(nh\right)^{-\delta/2}\right),
\]
for every $\eta>0$. This verifies the Lindeberg envelope condition.

The class (\ref{eq:indicator class 2}) is VC-subgraph with index
$2$. By \citet[Lemma 9.6]{Kosorok2007}, the function class $\left\{ \left(x,z\right)\mapsto\left(z-\mu_{Z}\left(x\right)\right)^{\top}b:b\in\mathbb{R}^{d_{z}}\right\} $
is VC-subgraph with index no more than $d_{z}+2$. By \citet[Theorem 9.3 and Lemma 9.9(vi)]{Kosorok2007}
and \citet[Lemma A.6]{chernozhukov2014gaussian}, $\mathscr{G}_{7}$
is VC-type with the envelope $G_{7}$ and characteristics independent
of $n$. It follows that $\mathscr{G}_{7}$ has a bounded uniform
entropy integral. Hence the convergence condition for the uniform
entropy integral in the statement of \citet[Theorem 19.28]{VanDerVaart1998}
is satisfied.

Next, we show that the covariance convergence condition is satisfied.
For any $\left(\tau,\tau'\right)\in\left[\underline{\tau},\overline{\tau}\right]^{2}$,
we write 
\begin{eqnarray*}
 &  & \mathbb{P}\left(g_{7}\left(\cdot\mid\tau\right)g_{7}\left(\cdot\mid\tau'\right)\right)\\
 & = & \sum_{s\in\left\{ +,-\right\} }\left\{ \mathrm{E}\left[\frac{1}{h}\left(\frac{\bar{W}^{2}_{s,p}}{\varphi_{s}\left(\tau\right)\varphi_{s}\left(\tau'\right)}\right)U\left(\tau\right)U\left(\tau'\right)\right]+\mathrm{E}\left[\frac{1}{h}\bar{W}^{2}_{s,p}\left(\xi^{\top}\gamma\left(\tau\right)\right)\left(\xi^{\top}\gamma\left(\tau'\right)\right)\right]\right.\\
 &  & \left.-\mathrm{E}\left[\frac{1}{h}\left(\frac{\bar{W}^{2}_{s,p}}{\varphi_{s}\left(\tau\right)}\right)U\left(\tau\right)\xi^{\top}\gamma\left(\tau'\right)\right]-\mathrm{E}\left[\frac{1}{h}\left(\frac{\bar{W}^{2}_{s,p}}{\varphi_{s}\left(\tau'\right)}\right)U\left(\tau'\right)\xi^{\top}\gamma\left(\tau\right)\right]\right\} .
\end{eqnarray*}
Note that $\mathrm{E}\left[U\left(\tau\right)U\left(\tau'\right)\mid X\right]=\min\left\{ \tau,\tau'\right\} -\tau\tau'$,
$\mathrm{E}\left[\xi U\left(\tau\right)\mid X\right]=\mu_{ZU\left(\tau\right)}\left(X\right)$,
and $\mathrm{E}\left[\xi\xi^{\top}\mid X\right]=\mathrm{Var}\left[Z\mid X\right]$.
By a change of variables, $\mathbb{P}\left(g_{7}\left(\cdot\mid\tau\right)g_{7}\left(\cdot\mid\tau'\right)\right)$
converges to $\mathscr{C}_{\star,p}\left(\tau,\tau'\right)$.

Write $\phi\left(\tau\right)=\phi_{u}\left(\tau\right)-\phi_{z}\left(\tau\right)$
with $\phi_{u}\left(\tau\right)\coloneqq\left(\bar{W}_{+,p}/\varphi_{+}\left(\tau\right)-\bar{W}_{-,p}/\varphi_{-}\left(\tau\right)\right)U\left(\tau\right)$
and $\phi_{z}\left(\tau\right)\coloneqq\bar{W}_{p}\xi^{\top}\gamma\left(\tau\right)$.
By adding and subtracting $\left(\bar{W}_{+,p}/\varphi_{+}\left(\tau\right)-\bar{W}_{-,p}/\varphi_{-}\left(\tau\right)\right)U\left(\tau'\right)$
in $\phi_{u}\left(\tau\right)-\phi_{u}\left(\tau'\right)$, we have
\begin{multline}
\mathrm{E}\left[\frac{1}{h}\left(\phi_{u}\left(\tau\right)-\phi_{u}\left(\tau'\right)\right)^{2}\right]\apprle\\
\sum_{s\in\left\{ +,-\right\} }\left\{ \varphi^{-2}_{s}\left(\tau\right)\mathrm{E}\left[\frac{1}{h}\bar{W}^{2}_{s,p}\left(U\left(\tau\right)-U\left(\tau'\right)\right)^{2}\right]+\left(\varphi^{-1}_{s}\left(\tau\right)-\varphi^{-1}_{s}\left(\tau'\right)\right)^{2}\mathrm{E}\left[\frac{1}{h}\bar{W}^{2}_{s,p}U^{2}\left(\tau'\right)\right]\right\} .\label{eq:phi_u inequality}
\end{multline}
By (\ref{eq:a.b indicator inequalities}), we have $\mathrm{E}\left[\left(U\left(\tau\right)-U\left(\tau'\right)\right)^{2}\mid X\right]\leq\left|\tau-\tau'\right|$.
Combining this with $\inf_{\tau\in\left[\underline{\tau},\overline{\tau}\right]}\varphi_{s}\left(\tau\right)>0$,
the continuity of $\varphi^{-1}_{s}$, $\left|U\left(\tau'\right)\right|\leq1$,
$\mathrm{E}\left[h^{-1}\bar{W}^{2}_{s,p}\right]=O\left(1\right)$
and (\ref{eq:phi_u inequality}), we obtain 
\begin{equation}
\sup_{\tau,\tau'\in\left[\underline{\tau},\overline{\tau}\right]:\left|\tau-\tau'\right|<\delta_{n}}\mathrm{E}\left[\frac{1}{h}\left(\phi_{u}\left(\tau\right)-\phi_{u}\left(\tau'\right)\right)^{2}\right]=o\left(1\right),\label{eq:phi_u continuity}
\end{equation}
for every $\delta_{n}\downarrow0$. Similarly, since 
\begin{equation}
\mathrm{E}\left[\frac{1}{h}\left(\phi_{z}\left(\tau\right)-\phi_{z}\left(\tau'\right)\right)^{2}\right]\leq\left(\mathrm{E}\left[\frac{1}{h}\bar{W}^{2}_{p}\left\Vert \xi\right\Vert ^{2}\right]\right)\left\Vert \gamma\left(\tau\right)-\gamma\left(\tau'\right)\right\Vert ^{2},\label{eq:phi_z inequality}
\end{equation}
the continuity of $\tau\mapsto\gamma\left(\tau\right)$ and $\mathrm{E}\left[h^{-1}\bar{W}^{2}_{p}\left\Vert \xi\right\Vert ^{2}\right]=O\left(1\right)$
yield that (\ref{eq:phi_u continuity}) also holds for $\phi_{z}$.
We have verified condition (19.27) of \citet{VanDerVaart1998} directly
under the Euclidean metric.

Now the hypotheses of \citet[Theorem 19.28]{VanDerVaart1998} are
all verified for the changing class $\mathscr{G}_{7}$. Therefore,
$\mathbb{Z}_{p,n}\rightsquigarrow\mathbb{Z}_{p}$ in $\ell^{\infty}\left[\underline{\tau},\overline{\tau}\right]$,
where $\mathbb{Z}_{p}$ is the centered tight Gaussian process with
covariance kernel $\mathscr{C}_{\star,p}$. The first conclusion follows
from this result, the uniform stochastic linearization (\ref{eq:uniform linearization QRD})
and Slutsky's lemma. The second conclusion follows from the first
one and the continuous mapping theorem applied to the evaluation functional. 
\end{proof}

\section{Proof of Theorem \ref{thm:QRD bootstrap}}

Let $\Pr_{*}\left[\cdot\right]$ denote probability under the resampling
distribution (the conditional distribution given the original sample).
The following notation follows \citet{Marmer:2012uq}. Let $O^{*}_{p}$
and $o^{*}_{p}$ denote the bootstrap analogues of $O_{p}$ and $o_{p}$.
For a sequence of random variables $\lambda_{n}$ and a deterministic
positive sequence $\alpha_{n}$, we write $\lambda_{n}=O^{*}_{p}\left(\alpha_{n}\right)$
if for all $\varepsilon>0$, there exists some $M>0$ such that $\limsup_{n\uparrow\infty}\Pr\left[\Pr_{*}\left[\left|\lambda_{n}/\alpha_{n}\right|\geq M\right]>\varepsilon\right]<\varepsilon$.
We write $\lambda_{n}=o^{*}_{p}\left(\alpha_{n}\right)$ if for all
$\varepsilon>0$, $\Pr_{*}\left[\left|\lambda_{n}/\alpha_{n}\right|>\varepsilon\right]\rightarrow_{p}0$
as $n\uparrow\infty$ (equivalently, $\Pr_{*}\left[\left|\lambda_{n}/\alpha_{n}\right|>\varepsilon\right]<\varepsilon$
wpa1). The usual sum and product rules for $O_{p}$ and $o_{p}$ (e.g.,
\citealp[Section 2.2]{VanDerVaart1998}) hold for $O^{*}_{p}$ and
$o^{*}_{p}$. $\lambda_{n}=O^{*}_{p}\left(\alpha_{n}\right)$ for
$\alpha_{n}\downarrow0$ implies $\lambda_{n}=o^{*}_{p}\left(1\right)$.
Further useful properties are summarized in the following lemma, which
extends \citet[Lemmas S.1 and S.2]{Marmer:2012uq} slightly.
\begin{lem}
\label{lem:O_p_star properties}(i) $\lambda_{n}=o_{p}\left(\alpha_{n}\right)$
implies $\lambda_{n}=o^{*}_{p}\left(\alpha_{n}\right)$. (ii) $\lambda_{n}=O_{p}\left(\alpha_{n}\right)$
implies $\lambda_{n}=O^{*}_{p}\left(\alpha_{n}\right)$. (iii) $\mathrm{E}_{*}\left[\left|\lambda_{n}\right|^{q}\right]=O_{p}\left(\alpha^{q}_{n}\right)$
for some $q\in\left[1,2\right]$ implies $\lambda_{n}=O^{*}_{p}\left(\alpha_{n}\right)$.
(iv) $\lambda_{n}=o^{*}_{p}\left(1\right)$ implies  $\mathrm{E}_{*}\left[\left|\lambda_{n}\right|\land M\right]=o_{p}\left(1\right)$,
for any $M>0$.
\end{lem}
\begin{proof}[Proof of Lemma \ref{lem:O_p_star properties}]
By Markov's inequality, for all $\varepsilon,\delta>0$,
\[
\Pr\left[\mathrm{Pr}_{*}\left[\left|\lambda_{n}/\alpha_{n}\right|>\varepsilon\right]>\delta\right]\leq\frac{\Pr\left[\left|\lambda_{n}/\alpha_{n}\right|>\varepsilon\right]}{\delta}.
\]
Part (i) follows from this result and $\Pr\left[\left|\lambda_{n}/\alpha_{n}\right|>\varepsilon\right]\rightarrow0$
if $\lambda_{n}=o_{p}\left(\alpha_{n}\right)$. Similarly, by Markov's
inequality, for all $M,\varepsilon>0$,
\[
\Pr\left[\mathrm{Pr}_{*}\left[\left|\lambda_{n}/\alpha_{n}\right|\geq M\right]>\varepsilon\right]\leq\frac{\Pr\left[\left|\lambda_{n}/\alpha_{n}\right|\geq M\right]}{\varepsilon}.
\]
Part (ii) follows from this result by choosing $M>0$ such that $\limsup_{n\uparrow\infty}\Pr\left[\left|\lambda_{n}/\alpha_{n}\right|\geq M\right]<\varepsilon^{2}$.
Part (iii) follows from essentially the same arguments as those in
the proof of \citet[Lemma S.2]{Marmer:2012uq}. For Part (iv), note
that for all $\varepsilon>0$, $\mathrm{E}_{*}\left[\left|\lambda_{n}\right|\land M\right]\leq\varepsilon/2+M\cdot\mathrm{Pr}_{*}\left[\left|\lambda_{n}\right|>\varepsilon/2\right]$.
Now, $\lambda_{n}=o^{*}_{p}\left(1\right)$ implies that $\mathrm{Pr}_{*}\left[\left|\lambda_{n}\right|>\varepsilon/2\right]\leq\varepsilon/\left(2M\right)$
wpa1 and it follows that $\mathrm{E}_{*}\left[\left|\lambda_{n}\right|\land M\right]\leq\varepsilon$
wpa1.
\end{proof}

Let $\bar{W}^{*}_{+,p,i}\coloneqq\mathrm{e}^{\top}_{p+1,1}\left(\varphi\Lambda_{+,p}\right)^{-1}K_{+,p}\left(X^{*}_{i}/h\right)$
denote the bootstrap analogue of $\bar{W}_{+,p,i}$. Let $\Pi^{*}_{+,p}$
be defined by the right-hand side of \eqref{eq:PI -} with $X_{i}$
replaced by $X^{*}_{i}$. Similarly, let $\left(\bar{W}^{*}_{-,p,i},\Pi^{*}_{-,p}\right)$
be the bootstrap analogues of $\left(\bar{W}_{-,p,i},\Pi_{-,p}\right)$
and $\bar{W}^{*}_{p,i}\coloneqq\bar{W}^{*}_{+,p,i}-\bar{W}^{*}_{-,p,i}$.
By Lemma \ref{lem:O_p_star properties}(iii) and Markov's inequality,
$\left\Vert \Pi^{*}_{+,p}-\Pi_{+,p}\right\Vert =o^{*}_{p}\left(1\right)$.
By this result, $\left\Vert \Pi_{+,p}-\varphi\Lambda_{+,p}\right\Vert =o_{p}\left(1\right)$,
and Lemma \ref{lem:O_p_star properties}(i), we have $\left\Vert \Pi^{*}_{+,p}-\varphi\Lambda_{+,p}\right\Vert =o^{*}_{p}\left(1\right)$.
Now this result and \eqref{eq:A_inv - B_inv} yield
\begin{equation}
\left\Vert \left(\Pi^{*}_{+,p}\right)^{-1}-\left(\varphi\Lambda_{+,p}\right)^{-1}\right\Vert =o^{*}_{p}\left(1\right).\label{eq:PI_star_inv - PI_inv rate}
\end{equation}
Similar results hold for $\Pi^{*}_{-,p}$. The following result is
a bootstrap analogue of Lemma \ref{lem: reg weights basic}.
\begin{lem}
\label{lem:bootstrap reg weights basic}Let $\left\{ \left(A^{*}_{i},X^{*}_{i}\right)\right\} ^{n}_{i=1}$
denote a bootstrap sample from $\left\{ \left(A_{i},X_{i}\right)\right\} ^{n}_{i=1}$.
Suppose that all assumptions and conditions in the statement of Lemma
\ref{lem: reg weights basic} hold. Then, for every $k\in\mathbb{N}$,
the following results hold, with $\delta$ and $r$ as defined there.
(i) 
\[
\mathrm{E}_{*}\left[\frac{1}{h}\left(\bar{W}^{*}_{+,p}\right)^{k}A^{*}\right]=\frac{\mu_{A,+}\omega^{0,k}_{+,p,0}}{\varphi^{k-1}}+o_{p}\left(1\right).
\]
(ii) 
\[
\frac{1}{nh}\sum_{i}W^{*}_{+,p,i}\mu_{A}\left(X^{*}_{i}\right)=\mu_{A,+}+\frac{\mu^{\left(p+1\right)}_{A,+}}{\left(p+1\right)!}\omega^{p+1,1}_{+,p,0}h^{p+1}+o^{*}_{p}\left(h^{p+1}\right).
\]
(iii) 
\[
\frac{1}{nh}\sum_{i}\left(W^{*}_{+,p,i}\right)^{k}A^{*}_{i}-\mathrm{E}_{*}\left[\frac{1}{h}\left(\bar{W}^{*}_{+,p}\right)^{k}A^{*}\right]=o^{*}_{p}\left(1\right).
\]
(iv) $\max_{i}\left|W^{*}_{+,p,i}A^{*}_{i}\right|=o^{*}_{p}\left(\left(nh\right)^{1/r}\right)$.
Similar results with ``$+$'' replaced by ``$-$'' also hold.
\end{lem}
\begin{proof}[Proof of Lemma \ref{lem:bootstrap reg weights basic}]
Note that the resampling distribution of $\left(A^{*},X^{*}\right)$
is the empirical distribution of $\left\{ \left(A_{i},X_{i}\right)\right\} ^{n}_{i=1}$.
It follows from \eqref{eq:1+delta absolute moment bound} and Lemma
\ref{lem: reg weights basic}(iii) that 
\[
\mathrm{E}_{*}\left[\frac{1}{h}\left(\bar{W}^{*}_{+,p}\right)^{k}A^{*}\right]-\mathrm{E}\left[\frac{1}{h}\bar{W}^{k}_{+,p}A\right]=\frac{1}{nh}\sum_{i}\bar{W}^{k}_{+,p,i}A_{i}-\mathrm{E}\left[\frac{1}{h}\bar{W}^{k}_{+,p}A\right]=o_{p}\left(1\right).
\]
Part (i) follows from this result and Lemma \ref{lem: reg weights basic}(i).

By Lemma \ref{lem:O_p_star properties}(iii) and Markov's inequality,
$\left(nh\right)^{-1}\sum_{i}\left\Vert K_{+,p}\left(X^{*}_{i}/h\right)\left(X^{*}_{i}/h\right)^{p+1}\right\Vert =O^{*}_{p}\left(1\right)$.
It then follows from this result and \eqref{eq:PI_star_inv - PI_inv rate}
that
\[
\frac{1}{nh}\sum_{i}\left|W^{*}_{+,p,i}\left(\frac{X^{*}_{i}}{h}\right)^{p+1}\right|\leq\left\Vert \left(\Pi^{*}_{+,p}\right)^{-1}\right\Vert \left(\frac{1}{nh}\sum_{i}\left\Vert K_{+,p}\left(\frac{X^{*}_{i}}{h}\right)\left(\frac{X^{*}_{i}}{h}\right)^{p+1}\right\Vert \right)=O^{*}_{p}\left(1\right).
\]
By \eqref{eq:PI_star_inv - PI_inv rate}, Lemma \ref{lem:O_p_star properties}(i),
and Lemma \ref{lem:O_p_star properties}(iii),
\begin{eqnarray*}
\frac{1}{nh}\sum_{i}W^{*}_{+,p,i}\left(\frac{X^{*}_{i}}{h}\right)^{p+1} & = & \mathrm{e}^{\top}_{p+1,1}\left\{ \left(\varphi\Lambda_{+,p}\right)^{-1}+o^{*}_{p}\left(1\right)\right\} \left(\varphi\int K_{+,p}\left(u\right)u^{p+1}\mathrm{d}u+o^{*}_{p}\left(1\right)\right)\\
 & = & \omega^{p+1,1}_{+,p,0}+o^{*}_{p}\left(1\right).
\end{eqnarray*}
Then, the conclusion in Part (ii) follows from these results and a
Taylor expansion similar to that in \eqref{eq:smoothing bias decompose 1}. 

To prove Part (iii), we use a decomposition similar to \eqref{eq:decompose 3}:
\begin{eqnarray}
\frac{1}{nh}\sum_{i}\left(W^{*}_{+,p,i}\right)^{k}A^{*}_{i}-\mathrm{E}_{*}\left[\frac{1}{h}\left(\bar{W}^{*}_{+,p}\right)^{k}A^{*}\right] & = & \frac{1}{nh}\sum_{i}\left(\left(W^{*}_{+,p,i}\right)^{k}-\left(\bar{W}^{*}_{+,p,i}\right)^{k}\right)A^{*}_{i}\nonumber \\
 &  & +\left\{ \frac{1}{nh}\sum_{i}\left(\bar{W}^{*}_{+,p,i}\right)^{k}A^{*}_{i}-\mathrm{E}_{*}\left[\frac{1}{h}\left(\bar{W}^{*}_{+,p}\right)^{k}A^{*}\right]\right\} .\nonumber \\
 &  & \label{eq:decompose bootstrap}
\end{eqnarray}
From arguments similar to those used for the first term on the right-hand
side of \eqref{eq:decompose 3}, together with \eqref{eq:PI_star_inv - PI_inv rate},
Lemma \ref{lem:O_p_star properties}(iii), and Markov's inequality,
the first term on the right-hand side is $o^{*}_{p}\left(1\right)$.
Then, we can obtain a bootstrap analogue of the inequality in \eqref{eq:1+delta absolute moment bound}
by using the same arguments. By computing the population moment and
applying Markov's inequality, $\left(nh\right)^{-1}\sum_{i}\left|\bar{W}^{k}_{+,p,i}A_{i}\right|^{1+\delta}=O_{p}\left(1\right)$.
Now we have 
\[
\mathrm{E}_{*}\left[\left|\frac{1}{nh}\sum_{i}\left(\bar{W}^{*}_{+,p,i}\right)^{k}A^{*}_{i}-\mathrm{E}_{*}\left[\frac{1}{h}\left(\bar{W}^{*}_{+,p}\right)^{k}A^{*}\right]\right|^{1+\delta}\right]=O_{p}\left(\left(nh\right)^{-\delta}\right),
\]
and by Lemma \ref{lem:O_p_star properties}(iii), the second term
on the right-hand side of \eqref{eq:decompose bootstrap} is also
$o^{*}_{p}\left(1\right)$. 

For Part (iv), note that since every bootstrap observation is one
of the original observations, $\underset{i}{\max}\left\Vert K_{+,p}\left(X^{*}_{i}/h\right)\right\Vert \left|A^{*}_{i}\right|$
is bounded from above by $\max_{i}\left\Vert K_{+,p}\left(X_{i}/h\right)\right\Vert \left|A_{i}\right|=o_{p}\left(\left(nh\right)^{1/r}\right)$.
The conclusion follows from this observation, Lemma \ref{lem:O_p_star properties}(i),
and \eqref{eq:PI_star_inv - PI_inv rate}.
\end{proof}

Let $\widehat{\kappa}^{*}_{+,p}\left(\tau\right)$ and $R^{*}_{+}\left(\varDelta,\tau\right)$
denote the bootstrap analogues of $\widehat{\kappa}_{+,p}\left(\tau\right)$
and $R_{+}\left(\varDelta,\tau\right)$. The balancing weights are
taken to be $\widetilde{w}^{*}_{p,i}$. Let $\tilde{Y}^{*}_{+,i}$
be the bootstrap analogue of $\tilde{Y}_{+,i}$. Let $\mathbb{P}^{*}_{n}g\coloneqq n^{-1}\sum_{i}g\left(Y^{*}_{i},X^{*}_{i},Z^{*}_{i}\right)$
denote the average under the empirical distribution of the bootstrap
sample. Let $\mathbb{G}^{*}_{n}g\coloneqq\sqrt{n}\left(\mathbb{P}^{*}_{n}-\mathbb{P}_{n}\right)g$.
The following result is the bootstrap analogue of Lemma \ref{lem:R_delta linearization}.
\begin{lem}
\label{lem:bootstrap analogue 1}Suppose that the assumptions in the
statement of Theorem \ref{thm:QRD bootstrap} hold. For all $M>0$,
\begin{equation}
\sup_{\|\varDelta\|\leq M,\,\tau\in\left[\underline{\tau},\overline{\tau}\right]}\left\Vert R^{*}_{+}(\varDelta,\tau)+D_{+}(\tau)\varDelta-R^{*}_{+}(0,\tau)\right\Vert =o^{*}_{p}(1).\label{eq:asymptotic linearization-1}
\end{equation}
\end{lem}
\begin{proof}
We have
\begin{eqnarray*}
\widetilde{w}^{*}_{p,i} & = & \frac{1}{n}\cdot\frac{1+V^{*\top}_{p,i}\widetilde{\lambda}^{*}_{p}}{1+\left(n^{-1}\sum_{j}V^{*}_{p,j}\right)^{\top}\widetilde{\lambda}^{*}_{p}},\textrm{ where}\\
\widetilde{\lambda}^{*}_{p} & \coloneqq & -\left(\sum_{i}V^{*}_{p,i}V^{*\top}_{p,i}\right)^{-1}\left(\sum_{i}V^{*}_{p,i}\right).
\end{eqnarray*}
By Lemma \ref{lem:bootstrap reg weights basic}, we have
\begin{eqnarray}
\frac{1}{nh}\sum_{i}V^{*}_{p,i}V^{*\top}_{p,i} & = & \left(\frac{\omega^{0,2}_{p,0}}{\varphi}\right)\mu_{\bar{Z}\bar{Z}^{\top},\pm}+o^{*}_{p}\left(1\right)\nonumber \\
\frac{1}{nh}\sum_{i}W^{*}_{p,i}Z^{*}_{i} & = & \frac{1}{nh}\sum_{i}W^{*}_{p,i}\xi^{*}_{i}+\frac{\left(\mu^{\left(p+1\right)}_{Z,+}\omega^{p+1,1}_{+,p,0}-\mu^{\left(p+1\right)}_{Z,-}\omega^{p+1,1}_{-,p,0}\right)}{\left(p+1\right)!}h^{p+1}+o^{*}_{p}\left(h^{p+1}\right),\label{eq:VV* expansion}
\end{eqnarray}
where $\xi^{*}_{i}\coloneqq Z^{*}_{i}-\mu_{Z}\left(X^{*}_{i}\right)$.
By a simple variance calculation with respect to the resampling distribution,
Lemma \ref{lem:O_p_star properties}(iii), and Markov's inequality,
\[
\frac{1}{\sqrt{nh}}\sum_{i}K_{s,p}\left(\frac{X^{*}_{i}}{h}\right)\xi^{*\top}_{i}-\frac{1}{\sqrt{nh}}\sum_{i}K_{s,p}\left(\frac{X_{i}}{h}\right)\xi^{\top}_{i}=O^{*}_{p}\left(1\right),
\]
for $s\in\left\{ -,+\right\} $. Now, $\left(nh\right)^{-1/2}\sum_{i}K_{s,p}\left(X^{*}_{i}/h\right)\xi^{*\top}_{i}=O^{*}_{p}\left(1\right)$
follows from this result, the fact that $\left(nh\right)^{-1/2}\sum_{i}K_{s,p}\left(X_{i}/h\right)\xi^{\top}_{i}=O_{p}\left(1\right)$
and Lemma \ref{lem:O_p_star properties}(ii). Then, it follows from
this result, \eqref{eq:PI_star_inv - PI_inv rate}, and the bootstrap
analogue of the equality in \eqref{eq:W*ksi expansion} that $\left(nh\right)^{-1/2}\sum_{i}W^{*}_{p,i}\xi^{*}_{i}=O^{*}_{p}\left(1\right)$.
Since $\sum_{i}W^{*}_{p,i}=0$, the bootstrap analogue of \eqref{eq:V decompose}
holds. By this result and \eqref{eq:VV* expansion}, we have $\left(nh\right)^{-1/2}\sum_{i}V^{*}_{p,i}=O^{*}_{p}\left(1\right)$,
$\widetilde{\lambda}^{*}_{p}=O^{*}_{p}\left(\left(nh\right)^{-1/2}\right)$,
and $\widetilde{w}^{*}_{p,i}=n^{-1}\left(1+V^{*\top}_{p,i}\widetilde{\lambda}^{*}_{p}\right)\left(1+O^{*}_{p}\left(n^{-1}\right)\right)$.

Decompose
\[
R^{*}_{+}\left(\varDelta,\tau\right)-R^{*}_{+}\left(0,\tau\right)=T^{*}_{1}\left(\varDelta,\tau\right)\left\{ 1+O^{*}_{p}\left(n^{-1}\right)\right\} +T^{*}_{2}\left(\varDelta,\tau\right)\left\{ 1+O^{*}_{p}\left(n^{-1}\right)\right\} ,
\]
where
\begin{eqnarray*}
T^{*}_{1}\left(\varDelta,\tau\right) & \coloneqq & \frac{1}{\sqrt{nh}}\sum_{i}K_{+,p}\left(\frac{X^{*}_{i}}{h}\right)\left\{ \mathbbm{1}\left(\tilde{Y}^{*}_{+,i}\leq0\right)-\mathbbm{1}\left(\tilde{Y}^{*}_{+,i}\leq\frac{r^{\top}_{p}\left(X^{*}_{i}/h\right)\varDelta}{\sqrt{nh}}\right)\right\} \\
T^{*}_{2}\left(\varDelta,\tau\right) & \coloneqq & \frac{1}{\sqrt{nh}}\sum_{i}K_{+,p}\left(\frac{X^{*}_{i}}{h}\right)\left(V^{*\top}_{p,i}\widetilde{\lambda}^{*}_{p}\right)\left\{ \mathbbm{1}\left(\tilde{Y}^{*}_{+,i}\leq0\right)-\mathbbm{1}\left(\tilde{Y}^{*}_{+,i}\leq\frac{r^{\top}_{p}\left(X^{*}_{i}/h\right)\varDelta}{\sqrt{nh}}\right)\right\} 
\end{eqnarray*}
denote the bootstrap analogues of $T_{1}\left(\varDelta,\tau\right)$
and $T_{2}\left(\varDelta,\tau\right)$. Decomposing $T^{*}_{1}\left(\varDelta,\tau\right)$
as in \eqref{eq:T_1 decomposition}, we have $\mathrm{E}_{*}\left[T^{*}_{1}\left(\varDelta,\tau\right)\right]=T_{1}\left(\varDelta,\tau\right)$.
We can write the $j$-th coordinate of $T^{*}_{1}\left(\varDelta,\tau\right)-\mathrm{E}_{*}\left[T^{*}_{1}\left(\varDelta,\tau\right)\right]$
as $\mathbb{G}^{*}_{n}g_{1}\left(\cdot\mid\varDelta,\tau\right)$.
Applying \citet[Corollary 5.5]{Chen2020jackknife} to $\mathrm{E}_{*}\left[\left\Vert \mathbb{G}^{*}_{n}\right\Vert _{\mathscr{G}_{1}}\right]$,
we have
\begin{equation}
\mathrm{E}_{*}\left[\left\Vert \mathbb{G}^{*}_{n}\right\Vert _{\mathscr{G}_{1}}\right]\apprle\sqrt{\left(\sup_{\|\varDelta\|\leq M,\,\tau\in\left[\underline{\tau},\overline{\tau}\right]}\mathbb{P}_{n}g^{2}_{1}\left(\cdot\mid\varDelta,\tau\right)\right)\log\left(n\right)}+\frac{\log\left(n\right)}{\sqrt{nh}}.\label{eq:G1 expectation bound bootstrap}
\end{equation}
Let $\mathscr{G}^{\left(2\right)}_{1}\coloneqq\left\{ g^{2}_{1}\left(\cdot\mid\varDelta,\tau\right):\|\varDelta\|\leq M,\,\tau\in\left[\underline{\tau},\overline{\tau}\right]\right\} $.
By \citet[Corollary A.1]{chernozhukov2014gaussian}, $\mathscr{G}^{\left(2\right)}_{1}$
is VC-type with a constant envelope $G^{2}_{1}$ and characteristics
independent of $n$. By \eqref{eq:a.b indicator inequalities} and
calculations as in \eqref{eq:g_1^2 expectation bound}, we have $\mathbb{P}g^{4}_{1}\left(\cdot\mid\varDelta,\tau\right)=O\left(\left(nh^{3}\right)^{-1/2}\right)$,
uniformly in $\left\Vert \varDelta\right\Vert \leq M,\tau\in\left[\underline{\tau},\overline{\tau}\right]$.
Then, it follows from this result and \citet[Corollary 5.5]{Chen2020jackknife}
that $\left\Vert \mathbb{P}_{n}-\mathbb{P}\right\Vert _{\mathscr{G}^{\left(2\right)}_{1}}=O_{p}\left(\sqrt{\log\left(n\right)}/\left(nh\right)^{3/4}\right)$.
It now follows that $\mathbb{P}_{n}g^{2}_{1}\left(\cdot\mid\varDelta,\tau\right)=O_{p}\left(\left(nh\right)^{-1/2}\right)$,
uniformly in $\left\Vert \varDelta\right\Vert \leq M,\tau\in\left[\underline{\tau},\overline{\tau}\right]$.
By this result and \eqref{eq:G1 expectation bound bootstrap}, we
have $\mathrm{E}_{*}\left[\left\Vert \mathbb{G}^{*}_{n}\right\Vert _{\mathscr{G}_{1}}\right]=O_{p}\left(\sqrt{\log\left(n\right)}/\left(nh\right)^{1/4}\right)$
and $\left\Vert \mathbb{G}^{*}_{n}\right\Vert _{\mathscr{G}_{1}}=o^{*}_{p}\left(1\right)$
by Lemma \ref{lem:O_p_star properties}(iii). This gives that $T^{*}_{1}\left(\varDelta,\tau\right)-\mathrm{E}_{*}\left[T^{*}_{1}\left(\varDelta,\tau\right)\right]=o^{*}_{p}\left(1\right)$,
uniformly in $\left\Vert \varDelta\right\Vert \leq M,\tau\in\left[\underline{\tau},\overline{\tau}\right]$.
By this result, \eqref{eq:T_1 uniform approximation}, and Lemma \ref{lem:O_p_star properties}(i),
we have $T^{*}_{1}\left(\varDelta,\tau\right)=-D_{+}\left(\tau\right)\varDelta+o^{*}_{p}\left(1\right)$,
uniformly for all $\left\Vert \varDelta\right\Vert \leq M,\tau\in\left[\underline{\tau},\overline{\tau}\right]$.

For $T^{*}_{2}\left(\varDelta,\tau\right)$, we have
\[
\sup_{\|\varDelta\|\leq M}\left\Vert T^{*}_{2}\left(\varDelta,\tau\right)\right\Vert \leq\left(\max_{i}\left|V^{*\top}_{p,i}\widetilde{\lambda}^{*}_{p}\right|\right)\left(\sqrt{n}\cdot\mathbb{P}^{*}_{n}g_{2}\left(\cdot\mid\tau\right)\right).
\]
By Lemma \ref{lem:bootstrap reg weights basic}(iv), and $\widetilde{\lambda}^{*}_{p}=O^{*}_{p}\left(\left(nh\right)^{-1/2}\right)$,
we have $\max_{i}\left|V^{*\top}_{p,i}\widetilde{\lambda}^{*}_{p}\right|=o^{*}_{p}\left(1\right)$.
Decompose $\sqrt{n}\cdot\mathbb{P}^{*}_{n}g_{2}\left(\cdot\mid\tau\right)=\sqrt{n}\cdot\mathbb{P}_{n}g_{2}\left(\cdot\mid\tau\right)+\mathbb{G}^{*}_{n}g_{2}\left(\cdot\mid\tau\right)$.
We showed in the proof of Lemma \ref{lem:R_delta linearization} that
$\sqrt{n}\cdot\mathbb{P}_{n}g_{2}\left(\cdot\mid\tau\right)=O_{p}\left(1\right)$
uniformly in $\tau\in\left[\underline{\tau},\overline{\tau}\right]$.
Let $\mathscr{G}^{\left(2\right)}_{2}\coloneqq\left\{ g^{2}_{2}\left(\cdot\mid\tau\right):\tau\in\left[\underline{\tau},\overline{\tau}\right]\right\} $.
By \citet[Corollary A.1]{chernozhukov2014gaussian}, $\mathscr{G}^{\left(2\right)}_{2}$
is VC-type with a constant envelope $G^{2}_{2}\apprle h^{-1}$ and
characteristics independent of $n$. By \eqref{eq:a.b indicator inequalities},
the LIE, a bound analogous to \eqref{eq:F difference bound 1} and
a change of variables, as in \eqref{eq:g_1^2 expectation bound},
we have $\mathbb{P}g^{4}_{2}\left(\cdot\mid\tau\right)=O\left(\left(nh^{3}\right)^{-1/2}\right)$,
uniformly in $\tau\in\left[\underline{\tau},\overline{\tau}\right]$.
By \citet[Corollary 5.5]{Chen2020jackknife}, $\left\Vert \mathbb{P}_{n}-\mathbb{P}\right\Vert _{\mathscr{G}^{\left(2\right)}_{2}}=o_{p}\left(\left(nh\right)^{-1/2}\right)$,
and since $\sup_{\tau\in\left[\underline{\tau},\overline{\tau}\right]}\mathbb{P}g^{2}_{2}\left(\cdot\mid\tau\right)=O\left(\left(nh\right)^{-1/2}\right)$
as shown in the proof of Lemma \ref{lem:R_delta linearization}, we
have $\sup_{\tau\in\left[\underline{\tau},\overline{\tau}\right]}\mathbb{P}_{n}g^{2}_{2}\left(\cdot\mid\tau\right)=O_{p}\left(\left(nh\right)^{-1/2}\right)$.
Applying \citet[Corollary 5.5]{Chen2020jackknife} under the resampling
distribution, as in \eqref{eq:G1 expectation bound bootstrap}, we
have $\mathrm{E}_{*}\left[\left\Vert \mathbb{G}^{*}_{n}\right\Vert _{\mathscr{G}_{2}}\right]=o_{p}\left(1\right)$,
and $\left\Vert \mathbb{G}^{*}_{n}\right\Vert _{\mathscr{G}_{2}}=o^{*}_{p}\left(1\right)$
by Lemma \ref{lem:O_p_star properties}(iii). It follows from the
above results and Lemma \ref{lem:O_p_star properties}(ii) that $\sqrt{n}\cdot\mathbb{P}^{*}_{n}g_{2}\left(\cdot\mid\tau\right)=O^{*}_{p}\left(1\right)$,
uniformly in $\tau\in\left[\underline{\tau},\overline{\tau}\right]$.
Now we have $T^{*}_{2}\left(\varDelta,\tau\right)=o^{*}_{p}\left(1\right)$,
uniformly in $\left\Vert \varDelta\right\Vert \leq M,\tau\in\left[\underline{\tau},\overline{\tau}\right]$. 
\end{proof}

The following result is the bootstrap analogue of Lemma \ref{lem:R_0 linearization}.
\begin{lem}
\label{lem:bootstrap analogue 2}Suppose that the assumptions in the
statement of Theorem \ref{thm:QRD bootstrap} hold. We have the following
linearization for $R^{*}_{+}(0,\tau)$:
\begin{eqnarray*}
R^{*}_{+}\left(0,\tau\right) & = & \frac{1}{\sqrt{nh}}\sum_{i}K_{+,p}\left(\frac{X^{*}_{i}}{h}\right)U^{*}_{i}\left(\tau\right)-\left(\frac{\pi_{+}\varphi}{\omega^{0,2}_{p,0}}\right)\gamma^{\top}_{+}\left(\tau\right)\left\{ \frac{1}{\sqrt{nh}}\sum_{i}W^{*}_{p,i}\xi^{*}_{i}\right\} \\
 &  & +\bar{\mathscr{B}}_{+}\left(\tau\right)\sqrt{nh}h^{p+1}+o^{*}_{p}\left(1\right),
\end{eqnarray*}
uniformly in $\tau\in\left[\underline{\tau},\overline{\tau}\right]$,
where $U^{*}_{i}\left(\tau\right)\coloneqq\tau-\mathbbm{1}\left(Y^{*}_{i}\leq Q_{Y\mid X}\left(\tau\mid X^{*}_{i}\right)\right)$.
\end{lem}
\begin{proof}[Proof of Lemma \ref{lem:bootstrap analogue 2}]
Decompose
\begin{eqnarray}
R^{*}_{+}\left(0,\tau\right) & = & \sqrt{\frac{n}{h}}\sum_{i}\widetilde{w}^{*}_{p,i}K_{+,p}\left(\frac{X^{*}_{i}}{h}\right)\psi_{\tau}\left(\tilde{Y}^{*}_{+,i}\right)\nonumber \\
 & = & T^{*}_{1}\left(\tau\right)\left(1+O^{*}_{p}\left(n^{-1}\right)\right)+T^{*}_{2}\left(\tau\right)\left(1+O^{*}_{p}\left(n^{-1}\right)\right)\label{eq:R_+_star decomposition}
\end{eqnarray}
as in the proof of Lemma \ref{lem:R_0 linearization}, where
\begin{eqnarray*}
T^{*}_{1}\left(\tau\right) & \coloneqq & \frac{1}{\sqrt{nh}}\sum_{i}K_{+,p}\left(\frac{X^{*}_{i}}{h}\right)\psi_{\tau}\left(\tilde{Y}^{*}_{+,i}\right)\\
 & = & \frac{1}{\sqrt{nh}}\sum_{i}K_{+,p}\left(\frac{X^{*}_{i}}{h}\right)U^{*}_{i}\left(\tau\right)+T^{*}_{3}\left(\tau\right)\\
T^{*}_{2}\left(\tau\right) & \coloneqq & \frac{1}{\sqrt{nh}}\sum_{i}K_{+,p}\left(\frac{X^{*}_{i}}{h}\right)\psi_{\tau}\left(\tilde{Y}^{*}_{+,i}\right)V^{*\top}_{p,i}\widetilde{\lambda}^{*}_{p}\\
T^{*}_{3}\left(\tau\right) & \coloneqq & \frac{1}{\sqrt{nh}}\sum_{i}K_{+,p}\left(\frac{X^{*}_{i}}{h}\right)\left\{ \mathbbm{1}\left(Y^{*}_{i}\leq Q_{Y\mid X}\left(\tau\mid X^{*}_{i}\right)\right)-\mathbbm{1}\left(Y^{*}_{i}\leq r^{\top}_{p}\left(X^{*}_{i}\right)\kappa_{+}\left(\tau\right)\right)\right\} .
\end{eqnarray*}
Since $\mathbb{P}g_{3}\left(\cdot\mid\tau\right)=0$, we can write
\[
\frac{1}{\sqrt{nh}}\sum_{i}K_{+,p}\left(\frac{X^{*}_{i}}{h}\right)U^{*}_{i}\left(\tau\right)=\mathbb{G}^{*}_{n}g_{3}\left(\cdot\mid\tau\right)+\mathbb{G}_{n}g_{3}\left(\cdot\mid\tau\right).
\]
By \citet[Theorem 2.14.1]{VanDerVaartWellner1996}, $\mathrm{E}_{*}\left[\left\Vert \mathbb{G}^{*}_{n}\right\Vert _{\mathscr{G}_{3}}\right]\apprle\sqrt{\mathbb{P}_{n}G^{2}_{3}}=O_{p}\left(1\right)$.
It follows from this result, Lemma \ref{lem:O_p_star properties}(ii,
iii), and \eqref{eq:g_3 bound} that
\begin{equation}
\frac{1}{\sqrt{nh}}\sum_{i}K_{+,p}\left(\frac{X^{*}_{i}}{h}\right)U^{*}_{i}\left(\tau\right)=O^{*}_{p}\left(1\right),\label{eq:U_star bounded in probability}
\end{equation}
uniformly in $\tau\in\left[\underline{\tau},\overline{\tau}\right]$.

We write $T^{*}_{3}\left(\tau\right)=\left(T^{*}_{3}\left(\tau\right)-\mathrm{E}_{*}\left[T^{*}_{3}\left(\tau\right)\right]\right)+\mathrm{E}_{*}\left[T^{*}_{3}\left(\tau\right)\right]$,
where $\mathrm{E}_{*}\left[T^{*}_{3}\left(\tau\right)\right]=T_{3}\left(\tau\right)$
and $T_{3}\left(\tau\right)$ is defined by \eqref{eq:T_3 definition}.
We can also write the $j$-th coordinate of $T^{*}_{3}\left(\tau\right)-\mathrm{E}_{*}\left[T^{*}_{3}\left(\tau\right)\right]$
as $\mathbb{G}^{*}_{n}g_{4}\left(\cdot\mid\tau\right)$. Let $\mathscr{G}_{8}\coloneqq\left\{ g_{4}\left(\cdot\mid\tau\right):\tau\in\left[\underline{\tau},\overline{\tau}\right]\right\} $.
The class $\mathscr{G}_{8}$ is VC-type with characteristics independent
of $n$ and a constant envelope proportional to $h^{-1/2}$. By \citet[Corollary 5.5]{Chen2020jackknife},
we have
\begin{equation}
\mathrm{E}_{*}\left[\left\Vert \mathbb{G}^{*}_{n}\right\Vert _{\mathscr{G}_{8}}\right]\apprle\sqrt{\left(\sup_{\tau\in\left[\underline{\tau},\overline{\tau}\right]}\mathbb{P}_{n}g^{2}_{4}\left(\cdot\mid\tau\right)\right)\log\left(n\right)}+\frac{\log\left(n\right)}{\sqrt{nh}}.\label{eq:G_8 bootstrap expectation bound}
\end{equation}
Let $\mathscr{G}^{\left(2\right)}_{4}\coloneqq\left\{ g^{2}_{4}\left(\cdot\mid\tau\right):\tau\in\left[\underline{\tau},\overline{\tau}\right]\right\} $.
The class $\mathscr{G}^{\left(2\right)}_{4}$ is VC-type with characteristics
independent of $n$ and a constant envelope proportional to $h^{-1}$.
By \citet[Corollary 5.5]{Chen2020jackknife},
\[
\mathrm{E}\left[\left\Vert \mathbb{G}_{n}\right\Vert _{\mathscr{G}^{\left(2\right)}_{4}}\right]\apprle\sqrt{\left(\sup_{\tau\in\left[\underline{\tau},\overline{\tau}\right]}\mathbb{P}g^{4}_{4}\left(\cdot\mid\tau\right)\right)\log\left(n\right)}+\frac{\log\left(n\right)}{\sqrt{nh^{2}}}.
\]
By calculations using the LIE, \eqref{eq:F(Q) taylor expansion} and
\eqref{eq:a.b indicator inequalities}, we have $\mathbb{P}g^{4}_{4}\left(\cdot\mid\tau\right)=O\left(h^{p}\right)$,
uniformly in $\tau\in\left[\underline{\tau},\overline{\tau}\right]$.
It follows that 
\[
\left\Vert \mathbb{P}_{n}-\mathbb{P}\right\Vert _{\mathscr{G}^{\left(2\right)}_{4}}=O_{p}\left(\sqrt{\frac{h^{p}\log\left(n\right)}{n}}+\frac{\log\left(n\right)}{nh}\right).
\]
By this result, \eqref{eq:(g_4-g_5)^2 expectation bound}, and \eqref{eq:G_8 bootstrap expectation bound},
we have $\mathrm{E}_{*}\left[\left\Vert \mathbb{G}^{*}_{n}\right\Vert _{\mathscr{G}_{8}}\right]=o_{p}\left(1\right)$
and $\left\Vert \mathbb{G}^{*}_{n}\right\Vert _{\mathscr{G}_{8}}=o^{*}_{p}\left(1\right)$.
This shows that $T^{*}_{3}\left(\tau\right)-\mathrm{E}_{*}\left[T^{*}_{3}\left(\tau\right)\right]=o^{*}_{p}\left(1\right)$,
uniformly in $\tau\in\left[\underline{\tau},\overline{\tau}\right]$.
In the proof of Lemma \ref{lem:R_0 linearization}, we showed that
\[
T_{3}\left(\tau\right)=\sqrt{nh}h^{p+1}\left\{ \left(\frac{\int^{1}_{0}K_{+,p}\left(u\right)u^{p+1}\mathrm{d}u}{\left(p+1\right)!}\right)\varphi_{+}\left(\tau\right)\kappa^{\left(p+1\right)}_{Y,+}\left(\tau\right)\varphi+o\left(1\right)\right\} +o_{p}\left(1\right),
\]
uniformly in $\tau\in\left[\underline{\tau},\overline{\tau}\right]$.
It now follows from these results and Lemma \ref{lem:O_p_star properties}(i)
that
\begin{equation}
T^{*}_{3}\left(\tau\right)=\sqrt{nh}h^{p+1}\left\{ \left(\frac{\int^{1}_{0}K_{+,p}\left(u\right)u^{p+1}\mathrm{d}u}{\left(p+1\right)!}\right)\varphi_{+}\left(\tau\right)\kappa^{\left(p+1\right)}_{Y,+}\left(\tau\right)\varphi+o\left(1\right)\right\} +o^{*}_{p}\left(1\right),\label{eq:T_3_star approximation}
\end{equation}
uniformly in $\tau\in\left[\underline{\tau},\overline{\tau}\right]$. 

By \eqref{eq:PI_star_inv - PI_inv rate} and Lemma \ref{lem:O_p_star properties}(iii),
we have a bootstrap analogue of \eqref{eq:K to K_til approximation}.
Write 
\[
\frac{1}{nh}\sum_{i}\tilde{K}_{+,p}\left(\frac{X^{*}_{i}}{h}\right)\psi_{\tau}\left(\tilde{Y}^{*}_{+,i}\right)\bar{Z}^{*\top}_{i}=T^{*}_{6}\left(\tau\right)+T^{*}_{7}\left(\tau\right),
\]
where $T^{*}_{6}\left(\tau\right)$ and $T^{*}_{7}\left(\tau\right)$
are bootstrap analogues of $T_{6}\left(\tau\right)$ and $T_{7}\left(\tau\right)$
defined by \eqref{eq:T_6 and T_7 definitions}. Write $T^{*}_{6}\left(\tau\right)=\left(T^{*}_{6}\left(\tau\right)-\mathrm{E}_{*}\left[T^{*}_{6}\left(\tau\right)\right]\right)+\mathrm{E}_{*}\left[T^{*}_{6}\left(\tau\right)\right]$,
where $\mathrm{E}_{*}\left[T^{*}_{6}\left(\tau\right)\right]=T_{6}\left(\tau\right)$.
The $jk$-th coordinate of $T^{*}_{6}\left(\tau\right)-\mathrm{E}_{*}\left[T^{*}_{6}\left(\tau\right)\right]$
can be written as $\mathbb{G}^{*}_{n}g_{6}\left(\cdot\mid\tau\right)/\sqrt{nh}$.
By \citet[Theorem 2.14.1]{VanDerVaartWellner1996} and Markov's inequality,
$\mathrm{E}_{*}\left[\left\Vert \mathbb{G}^{*}_{n}\right\Vert _{\mathscr{G}_{6}}\right]=O_{p}\left(1\right)$.
It follows that $T^{*}_{6}\left(\tau\right)-\mathrm{E}_{*}\left[T^{*}_{6}\left(\tau\right)\right]=o^{*}_{p}\left(1\right)$,
uniformly in $\tau\in\left[\underline{\tau},\overline{\tau}\right]$.
By similar arguments, $T^{*}_{7}\left(\tau\right)-\mathrm{E}_{*}\left[T^{*}_{7}\left(\tau\right)\right]=o^{*}_{p}\left(1\right)$,
uniformly in $\tau\in\left[\underline{\tau},\overline{\tau}\right]$.
In the proof of Lemma \ref{lem:R_0 linearization}, we showed that
$T_{6}\left(\tau\right)=\pi_{+}\varphi\mu_{U\left(\tau\right)\bar{Z}^{\top},+}+o_{p}\left(1\right)$
and $\mathrm{E}_{*}\left[T^{*}_{7}\left(\tau\right)\right]=T_{7}\left(\tau\right)=o_{p}\left(1\right)$,
uniformly in $\tau\in\left[\underline{\tau},\overline{\tau}\right]$.
Now, by Lemma \ref{lem:O_p_star properties}(i), $T^{*}_{6}\left(\tau\right)=\pi_{+}\varphi\mu_{U\left(\tau\right)\bar{Z}^{\top},+}+o^{*}_{p}\left(1\right)$,
uniformly in $\tau\in\left[\underline{\tau},\overline{\tau}\right]$.
Now we have
\[
\frac{1}{nh}\sum_{i}K_{+,p}\left(\frac{X^{*}_{i}}{h}\right)\psi_{\tau}\left(\tilde{Y}^{*}_{+,i}\right)V^{*\top}_{p,i}=\pi_{+}\mu_{U\left(\tau\right)\bar{Z}^{\top},+}+o^{*}_{p}\left(1\right),
\]
uniformly in $\tau\in\left[\underline{\tau},\overline{\tau}\right]$.
By this result and \eqref{eq:VV* expansion}, we have
\[
T^{*}_{2}\left(\tau\right)=-\left(\frac{\pi_{+}\varphi}{\omega^{0,2}_{p,0}}\right)\gamma^{\top}_{+}\left(\tau\right)\left\{ \frac{1}{\sqrt{nh}}\sum_{i}W^{*}_{p,i}\xi^{*}_{i}+\frac{\left(\mu^{\left(p+1\right)}_{Z,+}\omega^{p+1,1}_{+,p,0}-\mu^{\left(p+1\right)}_{Z,-}\omega^{p+1,1}_{-,p,0}\right)}{\left(p+1\right)!}\sqrt{nh}h^{p+1}\right\} +o^{*}_{p}\left(1\right),
\]
uniformly in $\tau\in\left[\underline{\tau},\overline{\tau}\right]$.
The conclusion follows from this result, \eqref{eq:R_+_star decomposition},
and \eqref{eq:T_3_star approximation}.
\end{proof}

Let $\varDelta^{*}_{+}\left(\tau\right)\coloneqq\sqrt{nh}\mathrm{H}\left(\widehat{\kappa}^{*}_{+,p}\left(\tau\right)-\kappa_{+}\left(\tau\right)\right)$
and let $\varDelta^{*}_{-}\left(\tau\right)$ be defined analogously.
Then, we have
\begin{equation}
\sqrt{nh}\left(\widehat{\vartheta}^{*}_{p}\left(\tau\right)-\widehat{\vartheta}^{\mathit{eb}}_{p}\left(\tau\right)\right)=\mathrm{e}^{\top}_{p+1,1}\left(\varDelta^{*}_{+}\left(\tau\right)-\varDelta^{*}_{-}\left(\tau\right)\right)-\mathrm{e}^{\top}_{p+1,1}\left(\varDelta_{+}\left(\tau\right)-\varDelta_{-}\left(\tau\right)\right).\label{eq:theta_hat_star - theta_hat decomposition}
\end{equation}
The following result is a bootstrap analogue of Lemma \ref{lem:delta linearization}.
A similar result holds for $\varDelta^{*}_{-}\left(\tau\right)$.
\begin{lem}
\label{lem:delta linearization bootstrap analogue}Suppose that the
assumptions in the statement of Theorem \ref{thm:QRD bootstrap} hold.
We have the following linearization:
\[
\varDelta^{*}_{+}\left(\tau\right)=D^{-1}_{+}(\tau)R^{*}_{+}(0,\tau)+o^{*}_{p}\left(1\right),
\]
uniformly in $\tau\in\left[\underline{\tau},\overline{\tau}\right]$. 
\end{lem}
\begin{proof}[Proof of Lemma \ref{lem:delta linearization bootstrap analogue}]
Since we obtained $\max_{i}\left|V^{*\top}_{p,i}\widetilde{\lambda}^{*}_{p}\right|=o^{*}_{p}\left(1\right)$
and 
\[
\widetilde{w}^{*}_{p,i}=n^{-1}\left(1+V^{*\top}_{p,i}\widetilde{\lambda}^{*}_{p}\right)\left(1+O^{*}_{p}\left(n^{-1}\right)\right)
\]
in the proof of Lemma \ref{lem:bootstrap analogue 1}, we have $\max_{i}\left|\widetilde{w}^{*}_{p,i}\right|=O^{*}_{p}\left(n^{-1}\right)$.
Moreover,
\[
\mathrm{Pr}_{*}\left[\widetilde{w}^{*}_{p,i}>0,\,\forall i\right]\geq\mathrm{Pr}_{*}\left[\max_{i}\left|V^{*\top}_{p,i}\widetilde{\lambda}^{*}_{p}\right|<\frac{1}{2}\right]\rightarrow_{p}1.
\]
When $\widetilde{w}^{*}_{p,i}>0$ for all $i$, by a slight adaptation
of \citet[Lemma A.2]{Cai2008} to the bootstrap sample,
\[
\left\Vert \sum_{i}\widetilde{w}^{*}_{p,i}K_{+,p}\left(\frac{X^{*}_{i}}{h}\right)\psi_{\tau}\left(Y^{*}_{i}-r^{\top}_{p}\left(X^{*}_{i}\right)\widehat{\kappa}^{*}_{+,p}\left(\tau\right)\right)\right\Vert \leq\left(p+1\right)E_{n}\underset{i}{\max}\left\Vert \widetilde{w}^{*}_{p,i}K_{+,p}\left(\frac{X^{*}_{i}}{h}\right)\right\Vert ,
\]
where $E_{n}$ is the maximal multiplicity of observations in the
bootstrap sample. By \citet[Theorem 3.1]{Motwani1995}, $E_{n}=o_{p}\left(\log\left(n\right)\right)$.
By this result and the union bound, for all $\varepsilon,M>0$,
\begin{multline*}
\Pr\left[\mathrm{Pr}_{*}\left[\frac{n}{\log\left(n\right)}\left\Vert \sum_{i}\widetilde{w}^{*}_{p,i}K_{+,p}\left(\frac{X^{*}_{i}}{h}\right)\psi_{\tau}\left(Y^{*}_{i}-r^{\top}_{p}\left(X^{*}_{i}\right)\widehat{\kappa}^{*}_{+,p}\left(\tau\right)\right)\right\Vert \geq M\right]>\varepsilon\right]\\
\leq\Pr\left[\mathrm{Pr}_{*}\left[\frac{nE_{n}\left(p+1\right)}{\log\left(n\right)}\underset{i}{\max}\left\Vert \widetilde{w}^{*}_{p,i}K_{+,p}\left(\frac{X^{*}_{i}}{h}\right)\right\Vert >M\right]>\frac{\varepsilon}{2}\right]+\Pr\left[\mathrm{Pr}_{*}\left[\exists i:\widetilde{w}^{*}_{p,i}\leq0\right]>\frac{\varepsilon}{2}\right].
\end{multline*}
By $\max_{i}\left|\widetilde{w}^{*}_{p,i}\right|=O^{*}_{p}\left(n^{-1}\right)$,
$E_{n}=o_{p}\left(\log\left(n\right)\right)$ and Lemma \ref{lem:O_p_star properties}(i),
one can choose $M>0$ such that the upper limit of the first term
on the right-hand side is bounded by $\varepsilon/2$. By $\Pr_{*}\left[\widetilde{w}^{*}_{p,i}>0,\,\forall i\right]\rightarrow_{p}1$,
the limit of the second term is zero. This gives $R^{*}_{+}(\varDelta^{*}_{+}\left(\tau\right),\tau)=O^{*}_{p}(\log\left(n\right)/\sqrt{nh})$,
uniformly in $\tau\in\left[\underline{\tau},\overline{\tau}\right]$.

By using the same arguments, we obtain a bootstrap analogue of \eqref{eq:step1-conclusion}.
By this result, for all $M,\eta,\varepsilon>0$, 
\begin{eqnarray}
 &  & \Pr\left[\mathrm{Pr}_{*}\left[\inf_{\|\varDelta\|=M,\,\tau\in\left[\underline{\tau},\overline{\tau}\right]}-\varDelta^{\top}R^{*}_{+}(\varDelta,\tau)<M\eta\right]>\varepsilon\right]\nonumber \\
 & \leq & \Pr\left[\mathrm{Pr}_{*}\left[\sup_{\|\varDelta\|=M,\,\tau\in\left[\underline{\tau},\overline{\tau}\right]}\|R^{*}_{+}(\varDelta,\tau)+D_{+}(\tau)\varDelta-R^{*}_{+}(0,\tau)\|\geq\eta\right]>\frac{\varepsilon}{2}\right]\nonumber \\
 &  & +\Pr\left[\mathrm{Pr}_{*}\left[\sup_{\tau\in\left[\underline{\tau},\overline{\tau}\right]}\|R^{*}_{+}(0,\tau)\|>M\left\{ \inf_{\tau\in\left[\underline{\tau},\overline{\tau}\right]}\mathrm{mineig}\left(D_{+}(\tau)\right)\right\} -2\eta\right]>\frac{\varepsilon}{2}\right].\label{eq:step1-conclusion-1}
\end{eqnarray}
In the proof of Lemma \ref{lem:bootstrap analogue 1}, we obtained
$\left(nh\right)^{-1/2}\sum_{i}W^{*}_{p,i}\xi^{*}_{i}=O^{*}_{p}\left(1\right)$.
It follows from this result, \eqref{eq:U_star bounded in probability},
and Lemma \ref{lem:bootstrap analogue 2} that $\sup_{\tau\in\left[\underline{\tau},\overline{\tau}\right]}\|R^{*}_{+}(0,\tau)\|=O^{*}_{p}\left(1\right)$.
Now, by Lemmas \ref{lem:bootstrap analogue 1} and \ref{lem:bootstrap analogue 2},
for all $\eta,\varepsilon>0$, there exists an $M>0$ such that both
terms on the right-hand side of \eqref{eq:step1-conclusion-1} are
bounded by $\varepsilon/2$ for all sufficiently large $n$. Therefore,
for all $\eta,\varepsilon>0$, there exists an $M>0$ such that 
\begin{equation}
\limsup_{n\uparrow\infty}\Pr\left[\mathrm{Pr}_{*}\left[\inf_{\|\varDelta\|=M,\,\tau\in\left[\underline{\tau},\overline{\tau}\right]}-\varDelta^{\top}R^{*}_{+}(\varDelta,\tau)<M\eta\right]>\varepsilon\right]<\varepsilon.\label{eq:step1-final-bound bootstrap}
\end{equation}
Note that when $\widetilde{w}^{*}_{p,i}>0$ for all $i$, the bootstrap
analogue of \eqref{eq:R_+ inf bound} holds. By the union bound, for
all $M,\eta,\varepsilon>0$, 
\begin{multline*}
\Pr\left[\mathrm{Pr}_{*}\left[\inf_{\|\varDelta\|\geq M,\,\tau\in\left[\underline{\tau},\overline{\tau}\right]}\|R^{*}_{+}(\varDelta,\tau)\|<\eta\right]>\varepsilon\right]\\
\leq\Pr\left[\mathrm{Pr}_{*}\left[\inf_{\|\varDelta\|=M,\,\tau\in\left[\underline{\tau},\overline{\tau}\right]}-\varDelta^{\top}R^{*}_{+}(\varDelta,\tau)<M\eta\right]>\frac{\varepsilon}{2}\right]+\Pr\left[\mathrm{Pr}_{*}\left[\exists i:\widetilde{w}^{*}_{p,i}\leq0\right]>\frac{\varepsilon}{2}\right].
\end{multline*}
It follows from this result, $\Pr_{*}\left[\widetilde{w}^{*}_{p,i}>0,\,\forall i\right]\rightarrow_{p}1$,
and \eqref{eq:step1-final-bound bootstrap} that for all $\varepsilon,\eta>0$,
there exists $M>0$ such that 
\begin{equation}
\limsup_{n\uparrow\infty}\Pr\left[\mathrm{Pr}_{*}\left[\inf_{\|\varDelta\|\geq M,\,\tau\in\left[\underline{\tau},\overline{\tau}\right]}\|R^{*}_{+}(\varDelta,\tau)\|<\eta\right]>\varepsilon\right]<\varepsilon.\label{eq:step2-final-1}
\end{equation}
By the bootstrap analogue of \eqref{eq:step3}, \eqref{eq:step2-final-1}
and $\sup_{\tau\in\left[\underline{\tau},\overline{\tau}\right]}\|R^{*}_{+}(\varDelta^{*}_{+}\left(\tau\right),\tau)\|=o^{*}_{p}\left(1\right)$,
we have $\varDelta^{*}_{+}\left(\tau\right)=O^{*}_{p}(1)$, uniformly
in $\tau\in\left[\underline{\tau},\overline{\tau}\right]$.

By the bootstrap analogue of \eqref{eq:step3-final} and Lemma \ref{lem:bootstrap analogue 1},
for all $M,\eta,\varepsilon>0$,
\begin{multline*}
\limsup_{n\uparrow\infty}\Pr\left[\mathrm{Pr}_{*}\left[\sup_{\tau\in\left[\underline{\tau},\overline{\tau}\right]}\left\Vert R^{*}_{+}(\varDelta^{*}_{+}\left(\tau\right),\tau)+D_{+}(\tau)\varDelta^{*}_{+}\left(\tau\right)-R^{*}_{+}(0,\tau)\right\Vert >\eta\right]>\varepsilon\right]\\
\leq\limsup_{n\uparrow\infty}\Pr\left[\mathrm{Pr}_{*}\left[\sup_{\tau\in\left[\underline{\tau},\overline{\tau}\right]}\left\Vert \varDelta^{*}_{+}\left(\tau\right)\right\Vert >M\right]>\frac{\varepsilon}{2}\right].
\end{multline*}
By $\sup_{\tau\in\left[\underline{\tau},\overline{\tau}\right]}\left\Vert \varDelta^{*}_{+}\left(\tau\right)\right\Vert =O^{*}_{p}(1)$,
one can choose $M>0$ such that the right-hand side is bounded by
$\varepsilon/2$. Now it follows that
\[
R^{*}_{+}(\varDelta^{*}_{+}\left(\tau\right),\tau)+D_{+}(\tau)\varDelta^{*}_{+}\left(\tau\right)-R^{*}_{+}(0,\tau)=o^{*}_{p}\left(1\right),
\]
uniformly in $\tau\in\left[\underline{\tau},\overline{\tau}\right]$.
The conclusion follows from this result and $\sup_{\tau\in\left[\underline{\tau},\overline{\tau}\right]}\|R^{*}_{+}(\varDelta^{*}_{+}\left(\tau\right),\tau)\|=o^{*}_{p}\left(1\right)$. 
\end{proof}

\begin{proof}[Proof of Theorem \ref{thm:QRD bootstrap}]
By Lemmas \ref{lem:bootstrap analogue 2} and \ref{lem:delta linearization bootstrap analogue},
$\mathrm{e}^{\top}_{p+1,1}\varDelta^{*}_{+}\left(\tau\right)=\mathrm{e}^{\top}_{p+1,1}D^{-1}_{+}\left(\tau\right)R^{*}_{+}\left(0,\tau\right)+o^{*}_{p}\left(1\right)$,
uniformly in $\tau\in\left[\underline{\tau},\overline{\tau}\right]$,
since $\sup_{\tau\in\left[\underline{\tau},\overline{\tau}\right]}\left\Vert D^{-1}_{+}\left(\tau\right)\right\Vert <\infty$
by Assumption \ref{assu:QRD smoothness}(iv). By \eqref{eq:PI_star_inv - PI_inv rate}
and $\left(nh\right)^{-1/2}\sum_{i}K_{s,p}\left(X^{*}_{i}/h\right)\xi^{*\top}_{i}=O^{*}_{p}\left(1\right)$,
$s\in\left\{ -,+\right\} $, obtained in the proof of Lemma \ref{lem:bootstrap analogue 1},
we have $\left(nh\right)^{-1/2}\sum_{i}\left(W^{*}_{p,i}-\bar{W}^{*}_{p,i}\right)\xi^{*\top}_{i}=o^{*}_{p}\left(1\right)$.
Substituting the linearization of Lemma \ref{lem:bootstrap analogue 2},
we obtain the bootstrap analogue of \eqref{eq:Delta_+ linearization}.
The analogue for $\mathrm{e}^{\top}_{p+1,1}\varDelta^{*}_{-}\left(\tau\right)$
follows in the same way. Combining these with \eqref{eq:theta_hat_star - theta_hat decomposition}
and \eqref{eq:Delta_+ linearization}, the bias terms cancel, and
the $o_{p}\left(1\right)$ remainders of \eqref{eq:Delta_+ linearization}
and of its counterpart for $\varDelta_{-}\left(\tau\right)$ are $o^{*}_{p}\left(1\right)$
by Lemma \ref{lem:O_p_star properties}(i). Therefore, we have 
\begin{eqnarray}
\sqrt{nh}\left(\widehat{\vartheta}^{*}_{p}\left(\tau\right)-\widehat{\vartheta}^{\mathit{eb}}_{p}\left(\tau\right)\right) & = & \mathbb{Z}^{*}_{p,n}\left(\tau\right)+o^{*}_{p}\left(1\right),\mbox{ where}\nonumber \\
\mathbb{Z}^{*}_{p,n}\left(\tau\right) & \coloneqq & \frac{1}{\sqrt{n}}\sum_{i}\left\{ \frac{\phi^{*}_{i}\left(\tau\right)}{\sqrt{h}}-\frac{1}{n}\sum_{j}\frac{\phi_{j}\left(\tau\right)}{\sqrt{h}}\right\} ,\label{eq:theta_hat_star - theta_hat linearization}
\end{eqnarray}
uniformly in $\tau\in\left[\underline{\tau},\overline{\tau}\right]$,
where $\phi^{*}_{i}\left(\tau\right)$ denotes the bootstrap analogue
of $\phi_{i}\left(\tau\right)$ defined by \eqref{eq:phi_i definition}.
Clearly, $\mathbb{Z}^{*}_{p,n}\left(\tau\right)=\mathbb{G}^{*}_{n}g_{7}\left(\cdot\mid\tau\right)$,
and $\mathbb{Z}^{*}_{p,n}\left(\tau\right)$ is the bootstrap analogue
of the leading term $\mathbb{Z}_{p,n}\left(\tau\right)$ in \eqref{eq:uniform linearization QRD}.
By \citet[Theorem 1]{Mammen1992} and the Cramér--Wold device, 
\begin{equation}
\left(\mathbb{Z}^{*}_{p,n}\left(\tau_{1}\right),\ldots,\mathbb{Z}^{*}_{p,n}\left(\tau_{J}\right)\right)\rightsquigarrow_{*}\left(\mathbb{Z}_{p}\left(\tau_{1}\right),\ldots,\mathbb{Z}_{p}\left(\tau_{J}\right)\right)\mbox{ in }\mathbb{R}^{J},\label{eq:bootstrap weak convergence marginals}
\end{equation}
for all $\tau_{1}<\cdots<\tau_{J}$ in $\left[\underline{\tau},\overline{\tau}\right]$.

Let $\delta\in\left(0,1\right)$. Consider a finite $\delta$-net
$\left\{ \tau_{1},\ldots,\tau_{N_{\delta}}\right\} $ of $\left[\underline{\tau},\overline{\tau}\right]$,
and for any $\tau\in\left[\underline{\tau},\overline{\tau}\right]$,
let $\pi_{\delta}\left(\tau\right)\in\left\{ \tau_{1},\ldots,\tau_{N_{\delta}}\right\} $
be chosen such that $\left|\pi_{\delta}\left(\tau\right)-\tau\right|<\delta$.
Let $\mathbb{Z}^{*}_{p,n,\delta}\left(\tau\right)\coloneqq\mathbb{Z}^{*}_{p,n}\left(\pi_{\delta}\left(\tau\right)\right)$
and $\mathbb{Z}_{p,\delta}\left(\tau\right)\coloneqq\mathbb{Z}_{p}\left(\pi_{\delta}\left(\tau\right)\right)$.
Then, by the triangle inequality,
\begin{multline}
\underset{f\in\mathit{BL}_{1}\left(\ell^{\infty}\left[\underline{\tau},\overline{\tau}\right]\right)}{\sup}\left|\mathrm{E}_{*}\left[f\left(\mathbb{Z}^{*}_{p,n}\right)\right]-\mathrm{E}\left[f\left(\mathbb{Z}_{p}\right)\right]\right|\leq\left|\mathrm{E}_{*}\left[\left\Vert \mathbb{Z}^{*}_{p,n}-\mathbb{Z}^{*}_{p,n,\delta}\right\Vert _{\infty}\land2\right]\right|\\
+\underset{f\in\mathit{BL}_{1}\left(\ell^{\infty}\left[\underline{\tau},\overline{\tau}\right]\right)}{\sup}\left|\mathrm{E}_{*}\left[f\left(\mathbb{Z}^{*}_{p,n,\delta}\right)\right]-\mathrm{E}\left[f\left(\mathbb{Z}_{p,\delta}\right)\right]\right|+\left|\mathrm{E}\left[\left\Vert \mathbb{Z}_{p}-\mathbb{Z}_{p,\delta}\right\Vert _{\infty}\land2\right]\right|.\label{eq:bootstrap weak convergence decompose}
\end{multline}
By \eqref{eq:bootstrap weak convergence marginals} and the continuous
mapping theorem, $\mathbb{Z}^{*}_{p,n,\delta}\rightsquigarrow_{*}\mathbb{Z}_{p,\delta}$
in $\ell^{\infty}\left[\underline{\tau},\overline{\tau}\right]$.
Therefore, the second term on the right-hand side is $o_{p}\left(1\right)$.
Let $\mathcal{T}_{\delta}\coloneqq\left\{ \left(\tau,\tau'\right)\in\left[\underline{\tau},\overline{\tau}\right]^{2}:\left|\tau-\tau'\right|<\delta\right\} $.
For the third term, note that
\[
\left\Vert \mathbb{Z}_{p}-\mathbb{Z}_{p,\delta}\right\Vert _{\infty}\leq\underset{\left(\tau,\tau'\right)\in\mathcal{T}_{\delta}}{\sup}\left|\mathbb{Z}_{p}\left(\tau\right)-\mathbb{Z}_{p}\left(\tau'\right)\right|.
\]
It follows from this result, the DCT and the fact that $\mathbb{Z}_{p}$
concentrates on the set of continuous functions on $\left[\underline{\tau},\overline{\tau}\right]$
with probability one that 
\begin{equation}
\underset{\delta\downarrow0}{\lim}\mathrm{E}\left[\left\Vert \mathbb{Z}_{p}-\mathbb{Z}_{p,\delta}\right\Vert _{\infty}\land2\right]=0.\label{eq:Gaussian discretization}
\end{equation}

For the first term on the right-hand side of \eqref{eq:bootstrap weak convergence decompose},
note that
\[
\left\Vert \mathbb{Z}^{*}_{p,n}-\mathbb{Z}^{*}_{p,n,\delta}\right\Vert _{\infty}\leq\underset{\left(\tau,\tau'\right)\in\mathcal{T}_{\delta}}{\sup}\left|\mathbb{Z}^{*}_{p,n}\left(\tau\right)-\mathbb{Z}^{*}_{p,n}\left(\tau'\right)\right|.
\]
Let $\mathscr{G}_{\delta}\coloneqq\left\{ g_{7}\left(\cdot\mid\tau\right)-g_{7}\left(\cdot\mid\tau'\right):\left(\tau,\tau'\right)\in\mathcal{T}_{\delta}\right\} $.
By \citet[Lemma A.6]{chernozhukov2014gaussian}, $\mathscr{G}_{\delta}$
is VC-type with the envelope $2G_{7}$ and characteristics independent
of $n$. Then, we have
\begin{eqnarray*}
\mathrm{E}_{*}\left[\underset{\left(\tau,\tau'\right)\in\mathcal{T}_{\delta}}{\sup}\left|\mathbb{Z}^{*}_{p,n}\left(\tau\right)-\mathbb{Z}^{*}_{p,n}\left(\tau'\right)\right|\right] & = & \mathrm{E}_{*}\left[\left\Vert \mathbb{G}^{*}_{n}\right\Vert _{\mathscr{G}_{\delta}}\right]\mbox{, and}\\
\underset{f\in\mathscr{G}_{\delta}}{\sup}\mathbb{P}_{n}\left(f-\mathbb{P}_{n}f\right)^{2} & \leq & \underset{\left(\tau,\tau'\right)\in\mathcal{T}_{\delta}}{\sup}\frac{1}{nh}\sum_{i}\left(\phi_{i}\left(\tau\right)-\phi_{i}\left(\tau'\right)\right)^{2}.
\end{eqnarray*}
By \eqref{eq:phi_u inequality} and \eqref{eq:phi_z inequality},
the right-hand side of the above inequality can be written as $\varpi^{2}\left(\delta\right)+\chi_{n}\left(\delta\right)$,
where $\varpi\left(\delta\right)$ is deterministic with $\lim_{\delta\downarrow0}\varpi\left(\delta\right)=0$
and $\chi_{n}\left(\delta\right)=o_{p}\left(1\right)$ for all $\delta>0$.
By arguments similar to those in the proofs of Lemma \ref{lem: reg weights basic}(i)
and (iii), $\mathbb{P}_{n}G^{2}_{7}=c^{2}+o_{p}\left(1\right)$ for
some $c>0$, where $G_{7}$ is defined by \eqref{eq:G_7 definition}.
Let $\delta$ be an arbitrarily small number satisfying $2\varpi^{2}\left(\delta\right)\leq c^{2}/2$.
Then, for any such fixed $\delta$, since $\chi_{n}\left(\delta\right)=o_{p}\left(1\right)$
and $\mathbb{P}_{n}G^{2}_{7}=c^{2}+o_{p}\left(1\right)$, we have
$\left|\chi_{n}\left(\delta\right)\right|\leq\varpi^{2}\left(\delta\right)$
and $c^{2}/2\leq\mathbb{P}_{n}G^{2}_{7}\leq2c^{2}$ wpa1. By this
result, Lemma A.6, and Corollary 5.1 of \citet{chernozhukov2014gaussian}
with $\sigma=\sqrt{2}\varpi\left(\delta\right)$, we have
\begin{eqnarray}
\mathrm{E}_{*}\left[\left\Vert \mathbb{G}^{*}_{n}\right\Vert _{\mathscr{G}_{\delta}}\right] & = & \mathrm{E}_{*}\left[\left\Vert \mathbb{G}^{*}_{n}\right\Vert _{\mathscr{G}_{\delta}}\right]\mathbbm{1}\left(\left|\chi_{n}\left(\delta\right)\right|\leq\varpi^{2}\left(\delta\right),\frac{1}{2}c^{2}\leq\mathbb{P}_{n}G^{2}_{7}\leq2c^{2}\right)+o_{p}\left(1\right)\nonumber \\
 & \apprle & \varpi\left(\delta\right)\sqrt{\log\left(\frac{C}{\varpi\left(\delta\right)}\right)}+\frac{\sqrt{\mathrm{E}_{*}\left[\max_{i}G^{2}_{7}\left(Y^{*}_{i},X^{*}_{i},Z^{*}_{i}\right)\right]}}{\sqrt{n}}\log\left(\frac{C}{\varpi\left(\delta\right)}\right)+o_{p}\left(1\right),\nonumber \\
 &  & \label{eq:Z_star increment bound}
\end{eqnarray}
where $C$ is a positive constant that depends on $c$ and on the
VC characteristics of $\mathscr{G}_{7}$ defined in the proof of Theorem
\ref{thm:sharp QRD}. By arguments similar to those in the proof of
Lemma \ref{lem: reg weights basic}(iv), we have
\begin{equation}
\frac{\sqrt{\mathrm{E}_{*}\left[\max_{i}G^{2}_{7}\left(Y^{*}_{i},X^{*}_{i},Z^{*}_{i}\right)\right]}}{\sqrt{n}}\leq\frac{\max_{i}G_{7}\left(Y_{i},X_{i},Z_{i}\right)}{\sqrt{n}}=o_{p}\left(1\right).\label{eq:max_G7 upper bound}
\end{equation}
Since $\lim_{\delta\downarrow0}\varpi\left(\delta\right)\sqrt{\log\left(C/\varpi\left(\delta\right)\right)}=0$
and \eqref{eq:Gaussian discretization}, for any $\varepsilon>0$,
one can choose some sufficiently small $\delta$ such that the sum
of the right-hand side of \eqref{eq:Z_star increment bound} and $\left|\mathrm{E}\left[\left\Vert \mathbb{Z}_{p}-\mathbb{Z}_{p,\delta}\right\Vert _{\infty}\land2\right]\right|$
is bounded by $\varepsilon/2$ wpa1. Since the second term on the
right-hand side of \eqref{eq:bootstrap weak convergence decompose}
is $o_{p}\left(1\right)$ for any fixed $\delta$, it follows that
the distance between $\mathbb{Z}^{*}_{p,n}$ and $\mathbb{Z}_{p}$
on the left-hand side of \eqref{eq:bootstrap weak convergence decompose}
is $o_{p}\left(1\right)$.

We have shown that the leading term $\mathbb{Z}^{*}_{p,n}$ in the
linearization \eqref{eq:theta_hat_star - theta_hat linearization}
converges in distribution to $\mathbb{Z}_{p}$, conditionally on the
original data. Note that by the triangle inequality and $\mathbb{Z}^{*}_{p,n}\rightsquigarrow_{*}\mathbb{Z}_{p}$
in $\ell^{\infty}\left[\underline{\tau},\overline{\tau}\right]$,
\begin{multline*}
\underset{f\in\mathit{BL}_{1}\left(\ell^{\infty}\left[\underline{\tau},\overline{\tau}\right]\right)}{\sup}\left|\mathrm{E}_{*}\left[f\left(\sqrt{nh}\left(\widehat{\vartheta}^{*}_{p}\left(\cdot\right)-\widehat{\vartheta}^{\mathit{eb}}_{p}\left(\cdot\right)\right)\right)\right]-\mathrm{E}\left[f\left(\mathbb{Z}_{p}\right)\right]\right|\\
\leq\left|\mathrm{E}_{*}\left[\underset{\tau\in\left[\underline{\tau},\overline{\tau}\right]}{\sup}\left|\sqrt{nh}\left(\widehat{\vartheta}^{*}_{p}\left(\tau\right)-\widehat{\vartheta}^{\mathit{eb}}_{p}\left(\tau\right)\right)-\mathbb{Z}^{*}_{p,n}\left(\tau\right)\right|\land2\right]\right|+o_{p}\left(1\right).
\end{multline*}
The conclusion follows from this result, \eqref{eq:theta_hat_star - theta_hat linearization}
and Lemma \ref{lem:O_p_star properties}(iv).
\end{proof}

\section{Proof of Theorem \ref{thm:QRD varying bandwidth}}

We now sketch the proof of the first part of Theorem \ref{thm:QRD varying bandwidth},
the weak convergence of $\mathbb{Z}^{\mathit{eb}}_{p}\left(\cdot\mid h\left(\cdot\right)\right)$.
The argument extends the proof of Theorem \ref{thm:sharp QRD} by
indexing every sample object by the pair $\left(\tau,h\right)$ and
establishing each intermediate bound uniformly over the rectangle
$\left(\tau,h\right)\in\left[\underline{\tau},\overline{\tau}\right]\times\mathcal{H}$,
where $\mathcal{H}\coloneqq\left[\underline{c}h_{0},\overline{c}h_{0}\right]$.
Assumption \ref{assu:bounded variation} enters through the following
fact: for $g$ of bounded variation supported on $\left[-1,1\right]$,
the function class $\left\{ x\mapsto g\left(x/h\right):h\in\mathcal{H}\right\} $
is VC-type with a constant envelope and VC characteristics depending
only on $g$. We add the bandwidth index $h\in\mathcal{H}$ to the
kernel weights and the Lagrange multiplier and show that the remainder
terms in the statements of Lemmas \ref{lem: reg weights basic} and
\ref{lem:reg weights lambda} are uniform in $h\in\mathcal{H}$. We
then add the index $h\in\mathcal{H}$ to any function class appearing
in the proofs of Lemmas \ref{lem:R_delta linearization}--\ref{lem:delta linearization}.
The proofs of Lemmas \ref{lem:R_delta linearization}--\ref{lem:delta linearization}
extend to uniform-in-bandwidth statements: the remainder terms in
the statements of Lemmas \ref{lem:R_delta linearization}--\ref{lem:delta linearization}
are all uniform in $h\in\mathcal{H}$. These lead to $\mathbb{Z}^{\mathit{eb}}_{p}\left(\tau\mid h\right)=\mathbb{Z}_{p,n}\left(\tau\mid h\right)+o_{p}\left(1\right)$,
uniformly in $\left(\tau,h\right)\in\left[\underline{\tau},\overline{\tau}\right]\times\mathcal{H}$,
where $\phi_{i}\left(\tau\mid h\right)$ denotes $\phi_{i}\left(\tau\right)$
computed at bandwidth $h$ and $\mathbb{Z}_{p,n}\left(\tau\mid h\right)$
denotes $\left(nh\right)^{-1/2}\sum_{i}\phi_{i}\left(\tau\mid h\right)$.
Now since $\left\{ \left(\tau,h\left(\tau\right)\right):\tau\in\left[\underline{\tau},\overline{\tau}\right]\right\} $
is contained in the rectangle $\left[\underline{\tau},\overline{\tau}\right]\times\mathcal{H}$,
by taking $h=h\left(\tau\right)$, we have 
\begin{equation}
\mathbb{Z}^{\mathit{eb}}_{p}\left(\tau\mid h\left(\tau\right)\right)=\frac{1}{\sqrt{nh\left(\tau\right)}}\sum_{i}\phi_{i}\left(\tau\mid h\left(\tau\right)\right)+o_{p}\left(1\right),\label{eq:uniform linearization QRD vb}
\end{equation}
uniformly in $\tau\in\left[\underline{\tau},\overline{\tau}\right]$.

Let $\tilde{g}_{7}\left(\cdot\mid\tau\right)$ denote $g_{7}\left(\cdot\mid\tau\right)$
with $h$ replaced by $h\left(\tau\right)$, so that $\mathbb{G}_{n}\tilde{g}_{7}\left(\cdot\mid\tau\right)$
equals the leading term of (\ref{eq:uniform linearization QRD vb}).
An envelope function proportional to $h^{-1/2}_{0}\mathbbm{1}\left(\left|x\right|\leq\overline{c}h_{0}\right)\left\{ 1+\left\Vert z-\mu_{Z}\left(x\right)\right\Vert \right\} $
satisfies the same moment bounds as $G_{7}$ with $h$ replaced by
$h_{0}$, so the Lindeberg envelope condition follows as in the proof
of Theorem \ref{thm:sharp QRD}, and $\tilde{\mathscr{G}}_{7}\coloneqq\left\{ \tilde{g}_{7}\left(\cdot\mid\tau\right):\tau\in\left[\underline{\tau},\overline{\tau}\right]\right\} $
is VC-type with characteristics independent of $n$.

By the change of variables $x=h_{0}u$ and the DCT, every kernel second
moment that produced $\omega^{0,2}_{p,0}/\varphi$ in the corresponding
calculation in the proof of Theorem \ref{thm:sharp QRD} now converges
to $\Omega_{p}\left(c\left(\tau\right),c\left(\tau'\right)\right)/\varphi$.
Hence $\mathbb{P}\left(\tilde{g}_{7}\left(\cdot\mid\tau\right)\tilde{g}_{7}\left(\cdot\mid\tau'\right)\right)\rightarrow\mathscr{C}^{c}_{\star,p}\left(\tau,\tau'\right)$.

For the equicontinuity condition (19.27) of \citet{VanDerVaart1998},
we decompose $\phi\left(\tau\mid h\left(\tau\right)\right)=\phi_{u}\left(\tau\mid h\left(\tau\right)\right)-\phi_{z}\left(\tau\mid h\left(\tau\right)\right)$.
Then, we have 
\begin{eqnarray}
 &  & \mathrm{E}\left[\left(\frac{1}{\sqrt{h\left(\tau\right)}}\phi_{u}\left(\tau\mid h\left(\tau\right)\right)-\frac{1}{\sqrt{h\left(\tau'\right)}}\phi_{u}\left(\tau'\mid h\left(\tau'\right)\right)\right)^{2}\right]\nonumber \\
 & \apprle & \sum_{s\in\left\{ +,-\right\} }\mathrm{E}\left[\left(\frac{\bar{W}_{s,p}\left(h\left(\tau\right)\right)}{\sqrt{h\left(\tau\right)}\varphi_{s}\left(\tau\right)}U\left(\tau\right)-\frac{\bar{W}_{s,p}\left(h\left(\tau'\right)\right)}{\sqrt{h\left(\tau'\right)}\varphi_{s}\left(\tau'\right)}U\left(\tau'\right)\right)^{2}\right]\nonumber \\
 & \apprle & \sum_{s\in\left\{ +,-\right\} }\left\{ \mathrm{E}\left[\frac{\bar{W}^{2}_{s,p}\left(h\left(\tau\right)\right)}{h\left(\tau\right)\varphi^{2}_{s}\left(\tau\right)}\left(U\left(\tau\right)-U\left(\tau'\right)\right)^{2}\right]+\left(\frac{1}{\varphi_{s}\left(\tau\right)}-\frac{1}{\varphi_{s}\left(\tau'\right)}\right)^{2}\mathrm{E}\left[\frac{\bar{W}^{2}_{s,p}\left(h\left(\tau\right)\right)}{h\left(\tau\right)}U^{2}\left(\tau'\right)\right]\right.\nonumber \\
 &  & \left.+\frac{1}{\varphi^{2}_{s}\left(\tau'\right)}\mathrm{E}\left[\left(\frac{\bar{W}_{s,p}\left(h\left(\tau\right)\right)}{\sqrt{h\left(\tau\right)}}-\frac{\bar{W}_{s,p}\left(h\left(\tau'\right)\right)}{\sqrt{h\left(\tau'\right)}}\right)^{2}U^{2}\left(\tau'\right)\right]\right\} .\label{eq:vary btw decompose}
\end{eqnarray}
By similar arguments, the suprema of the first and second terms on
the right-hand side of the second inequality over $\left|\tau-\tau'\right|<\delta_{n}$
are both $o\left(1\right)$ for every $\delta_{n}\downarrow0$. For
the third term, note that since $U^{2}\left(\tau'\right)<1$, by a
change of variables 
\begin{multline*}
\mathrm{E}\left[\left(\frac{\bar{W}_{+,p}\left(h\left(\tau\right)\right)}{\sqrt{h\left(\tau\right)}}-\frac{\bar{W}_{+,p}\left(h\left(\tau'\right)\right)}{\sqrt{h\left(\tau'\right)}}\right)^{2}U^{2}\left(\tau'\right)\right]\\
\leq\frac{1}{\varphi^{2}}\int^{\overline{x}/h_{0}}_{0}\left\{ c^{-1/2}\left(\tau\right)\mathcal{K}_{+,p,0}\left(\frac{u}{c\left(\tau\right)}\right)-c^{-1/2}\left(\tau'\right)\mathcal{K}_{+,p,0}\left(\frac{u}{c\left(\tau'\right)}\right)\right\} ^{2}f_{X}\left(h_{0}u\right)\mathrm{d}u.
\end{multline*}
Note that $\mathcal{K}_{+,p,0}\left(\cdot\right)$ is also Lipschitz
on $\mathbb{R}$, under the Lipschitz condition in Assumption \ref{assu:bounded variation}
and the bounded support condition in Assumption \ref{assu:kernel}
for $K\left(\cdot\right)$. The right-hand side of the above inequality
can be bounded by the product of an $O\left(1\right)$ term and $\left|\tau-\tau'\right|$,
by the Lipschitz continuity of $\mathcal{K}_{+,p,0}\left(\cdot\right)$
and $c\left(\cdot\right)$. The same result holds when $+$ is replaced
by $-$.

By adding and subtracting $\bar{W}_{s,p}\left(h\left(\tau\right)\right)\xi^{\top}\gamma\left(\tau'\right)/\sqrt{h\left(\tau\right)}$
for $s\in\left\{ +,-\right\} $, we have 
\begin{eqnarray*}
 &  & \mathrm{E}\left[\left(\frac{1}{\sqrt{h\left(\tau\right)}}\phi_{z}\left(\tau\mid h\left(\tau\right)\right)-\frac{1}{\sqrt{h\left(\tau'\right)}}\phi_{z}\left(\tau'\mid h\left(\tau'\right)\right)\right)^{2}\right]\\
 & \apprle & \sum_{s\in\left\{ +,-\right\} }\left\{ \mathrm{E}\left[\frac{\bar{W}^{2}_{s,p}\left(h\left(\tau\right)\right)}{h\left(\tau\right)}\left(\xi^{\top}\left(\gamma\left(\tau\right)-\gamma\left(\tau'\right)\right)\right)^{2}\right]\right.\\
 &  & \left.+\mathrm{E}\left[\left(\frac{\bar{W}_{s,p}\left(h\left(\tau\right)\right)}{\sqrt{h\left(\tau\right)}}-\frac{\bar{W}_{s,p}\left(h\left(\tau'\right)\right)}{\sqrt{h\left(\tau'\right)}}\right)^{2}\left(\xi^{\top}\gamma\left(\tau'\right)\right)^{2}\right]\right\} .
\end{eqnarray*}
By the Cauchy--Schwarz inequality and the bounded conditional second
moment of $\xi$ (Assumption \ref{assu:data generating process}(iv)),
the first term is bounded by the product of an $O\left(1\right)$
term and $\left\Vert \gamma\left(\tau\right)-\gamma\left(\tau'\right)\right\Vert ^{2}$.
As the third term on the right-hand side of the second inequality
in \eqref{eq:vary btw decompose}, the second term admits an upper
bound of the product of an $O\left(1\right)$ term and $\left|\tau-\tau'\right|$.
All hypotheses of \citet[Theorem 19.28]{VanDerVaart1998} are thus
verified: $\mathbb{G}_{n}\tilde{g}_{7}\rightsquigarrow\mathbb{Z}^{c}_{p}$
in $\ell^{\infty}\left[\underline{\tau},\overline{\tau}\right]$,
a centered tight Gaussian process with covariance kernel $\mathscr{C}^{c}_{\star,p}$.
This completes the proof of the first part.

The second part follows from similar uniform-in-bandwidth modifications
to the proof of Theorem \ref{thm:QRD bootstrap}.

\section{Proof of Theorem \ref{thm:kink}}

Let $\widehat{\gamma}^{\mathit{eb}}_{p}$ denote the regression coefficient
of $Z_{i}$ in (\ref{eq:adjusted CCFT kink}). Let 
\begin{equation}
\widehat{\mu}_{+,p}\coloneqq\underset{\beta}{\arg\min}\sum_{i}\widehat{w}_{p,i}K\left(\frac{X_{i}}{h}\right)I_{i}\left\{ \widehat{\epsilon}_{i}-r^{\top}_{p}\left(X_{i}\right)\beta\right\} ^{2},\label{eq:mu_hat_rho definition}
\end{equation}
where $\widehat{\epsilon}_{i}\coloneqq Y_{i}-Z^{\top}_{i}\widehat{\gamma}^{\mathit{eb}}_{p}$.
Let $\widehat{\mu}_{-,p}$ be defined analogously. Clearly, $\widehat{\vartheta}^{\mathit{eb}}_{\dagger,p}=\mathrm{e}^{\top}_{p+1,2}\left(\widehat{\mu}_{+,p}-\widehat{\mu}_{-,p}\right)$.

Let 
\[
W^{\dagger\mathit{eb}}_{+,p,i}\coloneqq\mathrm{e}^{\top}_{p+1,2}\left(\Pi^{\mathit{eb}}_{+,p}\right)^{-1}K_{+,p}\left(\frac{X_{i}}{h}\right)\textrm{ and }\bar{W}^{\dagger}_{+,p,i}\coloneqq\mathrm{e}^{\top}_{p+1,2}\left(\varphi\Lambda_{+,p}\right)^{-1}K_{+,p}\left(\frac{X_{i}}{h}\right),
\]
and let $\left(W^{\dagger\mathit{eb}}_{-,p,i},\bar{W}^{\dagger}_{-,p,i}\right)$
be defined analogously. Denote $W^{\dagger\mathit{eb}}_{p,i}\coloneqq W^{\dagger\mathit{eb}}_{+,p,i}-W^{\dagger\mathit{eb}}_{-,p,i}$
and $\bar{W}^{\dagger}_{p,i}\coloneqq\bar{W}^{\dagger}_{+,p,i}-\bar{W}^{\dagger}_{-,p,i}$. 
\begin{lem}
\label{lem:kink}Suppose that the assumptions in the statement of
Lemma \ref{lem: reg weights basic} hold. The following results hold
for all $k\in\mathbb{N}$. (i) If $\mu_{A}$ is uniformly continuous
on $\mathbb{B}\setminus\left\{ 0\right\} $, 
\[
\mathrm{E}\left[\frac{1}{h}\left(\bar{W}^{\dagger}_{+,p}\right)^{k}A\right]=\frac{\mu_{A,+}\omega^{0,k}_{p,1}}{\varphi^{k-1}}+o\left(1\right)\textrm{ and }\mathrm{E}\left[\frac{1}{h}\bar{W}_{+,p}\bar{W}^{\dagger}_{+,p}A\right]=\frac{\mu_{A,+}\varsigma_{p}}{\varphi}+o\left(1\right),
\]
and 
\[
\mathrm{E}\left[\frac{1}{h}\left(\bar{W}^{\dagger}_{-,p}\right)^{k}A\right]=\frac{\mu_{A,-}\omega^{0,k}_{p,1}}{\varphi^{k-1}}+o\left(1\right)\textrm{ for even }k\textrm{,}\textrm{ and }\mathrm{E}\left[\frac{1}{h}\bar{W}_{-,p}\bar{W}^{\dagger}_{-,p}A\right]=-\frac{\mu_{A,-}\varsigma_{p}}{\varphi}+o\left(1\right),
\]
where the sign flip follows from $\int^{0}_{-1}\mathcal{K}_{-,p,0}\left(t\right)\mathcal{K}_{-,p,1}\left(t\right)\mathrm{d}t=-\varsigma_{p}$.
(ii) If $\mu_{A}$ is $\left(p+1\right)$-times continuously differentiable
with uniformly continuous $\mu^{\left(p+1\right)}_{A}$ on $\mathbb{B}\setminus\left\{ 0\right\} $,
\[
\frac{1}{h^{2}}\sum_{i}\widehat{w}_{p,i}W^{\dagger\mathit{eb}}_{+,p,i}\mu_{A}\left(X_{i}\right)=\mu^{\left(1\right)}_{A,+}+\frac{\mu^{\left(p+1\right)}_{A,+}}{\left(p+1\right)!}\omega^{p+1,1}_{+,p,1}h^{p}+o_{p}\left(h^{p}\right).
\]
(iii) If for some $\delta>0$, $\mu_{\left|A\right|^{1+\delta}}$
is bounded on $\mathbb{B}\setminus\left\{ 0\right\} $, 
\[
\frac{1}{nh}\sum_{i}W_{+,p,i}W^{\dagger\mathit{eb}}_{+,p,i}A_{i}-\mathrm{E}\left[\frac{1}{h}\bar{W}_{+,p}\bar{W}^{\dagger}_{+,p}A\right]=o_{p}\left(1\right).
\]
The results in (ii) and (iii) with ``$+$'' replaced by ``$-$''
also hold. 
\end{lem}
\begin{proof}[Proof of Lemma \ref{lem:kink}]
Part (i) follows from the same arguments used to prove Lemma \ref{lem: reg weights basic}(i).
For Part (ii), by (\ref{eq:mu Taylor}), 
\begin{equation}
\frac{1}{h^{2}}\sum_{i}\widehat{w}_{p,i}W^{\dagger\mathit{eb}}_{+,p,i}\mu_{A}\left(X_{i}\right)=\frac{1}{h^{2}}\sum_{i}\widehat{w}_{p,i}W^{\dagger\mathit{eb}}_{+,p,i}\left(r^{\top}_{p}\left(X_{i}\right)\mu_{+}\right)+\frac{1}{h^{2}}\sum_{i}\widehat{w}_{p,i}W^{\dagger\mathit{eb}}_{+,p,i}\frac{\mu^{\left(p+1\right)}_{A}\left(\dot{t}X_{i}\right)}{\left(p+1\right)!}X^{p+1}_{i}.\label{eq:smoothing bias decompose 1 ted}
\end{equation}
By the definition of $W^{\dagger\mathit{eb}}_{+,p,i}$, $h^{-2}\sum_{i}\widehat{w}_{p,i}W^{\dagger\mathit{eb}}_{+,p,i}\left(r^{\top}_{p}\left(X_{i}\right)\mu_{+}\right)=\mu^{(1)}_{A,+}$.
Write 
\begin{eqnarray}
\frac{1}{h^{2}}\sum_{i}\widehat{w}_{p,i}W^{\dagger\mathit{eb}}_{+,p,i}\frac{\mu^{\left(p+1\right)}_{A}\left(\dot{t}X_{i}\right)}{\left(p+1\right)!}X^{p+1}_{i} & = & \frac{1}{h}\sum_{i}\widehat{w}_{p,i}W^{\dagger\mathit{eb}}_{+,p,i}\frac{\left(\mu^{\left(p+1\right)}_{A}\left(\dot{t}X_{i}\right)-\mu^{\left(p+1\right)}_{A,+}\right)}{\left(p+1\right)!}\left(\frac{X_{i}}{h}\right)^{p+1}h^{p}\nonumber \\
 &  & +\left(\frac{\mu^{\left(p+1\right)}_{A,+}}{\left(p+1\right)!}\right)\left(\frac{1}{h}\sum_{i}\widehat{w}_{p,i}W^{\dagger\mathit{eb}}_{+,p,i}\left(\frac{X_{i}}{h}\right)^{p+1}\right)h^{p}.\label{eq:smoothing bias decompose 2 ted}
\end{eqnarray}
By (\ref{eq:w_hat Taylor expansion}), Lemma \ref{lem:reg weights lambda},
and arguments similar to those used in the proof of Lemma \ref{lem: reg weights basic}(ii),
one can show that $h^{-1}\sum_{i}\widehat{w}_{p,i}W^{\dagger\mathit{eb}}_{+,p,i}\left(X_{i}/h\right)^{p+1}=\omega^{p+1,1}_{+,p,1}+o_{p}\left(1\right)$
and the first term on the right-hand side of (\ref{eq:smoothing bias decompose 2 ted})
is $o_{p}\left(h^{p}\right)$.

Part (iii) follows from (\ref{eq:PI_hat_inv - PI_inv rate}), \eqref{eq:PI_eb_inv}
and arguments similar to those in the proof of Lemma \ref{lem: reg weights basic}(iii). 
\end{proof}

\begin{proof}[Proof of Theorem \ref{thm:kink}]
By the partitioned regression argument as in CCFT, 
\begin{multline}
\widehat{\gamma}^{\mathit{eb}}_{p}=\\
\left\{ \frac{1}{h}\sum_{i}\widehat{w}_{p,i}K\left(\frac{X_{i}}{h}\right)Z_{i}Z^{\top}_{i}-\sum_{s\in\{+,-\}}\Bigl(\sum_{i}\frac{1}{h}\widehat{w}_{p,i}Z_{i}K^{\top}_{s,p}\left(\frac{X_{i}}{h}\right)\Bigr)\left(\Pi^{\mathit{eb}}_{s,p}\right)^{-1}\Bigl(\sum_{i}\frac{1}{h}\widehat{w}_{p,i}K_{s,p}\left(\frac{X_{i}}{h}\right)Z^{\top}_{i}\Bigr)\right\} ^{-1}\\
\times\left\{ \frac{1}{h}\sum_{i}\widehat{w}_{p,i}K\left(\frac{X_{i}}{h}\right)Z_{i}Y_{i}-\sum_{s\in\{+,-\}}\Bigl(\sum_{i}\frac{1}{h}\widehat{w}_{p,i}Z_{i}K^{\top}_{s,p}\left(\frac{X_{i}}{h}\right)\Bigr)\left(\Pi^{\mathit{eb}}_{s,p}\right)^{-1}\Bigl(\sum_{i}\frac{1}{h}\widehat{w}_{p,i}K_{s,p}\left(\frac{X_{i}}{h}\right)Y_{i}\Bigr)\right\} .\label{eq:gamma_hat FWL}
\end{multline}
By \eqref{eq:w_hat Taylor expansion}, Lemma \ref{lem: reg weights basic}(iv)
and Lemma \ref{lem:reg weights lambda}, we have $\max_{i}\left|n\widehat{w}_{p,i}-1\right|=o_{p}\left(1\right)$.
By this result, 
\[
\frac{1}{h}\sum_{i}\widehat{w}_{p,i}K\left(\frac{X_{i}}{h}\right)Z_{i}Z^{\top}_{i}=\frac{1}{nh}\sum_{i}K\left(\frac{X_{i}}{h}\right)Z_{i}Z^{\top}_{i}+o_{p}\left(1\right).
\]
By the von Bahr--Esseen inequality, the probability limit of the
leading term on the right-hand side is $\mu_{ZZ^{\top},\pm}\varphi/2$.
By similar arguments, we derive the probability limits of all other
terms on the right-hand side of \eqref{eq:gamma_hat FWL}. Then, we
have $\widehat{\gamma}^{\mathit{eb}}_{p}\rightarrow_{p}\gamma$. Let
$\epsilon_{i}\coloneqq Y_{i}-Z^{\top}_{i}\gamma$. By using the first-order
conditions for (\ref{eq:mu_hat_rho definition}), we can write 
\begin{equation}
\widehat{\vartheta}^{\mathit{eb}}_{\dagger,p}=\frac{1}{h^{2}}\sum_{i}\widehat{w}_{p,i}W^{\dagger\mathit{eb}}_{p,i}\epsilon_{i}-\left(\frac{1}{h^{2}}\sum_{i}\widehat{w}_{p,i}W^{\dagger\mathit{eb}}_{p,i}Z^{\top}_{i}\right)\left(\widehat{\gamma}^{\mathit{eb}}_{p}-\gamma\right).\label{eq:theta_hat_dag decompose}
\end{equation}
By \eqref{eq:w_hat Taylor expansion}, Lemma \ref{lem:reg weights lambda},
and Chebyshev's inequality, 
\[
\frac{1}{h^{2}}\sum_{i}\widehat{w}_{p,i}K_{+,p}\left(\frac{X_{i}}{h}\right)\xi_{i}=O_{p}\left(\left(nh^{3}\right)^{-1/2}\right),
\]
where $\xi_{i}\coloneqq Z_{i}-\mu_{Z}\left(X_{i}\right)$. Then, by
this result, (\ref{eq:PI_hat_inv - PI_inv rate}), and \eqref{eq:PI_eb_inv}
\[
\left|\frac{1}{h^{2}}\sum_{i}\widehat{w}_{p,i}W^{\dagger\mathit{eb}}_{+,p,i}\xi_{i}\right|\leq\left\Vert \left(\Pi^{\mathit{eb}}_{+,p}\right)^{-1}\right\Vert \left\Vert \frac{1}{h^{2}}\sum_{i}\widehat{w}_{p,i}K_{+,p}\left(\frac{X_{i}}{h}\right)\xi_{i}\right\Vert =O_{p}\left(\left(nh^{3}\right)^{-1/2}\right).
\]
By Lemma \ref{lem:kink}(ii), the balancing condition $\mu^{\left(1\right)}_{Z,+}=\mu^{\left(1\right)}_{Z,-}$,
and the above result, the coefficients of $\widehat{\gamma}^{\mathit{eb}}_{p}-\gamma$
in \eqref{eq:theta_hat_dag decompose} are $O_{p}\left(\left(nh^{3}\right)^{-1/2}\right)$.
It follows that the second term on the right-hand side of (\ref{eq:theta_hat_dag decompose})
is $o_{p}\left(\left(nh^{3}\right)^{-1/2}\right)$. By Lemma \ref{lem:kink}(ii),
\[
\frac{1}{h^{2}}\sum_{i}\widehat{w}_{p,i}W^{\dagger\mathit{eb}}_{+,p,i}\epsilon_{i}=\mu^{\left(1\right)}_{\epsilon,+}+\frac{\mu^{\left(p+1\right)}_{\epsilon,+}}{\left(p+1\right)!}\omega^{p+1,1}_{+,p,1}h^{p}+\frac{1}{h^{2}}\sum_{i}\widehat{w}_{p,i}W^{\dagger\mathit{eb}}_{+,p,i}\zeta_{i}+o_{p}\left(\left(nh^{3}\right)^{-1/2}\right),
\]
where $\zeta_{i}\coloneqq\epsilon_{i}-\mu_{\epsilon}\left(X_{i}\right)$.
A similar result with $+$ replaced by $-$ also holds. It now follows
from the above result, (\ref{eq:theta_hat_dag decompose}) and the
fact that $\mu^{\left(1\right)}_{\epsilon,+}-\mu^{\left(1\right)}_{\epsilon,-}=\vartheta_{\dagger}$
that 
\begin{equation}
\widehat{\vartheta}^{\mathit{eb}}_{\dagger,p}=\vartheta_{\dagger}+\mathscr{B}_{\dagger,p}h^{p}+\frac{1}{h^{2}}\sum_{i}\widehat{w}_{p,i}W^{\dagger\mathit{eb}}_{p,i}\zeta_{i}+o_{p}\left(\left(nh^{3}\right)^{-1/2}\right).\label{eq:kink leading appr 1}
\end{equation}
By (\ref{eq:w_hat Taylor expansion}) and Lemma \ref{lem:reg weights lambda},
\begin{eqnarray}
\frac{1}{h^{2}}\sum_{i}\widehat{w}_{p,i}W^{\dagger\mathit{eb}}_{p,i}\zeta_{i} & = & \frac{1}{nh^{2}}\sum_{i}\left\{ \iota_{\varrho}\left(1\right)+\iota_{\varrho}'\left(1\right)\left(V^{\top}_{p,i}\widehat{\lambda}_{p}\right)\right\} W^{\dagger\mathit{eb}}_{p,i}\zeta_{i}+o_{p}\left(\left(nh^{3}\right)^{-1/2}\right)\nonumber \\
 & = & \frac{1}{nh^{2}}\sum_{i}W^{\dagger\mathit{eb}}_{p,i}\zeta_{i}-\left(\frac{1}{nh}\sum_{i}W^{\dagger\mathit{eb}}_{p,i}\zeta_{i}V^{\top}_{p,i}\right)\left(\frac{1}{nh}\sum_{i}V_{p,i}V^{\top}_{p,i}\right)^{-1}\left(\frac{1}{nh^{2}}\sum_{i}V_{p,i}\right)\nonumber \\
 &  & +o_{p}\left(\left(nh^{3}\right)^{-1/2}\right).\label{eq:kink leading appr 2}
\end{eqnarray}
By Lemma \ref{lem:kink}(i, iii) applied coordinatewise, 
\begin{eqnarray*}
\frac{1}{nh}\sum_{i}W^{\dagger\mathit{eb}}_{p,i}\zeta_{i}V^{\top}_{p,i} & = & \frac{1}{nh}\sum_{i}\left(W_{+,p,i}W^{\dagger\mathit{eb}}_{+,p,i}+W_{-,p,i}W^{\dagger\mathit{eb}}_{-,p,i}\right)\zeta_{i}\bar{Z}^{\top}_{i}\\
 & = & \left(\frac{\varsigma_{p}}{\varphi}\right)\left(\mu_{\zeta\bar{Z}^{\top},+}-\mu_{\zeta\bar{Z}^{\top},-}\right)+o_{p}\left(1\right).
\end{eqnarray*}
By this result and \eqref{eq:VV limit}, 
\[
\left(\frac{1}{nh}\sum_{i}W^{\dagger\mathit{eb}}_{p,i}\zeta_{i}V^{\top}_{p,i}\right)\left(\frac{1}{nh}\sum_{i}V_{p,i}V^{\top}_{p,i}\right)^{-1}=\left(\frac{\varsigma_{p}}{\omega^{0,2}_{p,0}}\right)\left(\mu_{\zeta\bar{Z}^{\top},+}-\mu_{\zeta\bar{Z}^{\top},-}\right)\mu^{-1}_{\bar{Z}\bar{Z}^{\top},\pm}+o_{p}\left(1\right).
\]
By the above result and (\ref{eq:V_average rate}), 
\begin{eqnarray}
 &  & \left(\frac{1}{nh}\sum_{i}W^{\dagger\mathit{eb}}_{p,i}\zeta_{i}V^{\top}_{p,i}\right)\left(\frac{1}{nh}\sum_{i}V_{p,i}V^{\top}_{p,i}\right)^{-1}\left(\frac{1}{nh^{2}}\sum_{i}V_{p,i}\right)\nonumber \\
 & = & \left(\frac{\varsigma_{p}}{\omega^{0,2}_{p,0}}\right)\left(\mu_{\zeta\bar{Z}^{\top},+}-\mu_{\zeta\bar{Z}^{\top},-}\right)\mu^{-1}_{\bar{Z}\bar{Z}^{\top},\pm}\left(\frac{1}{nh^{2}}\sum_{i}V_{p,i}\right)+o_{p}\left(\left(nh^{3}\right)^{-1/2}\right).\label{eq:kink leading appr}
\end{eqnarray}
Since $\mathrm{E}\left[\zeta\mid X\right]=0$ and $\bar{Z}\coloneqq\left(1,Z^{\top}\right)^{\top}$,
we have $\mu_{\zeta\bar{Z}^{\top},s}=\left(0,\mathrm{Cov}_{s}\left[Z,\epsilon\right]^{\top}\right)$
for $s\in\left\{ -,+\right\} $. Combined with (\ref{eq:V decompose})
and (\ref{eq:block inversion}), the leading term equals 
\[
\left(\frac{\varsigma_{p}}{\omega^{0,2}_{p,0}}\right)\left(\mathrm{Cov}_{+}\left[Z,\epsilon\right]-\mathrm{Cov}_{-}\left[Z,\epsilon\right]\right)^{\top}\left(\mathrm{Var}_{\pm}\left[Z\right]\right)^{-1}\frac{1}{nh^{2}}\sum_{i}W_{p,i}Z_{i}=\frac{1}{nh^{2}}\sum_{i}W_{p,i}Z^{\top}_{i}\gamma_{\dagger,p},
\]
where the last equality uses the definition of $\gamma_{\dagger,p}$.
Recall the residualization $\xi_{i}\coloneqq Z_{i}-\mu_{Z}\left(X_{i}\right)$,
so that $Z_{i}=\mu_{Z}(X_{i})+\xi_{i}$. By (\ref{eq:PSI_Z_hat decompose}),
\[
\frac{1}{nh^{2}}\sum_{i}W_{p,i}Z^{\top}_{i}\gamma_{\dagger,p}=\frac{1}{nh^{2}}\sum_{i}W_{p,i}\xi^{\top}_{i}\gamma_{\dagger,p}+\frac{\left(\mu^{\left(p+1\right)}_{Z,+}\omega^{p+1,1}_{+,p,0}-\mu^{\left(p+1\right)}_{Z,-}\omega^{p+1,1}_{-,p,0}\right)^{\top}\gamma_{\dagger,p}}{\left(p+1\right)!}\cdot h^{p}+o_{p}\left(h^{p}\right),
\]
which converts the $\mathscr{B}_{\dagger,p}$ bias into $\mathscr{B}_{\star\dagger,p}$.
Then, 
\begin{eqnarray}
\widehat{\vartheta}^{\mathit{eb}}_{\dagger,p} & = & \vartheta_{\dagger}+\mathscr{B}_{\dagger,p}h^{p}+\frac{1}{nh^{2}}\sum_{i}W^{\dagger\mathit{eb}}_{p,i}\zeta_{i}-\frac{1}{nh^{2}}\sum_{i}W_{p,i}Z^{\top}_{i}\gamma_{\dagger,p}+o_{p}\left(\left(nh^{3}\right)^{-1/2}\right)\nonumber \\
 & = & \vartheta_{\dagger}+\mathscr{B}_{\star\dagger,p}h^{p}+\frac{1}{nh^{2}}\sum_{i}\left(W^{\dagger\mathit{eb}}_{p,i}\zeta_{i}-W_{p,i}\xi^{\top}_{i}\gamma_{\dagger,p}\right)+o_{p}\left(\left(nh^{3}\right)^{-1/2}\right).\label{eq:theta_hat_dag decompose 2}
\end{eqnarray}
By (\ref{eq:PI_hat_inv - PI_inv rate}), (\ref{eq:PI_eb_inv}), and
Chebyshev's inequality, we have 
\begin{eqnarray*}
\frac{1}{\sqrt{nh}}\sum_{i}\left(W^{\dagger\mathit{eb}}_{p,i}\zeta_{i}-W_{p,i}\xi^{\top}_{i}\gamma_{\dagger,p}\right) & = & \frac{1}{\sqrt{nh}}\sum_{i}\left(W^{\dagger}_{p,i}\zeta_{i}-W_{p,i}\xi^{\top}_{i}\gamma_{\dagger,p}\right)+o_{p}(1)\\
 & = & \frac{1}{\sqrt{nh}}\sum_{i}\left(\bar{W}^{\dagger}_{p,i}\zeta_{i}-\bar{W}_{p,i}\xi^{\top}_{i}\gamma_{\dagger,p}\right)+o_{p}(1).
\end{eqnarray*}
By this result and (\ref{eq:theta_hat_dag decompose 2}), 
\[
\sqrt{nh^{3}}\left(\widehat{\vartheta}^{\mathit{eb}}_{\dagger,p}-\vartheta_{\dagger}-\mathscr{B}_{\star\dagger,p}h^{p}\right)=\frac{1}{\sqrt{nh}}\sum_{i}\left(\bar{W}^{\dagger}_{p,i}\zeta_{i}-\bar{W}_{p,i}\xi^{\top}_{i}\gamma_{\dagger,p}\right)+o_{p}(1).
\]

Since $\mathrm{E}\left[\zeta\mid X\right]=0$ and $\mathrm{E}\left[\xi\mid X\right]=0$,
while $\bar{W}^{\dagger}_{p}$ and $\bar{W}_{p}$ are functions of
$X$ alone, each summand in the preceding display has mean zero by
the LIE. As the summands are i.i.d., the variance of their normalized
sum equals the expectation of the squared summand divided by $h$,
which expands as 
\begin{eqnarray*}
\mathrm{Var}\left[\frac{1}{\sqrt{nh}}\sum_{i}\left(\bar{W}^{\dagger}_{p,i}\zeta_{i}-\bar{W}_{p,i}\xi^{\top}_{i}\gamma_{\dagger,p}\right)\right] & = & \mathrm{E}\left[\frac{1}{h}\left(\bar{W}^{\dagger}_{p}\right)^{2}\zeta^{2}\right]-2\cdot\mathrm{E}\left[\frac{1}{h}\bar{W}^{\dagger}_{p}\bar{W}_{p}\zeta\xi^{\top}\gamma_{\dagger,p}\right]\\
 &  & +\mathrm{E}\left[\frac{1}{h}\bar{W}^{2}_{p}\left(\xi^{\top}\gamma_{\dagger,p}\right)^{2}\right].
\end{eqnarray*}
We evaluate the three pieces. By Lemma \ref{lem:kink}(i) applied
with $A=\zeta^{2}$ (so that $\mu_{\zeta^{2},\pm}=\sigma^{2}$) on
each side and summed, 
\[
\mathrm{E}\left[\frac{1}{h}\left(\bar{W}^{\dagger}_{p}\right)^{2}\zeta^{2}\right]=\frac{\omega^{0,2}_{p,1}\sigma^{2}}{\varphi}+o\left(1\right).
\]
By Lemma \ref{lem: reg weights basic}(i) applied with $A=\left(\xi^{\top}\gamma_{\dagger,p}\right)^{2}$,
\[
\mathrm{E}\left[\frac{1}{h}\bar{W}^{2}_{p}\left(\xi^{\top}\gamma_{\dagger,p}\right)^{2}\right]=\frac{\omega^{0,2}_{p,0}\gamma^{\top}_{\dagger,p}\mathrm{Var}_{\pm}\left[Z\right]\gamma_{\dagger,p}}{\varphi}+o\left(1\right).
\]
For the cross-moment, by Lemma \ref{lem:kink}(i) applied with $A=\zeta\xi^{\top}\gamma_{\dagger,p}$
(so that $\mu_{A,\pm}=\mathrm{Cov}_{\pm}\left[Z,\epsilon\right]^{\top}\gamma_{\dagger,p}$),
\[
\mathrm{E}\left[\frac{1}{h}\bar{W}^{\dagger}_{p}\bar{W}_{p}\zeta\xi^{\top}\gamma_{\dagger,p}\right]=\left(\frac{\varsigma_{p}}{\varphi}\right)\left(\mathrm{Cov}_{+}\left[Z,\epsilon\right]-\mathrm{Cov}_{-}\left[Z,\epsilon\right]\right)^{\top}\gamma_{\dagger,p}+o\left(1\right).
\]
By the definition of $\gamma_{\dagger,p}$, $\mathrm{Cov}_{+}[Z,\epsilon]-\mathrm{Cov}_{-}[Z,\epsilon]=(\omega^{0,2}_{p,0}/\varsigma_{p})\mathrm{Var}_{\pm}[Z]\gamma_{\dagger,p}$;
hence 
\[
2\cdot\mathrm{E}\left[\frac{1}{h}\bar{W}^{\dagger}_{p}\bar{W}_{p}\zeta\xi^{\top}\gamma_{\dagger,p}\right]=\frac{2\omega^{0,2}_{p,0}\gamma^{\top}_{\dagger,p}\mathrm{Var}_{\pm}\left[Z\right]\gamma_{\dagger,p}}{\varphi}+o\left(1\right).
\]
Collecting terms, we have 
\begin{equation}
\mathrm{Var}\left[\frac{1}{\sqrt{nh}}\sum_{i}\left(\bar{W}^{\dagger}_{p,i}\zeta_{i}-\bar{W}_{p,i}\xi^{\top}_{i}\gamma_{\dagger,p}\right)\right]=\mathscr{V}_{\star\dagger,p}+o\left(1\right).\label{eq:variance limit}
\end{equation}
Now we show that $\mathscr{V}_{\star\dagger,p}$ is strictly positive.
To see this, note that $\varphi\mathscr{V}_{\star\dagger,p}$ is the
sum over $s$ of the integrals of $\mathrm{Var}_{s}\left[\mathcal{K}_{s,p,1}\left(t\right)\epsilon-\mathcal{K}_{s,p,0}\left(t\right)Z^{\top}\gamma_{\dagger,p}\right]$
over $\left[0,1\right]$ for $s=+$ and over $\left[-1,0\right]$
for $s=-$. The integrand can be written as $a^{\top}_{s}\left(t\right)\mathrm{Var}_{s}\left[B\right]a_{s}\left(t\right)$
for $a_{s}\left(t\right)\coloneqq\left(\begin{array}{cc}
\mathcal{K}_{s,p,1}\left(t\right) & -\mathcal{K}_{s,p,1}\left(t\right)\gamma^{\top}-\mathcal{K}_{s,p,0}\left(t\right)\gamma^{\top}_{\dagger,p}\end{array}\right)^{\top}$, which is bounded below by $\mathrm{mineig}\left(\mathrm{Var}_{s}\left[B\right]\right)\mathcal{K}^{2}_{s,p,1}\left(t\right)$.
It follows from Assumption \ref{assu:data generating process}(iii)
that $\mathscr{V}_{\star\dagger,p}>0$. Let $\delta>0$ be the moment
exponent in Assumption \ref{assu:data generating process}(iv). By
the $c_{r}$ inequality, 
\begin{eqnarray*}
\sum_{i}\mathrm{E}\left[\left|\frac{1}{\sqrt{nh}}\left(\bar{W}^{\dagger}_{p,i}\zeta_{i}-\bar{W}_{p,i}\xi^{\top}_{i}\gamma_{\dagger,p}\right)\right|^{2+\delta}\right] & = & \frac{1}{\left(nh\right)^{1+\delta/2}}\sum_{i}\mathrm{E}\left[\left|\bar{W}^{\dagger}_{p,i}\zeta_{i}-\bar{W}_{p,i}\xi^{\top}_{i}\gamma_{\dagger,p}\right|^{2+\delta}\right]\\
 & \apprle & \frac{\mathrm{E}\left[h^{-1}\left|\bar{W}^{\dagger}_{p}\right|^{2+\delta}\left|\zeta\right|^{2+\delta}\right]+\mathrm{E}\left[h^{-1}\left|\bar{W}_{p}\right|^{2+\delta}\left|\xi^{\top}\gamma_{\dagger,p}\right|^{2+\delta}\right]}{\left(nh\right)^{\delta/2}}.
\end{eqnarray*}
By a change of variables and the envelope condition in Assumption
\ref{assu:data generating process}(iv), the two expectations in the
numerator are $O\left(1\right)$. By this result, \eqref{eq:variance limit}
and $\mathscr{V}_{\star\dagger,p}>0$, Lyapunov's condition 
\[
\frac{\sum_{i}\mathrm{E}\left[\left|\left(nh\right)^{-1/2}\left(\bar{W}^{\dagger}_{p,i}\zeta_{i}-\bar{W}_{p,i}\xi^{\top}_{i}\gamma_{\dagger,p}\right)\right|^{2+\delta}\right]}{\left(\mathrm{Var}\left[\left(nh\right)^{-1/2}\sum_{i}\left(\bar{W}^{\dagger}_{p,i}\zeta_{i}-\bar{W}_{p,i}\xi^{\top}_{i}\gamma_{\dagger,p}\right)\right]\right)^{1+\delta/2}}\rightarrow0
\]
holds. The conclusion follows from Lyapunov's central limit theorem
and Slutsky's lemma. 
\end{proof}

\end{document}